\documentclass[11pt]{article}
\usepackage[margin=1in]{geometry}
\usepackage[T1]{fontenc}
\usepackage{lmodern,microtype}
\usepackage{amsmath,amssymb,amsthm}
\usepackage{placeins}
\usepackage{tikz}
\usetikzlibrary{arrows.meta,positioning,calc}
\usepackage[colorlinks=true,linkcolor=blue!50!black,citecolor=blue!50!black,urlcolor=blue!50!black]{hyperref}
\usepackage{adjustbox,bookmark}
\hypersetup{pdftitle={Towards the spin-glass transition in finite dimensions via blue percolation},pdfauthor={Yan Ru Pei},pdfkeywords={spin glass, CMR representation, dependent percolation, overlap susceptibility}}
\newtheorem{theorem}{Theorem}[section]
\newtheorem{proposition}[theorem]{Proposition}
\newtheorem{lemma}[theorem]{Lemma}
\newtheorem{corollary}[theorem]{Corollary}
\newtheorem{hypothesis}[theorem]{Hypothesis}
\theoremstyle{remark}
\newtheorem{remark}[theorem]{Remark}
\newcommand{\E}{\mathbb E}
\newcommand{\Pp}{\mathbb P}
\newcommand{\Z}{\mathbb Z}
\newcommand{\sech}{\operatorname{sech}}
\newcommand{\blue}{\mathrel{\longleftrightarrow}_{\!\mathrm b}}
\newcommand{\pc}{p_c^{\mathrm{site}}}
\newcommand{\thsite}{\theta_{\mathrm{site}}}
\newcommand{\starv}{\operatorname{star}}
\newcounter{conceptalgorithm}

\title{Towards the spin-glass transition in finite dimensions\\via blue percolation}
\author{Yan Ru Pei\\[0.3em]\small\href{mailto:yanrpei@gmail.com}{\texttt{yanrpei@gmail.com}}}
\date{September 29, 2026}
\begin{document}
\maketitle

\begin{abstract}
We prove blue percolation in the Chayes--Machta--Redner representation of the
zero-field Edwards--Anderson spin glass in finite dimensions. For fixed
symmetric iid couplings with $0<m=\E|J|<\infty$, the first blue-percolation
inverse temperature $\beta_{\mathrm b}(d)$ satisfies $2dm\,\beta_{\mathrm b}(d)\to1$
as $d\to\infty$, in periodic joint limits. For symmetric $\pm1$ couplings,
both overlap signs percolate at explicit temperatures in every dimension from
$7$ to $12$. That proof reveals the overlap field, shows that it dominates a
weakly coupled Ising field, and bounds a weighted second moment of open paths;
in dimensions $7$ to $9$ the paths make short lateral excursions, and an
interaction-matrix criterion controls the second moment. These proofs are
computer-assisted: finitely many inequalities between rigorously bounded
quantities are decided in exact rational arithmetic. For exponential-moment disorder, we
also establish a percolating regime with exactly equal infinite-blue-sector
densities and finite spin-glass susceptibility. In these periodic limits,
persistent density imbalance would imply distinct spin-flip-related Gibbs states
and ordinary spin-glass order: positive spatial-overlap variance, infinite
spin-glass susceptibility and a cusp of the replica-coupling pressure, without
requiring finite-blue cancellation. Proving imbalance remains open.
\end{abstract}

\smallskip


\section{Introduction and results}\label{sec:intro}

In the Edwards--Anderson spin glass, each site of $\Z^d$ carries a spin
taking values $+1$ or $-1$, and only nearest neighbors interact. A positive
coupling favors equal neighboring spins; a negative coupling favors opposite
spins. These preferences can conflict around a loop, so a configuration need
not be able to satisfy them all. This competition, called frustration, makes
the average spin, or magnetization, an incomplete description of possible order.
Whether the model, with no external magnetic field, has an ordering transition
at positive temperature in a fixed finite dimension remains a major open problem \cite{IMT2024}. One
precise target is to prove that the same couplings support distinct
infinite-volume equilibrium states, or Gibbs states.

Two equilibrium samples of the same disorder provide a way to look for order
without choosing a preferred spin pattern in advance. Fix the couplings and
sample two configurations $\sigma$ and $\tau$ independently at the same
temperature; these are the two \emph{replicas}. Their local overlap
$q_x=\sigma_x\tau_x$ is $+1$ where they agree and $-1$ where they disagree.
The two sets of sites with these signs are the \emph{overlap sectors}; neither
set is assumed connected. The spatial overlap is the average of $q_x$ over
large boxes, measuring the difference between the fractions of agreeing and
disagreeing sites. Reversing all spins in one replica exchanges the sectors,
while leaving its zero-field energy unchanged.

The Chayes--Machta--Redner (CMR) representation turns this agreement pattern
into a random network \cite{CMR1998,MNS2008}. An edge is \emph{satisfied} in
a replica when its endpoint spins have the relative orientation favored by
its coupling. On edges satisfied in both replicas, auxiliary random choices
activate \emph{blue~bonds}, with probabilities depending on temperature and
coupling strength. A blue edge joins sites with the same overlap sign:
both replicas realize the same preferred relative orientation along that edge.
Every connected component of blue bonds, called a \emph{blue cluster},
therefore lies in one overlap sector. Blue percolation means the existence
of an infinite such component.

This construction separates two geometric questions. Do infinite blue
clusters exist, and do they occupy more of one overlap sector than the other?
For each sign, its \emph{infinite-blue-sector density} is the fraction of all
lattice sites belonging to infinite blue clusters of that sign. Both signs
can percolate while these two densities are equal. Their difference also need
not equal the spatial overlap, because vertices in finite blue clusters
contribute to that overlap as well. Thus forming an infinite network and
breaking the symmetry between agreement and disagreement are distinct steps.

In the fully connected Sherrington--Kirkpatrick (SK) model, where every pair
of spins interacts, the two largest blue-cluster densities can be related to
the fractions of agreeing and disagreeing sites \cite{MNS2007,MNS2008}.
Numerical evidence supports a related picture for the short-range model
\cite{MNS2007,MW2023,MW2026}, but the lattice introduces a geometric obstacle.
On a complete graph, the number of sites in a sector determines all its
available edges. On $\Z^d$, the placement of those sites matters: even a
positive-density set can have only finite connected components. A percolation
proof must therefore control the local arrangement of the replicas, as well
as their overall agreement.

We establish the first of these geometric steps, existence, in finite
dimensions. Apart from the first version of this paper, which reached
dimension $22$ \cite{PeiV1}, we are not aware of an earlier proof of blue
percolation on a finite-dimensional lattice.\footnote{An earlier preprint by the author and
Di~Ventra \cite{PeiDiVentra2021} claimed, by a contour argument, an infinite
blue cluster on $\Z^d$, $d\ge2$, at low temperature. Its fourth arXiv version
(2022) records an error in the contour estimate of its Appendix~D.2 and states
that the result does not immediately follow.}
Write $\beta$ for inverse temperature, so cooling increases $\beta$, and
let $m=\E|J|$ be the mean coupling magnitude. For independent couplings
with any fixed common distribution symmetric about zero and $0<m<\infty$,
we prove that the first
blue-percolation onset is asymptotic to $\beta=(2dm)^{-1}$ as $d$ grows.
This is the lattice counterpart of the complete-graph threshold: in the SK
model with couplings of size $N^{-1/2}$, blue percolation sets in at
$\beta=N^{-1/2}$, with one giant blue cluster in each overlap sector
\cite[Section~III]{MNS2009}. In both cases the number of neighbors times the
mean coupling magnitude times $\beta$ is asymptotically one at the threshold.
At any fixed factor above this inverse-temperature scale, both overlap signs
percolate once $d$ is sufficiently large; at fixed factors below it they do not
(Theorem~\ref{thm:headline}). For equally likely unit couplings $J=\pm1$, we
prove that both overlap signs percolate at explicit temperatures in every
dimension from $7$ to $12$ (Theorem~\ref{thm:explicit}). These proofs are
computer-assisted: finitely many inequalities between rigorously bounded
quantities are decided in exact rational arithmetic.
We work throughout with limits of larger and larger periodic boxes, sampling
the couplings, replicas and blue bonds together; Section~\ref{sec:model}
specifies these limiting laws and the quantifiers of both theorems.

The resulting percolation can coexist with controlled overlap fluctuations.
The spin-glass susceptibility sums disorder-averaged squared spin correlations.
On a periodic box it is also the number of sites times the disorder-averaged
second moment of the spatial overlap. Squaring prevents cancellation merely
from changes of sign.
If the coupling law has a finite exponential moment, meaning
$\E e^{a|J|}<\infty$ for some $a>0$, we prove a regime in sufficiently high
dimensions where both signs percolate, their infinite-blue-sector densities
are exactly equal, the spatial overlap vanishes and this susceptibility stays
finite. With $s_2=\E J^2$ and any fixed $\varepsilon>0$, this regime contains
the whole interval $(1+\varepsilon)/(2dm)\le\beta\le(6ds_2)^{-1/2}$ once $d$
is large, and there the susceptibility is at most $2$
(Corollaries~\ref{cor:separation}, \ref{cor:exact-balance}
and~\ref{cor:window}). The susceptibility stays bounded, uniformly in the
dimension, up to $\beta\approx0.741\,(2ds_2)^{-1/2}$: in every dimension for
$\pm1$ and Gaussian couplings, and in all sufficiently large dimensions under an
exponential moment (Proposition~\ref{prop:sg-uniform}). Along temperatures $\beta_d$ with
$dm\beta_d\to\infty$ and $d\beta_d^2\to0$, the probability that a site belongs
to an infinite blue cluster tends to one (Corollary~\ref{cor:balanced}). These
conclusions show how extensive blue connectivity can precede an
overlap-ordering transition (Section~\ref{sec:response} and
Remark~\ref{rem:onset-ordering}). For unit couplings the same picture holds at
every explicit temperature of Theorem~\ref{thm:explicit}: in each dimension
from $7$ to $12$, both signs percolate with exactly equal densities and
$\chi_{\rm SG}<7/2$ (Corollary~\ref{cor:explicit-balance}).

The remaining geometric target is a persistent difference between the two
infinite-blue-sector densities. Such an imbalance in a periodic limit would
force distinct Gibbs states related by a global spin reversal
\cite[Section II C]{MNS2007}: for almost every disorder, two distinct
single-replica Gibbs states exchanged by the global flip. This is a statement
about spin-flip symmetry breaking, not about the number of states.
Section~\ref{sec:discussion} proves this implication without assuming that
finite blue clusters cancel in the spatial overlap
(Proposition~\ref{prop:imbalance}). It also shows that, in the selected
periodic ensemble, imbalance forces ordinary spin-glass order: the spatial
overlap has positive second moment, the spin-glass susceptibility diverges,
and the quenched two-replica pressure with the added exponent $h\sum_xq_x$
has a cusp at $h=0$ (Corollary~\ref{cor:imbalance-overlap}). The route passes
through an odd local order channel (Proposition~\ref{prop:odd-order}) and a
local comparison of odd channels with the site overlap
(Proposition~\ref{prop:odd-to-site}). The imbalance estimate itself remains open
(Hypothesis~\ref{hyp:imbalance}). The purpose of the blue network is
therefore twofold: it gives a percolation problem that can be solved at high
dimension, and it identifies a further geometric statement that would imply
spin-glass state coexistence.

\subsection{Model, thermodynamic limits and main results}\label{sec:model}

Fix a symmetric probability law on $\mathbb R$, used unchanged as the dimension
varies. Write its iid couplings as $J_e=\epsilon_e W_e$, where the signs
$\epsilon_e$ are independent and fair and are independent of the iid magnitudes
$W_e\ge0$. An auxiliary sign at $W_e=0$ has no effect. Our standing moment assumption is
\begin{equation}\label{eq:moment}
 0<m:=\E W<\infty.
\end{equation}
This includes Gaussian and bounded symmetric laws, as well as laws of infinite
variance. Symmetry is essential to our comparison: its restoration estimate
averages a fair coupling sign to cancel the term linear in $\beta$
(Lemma~\ref{lem:insertion}), and a mean-zero assumption alone does not provide
this cancellation.

On a finite simple graph $G=(V,E)$, let
\[
 \mu_{G,J,\beta}(\sigma)
 =Z_{G,J,\beta}^{-1}
   \exp\!\left(\beta\sum_{e=xy\in E}J_e\sigma_x\sigma_y\right),
 \qquad \beta\ge0.
\]
Here $\sigma\in\{-1,1\}^V$, and the partition function $Z_{G,J,\beta}$
normalizes the probabilities. After sampling $J$, sample $\sigma$ and $\tau$
independently from this zero-field Gibbs measure. Given these variables, sample independent blue indicators $B_e$,
with
\begin{equation}\label{eq:blue}
 \Pp(B_e=1\mid J,\sigma,\tau)
 =(1-e^{-4\beta|J_e|})
   \mathbf1\{J_e\sigma_x\sigma_y>0,\ J_e\tau_x\tau_y>0\}.
\end{equation}
This is the blue marginal of the CMR representation \cite{CMR1998,MNS2008};
only these blue indicators enter the exploration. Denote the joint disorder,
spin and blue law by $\nu_{G,\beta}$. The disorder is \emph{quenched}: it is
held fixed while each Gibbs measure is normalized and the spins are sampled.
Thus the factor $Z_{G,J,\beta}^{-2}$ is retained before disorder averaging.
We use \emph{annealed} for this averaged joint law,
not for a Gibbs measure obtained by averaging its unnormalized weights.

Let $\mathbb T_L^d=(\mathbb Z/L\mathbb Z)^d$, $L\ge3$, have its nearest-neighbor
edges. Write $\mathcal L_{d,\beta}$ for all joint local accumulation points of
$\nu_{\mathbb T_L^d,\beta}$ as $L\to\infty$. This family is nonempty: the coupling
marginals are the fixed iid law, and the local spin and bond spaces are finite,
so tightness and diagonal extraction apply. Write $C_{\mathrm b}(x)$ for the blue
cluster of $x$, $q_x=\sigma_x\tau_x$ for the overlap sign, and $o=0$ for the
origin. The notation $x\blue y$ means that a blue path joins $x$ to $y$,
and $x\blue\infty$ means $|C_{\mathrm b}(x)|=\infty$.
Write $B_n$ for the graph-distance ball of radius $n$ about the origin and
$\partial B_n$ for the vertices at graph distance $n$. Connectivity to $\partial B_n$
is understood using paths inside $B_n$.

\ifdefined\Needspace\Needspace{10\baselineskip}\fi
The dimension thresholds below may depend on the coupling law and on $c$.
\begin{theorem}\label{thm:headline}
Under \eqref{eq:moment}, the following statements hold.
\begin{enumerate}
\item For each fixed $c>1$, all sufficiently large $d$ have the property that
every $\nu\in\mathcal L_{d,c/(2dm)}$ satisfies
\[
 \nu(0\blue\infty,\ q_0=s)>0\quad(s=\pm1).
\]
Moreover, each of the two overlap signs contains an infinite blue component
$\nu$-almost surely.
\item For each fixed $0<c<1$, there exist $d_0<\infty$ and $\rho<1$ such that,
for $d\ge d_0$, $0\le\beta\le c/(2dm)$, every $\nu\in\mathcal L_{d,\beta}$,
and every $n\ge1$,
\[
 \nu(0\blue\partial B_n)\le2\rho^n.
\]
In particular, these laws have no infinite blue component almost surely.
\end{enumerate}
\end{theorem}

Define the first blue-percolation onset by
\begin{equation}\label{eq:onsetdef}
 \beta_{\mathrm b}(d)
 :=\inf\{\beta\ge0:\text{some }\nu\in\mathcal L_{d,\beta}
                     \text{ has }\nu(0\blue\infty)>0\},
 \qquad \inf\varnothing=+\infty.
\end{equation}
The theorem immediately gives
\begin{equation}\label{eq:onset}
 2dm\,\beta_{\mathrm b}(d)\longrightarrow1.
\end{equation}
Indeed, for each fixed $0<\varepsilon<1$, the interval below
$(1-\varepsilon)/(2dm)$ is excluded and $(1+\varepsilon)/(2dm)$ is included for
large $d$. No monotonicity of the blue law in $\beta$ is needed for this squeeze.

Theorem~\ref{thm:headline} does not name a dimension: its thresholds come from
asymptotic estimates. For unit couplings we obtain explicit dimensions and
temperatures, parametrized by $t=\tanh\beta$.

\begin{theorem}[Explicit dimensions for unit couplings]\label{thm:explicit}
Let the couplings be iid and uniform on $\{-1,1\}$, let $(d,t)$ be one of the
seventeen pairs listed in Table~\ref{tab:explicit}, and put
$\beta=\operatorname{atanh}t$. Let $\theta_*=\theta_*(d,t)>0$ be the explicit
constant of \eqref{eq:ov-theta} or \eqref{eq:ms-theta}, bounded from below in
Table~\ref{tab:explicit}; where a pair has two rows, either constant may be
used. Then every $\nu\in\mathcal L_{d,\beta}$ satisfies
\[
 \nu(0\blue\infty,\ q_0=s)\ \ge\ \theta_*\qquad(s=\pm1),
\]
and $\nu$-almost surely each of the two overlap signs contains an infinite blue
component.
\end{theorem}

\begin{table}[!t]
\centering
\small
\begin{tabular}{cclll}
\hline
$d$ & $t=\tanh\beta$ & $\theta_*\ge$ & global step & local inputs\\ \hline
$12$ & $3/25$   & $0.0362$ & oriented & Sections~\ref{sec:ov-floor}--\ref{sec:ov-holley}\\
$12$ & $11/100$ & $0.0312$ & oriented & Sections~\ref{sec:ov-floor}--\ref{sec:ov-holley}\\
$12$ & $23/200$ & $0.0354$ & oriented & Sections~\ref{sec:ov-floor}--\ref{sec:ov-holley}\\
$12$ & $1/8$    & $0.0349$ & oriented & Sections~\ref{sec:ov-floor}--\ref{sec:ov-holley}\\
$12$ & $13/100$ & $0.0306$ & oriented & Sections~\ref{sec:ov-floor}--\ref{sec:ov-holley}\\
$11$ & $13/100$ & $0.0173$ & oriented & Sections~\ref{sec:ov-floor}--\ref{sec:ov-holley}\\
$11$ & $3/25$   & $0.0158$ & oriented & Sections~\ref{sec:ov-floor}--\ref{sec:ov-holley}\\
$10$ & $27/200$ & $0.00195$ & oriented, shared edges & Sections~\ref{sec:ov-floor}--\ref{sec:ov-holley}\\
$10$ & $13/100$ & $0.00104$ & oriented, shared edges & Sections~\ref{sec:ov-floor}--\ref{sec:ov-holley}\\ \hline
$9$ & $7/50$   & $0.000950$ & oriented, covariance   & noisy cavities, Section~\ref{sec:ov-noisy}\\
$9$ & $29/200$ & $0.00200$  & oriented, shared edges & noisy cavities, Section~\ref{sec:ov-noisy}\\
$9$ & $3/20$   & $0.00317$  & oriented, shared edges & noisy cavities, Section~\ref{sec:ov-noisy}\\
$9$ & $31/200$ & $0.00280$  & oriented, shared edges & noisy cavities, Section~\ref{sec:ov-noisy}\\
$9$ & $4/25$   & $0.000899$ & oriented, shared edges & noisy cavities, Section~\ref{sec:ov-noisy}\\
$9$ & $7/50$   & $0.00190$  & macrostep & Sections~\ref{sec:ov-floor}--\ref{sec:ov-holley}\\ \hline
$8$ & $3/20$   & $0.000660$ & macrostep & Sections~\ref{sec:ov-floor}--\ref{sec:ov-holley}\\
$8$ & $3/20$   & $0.00193$  & macrostep & noisy cavities, Section~\ref{sec:ov-noisy}\\ \hline
$7$ & $31/200$ & $0.000326$ & macrostep & noisy cavities, Section~\ref{sec:ov-noisy}\\
$7$ & $3/20$   & $0.000227$ & macrostep & noisy cavities, Section~\ref{sec:ov-noisy}\\ \hline
\end{tabular}
\caption{The seventeen pairs of Theorem~\ref{thm:explicit}, with lower bounds on
$\theta_*$ (rounded down) and the ingredients of each proof. The oriented global
step is the criterion of Theorem~\ref{thm:ov-criterion}, with the shared-edge
refinement of Lemma~\ref{lem:ov-shared} or the covariance refinement of
Lemma~\ref{lem:ov-covariance} where indicated; the macrostep global step is
Theorem~\ref{thm:ms-engine}. The pairs $(9,7/50)$ and $(8,3/20)$ are proved
twice, with different inputs. The certified constants are in
Tables~\ref{tab:overlap-constants} ($d\ge10$), \ref{tab:ov-noisy} ($d=9$,
oriented) and~\ref{tab:ms-local}--\ref{tab:ms-certificates} (macrostep).}
\label{tab:explicit}
\end{table}

The constants $\theta_*$ and the certified inputs behind them are defined in
Sections~\ref{sec:overlap-route} and~\ref{sec:macrostep}. In dimensions $9$ to
$12$ the global step is a scalar second-moment criterion $\mathcal S<1$ along
oriented paths (Theorem~\ref{thm:ov-criterion}). The certified margin
$1-\mathcal S$ is about $9\%$ at $(d,t)=(12,3/25)$ and just under $5\%$ at
$(11,13/100)$. In dimension $10$ the margin is below $1\%$, and only a refined
criterion, which gains from correlations along edges shared by two paths
(Lemma~\ref{lem:ov-shared}), certifies the two temperatures. In dimension $9$
the criterion certifies five temperatures once both local inputs are sharpened
(Section~\ref{sec:ov-noisy}), with margins between $0.24\%$ and $0.93\%$. In dimensions $7$ to $9$, Section~\ref{sec:macrostep} replaces the
oriented paths by paths with short lateral excursions and bounds their second
moment by a finite interaction-matrix criterion (Theorem~\ref{thm:ms-engine}).
Its certified contraction factors $\lambda$ leave margins $1-\lambda$ of about
$22\%$ at $(9,7/50)$, $10.7\%$ and $14\%$ at $(8,3/20)$ with the two kinds of
local inputs, and $7.5\%$ and $6.4\%$ at $t=31/200$ and $t=3/20$ in dimension
$7$, which is certified only with the sharpened inputs
(Remark~\ref{rem:ms-inputs}). These are exact inequalities, so thin
margins do not weaken the proofs. Each conclusion concerns the stated dimension
and temperature: no monotonicity in $\beta$ is available to fill in intervals,
and nothing is claimed in dimension $6$ or below (Remark~\ref{rem:ms-limits}).
At the listed points $2d\beta$ lies between $2.1$ and $3.2$, about two to three
times the asymptotic onset scale of Theorem~\ref{thm:headline}. They are points of
percolation, not estimates of the onset $\beta_{\mathrm b}(d)$, and we make no
claim about their position relative to a spin-glass transition.

The proofs of Theorem~\ref{thm:explicit} are computer-assisted, at two scales.
The local inputs are explicit finite sums of rationals, decided exactly. At
$d=12$, for instance, each temperature needs $24$ double sums of at most $156$
terms for the bond floor and $325$ numerator sums for the Ising comparison; in
dimension $9$ the sharpened inputs need, per temperature, $1377$ chord
polynomials of degree $5$ and $18$ class sums for the bond floor and $1330$
two-block classes for the Ising comparison. The oriented criterion involves finitely many constants, replaced by
rigorous rational enclosures. The macrostep criterion is a larger computation.
Its certificate is a vector on $46$ to $804$ symmetry classes of relative
positions of two paths, and for each class a per-step weight is summed over all
pairs of macrosteps, between $425{,}104$ and $956{,}484$ of them, about
$3.4\times10^8$ terms for the largest certificate. These sums are rigorous upper
bounds computed in floating-point arithmetic with directed rounding or an a
priori error analysis, and the final inequalities are decided in exact rational
arithmetic. Two independently written programs compute every macrostep
certificate. Sections~\ref{sec:ov-constants} and~\ref{sec:ms-certificates} describe the
computations, and the section on computational reproducibility lists the
programs.
Given the couplings, the spin and blue laws depend on $\beta$ and $J$ only
through the products $\beta J_e$, so the same conclusions hold for couplings
$\pm J_0$ at $\beta J_0=\operatorname{atanh}t$. Since the $\pm1$ law has no
atom at zero, the at-most-two theorem of the author's earlier paper
\cite[Theorem~1.9(b)]{Pei2026}, applied as in Section~\ref{sec:discussion},
shows that at these points every selected limit has almost surely exactly two
infinite blue components, one in each overlap sector. The same holds under
Theorem~\ref{thm:headline}(i) whenever the coupling law has no atom at zero.
The proof of Theorem~\ref{thm:explicit} does not use this upper bound.

In a selected limit, the conditional law of the two replicas given the
couplings need not be a product of two copies of one Gibbs state. The
probability-one statements in both theorems hold conditionally for almost
every disorder sample under each selected joint limit's own spin/bond kernel,
and the construction does not require the two replicas to remain independent
after conditioning only on the disorder in that limit. Our conclusions concern these
periodic joint limits; they do not quantify over arbitrary selections of
quenched Gibbs states. Both theorems establish coexistence of opposite-sign
infinite blue components, without determining their densities or proving
spin-glass order.

\subsection{Proof roadmap: local normalization and global geometry}\label{sec:roadmap}

The central difficulty is dependence. Blue bonds are independent once the
couplings and both spin configurations are fixed, but their averaged law is
not independent. Revealing one failed opening changes the information
available about later openings. The proof needs a lower bound that remains
valid after an exploration has recorded its successes and failures, rather
than just an estimate of a single edge's unconditioned probability.

For the first comparison, assume a finite exponential moment.
Start with a vertex already connected to the root by blue bonds and consider
an edge to a fresh target $v$. The \emph{star} of $v$ is the set of all edges
incident to $v$. Temporarily delete all interactions in this star except the
candidate edge, making $v$ a \emph{leaf}, a vertex with just one incident edge.
Its conditional blue-opening probability is then exactly
$\E\tanh(\beta|J|)$. Restore the deleted interactions one at a time and bound
how much this probability can change. Averaging a fair coupling sign cancels
the term linear in $\beta$, leaving a cost of order $\beta^2$ per restored
edge. Both fixed-disorder Gibbs normalizations must be retained in this
calculation. For the single-replica Fortuin--Kasteleyn (FK) bond
representation, uniform conditional bond bounds obtained by averaging fair
coupling signs are known \cite[Theorem~3]{GKN1992},
\cite[Lemma~3.1]{DeSantisGandolfi1999}.

To turn the conditional bound into an independent comparison, test each
target only once. When a physical blue opening is found, accept its endpoint
with an extra independent coin chosen to make the total conditional acceptance
probability a fixed number $p$. This extra rejection step is called
\emph{thinning}. Accepted vertices remain connected to the root by actual blue
bonds, while their acceptance labels have the law of independent occupied
sites \cite[proof of Lemma 5]{GrimmettStacey1998}.

\begin{remark}[A direct complete percolation proof]\label{rem:short-route}
For a reader interested first in blue percolation in some finite dimension,
Section~\ref{sec:cmr-verification} already proves Theorem~\ref{thm:headline}
for exponential-moment disorder, including unit symmetric and Gaussian couplings.
The leaf probability is $\E\tanh(\beta W)\sim m\beta$; each restored edge
costs a factor $1+O(\beta^2)$. At $\beta=c/(2dm)$ the total cost tends to one,
so the independent site parameter satisfies $2dp\to c$. Write $\pc(\Z^d)$
for the critical occupation probability of independent site percolation, and
$\thsite(p)$ for the probability that the origin lies in an infinite occupied
cluster when the sites of $\Z^d$ are occupied independently with probability
$p$; we write $\theta_d^{\rm site}(p)$ when the dimension varies.
Kesten's threshold
$2d\pc(\Z^d)\to1$ proves percolation for $c>1$; the matching upper estimate
and path counting exclude it for $c<1$. Local transfer and exterior conditioning
give the periodic-limit statements. The overlap-revealed exploration of
Section~\ref{sec:overlap-route} and the star estimates of
Appendix~\ref{app:fresh-star} are needed only for explicit dimensions, not for
existence in some finite dimension.
This is a geometric transition; Gibbs-state nonuniqueness remains conditional
on the imbalance estimate in Section~\ref{sec:discussion}.
\end{remark}

Appendix~\ref{sec:finite-mean} extends this argument beyond exponential-moment
laws. The exploration avoids targets incident to exceptionally large couplings
and bounds the loss using auxiliary independent marks. The physical
interactions themselves remain present. Appendix~\ref{app:probability}
isolates the general comparison and the passage from finite graphs to periodic
limits. An alternative proof for exponential-moment disorder works first
within the two overlap sectors, using the normalized bond comparison of
\cite[Lemmas 4.1--4.2]{MNS2008}. Appendix~\ref{sec:mns} supplies the further
conditional control of the overlap sites needed to use that comparison on the
lattice.

The explicit dimensions of Theorem~\ref{thm:explicit} require sharper inputs at
both levels. The comparison above reveals both replica configurations outside a
target's star, so its opening bound must hold for the least favorable exterior
configuration, and Kesten's threshold is only asymptotic.
Section~\ref{sec:overlap-route} instead reveals the whole overlap field $q$ at
the outset, including the overlap of each target, and queries blue bonds only
between vertices of overlap $+1$. This refines the overlap-sector comparison of
Appendix~\ref{sec:mns}, which conditions on both replicas and uses
\cite[Lemmas~4.1--4.2]{MNS2008}. In the gauge variables
$s_e=J_e\sigma_x\sigma_y$ \cite[Section~4.2]{Nishimori2001}, the first replica
is uniform and independent of the overlap and blue fields, so revealing it would
change no conditional law, and the two quenched normalizations combine into one
factor $Z(s)^{-2}$ (Lemma~\ref{lem:ov-gauge}). Given $q$ and all gauge signs and
activations off a target's star, the gauge signs of the star have an explicit
law: a product law reweighted by the part of $Z(s)^{-2}$ that depends on them
(Proposition~\ref{prop:ov-star-law}).

In that product law, the signs toward neighbors of opposite overlap are fair,
and those toward agreeing neighbors are tilted toward satisfaction. A failed
query turns a tilted sign into a fair one. Indeed, the activation probability
$p_A=1-e^{-4\beta}$ of \eqref{eq:blue} satisfies $e^{2\beta}(1-p_A)=e^{-2\beta}$,
so a failed agreeing edge carries the weight $e^{-2\beta}$ whatever its sign
\eqref{eq:ov-failed-fair}; this is the vacant term of the CMR edge weight
\cite[Eq.~(6)]{CMR1998}. A mediant inequality shows that fair signs cannot
lower the worst case over \emph{cavity laws}, the possible laws of the
neighboring spins once the target is removed. One bond floor, about
$0.9\tanh2\beta$, therefore survives any number of failed queries
(Theorem~\ref{thm:ov-extremal}); for independent fair neighboring spins the
opening probability is exactly $\tanh2\beta$. A quantile coupling, with no thinning, then
places the cluster of the origin for independent Bernoulli bonds on the edges
inside $\{q=+1\}$ within its blue cluster.

The price is that the plus sites of $q$ form a dependent random set. Their
single-site conditional odds admit an explicit lower bound, again uniform over
cavity laws, whose logarithm is linear in the sum of the neighboring overlaps;
we call it the \emph{Holley line}. A one-sided form of Holley's theorem
\cite{Holley1974} then shows that $q$ stochastically dominates an Ising field
with a small ferromagnetic coupling and a negative field. A classical mean-field bound on the
magnetization \cite{Pearce1981}, \cite[Theorem~3.53]{FriedliVelenik2017}
keeps its plus-density above $0.40$ at the listed temperatures.

Both local bounds must hold uniformly over cavity laws, and their certificates
are decided at frozen laws, under which the neighboring spins are deterministic
up to a global flip. Actual cavity laws are never frozen: given the spins at distance
two, the neighbors of the target are independent, and each has conditional mean
at most $\tanh((2d-1)\beta)$ in absolute value. As for an Ising edge on a tree
\cite[Section~3]{Mossel2004}, such a neighbor acts on the target through a
noisy binary channel, and this noise can be moved into the star interaction.
The cavity laws are then replaced by a smaller class at a smaller effective
coupling, at the price of a constant factor that cancels in every ratio used
(Lemma~\ref{lem:ov-noisy}). Over the smaller class, an exact certificate for
the floor, which pairs opposite sign patterns of the star, and a symmetrized
certificate for the Holley line give sharper local inputs
(Section~\ref{sec:ov-noisy}). At every point where they are used, the resulting
floors exceed every floor that is uniform over all cavity laws
(Remark~\ref{rem:ov-nesting}).

In the resulting comparison model, the plus sites of this Ising field are
joined by independent Bernoulli bonds. A weighted form of Kesten's second
moment for oriented paths \cite[Section~2]{CoxDurrett1983} bounds the
probability that this model has an open path from the origin to distance $n$
from below, uniformly in $n$ and in the torus size. Normalizing each path by its
own probability, as in \cite[Eq.~(5.16)]{LyonsPeres2016}, makes the pair law of
two paths uniform. The Ising field then
enters only through a pair ratio, which the FKG inequality \cite{FKG1971} and
Dobrushin's comparison theorem \cite{Dobrushin1970,Follmer1982} bound in terms
of the vertices shared by the two paths and an explicit correlation term
between the unshared ones. Along the difference of two independent oriented
walks this becomes a renewal sum, which a Schur test \cite{Schur1911} and
Khas'minskii's lemma \cite{Khasminskii1959} bound by an explicit scalar
criterion (Theorem~\ref{thm:ov-criterion} and
Proposition~\ref{prop:overlap-central}). Local transfer passes the resulting
finite-volume bound to every selected limit, and an adaptation of the exterior
conditioning of Proposition~\ref{prop:as-existence}, with Holley domination
inside finite boxes, gives almost-sure coexistence
(Section~\ref{sec:ov-assembly}). This gives dimensions $10$ to $12$ with the
first local inputs, and dimension $9$ with the sharpened ones. At each listed
temperature, the local inputs are verified by exact rational evaluation of
explicit finite sums, and the global criterion with rigorous rational
enclosures of finitely many constants; two independently written programs
certify every temperature (Section~\ref{sec:ov-constants}).

Below dimension $9$ the oriented criterion is not certified. Two oriented paths
meet only at equal times, and after a meeting they share their next step with
probability $1/d$. Each shared step costs $(\rho_ep)^{-1}$, where $p$ is the
bond floor and $\rho_e$ a lower bound on the Ising plus-density along a shared
edge, so the criterion needs $d\rho_ep>1$. With the refinement $\rho_e=\rho_c$
of Lemma~\ref{lem:ov-shared}, this fails at the local inputs certified in
dimension $8$ (Remark~\ref{rem:ov-dim9}). Section~\ref{sec:macrostep}
therefore changes the path family. A macrostep makes a short excursion in three lateral coordinates,
along a word in which no coordinate is used with both signs, and then one step
in one of the remaining $d-3$ coordinates. Such paths go back to Kesten's block
paths \cite[Section~2]{Kesten1990} and to the macrosteps of Jiang and Lang
\cite[Section~4]{JiangLang2026}; the monotone words keep them self-avoiding. Two
macrostep paths still meet only at equal levels, and they share a forward step
only if their excursions end at the same point and their forward directions
agree. Their relative
position is a Markov chain, and the Ising correlations between the two paths
become an additive functional of this chain (Lemma~\ref{lem:ms-boost}). The
resulting product is bounded by an interaction-matrix expansion, as in
\cite[Theorem~5]{JiangLang2026}. Since the Ising correlations never vanish, its
kernel is not finitely supported, and a supersolution outside a finite
\emph{near set} of relative positions of the two paths, together with
Khas'minskii's lemma for the forward walk, controls the rest
(Theorem~\ref{thm:ms-criterion}). A finite certificate for this criterion gives
the uniform torus bound (Theorem~\ref{thm:ms-engine}), and the local inputs and
the assembly of Section~\ref{sec:overlap-route} give dimensions $7$, $8$
and~$9$.

Appendix~\ref{app:fresh-star} gives a second route, for unit couplings. It
keeps the one-attempt exploration and replaces the leaf floor by the joint star
estimate of \cite{PeiV1} (Proposition~\ref{prop:unit-star}), which sums over
the target's central spins with both quenched normalizations retained. Kesten's
asymptotic threshold is replaced by the explicit oriented bound
$\pc(\Z^d)\le F_d$, where $F_d$ is the probability that two independent
oriented random walks meet again. A single floor above $F_d$ then suffices
(Theorem~\ref{thm:single-floor}). This gives dimensions $16$ to $20$ and
reproves dimension $22$, with the floor exceeding $F_{22}$ by a factor above
$1.2$ (Corollary~\ref{cor:single-floor}); dimension $21$ is not certified here.
A floor free of finite sums gives dimensions $25$ and $26$
(Corollary~\ref{cor:fs-closed-form}). Remark~\ref{rem:v1-certificate} compares
the longer dimension-$22$ argument of \cite{PeiV1}.

The response argument in Section~\ref{sec:response} uses squared spin
correlations in place of blue-opening events. Integrating out one spin and
averaging the fair signs of its whole star gives a mean-square cavity
inequality \cite{FZ1987}. Summed over endpoints before it is expanded, the
inequality closes a recursion for row sums of squared correlations and bounds
the susceptibility on the mean-field temperature scale, uniformly in large
dimensions. Finite susceptibility forces the
spatial overlap to vanish, and a local coupling-interpolation argument
(Section~\ref{sec:discussion}) shows that vanishing overlap already forces
exact balance of the infinite-blue-sector densities. Combining these estimates
with percolation gives the balanced regime described above, including every
explicit point of Theorem~\ref{thm:explicit}, and shows that any
positive-temperature susceptibility transition must lie below the first
blue-percolation temperature (Remark~\ref{rem:onset-ordering}). The estimates are first proved in finite
volume and passed through observables on fixed finite sets to the periodic
limit. Only then are infinite connectivity and density events used; this
order does not require limiting replicas to remain independent conditional
only on the disorder.

Figure~\ref{fig:architecture} summarizes these dependencies. The blue
exploration and susceptibility calculation meet only when comparing their
conclusions about temperature.

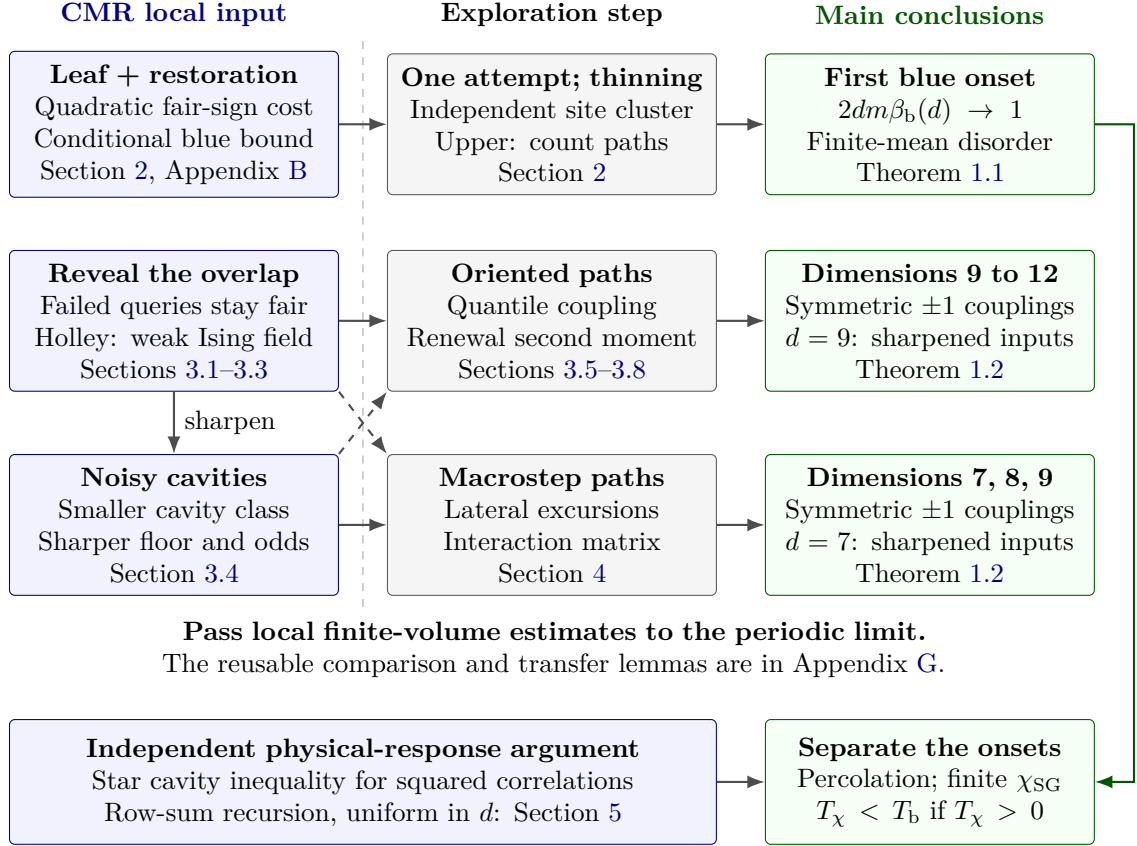
\begin{figure}[!htbp]
\centering
\begin{adjustbox}{max width=\linewidth}
\begin{tikzpicture}[>=Latex,
 box/.style={draw,rounded corners=2pt,align=center,font=\small,
   text width=4.0cm,minimum height=1.85cm,inner sep=5pt},
 specific/.style={box,fill=blue!5,draw=blue!55!black},
 general/.style={box,fill=gray!8,draw=black!65},
 result/.style={box,fill=green!4,draw=green!35!black},
 flow/.style={->,thick,draw=black!70}]
\node[font=\small\bfseries,text=blue!55!black] at (0,1.45)
 {CMR local input};
\node[font=\small\bfseries] at (5.0,1.45)
 {Exploration step};
\node[font=\small\bfseries,text=green!35!black] at (10.0,1.45)
 {Main conclusions};
\draw[dashed,black!35] (2.5,1.1)--(2.5,-6.4);
\node[specific] (local) at (0,0)
 {\textbf{Leaf + restoration}\\Quadratic fair-sign cost\\
 Conditional blue bound\\
 Section~\ref{sec:cmr-verification}, Appendix~\ref{sec:finite-mean}};
\node[general] (asymptotic) at (5.0,0)
 {\textbf{One attempt; thinning}\\Independent site cluster\\
 Upper: count paths\\Section~\ref{sec:general}};
\node[result] (onset) at (10.0,0)
 {\textbf{First blue onset}\\$2dm\beta_{\mathrm b}(d)\to1$\\
 Finite-mean disorder\\Theorem~\ref{thm:headline}};
\node[specific] (overlap) at (0,-2.6)
 {\textbf{Reveal the overlap}\\Failed queries stay fair\\
 Holley: weak Ising field\\
 Sections~\ref{sec:ov-gauge}--\ref{sec:ov-holley}};
\node[general] (moment) at (5.0,-2.6)
 {\textbf{Oriented paths}\\Quantile coupling\\
 Renewal second moment\\
 Sections~\ref{sec:ov-density}--\ref{sec:ov-assembly}};
\node[result] (dimension) at (10.0,-2.6)
 {\textbf{Dimensions 9 to 12}\\Symmetric $\pm1$ couplings\\
 $d=9$: sharpened inputs\\Theorem~\ref{thm:explicit}};
\node[specific] (noisy) at (0,-5.3)
 {\textbf{Noisy cavities}\\Smaller cavity class\\
 Sharper floor and odds\\
 Section~\ref{sec:ov-noisy}};
\node[general] (macro) at (5.0,-5.3)
 {\textbf{Macrostep paths}\\Lateral excursions\\
 Interaction matrix\\
 Section~\ref{sec:macrostep}};
\node[result] (low) at (10.0,-5.3)
 {\textbf{Dimensions 7, 8, 9}\\Symmetric $\pm1$ couplings\\
 $d=7$: sharpened inputs\\Theorem~\ref{thm:explicit}};
\draw[flow] (local)--(asymptotic);
\draw[flow] (asymptotic)--(onset);
\draw[flow] (overlap)--(moment);
\draw[flow] (moment)--(dimension);
\draw[flow] (noisy)--(macro);
\draw[flow] (macro)--(low);
\draw[flow] (overlap.south)--node[right,font=\small]{sharpen}(noisy.north);
\draw[flow,dashed] (overlap.south east)--(macro.north west);
\draw[flow,dashed] (noisy.north east)--(moment.south west);
\node[font=\small,align=center,text width=14cm] at (5.0,-6.95)
 {\textbf{Pass local finite-volume estimates to the periodic limit.}\\
 The reusable comparison and transfer lemmas are in Appendix~\ref{app:probability}.};
\node[specific,text width=9.0cm,minimum height=1.65cm]
 (response) at (2.5,-8.7)
 {\textbf{Independent physical-response argument}\\
 Star cavity inequality for squared correlations\\
 Row-sum recursion, uniform in $d$: Section~\ref{sec:response}};
\node[result,minimum height=1.65cm] (separation) at (10.0,-8.7)
 {\textbf{Separate the onsets}\\
 Percolation; finite $\chi_{\rm SG}$\\
 $T_\chi<T_{\mathrm b}$ if $T_\chi>0$};
\draw[flow] (response)--(separation);
\draw[flow,draw=green!35!black]
 (onset.east)--(12.7,0)--(12.7,-8.7)--(separation.east);
\end{tikzpicture}
\end{adjustbox}
\caption{CMR proof roadmap. The dashed line separates the local CMR estimates
from the exploration comparisons. The first row leads to
Theorem~\ref{thm:headline}. The second and third rows are the explicit-dimension
routes for unit couplings. With the overlap field revealed
(Section~\ref{sec:overlap-route}), failed blue queries cannot lower the bond
floor, and the overlap field dominates a weak Ising field; noisy cavities
sharpen both local inputs. Oriented paths give dimensions $10$ to $12$ with the
first inputs and dimension $9$ with the sharpened ones; macrostep paths give
dimensions $8$ and $9$ with the first inputs and dimensions $7$ and $8$ with the
sharpened ones. The dashed arrows mark these cross combinations. The fresh-star
route of Appendix~\ref{app:fresh-star}, not shown, sharpens the first row by a
joint star estimate and an oriented site comparison.
The outer arrow combines the onset theorem with the independent response bound
for exponential-moment disorder in sufficiently high dimension. Here
$\chi_{\rm SG}$ is the spin-glass susceptibility, $T_{\mathrm b}=1/\beta_{\mathrm b}(d)$
and $T_\chi=1/\beta_\chi(d)$, where $\beta_\chi(d)$ is the first inverse
temperature at which some selected limit has infinite $\chi_{\rm SG}$
(Remark~\ref{rem:onset-ordering}). A
susceptibility transition is not proved to exist. At the explicit points of
Theorem~\ref{thm:explicit}, finite susceptibility and exact balance follow
from Corollary~\ref{cor:explicit-balance}.}
\label{fig:architecture}
\end{figure}

\paragraph{Reading guide.}
The paper offers several levels of detail.
\begin{itemize}
\item \emph{Existence in a finite dimension:} the model, Theorem~\ref{thm:headline},
and Section~\ref{sec:cmr-verification} give the direct complete proof just described.
\item \emph{The sharp onset scale for general disorder:}
Appendix~\ref{sec:finite-mean} removes the exponential-moment assumption and
also determines the subcritical blue susceptibility.
\item \emph{Explicit dimensions for unit couplings:}
Section~\ref{sec:overlap-route} proves Theorem~\ref{thm:explicit} in dimensions
$9$ to $12$ and can be read directly after Section~\ref{sec:cmr-verification};
its Section~\ref{sec:ov-noisy} sharpens the local inputs, and the certified
constants of both engines are collected in Appendix~\ref{app:certificates}. Section~\ref{sec:macrostep} proves dimensions $7$ to $9$ with a
different global step; from Section~\ref{sec:overlap-route} it uses the local
inputs, the lemmas on the Ising field and the difference walk in
Sections~\ref{sec:ov-density}--\ref{sec:ov-second-moment}, and the assembly.
Appendix~\ref{app:fresh-star} gives the fresh-star route to dimensions $16$
to $20$ and $22$ for unit couplings, and a closed-form bound in dimensions $25$
and $26$.
\item \emph{What this means for spin-glass ordering:}
Sections~\ref{sec:response}--\ref{sec:discussion} separate percolation from
ordinary overlap order, prove the conditional Gibbs-state implication, and
state the missing imbalance estimate.
\end{itemize}

\FloatBarrier
\section{A direct CMR percolation proof}\label{sec:cmr-verification}\label{sec:general}

The central task is to lower-bound a blue opening after the exploration history.
We first do this for disorder with an exponential moment. Making the target a
leaf gives an exact probability; restoring its other interactions costs only
quadratically in $\beta$. A classical one-attempt exploration then proves the
high-dimensional theorem. The finite-mean extension uses the same
local-to-global structure, and the explicit-dimension proofs of
Section~\ref{sec:overlap-route} and Appendix~\ref{app:fresh-star} sharpen both
of its levels.

\subsection{The leaf identity and the cost of restoration}

The retained coordinates are the two spins at each vertex and the coupling
and raw activation uniform on each present edge. Deletion retains all vertex
coordinates and forgets only the removed edge's coupling and uniform. The blue
indicators on surviving edges are unchanged functions of the retained coordinates.
The laws on the resulting graphs are their separately normalized, same-disorder
two-replica laws. They are not assumed consistent under graph restriction.

For a vertex $v$, write $\starv(v)$ for its incident edges and set
\[
 \mathcal F_v^{\mathrm{ext}}
 :=\sigma\bigl((\sigma_x,\tau_x)_{x\ne v},(B_e)_{e\notin\starv(v)}\bigr).
\]
\begin{lemma}[Two-replica leaf identity]\label{lem:cmr-leaf}
If $v$ is a leaf with incident edge $e=uv$, then conditional on the exterior
spins and retained exterior variables, and on $W_e$, its blue probability is
$\tanh(\beta W_e)$. The exterior marginal is independent of $J_e$.
If $v$ is isolated, its overlap sign is fair independently of the exterior.
\end{lemma}

\begin{proof}
Conditional on the exterior spins and on $W_e>0$,
each replica satisfies this leaf edge with probability $(1+t_e)/2$, where
$t_e=\tanh(\beta W_e)$. Since its activation probability is
$4t_e/(1+t_e)^2$, the leaf blue probability is exactly $t_e$.
At $W_e=0$ the probability is zero, so the same identity holds.
Summing the leaf spins cancels a factor $2\cosh(\beta W_e)$ in each partition
function; the entire exterior marginal is independent of $J_e$.

When $v$ is isolated, each replica spin at $v$ is independent and uniform,
which proves the last assertion.
\end{proof}

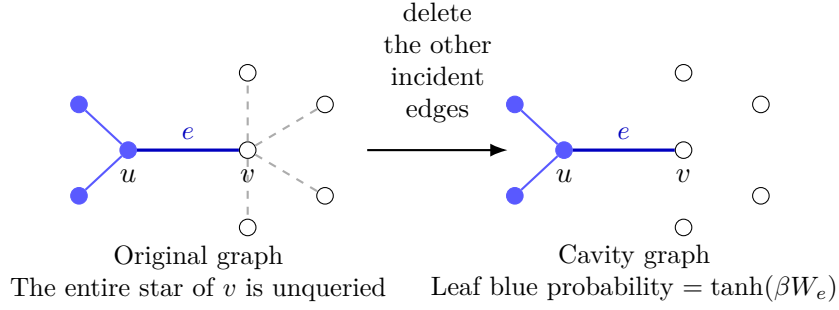
\begin{figure}[tb]
\centering
\begin{adjustbox}{max width=\linewidth}
\begin{tikzpicture}[scale=.93,vertex/.style={circle,draw,fill=white,inner sep=2.2pt},
 reached/.style={vertex,fill=blue!65,draw=blue!65},
 candidate/.style={blue!75!black,very thick},
 untouched/.style={gray!65,dashed,thick},>=Latex]
\begin{scope}
\node[reached,label=below:$u$] (u) at (0,0) {};
\node[vertex,label=below:$v$] (v) at (1.7,0) {};
\node[vertex] (a) at (1.7,1.1) {};
\node[vertex] (b) at (2.8,.65) {};
\node[vertex] (c) at (2.8,-.65) {};
\node[vertex] (d) at (1.7,-1.1) {};
\node[reached] (w) at (-.7,.65) {};
\node[reached] (z) at (-.7,-.65) {};
\draw[blue!65,thick] (w)--(u)--(z);
\draw[candidate] (u)--node[above,font=\small] {$e$}(v);
\foreach \x in {a,b,c,d} {\draw[untouched] (v)--(\x);}
\node[align=center,font=\small] at (1,-1.75)
 {Original graph\\The entire star of $v$ is unqueried};
\end{scope}
\draw[->,thick] (3.4,0)--(5.4,0);
\node[align=center,font=\small,text width=2cm] at (4.4,1.25)
 {delete the other\\incident edges};
\begin{scope}[xshift=6.2cm]
\node[reached,label=below:$u$] (u2) at (0,0) {};
\node[vertex,label=below:$v$] (v2) at (1.7,0) {};
\node[vertex] at (1.7,1.1) {};
\node[vertex] at (2.8,.65) {};
\node[vertex] at (2.8,-.65) {};
\node[vertex] at (1.7,-1.1) {};
\node[reached] (w2) at (-.7,.65) {};
\node[reached] (z2) at (-.7,-.65) {};
\draw[blue!65,thick] (w2)--(u2)--(z2);
\draw[candidate] (u2)--node[above,font=\small] {$e$}(v2);
\node[align=center,font=\small] at (1,-1.75)
 {Cavity graph\\Leaf blue probability $=\tanh(\beta W_e)$};
\end{scope}
\end{tikzpicture}
\end{adjustbox}
\caption{The target comparison. Solid edges to the left illustrate an already
reached blue cluster; the highlighted candidate $e$ has not yet been tested.
Dashed edges are unqueried incident edges, not observed closed bonds. Temporarily
removing them makes $v$ a leaf without changing the exterior-history observable.
The finite-mean extension in Appendix~\ref{sec:finite-mean} further restricts
exploration to targets whose incident cutoff marks are all good.
Sign-averaged insertion controls the cost of restoring them. In the one-attempt exploration,
a failed target is never tried again, preserving this untouched-star condition
at every subsequent attempt.}
\label{fig:leaf}
\end{figure}

\begin{lemma}[Quadratic sign-averaged insertion]\label{lem:insertion}
Delete an edge $f=xy$ of a finite graph, retaining its vertices. Fix all retained
disorder and the new magnitude $W_f$, and average the new fair sign. The density
of the restored two-replica spin law relative to the cavity law lies between
$\sech^2(\beta W_f)$ and $\cosh(4\beta W_f)$.
The same bounds hold for expectations of every nonnegative retained observable
that does not use the inserted coupling or its activation coin.
\end{lemma}

\begin{proof}
Set $a=\beta W_f$, $t=\tanh a$, and let
$z=\langle\sigma_x\sigma_y\rangle_{G\setminus f}$ be the one-replica cavity
correlation. For $r=\sigma_x\sigma_y$ and $s=\tau_x\tau_y$, the density is
\begin{align}
 K_t(r,s;z)
 &=\frac12\frac{(1+tr)(1+ts)}{(1+tz)^2}
  +\frac12\frac{(1-tr)(1-ts)}{(1-tz)^2}\label{eq:density}\\
 &=\frac{\cosh^2h}{\cosh^2a}
   \cosh\!\bigl(a(r+s)-2h\bigr),
 \qquad h=\operatorname{atanh}(tz).\label{eq:cosh}
\end{align}
The denominators in \eqref{eq:density} retain the two quenched normalizations.
Since $|h|\le a$, the prefactor in \eqref{eq:cosh} belongs to
$[\sech^2a,1]$ and the last argument has absolute value at most $4a$.
This proves both bounds, including $a=0$. Carry along the raw uniforms of
retained edges. Their blue-indicator functions remain unchanged, so integration
gives the retained-observable assertion, including events specifying closed edges.
\end{proof}

\begin{corollary}[Conditional CMR opening bounds]\label{cor:cmr-local}
On a finite graph of maximum degree $\Delta$, suppose $M=\E\cosh(4\beta W)<\infty$,
and put $b=\E\tanh(\beta W)$, $\ell=\E\sech^2(\beta W)$ and $r=\ell/M$.
For $e=uv$, every positive-probability $H\in\mathcal F_v^{\rm ext}$ satisfies
\[
 br^{\Delta-1}\le\nu_{G,\beta}(B_e=1\mid H)
 \le br^{-(\Delta-1)},\qquad
 \nu_{G,\beta}(q_v=s\mid H)\ge\tfrac12r^\Delta\quad(s=\pm1).
\]
\end{corollary}

\begin{proof}
Make $v$ a leaf by deleting its other incident edges. The history $H$ and the
success test $\mathbf1_H B_e$ remain the same functions of the surviving
coordinates. In the leaf law their expectations have ratio $b$, by
Lemma~\ref{lem:cmr-leaf}. Restore the deleted edges one at a time.
After averaging each new magnitude, Lemma~\ref{lem:insertion} multiplies either
retained expectation by a factor between $\ell$ and $M$. Comparing the lower
bound for the success mass with the upper bound for the history mass gives
$br^{\Delta-1}$; interchanging the bounds gives the upper estimate.
For the label bound delete the whole star, use the isolated probability $1/2$,
and restore with tests $\mathbf1_H$ and $\mathbf1_H\mathbf1_{\{q_v=s\}}$.
This is the CMR application of the general retained-observable comparison in
Lemma~\ref{lem:restoration}. In particular, the conditioning includes exterior
spins but reveals no target-star coupling or blue indicator.
\end{proof}

\paragraph{Cycle constraints and permitted histories.}
In the full CMR graph, blue edges have $q_xq_y=1$ and red edges have
$q_xq_y=-1$. Multiplication around a cycle shows that every grey
(blue or red) cycle has an even number of red edges. An all-blue cycle also
has $\prod_e\operatorname{sgn}J_e=1$, since every edge is satisfied in each
replica. These support constraints remain in the spin-generated law throughout
the comparison: restoring an interaction does not force its blue indicator
open. The history restrictions matter here. A failed blue test is not an
observed red bond; a red edge at the target could fix its overlap sign and
prohibit a proposed blue opening. Apart from the overlap field, which the
exploration of Section~\ref{sec:overlap-route} reveals at the outset before
testing only edges with $q_u=q_v=1$, our explorations query only blue
indicators. The one-attempt exploration below never reveals an edge at a target before
testing it, and the overlap-revealed exploration of
Section~\ref{sec:overlap-route} admits earlier failures at a target, which
cannot lower its bond floor (Theorem~\ref{thm:ov-extremal}).
The independent fields below describe thinned discoveries, not additional
physical bonds.

\subsection{One attempt at each target}

Let $G=(V,E)$ be a countable, locally finite simple graph, let $o\in V$, and
let $\omega\in\{0,1\}^E$ have an arbitrary law $P$. Write
$\starv(v)=\{e\in E:e\ni v\}$ and $C_\omega(o)$ for the open root cluster.
An \emph{off-star history at $v$} specifies finitely many edge indicators,
none belonging to $\starv(v)$. Both open and closed observations are allowed.

\begin{proposition}[One-attempt exploration]\label{prop:exploration}
Suppose $p\in[0,1]$ and, for every oriented edge $uv$ and every off-star
history $H$ at $v$ of positive probability,
\begin{equation}\label{eq:abstract-history}
 P(\omega_{uv}=1\mid H)\ge p.
\end{equation}
There is a coupling with independent Bernoulli($p$) site variables
$(X_v)_{v\ne o}$, with $X_o=1$, such that the occupied site root cluster
$C_X(o)$ is contained in $C_\omega(o)$. Each accepted vertex has an open
parent edge, so the containing paths can be chosen inside the explored set.
The statement also holds on a prescribed induced subgraph, and with
vertex-dependent bounds $p_v$ and independent sites of parameters $p_v$.

It suffices to assume \eqref{eq:abstract-history} for histories that the
exploration below can actually produce. Any extra initial information or
conditioning event must be included in the hypothesis; it is not free to reveal.
\end{proposition}

Figure~\ref{fig:one-attempt} shows the history restriction and the thinning step.

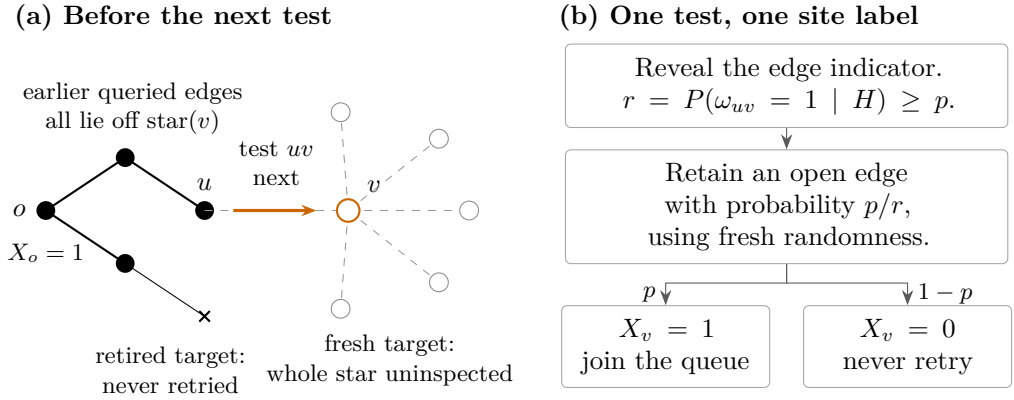
\begin{figure}[!htbp]
\centering
\begin{adjustbox}{max width=\linewidth}
\begin{tikzpicture}[
  x=1cm,y=1cm,font=\small,
  reached/.style={circle,fill=black,inner sep=2.5pt},
  fresh/.style={circle,draw=black!45,fill=white,inner sep=2.5pt},
  untouched/.style={draw=black!40,dashed},
  flow/.style={-{Stealth[length=2mm]},draw=black!65},
  box/.style={draw=black!35,rounded corners=2pt,align=center,inner sep=5pt}
]
% (a) The full star of the next target must be uninspected.
\node[anchor=west,font=\small\bfseries] at (0,4.9)
  {(a) Before the next test};
\coordinate (root) at (.55,2.35);
\coordinate (upper) at (1.6,3.05);
\coordinate (lower) at (1.6,1.65);
\coordinate (parent) at (2.65,2.35);
\coordinate (target) at (4.55,2.35);
\coordinate (retired) at (2.65,.95);
\draw[thick] (root)--(upper)--(parent);
\draw[thick] (root)--(lower);
\draw (lower)--(retired);
\node[reached,label=left:$o$] at (root) {};
\node[reached] at (upper) {};
\node[reached] at (lower) {};
\node[reached,label=above:$u$] at (parent) {};
\node[anchor=north,align=center,font=\footnotesize] at (.55,2.05)
  {$X_o=1$};
\draw[thick] ($(retired)+(-.07,-.07)$)--($(retired)+(.07,.07)$)
  ($(retired)+(-.07,.07)$)--($(retired)+(.07,-.07)$);
\node[anchor=north,align=center,font=\footnotesize] at (2.2,.68)
  {retired target:\\ never retried};
\node[align=center,font=\footnotesize] at (1.7,3.7)
  {earlier queried edges\\ all lie off $\starv(v)$};
\foreach \a/\b in {4.45/3.65,5.75/3.3,6.15/2.35,5.75/1.4,4.45/1.05}{
  \draw[untouched] (target)--(\a,\b);
  \node[fresh] at (\a,\b) {};
}
\draw[untouched] (parent)--(target);
\draw[-{Stealth[length=2mm]},very thick,orange!80!black]
  (3.03,2.35)--(4.13,2.35);
\node[above,font=\footnotesize,align=center] at (3.6,2.55)
  {test $uv$\\ next};
\node[circle,draw=orange!80!black,thick,fill=white,inner sep=3pt,
  label=above right:$v$] at (target) {};
\node[align=center,font=\footnotesize] at (5.1,.37)
  {fresh target:\\ whole star uninspected};
% (b) Conditional thinning gives one iid site label per target.
\begin{scope}[xshift=7.2cm]
\node[anchor=west,font=\small\bfseries] at (0,4.9)
  {(b) One test, one site label};
\node[box,text width=5.5cm] (physical) at (3.15,4.0)
  {Reveal the edge indicator.\\
   $r=P(\omega_{uv}=1\mid H)\ge p$.};
\node[box,text width=5.5cm] (thin) at (3.15,2.4)
  {Retain an open edge with probability $p/r$,\\
   using fresh randomness.};
\draw[flow] (physical)--(thin);
\node[box,text width=2.4cm] (accept) at (1.55,.55)
  {$X_v=1$\\ join the queue};
\node[box,text width=2.4cm] (reject) at (4.75,.55)
  {$X_v=0$\\ never retry};
\draw[flow] (thin.south)--++(0,-.25)-|
  node[pos=.75,left,font=\footnotesize] {$p$} (accept.north);
\draw[flow] (thin.south)--++(0,-.25)-|
  node[pos=.75,right,font=\footnotesize] {$1-p$} (reject.north);
\end{scope}
\end{tikzpicture}
\end{adjustbox}
\caption{One-attempt exploration (Proposition~\ref{prop:exploration}).
All earlier queries miss the fresh target's entire star. Thinning gives each
new site label conditional success probability $p$. Targets are never retried,
including physical successes discarded by thinning. Thus an iid site root
cluster, with root forced occupied, lies inside the physical open cluster.
This compares rooted clusters, not the entire edge field.
When $r=p=0$, set $X_v=0$.}
\label{fig:one-attempt}
\end{figure}

\begin{proof}
Sample the physical configuration $\omega$ from $P$ and reveal only the
indicators requested below. Initialize both the reached set and the set of
already-attempted vertices as $\{o\}$. Process reached vertices in a first-in,
first-out queue, with a fixed ordering of each finite neighbor set. When a
reached vertex $u$ has a neighbor $v$ never previously attempted, test $uv$.
Declare $v$ attempted immediately. It will never be tested again, even if
another reached vertex later becomes its neighbor.

Let $r$ be the conditional probability that $\omega_{uv}=1$ given the history
so far. Reveal this indicator and take a fresh uniform $U$, independent of
the physical configuration and all previous uniforms. Set
\[
 X_v=\mathbf1\{\omega_{uv}=1\}\mathbf1\{U\le p/r\}.
\]
If $r=0$, necessarily $p=0$, and set $X_v=0$. Add $v$ to the reached queue
exactly when $X_v=1$. Thus an open tested edge may also be discarded by thinning.

No earlier queried edge touches a never-attempted target: the other endpoint
of every previous test was reached and its target was immediately retired or
reached. Consequently the physical history lies off the entire star of $v$.
Past auxiliary uniforms reveal no additional information about unqueried edges
once the previous queried indicators are specified. More explicitly, for a fixed
transcript their likelihood is a function only of those indicators and the
independent randomization already used. This deferred-decisions property
preserves the physical conditional law, and \eqref{eq:abstract-history} gives
$r\ge p$. The next site label therefore has conditional success probability
exactly $p$, independently of the preceding transcript.

The next target is selected using only the assigned site labels and the fixed
queue rule. Thus the site-label transcript has exactly the law obtained by
revealing an iid site field with this exploration. Completing unattempted
labels independently gives the product site law. The physical configuration
retains its original law $P$.
The queue is fair because degrees are finite: every reached vertex is eventually
processed. Every occupied neighbor of the reached set is then reached, so its
final set is precisely $C_X(o)$. Each addition used an open parent edge, proving
the containment. The same proof permits deterministic restrictions of the graph
and replaces $p/r$ by $p_v/r$.
\end{proof}

\begin{remark}[Fresh targets and thinning]\label{rem:exploration}
For independent site/bond percolation with parameters $a,b$ and an occupied
root, testing a fresh target's site and its first candidate edge produces an
independent site root cluster of parameter $ab$
\cite[proof of Lemma 5, pp.~1794--1795]{GrimmettStacey1998}.
Here the conditional success probability may vary with the exploration
history; thinning makes the next site label independent of that history.
The comparison concerns rooted clusters, not all edge indicators. Retrying a
failed target would reveal part of its star and invalidate the stated hypothesis.
\end{remark}

\subsection{Percolation in the periodic ensemble}

\begin{corollary}[Direct CMR comparison]\label{cor:direct-cmr}
Suppose $\E\cosh(4\beta W)<\infty$, and set
\[
\begin{gathered}
 b=\E\tanh(\beta W),\quad \ell=\E\sech^2(\beta W),\quad
 M=\E\cosh(4\beta W),\\
 p=b(\ell/M)^{2d-1},\quad u=b(M/\ell)^{2d-1}.
\end{gathered}
\]
Every periodic joint limit has a blue root cluster containing the root-forced
independent site cluster of parameter $p$. If $p>\pc(\Z^d)$, both overlap
signs have infinite blue components almost surely. If $(2d-1)u<1$, the radius
and susceptibility estimates of Proposition~\ref{prop:path-bound} hold.
\end{corollary}

\begin{proof}
Corollary~\ref{cor:cmr-local} gives the finite-volume conditional bounds.
Write them without division, for example $\E[hB_e]\ge p\E[h]$ for every
nonnegative finite exterior-cylinder test $h$. Pass this fixed-coordinate
inequality to the selected periodic limit, then extend to the exterior
sigma-field by a monotone-class argument. Proposition~\ref{prop:exploration}
now supplies the independent site cluster. Equivalently, explore inside a fixed
ball, keep its full ambient law, pass to the limit and only then increase the
radius, as in Figure~\ref{fig:local-transfer}.

To upgrade positive root probability to almost-sure existence in sign $s$,
let $T_s$ be absence of any infinite blue component with that sign.
Deleting one vertex and its finite star cannot change this event, so $T_s$
is exterior to every target. If $\nu(T_s)>0$, the isolated-label bound gives
$\nu(T_s,q_o=s)>0$. Under this conditioning the exploration bound survives
and yields an infinite root cluster with positive probability, contradicting
$T_s$. These reusable arguments are stated in
Propositions~\ref{prop:local-transfer} and \ref{prop:as-existence}.
Global reversal of one replica preserves blue bonds and exchanges the root
overlap signs, so the positive root-percolation probability splits equally
between them.

For the upper bound, expose a prescribed self-avoiding path in order.
Its preceding edges are off the next target's star, so a path of length $n$
has probability at most $u^n$. There are at most $2d(2d-1)^{n-1}$ such paths
from a root. Summing gives Proposition~\ref{prop:path-bound}'s radius and
cluster-size estimates, without a monotonicity or boundary comparison.
\end{proof}

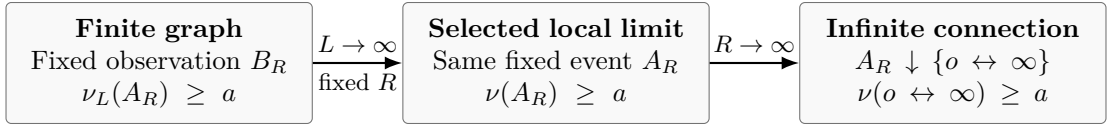
\begin{figure}[!htbp]
\centering
\begin{adjustbox}{max width=\linewidth}
\begin{tikzpicture}[>=Latex,
 box/.style={draw=black!55,rounded corners=2pt,fill=gray!5,
   align=center,font=\small,text width=3.7cm,minimum height=1.6cm,inner sep=5pt}]
\node[box] (finite) at (0,0)
 {\textbf{Finite graph}\\Fixed observation $B_R$\\$\nu_L(A_R)\ge a$};
\node[box] (local) at (5.25,0)
 {\textbf{Selected local limit}\\Same fixed event $A_R$\\$\nu(A_R)\ge a$};
\node[box] (infinite) at (10.5,0)
 {\textbf{Infinite connection}\\$A_R\downarrow\{o\leftrightarrow\infty\}$\\$\nu(o\leftrightarrow\infty)\ge a$};
\draw[->,thick] (finite)--node[above,font=\footnotesize] {$L\to\infty$}
 node[below,font=\footnotesize] {fixed $R$}(local);
\draw[->,thick] (local)--node[above,font=\footnotesize] {$R\to\infty$}(infinite);
\end{tikzpicture}
\end{adjustbox}
\caption{Transfer a finite observation before asking an infinite question.
Here $A_R$ is connection from $o$ to $\partial B_R$ by open edges inside $B_R$,
and $a$ is uniform in $R$ and in all sufficiently large $L$ for each $R$.
Local convergence passes
the fixed cylinder event to the selected law; only then does continuity from
above give infinite connection. The full ambient law is retained throughout,
without imposing a new boundary condition on $B_R$.}
\label{fig:local-transfer}
\end{figure}

\paragraph{The high-dimensional conclusion.}
Suppose $\E e^{aW}<\infty$ for some $a>0$. Then, as $\beta\downarrow0$,
$b=m\beta+O(\beta^3)$ and $\ell/M=1+O(\beta^2)$. At
$\beta=c/(2dm)$ the accumulated restoration cost therefore vanishes:
$2dp\to c$ and $(2d-1)u\to c$.
Kesten's classical theorem states
\begin{equation}\label{eq:kesten}
 \frac1{2d-1}\le p_c^{\mathrm{bond}}(\Z^d)
 \le\pc(\Z^d)
 \le\frac1{2d}+O\!\left(\frac{(\log\log d)^2}{d\log d}\right)
\end{equation}
\cite[Theorem 1, p.~220, Eq.~(1.3)]{Kesten1990}. Thus $p>\pc(\Z^d)$ for fixed $c>1$ and all sufficiently large finite $d$.
For $0<c<1$, take $\rho=(1+c)/2$; then $(2d-1)u\le\rho<1$ for large $d$.
The upper parameter $u$ is nondecreasing in $\beta$, giving this same
exponential path bound throughout $0\le\beta\le c/(2dm)$.
This proves both parts of Theorem~\ref{thm:headline} under exponential
integrability. Appendix~\ref{sec:finite-mean} removes that restriction by
skipping targets incident to a very large coupling, which leaves the
physical couplings and the Gibbs law unchanged, and determines the
subcritical blue susceptibility (Corollary~\ref{cor:chi});
Section~\ref{sec:overlap-route} and Appendix~\ref{app:fresh-star} give explicit
dimensions for unit couplings.

\section{Explicit dimensions through overlap-revealed exploration}\label{sec:overlap-route}

Theorem~\ref{thm:headline} locates the blue-percolation onset as the dimension
grows, but it does not name a dimension. For unit couplings we now prove
Theorem~\ref{thm:explicit} in dimensions $9$ to $12$: both overlap signs
percolate, for instance, in dimension $12$ at $\tanh\beta=3/25$ and in
dimension $9$ at $\tanh\beta=3/20$. Section~\ref{sec:macrostep} keeps the local
inputs and the assembly of this section, changes its global step, and reaches
dimensions $7$ to $9$.

The obstruction lies in what the exploration conditions on. The comparison of
Section~\ref{sec:cmr-verification} reveals both replica configurations outside a
target's star and accepts each target with one independent site label. Its
opening bound must hold for every exterior spin configuration, and the site
comparison pays for the least favorable one. For $\pm1$ couplings, a sharper
star bound of this kind combined with an oriented site comparison reaches
dimensions $16$--$20$ and $22$ (Appendix~\ref{app:fresh-star}). The
exploration below instead reveals the whole overlap field $q$ at the outset,
including the overlap of each target, and queries blue bonds only between
vertices of overlap~$+1$.

The comparison of Appendix~\ref{sec:mns} already works inside the overlap
sectors. It conditions on both replicas, hence on $q$, compares the blue bonds
on agreeing edges with independent bonds of parameter $L_\beta$, which is
$p_A/2$ for unit couplings \cite[Lemmas~4.1--4.2]{MNS2008}, and dominates the
overlap field by independent sites at its worst conditional probability, as in
\cite[Eq.~(8)]{ABL1987}; its explicit unit-coupling example lies in dimension
$180$ (Appendix~\ref{app:separation-certificate}). Revealing the replicas in
addition to $q$ changes no conditional law, by Lemma~\ref{lem:ov-gauge}(iii)
below. The gain comes from three refinements: a sharper bond floor that failures
cannot lower, domination by a weak Ising field instead of independent sites,
and a second moment adapted to that Ising field.

\emph{Failures are harmless once $q$ is known.} On an edge joining equal
overlaps, the blue indicator records only whether the edge is satisfied and an
independent activation. Given $q$ and everything off a target's star, the
satisfaction signs of the star have an explicit law
(Proposition~\ref{prop:ov-star-law}). In the product factor of that law, edges
to neighbors of equal overlap are tilted toward satisfaction, while edges to
neighbors of opposite overlap, and agreeing edges already observed to fail, are
exactly fair. A mediant inequality shows that fair coordinates cannot lower the
worst case over cavity laws: the floor is attained when every neighbor agrees
and nothing has failed. One bond floor $p_B$, about $0.9\tanh2\beta$, therefore
holds after any number of failures observed by a bond exploration inside
$\{q=+1\}$ (Section~\ref{sec:ov-floor}). At the lower level $p_A/2$, the same
robustness already follows from the coordinatewise monotone posterior density
in the proof of \cite[Lemma~4.1]{MNS2008} (Lemma~\ref{lem:mns}). The floor
$p_B$ is higher, and at $d=12$, $\tanh\beta=3/25$ the criterion below needs a
certified value above $p_A/2$ (Section~\ref{sec:ov-floor}).

\emph{The overlap field dominates a weak Ising field.} The price of revealing
$q$ is that the plus sites form a dependent random set. Their single-site
conditional odds admit an explicit lower bound $e^{2h+2K'S}$ that is uniform
over the unknown cavity laws, where $S$ is the sum of the neighboring overlaps;
we call it the Holley line. By Holley's theorem, $q$ dominates the Ising field
with the small coupling $K'$ and the negative field $h$
(Section~\ref{sec:ov-holley}). For
$d=12$ and $\tanh\beta=3/25$ we take $K'\approx0.0130$ and $h\approx-0.0900$,
and the plus-density is at least $0.435$ (Section~\ref{sec:ov-density}). We
are not aware of an earlier stochastic comparison of a spin-glass overlap field
with a ferromagnetic Ising field.

\emph{The comparison model percolates.} Plus sites of that Ising field, joined
by independent Bernoulli bonds, percolate along oriented paths by a weighted
second moment. The Ising law enters only through an upper bound on a pair ratio,
proved with the FKG inequality and Dobrushin's comparison theorem, and through a
correlation gain along shared edges. The resulting renewal series is summed with
a Schur test and Khas'minskii's lemma, which yields an explicit scalar
criterion (Section~\ref{sec:ov-second-moment}).

\emph{Noisy cavities sharpen both local inputs.} The floor and the Holley line
are uniform over all flip-invariant cavity laws, and their certificates are
decided at frozen laws. Actual cavity laws are not frozen: given the spins at distance
two, each neighbor of the target has conditional mean at most
$\tanh((2d-1)\beta)$ in absolute value. Moving this noise into the star
interaction replaces the cavity laws by a smaller class at a smaller effective
coupling (Lemma~\ref{lem:ov-noisy}). Over that class, an exact pair certificate
for the floor and a symmetrized certificate for the Holley line give sharper
inputs (Section~\ref{sec:ov-noisy}); with them the criterion reaches
dimension $9$.

Each constant is an exact rational number, or a rigorous rational enclosure of
finitely many irrational numbers, certified by small finite computations
(Section~\ref{sec:ov-constants}). The certified margin in the final criterion
is about $9\%$ at $(d,\tanh\beta)=(12,3/25)$ and just under $5\%$ at
$(11,13/100)$. In
dimension $10$ it is below $1\%$, and there only a refined criterion, with a
correlation gain along shared path edges (Lemma~\ref{lem:ov-shared}), certifies
the two temperatures of Theorem~\ref{thm:explicit}. In dimension $9$ the
sharpened inputs and the same refinement certify four temperatures, with margins
between $0.32\%$ and $0.93\%$, and a sharper covariance bound
(Lemma~\ref{lem:ov-covariance}) certifies a fifth. With the Ising law $\mu_L$
of \eqref{eq:ov-ising} and the event $\mathcal O_n$ of
Section~\ref{sec:ov-second-moment}, the argument is the chain, valid for every
$n\ge1$ and all sufficiently large $L$,
\begin{equation}\label{eq:ov-chain}
 \nu_{\mathbb T_L^d,\beta}\bigl(q_o=1,\ o\blue\partial B_n\bigr)
 \ \ge\ \bigl(\mu_L\otimes\mathrm{Ber}(p)^{\otimes E}\bigr)(\mathcal O_n)
 \ \ge\ \theta_* .
\end{equation}
The first inequality combines Sections~\ref{sec:ov-gauge}--\ref{sec:ov-holley}
through a quantile coupling; the second is the second-moment criterion.
Section~\ref{sec:ov-assembly} proves \eqref{eq:ov-chain} and passes it to the
periodic limits.

\paragraph{The statement proved.}
Fix a row $(d,t)$ of Table~\ref{tab:overlap-constants} or
Table~\ref{tab:ov-noisy} and put $\beta=\operatorname{atanh}t$. The row lists
the bond parameter $p$ and an Ising line $g_K=e^{2K'}$, $g_h=e^{2h}$, all
rational, together with certified bounds $\rho_-$ and $\eta'$ on the Ising
plus-density and a renewal constant, and bounds on the renewal scores of
Theorem~\ref{thm:ov-criterion}. Put $\kappa=(1-\rho_-)/\rho_-$ and
$\alpha_{\rm D}=2d\tanh K'$. Let $\mathcal S_\bullet$ be the score used:
$\mathcal S$ in dimensions $11$ and $12$, $\mathcal S_c$ in dimension $10$ and
at four of the five temperatures in dimension $9$, and $\mathcal S'_c$ at
$(d,t)=(9,7/50)$. Define
\begin{equation}\label{eq:ov-theta}
 \theta_*:=e^{-1/100}\,
 \frac{\rho_-(1-\eta')(1-\mathcal S_\bullet)}{1-1/d},
\end{equation}
\begin{equation}\label{eq:ov-L1}
 L_1(n):=\max\Bigl\{40,\ 2n+4,\
 2n+1+\Bigl\lceil\frac{\log\bigl(100\kappa n^2/(1-\alpha_{\rm D})\bigr)}
 {\log(1/\alpha_{\rm D})}\Bigr\rceil^+\Bigr\}.
\end{equation}

\begin{proposition}[Uniform torus bound]\label{prop:overlap-central}
For each row of Tables~\ref{tab:overlap-constants} and~\ref{tab:ov-noisy},
every $n\ge1$, every $L\ge L_1(n)$ and $s=\pm1$,
\[
 \nu_{\mathbb T_L^d,\beta}\bigl(q_o=s,\ o\blue\partial B_n\bigr)\ \ge\ \theta_*,
\]
where the blue connection uses paths inside $B_n$.
\end{proposition}

The row fixes $(d,\beta)$, the rational inputs, and hence $\theta_*$ and the
function $L_1$; then $n$, then $L\ge L_1(n)$, then $s$ are arbitrary. Only the
finiteness of $L_1(n)$ matters below. Theorem~\ref{thm:explicit} follows for
these rows in Section~\ref{sec:ov-assembly}. The finite-ball bound passes to every selected
periodic limit in the order of Figure~\ref{fig:local-transfer}, and an
adaptation of the exterior conditioning of Proposition~\ref{prop:as-existence}
gives almost-sure coexistence. Neither step uses independence of the limiting
replicas.

\subsection{Gauge representation and the overlap-revealed star law}\label{sec:ov-gauge}

Throughout this section the couplings are iid and uniform on $\{-1,1\}$, and
$\mathbb T_L$ abbreviates the torus $\mathbb T_L^d$ with $L\ge3$, so every vertex
has $2d$ distinct neighbors. Write $V$ and $E$ for its vertex and edge sets, and put
$\nu_L=\nu_{\mathbb T_L,\beta}$. Put
\[
 t=\tanh\beta,\qquad e^{2\beta}=\frac{1+t}{1-t},\qquad
 C_\beta=\cosh2\beta,\qquad
 a=\frac{e^{2\beta}}{2\cosh2\beta},\qquad
 p_A=1-e^{-4\beta}.
\]
Then $a/(1-a)=e^{4\beta}$, $p_Aa=\tanh2\beta$, $2C_\beta a=e^{2\beta}$ and
$2C_\beta(1-a)=e^{-2\beta}$; all these numbers are rational when $t$ is. We
realize the blue indicators of \eqref{eq:blue} as
$B_e=\mathbf1\{J_e\sigma_x\sigma_y=1,\ J_e\tau_x\tau_y=1\}A_e$, where the
activations $A_e$ are iid Bernoulli($p_A$) and independent of $(J,\sigma,\tau)$,
and keep writing $\nu_L$ for the joint law of $(J,\sigma,\tau,A,B)$. Its
disorder and spin marginal is
$2^{-|E|}\mu_{\mathbb T_L,J,\beta}(\sigma)\mu_{\mathbb T_L,J,\beta}(\tau)$, with both quenched
normalizations.

\begin{lemma}[Gauge representation]\label{lem:ov-gauge}
Put $s_e=J_e\sigma_x\sigma_y$ for $e=xy\in E$, and $q_x=\sigma_x\tau_x$. The map
$(J,\sigma,\tau)\mapsto(s,\sigma,q)$ is a bijection, and under $\nu_L$:
\begin{enumerate}\renewcommand{\labelenumi}{(\roman{enumi})}
\item $Z_{\mathbb T_L,J,\beta}=Z(s):=\sum_{\rho\in\{\pm1\}^V}
 \exp\bigl(\beta\sum_{xy\in E}s_{xy}\rho_x\rho_y\bigr)$;
\item $\nu_L(s,\sigma,q)=2^{-|E|}\exp\bigl(\beta\sum_{xy\in E}s_{xy}(1+q_xq_y)\bigr)
 Z(s)^{-2}$, and $A$ is independent of $(s,\sigma,q)$;
\item $\sigma$ is uniform on $\{\pm1\}^V$ and independent of $(s,q,A)$;
\item $B_e=\mathbf1\{s_e=1,\ q_x=q_y\}A_e$ for every $e=xy$.
\end{enumerate}
\end{lemma}

\begin{proof}
The inverse map is $J_e=s_e\sigma_x\sigma_y$, $\tau=\sigma q$. The substitution
$\rho=\sigma\rho'$ in $Z_{\mathbb T_L,J,\beta}$ gives (i). Since $J_e\sigma_x\sigma_y=s_e$
and $J_e\tau_x\tau_y=s_eq_xq_y$, the weight
$2^{-|E|}\exp(\beta\sum_eJ_e(\sigma_x\sigma_y+\tau_x\tau_y))Z_{\mathbb T_L,J,\beta}^{-2}$
equals the expression in (ii); a bijection of finite sets carries the law
without a Jacobian factor. The right side of (ii) does not depend on $\sigma$,
which gives (iii). Finally, blue requires $s_e=1$ and $s_eq_xq_y=1$.
\end{proof}

Lemma~\ref{lem:ov-gauge} is the gauge transformation
\cite[Section~4.2]{Nishimori2001} with the gauge taken from the first replica,
as in \cite[Lemma~4.1]{MNS2008} and Lemma~\ref{lem:mns}; the indicators
$\mathbf1\{s_e=1\}$ are the satisfaction variables of the first replica
\cite[Section~4.1]{NS2007}. By (iii), conditioning on $\sigma$ in addition to
$(s,q,A)$ changes no conditional law.
Thus the overlap field and the blue field are functions of $(s,q,A)$, and the
two quenched normalizations become the single factor $Z(s)^{-2}$. Fix a vertex
$v$ with neighbor set $N(v)$, and write $s_{\rm st}=(s_{vw})_{w\in N(v)}$ and
$s_{\rm off}$ for the other gauge signs. The \emph{cavity law}
$\chi=\chi_{v,s_{\rm off}}$ is the law of $(\rho_w)_{w\in N(v)}$ under the
zero-field Gibbs measure on $\mathbb T_L\setminus v$ with couplings $s_{\rm off}$.
Summing over $\rho_v$ first gives
\begin{equation}\label{eq:ov-star-factor}
 Z(s)=Z_{\mathbb T_L\setminus v}(s_{\rm off})\,R_\chi(s_{\rm st}),\qquad
 R_\chi(\varsigma):=\E_{\rho\sim\chi}\Bigl[2\cosh\Bigl(\beta\sum_{w\in N(v)}
 \varsigma_w\rho_w\Bigr)\Bigr].
\end{equation}
Every cavity law is invariant under the global flip $\rho\mapsto-\rho$. Write
$\mathcal N_{2d}$ for the class of all flip-invariant laws on
$\{\pm1\}^{N(v)}$. It is convex, and closed under permutations of $N(v)$ and
under the single-coordinate flips $\chi\mapsto\chi^{(w)}$, where $\chi^{(w)}$ is
the law of $\rho$ with $\rho_w$ replaced by $-\rho_w$. Its extreme points are
the frozen laws $\frac12(\delta_\rho+\delta_{-\rho})$. Moreover $R_\chi$ is
linear in $\chi$ and $R_\chi(-\varsigma)=R_\chi(\varsigma)$.

Given $q$, let $N_=(v)=\{w\in N(v):q_w=q_v\}$ be the agreeing neighbors and
$N_{\ne}(v)=N(v)\setminus N_=(v)$. For $\mathsf T\subseteq N(v)$, let
$\pi^{\mathsf T}$ be the product law on $\{\pm1\}^{N(v)}$ under which
$\varsigma_w=1$ has probability $a$ for $w\in\mathsf T$ (a \emph{tilted}
coordinate) and $\frac12$ otherwise (a \emph{fair} coordinate). For a product
law $\pi$ put $\pi_{R^{-2}}(\varsigma)=\pi(\varsigma)R_\chi(\varsigma)^{-2}/
\E_\pi[R_\chi^{-2}]$. Let
\[
 \mathcal G_v=\sigma\bigl(q,\ (s_e,A_e)_{e\notin\starv(v)}\bigr),\qquad
 \mathsf F_v(F)=\{B_{vw}=0\text{ for all }w\in F\}\quad(F\subseteq N(v)).
\]
The exploration and the limit arguments condition only on the coarser
$\sigma$-fields
\begin{equation}\label{eq:ov-exterior}
 \mathcal E_v=\sigma\bigl((q_x)_{x\ne v},\ (B_e)_{e\notin\starv(v)}\bigr),
 \qquad
 \mathcal E_v^+=\mathcal E_v\vee\sigma(q_v),
\end{equation}
which make sense in infinite volume; $\mathcal E_v$ is the exterior
$\sigma$-field of Proposition~\ref{prop:as-existence} with labels $q$. By
Lemma~\ref{lem:ov-gauge}(iv), $\mathcal E_v^+\subseteq\mathcal G_v$.

\begin{proposition}[Overlap-revealed star law]\label{prop:ov-star-law}
Let $g$ be an atom of $\mathcal G_v$ with $q_v=1$, and let $F\subseteq N_=(v)$.
Conditionally on $g\cap\mathsf F_v(F)$, the star signs $s_{\rm st}$ have law
$(\pi^{N_=(v)\setminus F})_{R^{-2}}$ with $\chi=\chi_{v,s_{\rm off}}$, and the
activations $(A_{vw})_{w\notin F}$ are iid Bernoulli($p_A$) and independent of
$s_{\rm st}$. In particular, for $u\in N_=(v)\setminus F$,
\begin{equation}\label{eq:ov-exploration-identity}
 \nu_L\bigl(B_{uv}=1\bigm| g\cap\mathsf F_v(F)\bigr)
 =p_A\,(\pi^{N_=(v)\setminus F})_{R^{-2}}(\varsigma_u=1).
\end{equation}
\end{proposition}

\begin{proof}
By Lemma~\ref{lem:ov-gauge}, \eqref{eq:ov-star-factor} and
$1+q_vq_w=2\cdot\mathbf1\{w\in N_=(v)\}$, the joint weight of
$(s_{\rm st},A_{\rm st})=(\varsigma,\alpha)$ on $g$ is a constant depending on
$g$ times
\[
 \prod_{w\in N_=(v)}e^{2\beta\varsigma_w}\cdot R_\chi(\varsigma)^{-2}\cdot
 \prod_{w\in N(v)}p_A^{\alpha_w}(1-p_A)^{1-\alpha_w}.
\]
For $w\in F$ we have $q_w=q_v$, hence $B_{vw}=\mathbf1\{\varsigma_w=1\}\alpha_w$.
Restricting to $\mathsf F_v(F)$ and summing over $\alpha_w$ replaces the factor
$e^{2\beta\varsigma_w}$ by
\begin{equation}\label{eq:ov-failed-fair}
 e^{2\beta\varsigma_w}\bigl(1-p_A\mathbf1\{\varsigma_w=1\}\bigr)=e^{-2\beta}
 \qquad(\varsigma_w=\pm1),
\end{equation}
a constant: \emph{a failed agreeing edge is exactly fair}. For
$w\in N_=(v)\setminus F$ the factor $e^{2\beta\varsigma_w}$ equals $2C_\beta a$
or $2C_\beta(1-a)$, a tilted coordinate, and for $w\in N_{\ne}(v)$ it is $1$, a
fair coordinate. The remaining weight factorizes between $\varsigma$ and
$(\alpha_w)_{w\notin F}$. For \eqref{eq:ov-exploration-identity} note that
$B_{uv}=\mathbf1\{\varsigma_u=1\}A_{uv}$.
\end{proof}

Proposition~\ref{prop:ov-star-law} is the star-local form of the posterior law
of the gauge couplings given both replicas \cite[Lemma~4.1]{MNS2008}. Identity
\eqref{eq:ov-failed-fair} is the vacant term of the CMR edge weight
\cite[Eq.~(6)]{CMR1998}: once the blue event is excluded, an agreeing edge keeps
a weight that does not depend on its sign. It uses $p_A=1-e^{-4\beta}$ exactly,
and the same cancellation gives the failed-star factor \eqref{eq:failed-factor}
of Appendix~\ref{app:fresh-star}. The difference lies in the conditioning. With
$q$ revealed, including the target's overlap, a failure turns a tilted
coordinate of a product law into a fair one, and the effect of that change can
be compared exactly (Lemma~\ref{lem:ov-mediant}).
The factor $R_\chi^{-2}$ carries the two quenched normalizations. If the cavity
law is iid fair, $R_\chi$ is constant and \eqref{eq:ov-exploration-identity}
equals $p_Aa=\tanh2\beta$. Figure~\ref{fig:ov-star} illustrates the law.

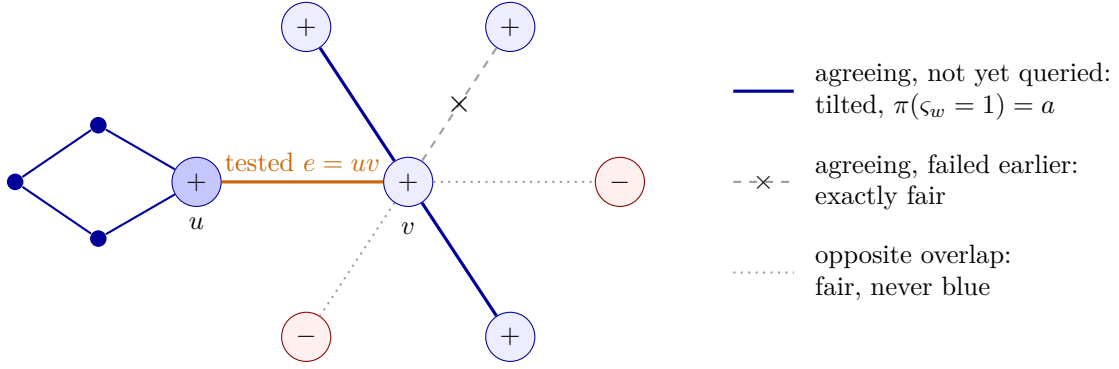
\begin{figure}[!htbp]
\centering
\begin{adjustbox}{max width=\linewidth}
\begin{tikzpicture}[>=Latex,font=\small,
 site/.style={circle,draw,minimum size=6.5mm,inner sep=0pt},
 plus/.style={site,fill=blue!7,draw=blue!60!black},
 minus/.style={site,fill=red!6,draw=red!55!black},
 reached/.style={circle,fill=blue!60!black,inner sep=2.2pt},
 tilted/.style={very thick,blue!60!black},
 fairopp/.style={thick,gray!75,dotted},
 fairfail/.style={thick,gray!75,dashed}]
\node[plus] (v) at (0,0) {$+$};
\node[font=\small] at (0,-0.62) {$v$};
\node[plus,fill=blue!22] (u) at (-2.8,0) {$+$};
\node[font=\small] at (-2.8,-0.55) {$u$};
\node[plus] (w1) at (-1.35,2.05) {$+$};
\node[plus] (w2) at (1.35,2.05) {$+$};
\node[minus] (w3) at (2.8,0) {$-$};
\node[plus] (w4) at (1.35,-2.05) {$+$};
\node[minus] (w5) at (-1.35,-2.05) {$-$};
\node[reached] (r1) at (-4.1,0.75) {};
\node[reached] (r2) at (-4.1,-0.75) {};
\node[reached] (r0) at (-5.2,0) {};
\draw[blue!60!black,thick] (r0)--(r1)--(u)--(r2)--(r0);
\draw[very thick,orange!80!black] (u)--node[above] {tested $e=uv$} (v);
\draw[tilted] (v)--(w1);
\draw[tilted] (v)--(w4);
\draw[fairfail] (v)--(w2);
\node[font=\large] at ($(v)!0.5!(w2)$) {$\times$};
\draw[fairopp] (v)--(w3);
\draw[fairopp] (v)--(w5);
\begin{scope}[xshift=4.3cm]
\draw[tilted] (0,1.2)--(0.8,1.2);
\node[anchor=west,align=left] at (0.95,1.2)
 {agreeing, not yet queried:\\ tilted, $\pi(\varsigma_w=1)=a$};
\draw[fairfail] (0,0)--(0.8,0);
\node at (0.4,0) {$\times$};
\node[anchor=west,align=left] at (0.95,0)
 {agreeing, failed earlier:\\ exactly fair};
\draw[fairopp] (0,-1.2)--(0.8,-1.2);
\node[anchor=west,align=left] at (0.95,-1.2)
 {opposite overlap:\\ fair, never blue};
\end{scope}
\end{tikzpicture}
\end{adjustbox}
\caption{The overlap-revealed star law (Proposition~\ref{prop:ov-star-law}) at a
target $v$ with $q_v=1$; the signs inside the vertices are the revealed
overlaps, and $u$ belongs to the explored blue cluster. Given $q$, the gauge
signs and activations off the star and the failures observed so far, the star
signs $\varsigma_w=s_{vw}$ have law proportional to
$\pi(\varsigma)R_\chi(\varsigma)^{-2}$ for a product law $\pi$ with the
coordinate types shown; the tested edge $uv$ is itself an agreeing coordinate
that has not been queried, hence tilted. Lemma~\ref{lem:ov-mediant} shows that
fair coordinates cannot lower the worst case, over cavity laws, of the success
probability of the tested edge, so the bond floor is attained when all
neighbors agree and no edge has failed (Theorem~\ref{thm:ov-extremal}).}
\label{fig:ov-star}
\end{figure}

\subsection{The bond floor and its extremal cavity}\label{sec:ov-floor}

For $u\in\mathsf T\subseteq N(v)$ and $\chi\in\mathcal N_{2d}$ put
\begin{equation}\label{eq:ov-pB}
 r(\chi;\mathsf T)=(\pi^{\mathsf T})_{R^{-2}}(\varsigma_u=1),\qquad
 p_B=p_A\inf_{\chi\in\mathcal N_{2d}}\ \inf_{\mathsf T\ni u}r(\chi;\mathsf T).
\end{equation}
Since $\mathcal N_{2d}$ is closed under permutations, $p_B$ depends only on $d$
and $\beta$. By \eqref{eq:ov-exploration-identity} with
$\mathsf T=N_=(v)\setminus F$, every conditional opening probability in
Proposition~\ref{prop:ov-star-law} is at least $p_B$, uniformly in $L$, the
vertex, the atom and the failure set. The class $\mathcal N_{2d}$ contains every
cavity law; taking the infimum over the whole class is what makes the bound
uniform.

\begin{lemma}[A fair coordinate is an average of two tilted ones]\label{lem:ov-mediant}
For $\chi\in\mathcal N_{2d}$ and $w\notin\mathsf T$,
\[
 r(\chi;\mathsf T)\ \ge\ \min\bigl\{r(\chi;\mathsf T\cup\{w\}),\
 r(\chi^{(w)};\mathsf T\cup\{w\})\bigr\}.
\]
\end{lemma}

\begin{proof}
Write $\varsigma=(\varsigma_u,\varsigma_w,\varsigma_{\rm r})$, and let
$\pi_{\rm r}$ be the law of $\varsigma_{\rm r}$, which is the same in all three
cases. For $\chi'\in\mathcal N_{2d}$ put
$X_{\chi'}=\E_{\pi_{\rm r}}[R_{\chi'}(1,1,\varsigma_{\rm r})^{-2}]$ and
$Y_{\chi'}=\E_{\pi_{\rm r}}[R_{\chi'}(-1,1,\varsigma_{\rm r})^{-2}]$. Flipping
$\varsigma_w$ has the same effect as flipping $\rho_w$, so
$R_\chi(\varsigma_u,-1,\varsigma_{\rm r})
=R_{\chi^{(w)}}(\varsigma_u,1,\varsigma_{\rm r})$. Hence each of the three
ratios equals $aX/(aX+(1-a)Y)$, where $(X,Y)$ is, respectively,
\[
 \tfrac12\bigl(X_\chi+X_{\chi^{(w)}},\,Y_\chi+Y_{\chi^{(w)}}\bigr),\qquad
 (X_1,Y_1)=a(X_\chi,Y_\chi)+(1-a)(X_{\chi^{(w)}},Y_{\chi^{(w)}}),
\]
\[
 (X_2,Y_2)=a(X_{\chi^{(w)}},Y_{\chi^{(w)}})+(1-a)(X_\chi,Y_\chi).
\]
The first pair is the average of the other two. The mediant inequality
$(Y_1+Y_2)/(X_1+X_2)\le\max(Y_1/X_1,Y_2/X_2)$ and the monotonicity of
$aX/(aX+(1-a)Y)$ in $Y/X$ prove the claim.
\end{proof}

\begin{theorem}[The extremal partition]\label{thm:ov-extremal}
$\displaystyle\inf_{\mathsf T\ni u}\ \inf_{\chi\in\mathcal N_{2d}}r(\chi;\mathsf T)
=\inf_{\chi\in\mathcal N_{2d}}r(\chi;N(v))$. Thus
$p_B=p_A\inf_\chi r(\chi;N(v))$: the worst case has every neighbor agreeing and
no failure, and earlier failures at a target, however many, never lower the
floor.
\end{theorem}

\begin{proof}
Apply Lemma~\ref{lem:ov-mediant} to the fair coordinates one at a time. This
gives
\[
 r(\chi;\mathsf T)\ \ge\ \min_{V'\subseteq N(v)\setminus\mathsf T}
 r\bigl(\chi^{(V')};N(v)\bigr),
\]
where $\chi^{(V')}\in\mathcal N_{2d}$ flips the coordinates in $V'$. The
reverse inequality is trivial.
\end{proof}

\begin{remark}[Observed successes are not harmless]\label{rem:ov-successes}
Lemma~\ref{lem:ov-mediant} concerns failures. An observed success on an edge
$wv$ would fix the coordinate $\varsigma_w=1$, and such frozen coordinates can
lower the floor. At $d=12$ and $t=3/25$, the aligned frozen cavity with
$f=0,1,2,3,4$ frozen successes gives opening probabilities at most $0.2210000$,
$0.2186899$, $0.2163712$, $0.2140597$ and $0.2117739$ (exact rationals rounded
up). From $f=3$ on these lie below the floor $0.2161$ used in
Table~\ref{tab:overlap-constants}. A bond exploration that discarded a success
by thinning and later queried the same target again would meet this situation.
The exploration of Lemma~\ref{lem:ov-exploration} never does: every observed
success reaches its target, so the history at an unreached target contains only
failures.
\end{remark}

By Theorem~\ref{thm:ov-extremal} only $\mathsf T=N(v)$ remains, where all $2d$
star signs are tilted. Put $\bar n=2d-1$ and $\varsigma=(\varsigma_u,\varsigma')$,
let $\pi'$ be the tilted product law of $\varsigma'\in\{\pm1\}^{\bar n}$, and put
$\Sigma'=\sum_x\varsigma'_x$ and $F_\chi(\varsigma')=R_\chi(1,\varsigma')^{-2}$.
Since $R_\chi(-1,\varsigma')=R_\chi(1,-\varsigma')$ and
$\pi'(-\varsigma')=e^{-4\beta\Sigma'}\pi'(\varsigma')$,
\begin{equation}\label{eq:ov-XY}
 r(\chi;N(v))=\frac{aX}{aX+(1-a)Y},\qquad
 \begin{aligned}
 X&=\E_{\pi'}F_\chi,\\
 Y&=\E_{\pi'}\bigl[R_\chi(-1,\varsigma')^{-2}\bigr]
 =\E_{\pi'}\bigl[e^{-4\beta\Sigma'}F_\chi\bigr].
 \end{aligned}
\end{equation}
For odd $z$ put $\omega(z)=2\cosh(\beta(1+z))$, which is rational when $t$ is.

\begin{proposition}[A finite certificate for the floor]\label{prop:ov-vertex}
Let $\underline r\in(0,1)$, put $\bar\lambda=e^{4\beta}(1-\underline r)/\underline r$,
and choose $c_m>0$ for each odd $m$ with $e^{-4\beta m}<\bar\lambda$. For
$j=0,\ldots,\bar n$ let
\[
 \Delta_j=\sum_{i=0}^{j}\sum_{i'=0}^{\bar n-j}\binom ji\binom{\bar n-j}{i'}
 a^{i+i'}(1-a)^{\bar n-i-i'}\bigl(e^{-4\beta m}-\bar\lambda\bigr)\Psi_m(z),
\]
where $\ell_+=2i-j$, $\ell_-=2i'-(\bar n-j)$, $m=\ell_++\ell_-$,
$z=\ell_+-\ell_-$, and $\Psi_m(z)=\omega(z)^{-2}$ if
$e^{-4\beta m}\ge\bar\lambda$, while $\Psi_m(z)=3c_m^{-2}-2\omega(z)c_m^{-3}$
otherwise. If $\Delta_j\le0$ for every $j$, then $r(\chi;N(v))\ge\underline r$
for all $\chi\in\mathcal N_{2d}$, and hence $p_B\ge p_A\underline r$.
\end{proposition}

\begin{proof}
By \eqref{eq:ov-XY},
$Y-\bar\lambda X=\sum_{\varsigma'}\pi'(\varsigma')(e^{-4\beta\Sigma'}-\bar\lambda)
F_\chi(\varsigma')$. Write $F_\chi=(\E_\chi W_\rho)^{-2}$ with
$W_\rho(\varsigma')=2\cosh(\beta(\rho_u+\sum_x\varsigma'_x\rho_x))$. Where the
coefficient is nonnegative use Jensen's inequality
$(\E_\chi W_\rho)^{-2}\le\E_\chi W_\rho^{-2}$. Where it is negative use the
tangent-line bound $y^{-2}\ge3c^{-2}-2yc^{-3}$ for $y,c>0$, which is linear in
$y$, at $c=c_{\Sigma'}$. Both bounds are linear in the law $\chi$. Hence
$Y-\bar\lambda X\le\E_\chi\Delta(\rho)\le\max_\rho\Delta(\rho)$, where
$\Delta(\rho)$ is the same sum with $W_\rho$ in place of $\E_\chi W_\rho$, for
one frozen pattern $\rho$. Since $\Delta(-\rho)=\Delta(\rho)$ we may take
$\rho_u=1$. If $\rho'$ has $j$ coordinates equal to $1$, let $\ell_+$ and
$\ell_-$ be the sums of $\varsigma'$ over these coordinates and over the others.
They are independent sums of $j$ and $\bar n-j$ tilted signs,
$\Sigma'=\ell_++\ell_-$ and $W_\rho=\omega(\ell_+-\ell_-)$. Thus
$\Delta(\rho)=\Delta_j$. If every
$\Delta_j\le0$, then $Y\le\bar\lambda X$ and
$r\ge a/(a+(1-a)\bar\lambda)=\underline r$, because $(1-a)e^{4\beta}=a$.
\end{proof}

Two explicit comparisons bracket the floor. First,
$\cosh(y-\beta)/\cosh(y+\beta)$ lies in $[e^{-2\beta},e^{2\beta}]$, so
$R_\chi(-1,\varsigma')\ge e^{-2\beta}R_\chi(1,\varsigma')$, $Y\le e^{4\beta}X$ and
\begin{equation}\label{eq:ov-closed-floor}
 p_B\ \ge\ p_A/2=(1-e^{-4\beta})/2 .
\end{equation}
This is the lower product parameter $L_\beta$ of \cite[Lemma~4.2]{MNS2008} at
unit couplings (Lemma~\ref{lem:overlap}).
Second, the aligned frozen cavity $\frac12(\delta_{\mathbf 1}+\delta_{-\mathbf 1})$
gives the upper bound $p_B\le p_Ar(\frac12(\delta_{\mathbf 1}+\delta_{-\mathbf 1});N(v))$,
an explicit single sum. At $d=12$, $t=3/25$ we have
$p_A=75/196$ and $\tanh2\beta=75/317\approx0.2366$, and \eqref{eq:ov-closed-floor}
gives only $75/392\approx0.1913$. Proposition~\ref{prop:ov-vertex} with
$\underline r=0.2161/p_A$ and tangent points $c_m=\varrho_m\omega(m)$ for
$m=1,3,\ldots,23$, where
\begin{gather*}
 (\varrho_1,\varrho_3,\ldots,\varrho_{11})=(1.080881,\ 1.042249,\ 0.984864,\
 0.916142,\ 0.842612,\ 0.768993),\\
 (\varrho_{13},\varrho_{15},\ldots,\varrho_{23})=(0.698220,\ 0.631872,\
 0.570633,\ 0.514653,\ 0.463784,\ 0.417731)
\end{gather*}
are exact decimals, gives $24$ negative exact rationals $\Delta_j$, the largest
being $\Delta_{20}\approx-1.95\times10^{-4}$. Together with the witness,
\[
 0.2161\ \le\ p_B\ \le\ 0.2210000\qquad(d=12,\ t=3/25).
\]
The finite certificate is needed: with \eqref{eq:ov-closed-floor} in place of
$p$, the score of Theorem~\ref{thm:ov-criterion} at this point would be about
$1.02$.

Only failures enter the exploration, and they enter through $\mathcal E_v^+$.
The following undivided form of the floor is the one used in
Section~\ref{sec:ov-assembly}.

\begin{lemma}[Bond floor after failures]\label{lem:ov-torus-floor}
Let $0\le p\le p_B$. For every ordered pair $(u,v)$ of neighbors in $\mathbb T_L$, every
$F\subseteq N(v)\setminus\{u\}$ and every $H\in\mathcal E_v^+$ with
$H\subseteq\{q_u=q_v=1\}$,
\begin{equation}\label{eq:ov-floor-undivided}
 \nu_L\bigl(H\cap\mathsf F_v(F)\cap\{B_{uv}=1\}\bigr)
 \ \ge\ p\,\nu_L\bigl(H\cap\mathsf F_v(F)\bigr).
\end{equation}
\end{lemma}

\begin{proof}
Since $\mathcal E_v^+\subseteq\mathcal G_v$, the event $H$ is a union of atoms
$g$ of $\mathcal G_v$, each contained in $\{q_u=q_v=1\}$ and each of positive
probability by Lemma~\ref{lem:ov-gauge}(ii). On $g$, edges to $N_{\ne}(v)$ are
never blue, so $g\cap\mathsf F_v(F)=g\cap\mathsf F_v(F')$ with
$F'=F\cap N_=(v)$, and $u\in N_=(v)\setminus F'$. By
Proposition~\ref{prop:ov-star-law} and \eqref{eq:ov-pB},
$\nu_L(B_{uv}=1\mid g\cap\mathsf F_v(F'))\ge p_B\ge p$. Multiply by
$\nu_L(g\cap\mathsf F_v(F))$ and sum over the atoms $g\subseteq H$.
\end{proof}

\subsection{Overlap odds and domination by a weak Ising field}\label{sec:ov-holley}

For $K'\ge0>h$ let $\mu_L$ be the Ising law on $\{\pm1\}^V$,
\begin{equation}\label{eq:ov-ising}
 \mu_L(X)\propto\exp\Bigl(K'\sum_{xy\in E}X_xX_y+h\sum_{x\in V}X_x\Bigr),
 \qquad
 \gamma_v(1\mid\eta)=\frac{e^{2h+2K'S_v(\eta)}}{1+e^{2h+2K'S_v(\eta)}},
\end{equation}
with $S_v(\eta)=\sum_{w\in N(v)}\eta_w$; here $\gamma_v$ is its single-site
conditional law, which is nondecreasing in $\eta$. The parameters are given
through the rationals $g_K=e^{2K'}$ and $g_h=e^{2h}$.

\begin{proposition}[Overlap odds]\label{prop:ov-odds}
Let $\eta\in\{\pm1\}^{V\setminus v}$, $N_\pm=\{w\in N(v):\eta_w=\pm1\}$,
$k_\pm=|N_\pm|$ and $S=k_+-k_-$. For gauge signs $s_{\rm off}$ off the star of
$v$, put $\Upsilon_\pm=C_\beta^{k_\pm}\E_{\pi^{N_\pm}}[R_\chi^{-2}]$ with
$\chi=\chi_{v,s_{\rm off}}$. Then
\[
 \nu_L\bigl(q_v=1\bigm| q_{V\setminus v}=\eta,\ s_{\rm off},\
 (A_e)_{e\notin\starv(v)}\bigr)
 =\frac{\Upsilon_+}{\Upsilon_++\Upsilon_-},\qquad
 \frac{\Upsilon_+}{\Upsilon_-}=C_\beta^{S}\,
 \frac{\E_{\pi^{N_+}}[R_\chi^{-2}]}{\E_{\pi^{N_-}}[R_\chi^{-2}]}.
\]
\end{proposition}

\begin{proof}
The activations are independent of $(s,q)$. By Lemma~\ref{lem:ov-gauge}(ii)--(iii)
and \eqref{eq:ov-star-factor}, summing over $\sigma$ and over the star signs
$\varsigma$ gives, for $\epsilon=\pm1$,
\[
 \nu_L(q_v=\epsilon,\ q_{V\setminus v}=\eta,\ s_{\rm off})
 =c(\eta,s_{\rm off})\sum_{\varsigma}\ \prod_{w:\,\eta_w=\epsilon}
 e^{2\beta\varsigma_w}\,R_\chi(\varsigma)^{-2},
\]
with $c$ independent of $\epsilon$, because
$e^{\beta\varsigma_w(1+\epsilon\eta_w)}=e^{2\beta\varsigma_w\mathbf1\{\eta_w=\epsilon\}}$.
Since $e^{2\beta\varsigma}/2$ equals $C_\beta a$ or $C_\beta(1-a)$, the sum
equals $2^{2d}\Upsilon_\epsilon$.
\end{proof}

The ratio must be bounded below uniformly over $\chi\in\mathcal N_{2d}$. After
this infimum only $k=k_+$ matters. We bound the numerator and the denominator
separately. Fix $k$ and $k_-=2d-k$, and for $z\equiv k\pmod 2$ put
$\Omega_k(z)=\E[2\cosh(\beta(z+\Sigma_{\rm f}))]$, where $\Sigma_{\rm f}$ is a sum of
$k_-$ independent fair signs. It is a finite average of values
$2\cosh(\beta\cdot\text{even integer})$, hence rational for rational $t$.

\begin{lemma}[Numerator]\label{lem:ov-numerator}
Let $c_0,\ldots,c_k>0$, and for $j=0,\ldots,k$ let
\[
 \widetilde X_k(j)=\sum_{i=0}^{j}\sum_{i'=0}^{k-j}\binom ji\binom{k-j}{i'}
 a^{i+i'}(1-a)^{k-i-i'}\Bigl[\frac3{c_{i+i'}^2}
 -\frac{2\,\Omega_k(\ell_+-\ell_-)}{c_{i+i'}^3}\Bigr],
\]
where $\ell_+=2i-j$ and $\ell_-=2i'-(k-j)$. Then $\E_{\pi^{N_+}}[R_\chi^{-2}]\ge\mathrm{Num}_k:=\min_j\widetilde X_k(j)$
for every $\chi\in\mathcal N_{2d}$.
\end{lemma}

\begin{proof}
Let $c(\varsigma)=c_i$, where $i$ is the number of $w\in N_+$ with
$\varsigma_w=1$. The tangent-line bound at $c(\varsigma)$, linearity of $R_\chi$
in $\chi$, and Fubini give
$\E_{\pi^{N_+}}R_\chi^{-2}\ge\E_\chi\E_{\pi^{N_+}}[3c^{-2}-2W_\rho c^{-3}]$ with
$W_\rho(\varsigma)=2\cosh(\beta\sum_w\varsigma_w\rho_w)$, and a mixture is at
least its worst frozen pattern. For a frozen pattern $\rho$, averaging $W_\rho$
over the fair coordinates on $N_-$ gives $\Omega_k(x)$ with
$x=\sum_{w\in N_+}\varsigma_w\rho_w$, whatever $\rho_{N_-}$ is, while
$c(\varsigma)$ depends only on $\varsigma_{N_+}$. If $\rho$ has $j$ coordinates
equal to $1$ on $N_+$, then $x=\ell_+-\ell_-$ as in
Proposition~\ref{prop:ov-vertex}, with $i+i'$ plus signs on $N_+$.
\end{proof}

\begin{lemma}[Denominator]\label{lem:ov-denominator}
Let $Y_k(j)=\E[(2\cosh\beta Z_j)^{-2}]$, where $Z_j$ is a sum of $j$ tilted,
$k_--j$ anti-tilted (equal to $1$ with probability $1-a$) and $k$ fair
independent signs. Then
\[
 \sup_{\chi\in\mathcal N_{2d}}\E_{\pi^{N_-}}\bigl[R_\chi^{-2}\bigr]
 =\mathrm{Den}_k:=Y_k(\lfloor k_-/2\rfloor).
\]
\end{lemma}

\begin{proof}
The map $\chi\mapsto\E_{\pi^{N_-}}R_\chi^{-2}$ is convex, so its supremum over
$\mathcal N_{2d}$ is attained at a frozen law. There
$\sum_w\varsigma_w\rho_w$ has the law of $Z_j$, where $j$ is the number of plus
coordinates of $\rho$ on $N_-$. It remains to show that $j\mapsto Y_k(j)$ is
maximal at $\lfloor k_-/2\rfloor$. Put $f(z)=\sech^2(\beta z)$, even and
nonincreasing in $|z|$, so that $Y_k(j)=\frac14\E f(Z_j)$ and
$Y_k(j)=Y_k(k_--j)$. For $2j+1\le k_-$ write $Z_j=\zeta+\varepsilon^-$ and
$Z_{j+1}\overset d=\zeta+\varepsilon^+$, with one anti-tilted sign
$\varepsilon^-$, one tilted sign $\varepsilon^+$ and a common independent part
$\zeta$. Then
\[
 Y_k(j+1)-Y_k(j)=\tfrac{2a-1}4\,\E\hat f(\zeta),\qquad
 \hat f(z)=f(z+1)-f(z-1),
\]
and $\hat f$ is odd and nonnegative on $(-\infty,0]$. Decompose
$\zeta=\zeta_{\rm s}+\zeta_{\rm a}$, where $\zeta_{\rm s}$ is the sum of $j$
pairs of one tilted and one anti-tilted sign and of $k$ fair signs, and
$\zeta_{\rm a}$ is a sum of $k_--2j-1$ anti-tilted signs. The law
$\mathsf p$ of $\zeta_{\rm s}$ is symmetric and unimodal on its lattice, that is,
$\mathsf p(y)\ge\mathsf p(y+2)$ for $y\ge0$. Indeed, this holds for a point
mass at $0$ and is preserved by adding a fair sign, which gives
$\mathsf p'(x)=\frac12(\mathsf p(x-1)+\mathsf p(x+1))$, and by adding a pair,
which takes the values $\pm2$ with probability $a(1-a)\le\frac14$
each. In both cases, for $x\ge0$, $\mathsf p'(x)-\mathsf p'(x+2)$ is a
combination of differences $\mathsf p(y)-\mathsf p(y+2)$ with $y\ge0$ and
nonnegative coefficients; for a pair at $x=0$ the coefficients are
$1-3a(1-a)$ and $a(1-a)$. For $z<0$,
\[
 \Pp(\zeta=z)-\Pp(\zeta=-z)=\sum_{y>0}\bigl[\Pp(\zeta_{\rm a}=y)-\Pp(\zeta_{\rm a}=-y)\bigr]
 \bigl[\Pp(\zeta_{\rm s}=z-y)-\Pp(\zeta_{\rm s}=z+y)\bigr]\ \ge\ 0,
\]
since $\Pp(\zeta_{\rm a}=y)/\Pp(\zeta_{\rm a}=-y)=((1-a)/a)^y<1$ and
$|z-y|>|z+y|$. Hence
$\E\hat f(\zeta)=\sum_{z<0}[\Pp(\zeta=z)-\Pp(\zeta=-z)]\hat f(z)\ge0$, and since
$a>\frac12$, $Y_k$ is nondecreasing on $\{j:2j+1\le k_-\}$. With the symmetry
$Y_k(j)=Y_k(k_--j)$, its maximum is at $\lfloor k_-/2\rfloor$.
\end{proof}

The same argument, for any even function that is nonincreasing in $|z|$, gives
Lemma~\ref{lem:fs-balanced-tilt} of Appendix~\ref{app:fresh-star}, whose tilted
signs have $\pi_+=a$.

Consequently $\Upsilon_+/\Upsilon_-\ge C_\beta^S\mathrm{Num}_k/\mathrm{Den}_k$ for
every cavity law. The numerator bound depends on tangent points; any positive
values are admissible. The certificates take $c_i$ to be the value of $R_\chi$,
averaged over the fair coordinates, at a cavity law that mixes the aligned
pattern with the uniform mixture of the $k$ patterns flipping one coordinate of
$N_+$. With $S_i=2i-k$ and mixing weight $w_k\in[0,1)$ this value is
\[
 c_i=(1-w_k)\,\Omega_k(S_i)
 +w_k\Bigl[\tfrac ik\,\Omega_k(S_i-2)+\tfrac{k-i}k\,\Omega_k(S_i+2)\Bigr].
\]
At $d=12$, $t=3/25$ the certificate uses $w_k=0$ for $k\le15$ and
\begin{gather*}
 (w_{16},\ldots,w_{20})=(0.0442,\ 0.0964,\ 0.1519,\ 0.2104,\ 0.2718),\\
 (w_{21},\ldots,w_{24})=(0.3359,\ 0.4027,\ 0.4719,\ 0.5435).
\end{gather*}

\begin{proposition}[Certified Holley line]\label{prop:ov-holley-line}
For each row of Table~\ref{tab:overlap-constants}, and for each row of
Table~\ref{tab:ms-local} with the local inputs of this section, there are
tangent points in Lemma~\ref{lem:ov-numerator} such that, for every
$k=0,\ldots,2d$, with $S=2k-2d$,
\begin{equation}\label{eq:ov-holley-line}
 C_\beta^{S}\,\mathrm{Num}_k/\mathrm{Den}_k\ \ge\ g_h\,g_K^{S}.
\end{equation}
Consequently, on every torus $\mathbb T_L$ with $L\ge3$, for every vertex $v$, every
$\eta\in\{\pm1\}^{V\setminus v}$ and every $s_{\rm off}$,
\[
 \nu_L\bigl(q_v=1\bigm| q_{V\setminus v}=\eta,\ s_{\rm off},\
 (A_e)_{e\notin\starv(v)}\bigr)\ \ge\ \gamma_v(1\mid\eta).
\]
\end{proposition}

\begin{proof}
Each inequality \eqref{eq:ov-holley-line} is a comparison of exact rationals;
Sections~\ref{sec:ov-constants} and~\ref{sec:ms-certificates} record its
evaluation in all $2d+1$ environments.
Then $\Upsilon_+/\Upsilon_-\ge g_hg_K^S=e^{2h+2K'S}$ by
Lemmas~\ref{lem:ov-numerator} and~\ref{lem:ov-denominator}, and
Proposition~\ref{prop:ov-odds} gives the claim, since $x\mapsto x/(1+x)$ is
increasing.
\end{proof}

At $d=12$, $t=3/25$ we have $g_K=1.02634$ and $g_h=0.83527$, so
$K'\approx0.013000$ and $h\approx-0.090000$. The binding environments are
$S=\pm2d$: the logarithmic slack in \eqref{eq:ov-holley-line} is at least
$0.001827$ at $S=24$ and $0.00192$ at $S=-24$, and at least $0.0057$ in the
other $23$ environments.

\begin{lemma}[Single-site bound]\label{lem:ov-torus-odds}
Under \eqref{eq:ov-holley-line}, for every vertex $v$ and every
$H\in\mathcal E_v$,
\[
 \nu_L\bigl(H\cap\{q_v=1\}\bigr)\ \ge\ \E_{\nu_L}\bigl[\mathbf1_H\,
 \gamma_v(1\mid q_{N(v)})\bigr].
\]
\end{lemma}

\begin{proof}
By Lemma~\ref{lem:ov-gauge}(iv), $\mathcal E_v$ is contained in
$\sigma((q_x)_{x\ne v},(s_e,A_e)_{e\notin\starv(v)})$, and
$\gamma_v(1\mid q_{N(v)})$ is $\mathcal E_v$-measurable. Apply
Proposition~\ref{prop:ov-holley-line} and the tower property.
\end{proof}

The comparison with the Ising field uses Holley's coupled heat bath
\cite{Holley1974}. We need a one-sided version, a routine variant of
\cite[Theorem~4.8]{GHM2001}: the dominating law may fail to
be strictly positive, since in Section~\ref{sec:ov-assembly} it is conditioned on
an event, and the dominated Ising law may carry a minus boundary condition. For
finite $\Lambda\subseteq\Z^d$, write $\mu^-_\Lambda$ for the Ising law with
parameters $(K',h)$ on $\Lambda$ and all spins outside $\Lambda$ fixed to $-1$.
On the torus with $\Lambda=V$ there is no boundary and $\mu^-_V=\mu_L$.

\begin{lemma}[One-sided Holley domination]\label{lem:ov-holley}
Let $\Lambda\subseteq\bar\Lambda$ be finite vertex sets of $\mathbb T_L$ or of $\Z^d$
with $N(v)\subseteq\bar\Lambda$ for all $v\in\Lambda$. Let $\pi$ be a law on
$\{\pm1\}^{\bar\Lambda}$ such that, for every $v\in\Lambda$ and every
$\xi\in\{\pm1\}^{\bar\Lambda\setminus v}$ with $\pi(\xi)>0$,
$\pi(q_v=1\mid\xi)\ge\gamma_v(1\mid\xi_{N(v)})$. Then there is a coupling of
$X\sim\mu^-_\Lambda$ and $q\sim\pi$ with $X\le q|_\Lambda$ almost surely.
\end{lemma}

\begin{proof}
Run a coupled random-scan heat bath: at each step choose $v\in\Lambda$
uniformly and an independent uniform $U\in[0,1]$, set $X_v=1$ if and only if
$U\le\gamma_v(1\mid X_{N(v)\cap\Lambda},-1\text{ on }N(v)\setminus\Lambda)$, and
set $q_v=1$ if and only if $U\le\pi(q_v=1\mid q_{\bar\Lambda\setminus v})$.
Start from $X\equiv-1$ and $q\sim\pi$. The $q$-chain is a Gibbs sampler of
$\pi$, so $q$ has law $\pi$ at all times, and its conditional probabilities are
evaluated only at configurations of positive mass. If $X\le q|_\Lambda$ before a
step, then $(X,-1\text{ off }\Lambda)\le q$ on $N(v)$, and monotonicity of
$\gamma_v$ and the hypothesis give $X\le q|_\Lambda$ after it. The $X$-chain
alone is the heat bath of the strictly positive law $\mu^-_\Lambda$. It is
irreducible and aperiodic on a finite space, so its law converges to
$\mu^-_\Lambda$ \cite[Theorem~4.9]{LPW2017}. Any limit point of the joint laws
is a coupling of $\mu^-_\Lambda$ and $\pi$ supported on the closed set
$\{X\le q|_\Lambda\}$.
\end{proof}

\subsection{Sharper local inputs from noisy cavities}\label{sec:ov-noisy}

Both local inputs are uniform over the class $\mathcal N_{2d}$ of all
flip-invariant laws of the neighboring spins, and both certificates are decided
at frozen laws $\frac12(\delta_\rho+\delta_{-\rho})$:
Proposition~\ref{prop:ov-vertex} and Lemma~\ref{lem:ov-numerator} bound a
mixture by its worst frozen pattern, and the supremum in
Lemma~\ref{lem:ov-denominator} is attained at a frozen law. A cavity law is
never frozen. Given the spins at
distance two from the target, its neighbors are independent, and each has
conditional mean at most $\tanh((2d-1)\beta)$ in absolute value, because it has
only $2d-1$ other neighbors. Such a neighbor acts on the target through a noisy
binary channel, as an Ising edge does on a tree
\cite[Section~3]{Mossel2004}, and this noise can be moved into the star
interaction. The cavity laws are thereby replaced by a smaller class at a
smaller effective coupling, at the price of a constant factor that cancels in
every ratio we use. This subsection proves the reduction and gives a
certificate for each local input over the smaller class. With both, the
oriented criterion is certified in dimension $9$ (Table~\ref{tab:ov-noisy}),
and the macrostep criterion of Section~\ref{sec:macrostep} in dimension~$7$,
where it is not certified without them (Remark~\ref{rem:ms-inputs}).

For $\beta'>0$ and a law $\chi$ on $\{\pm1\}^{N(v)}$ put
\[
 R^{\beta'}_\chi(\varsigma)=\E_{\rho\sim\chi}\Bigl[2\cosh\Bigl(\beta'
 \sum_{w\in N(v)}\varsigma_w\rho_w\Bigr)\Bigr],
\]
so that $R_\chi=R^\beta_\chi$ in \eqref{eq:ov-star-factor}. Like $R_\chi$, it
is linear in $\chi$ and even in $\varsigma$. Put
\begin{equation}\label{eq:ov-betac}
 \vartheta_d=\tanh\bigl((2d-1)\beta\bigr),\qquad
 \beta_c=\operatorname{atanh}(t\vartheta_d),\qquad
 e^{2\beta_c}=\frac{1+t\vartheta_d}{1-t\vartheta_d};
\end{equation}
$\vartheta_d$ and $e^{2\beta_c}$ are rational functions of $t$, and
$\beta_c<\beta$.

\begin{lemma}[Noisy-cavity reduction]\label{lem:ov-noisy}
Let $L\ge4$, let $v$ be a vertex of $\mathbb T_L$, let $s_{\rm off}$ be gauge
signs off $\starv(v)$, let $\chi=\chi_{v,s_{\rm off}}$, and let
$\beta'\in[\beta_c,\beta]$. There is $\tilde\chi\in\mathcal N_{2d}$ such that
\[
 R_\chi(\varsigma)=c_0\,R^{\beta'}_{\tilde\chi}(\varsigma)\quad
 (\varsigma\in\{\pm1\}^{N(v)}),\qquad
 c_0=\Bigl(\frac{\cosh\beta}{\cosh\beta'}\Bigr)^{2d}.
\]
The constant $c_0$ depends neither on $\varsigma$ nor on $v$ or $s_{\rm off}$.
\end{lemma}

\begin{proof}
Put $Z=V\setminus(\{v\}\cup N(v))$. Two distinct neighbors of $v$ are not
adjacent, since $L\ge4$: $v+e_i$ and $v\pm e_j$ with $j\ne i$ differ in two
coordinates, and $v+e_i$ and $v-e_i$ differ by $2e_i$, which is a lattice step
modulo $L$ only if $L\le3$. Moreover the $2d$ neighbors of $w=v+se_i$ are
distinct, and only $w-se_i$ equals $v$, since $w+se_i=v+2se_i\ne v$. Hence
every $w\in N(v)$ has $2d-1$ neighbors other than $v$, all in $Z$. (For $L=3$
this fails: $v+e_1$ and $v-e_1$ are adjacent.)

Let $\rho$ have the zero-field Gibbs law on $\mathbb T_L\setminus v$ with
couplings $s_{\rm off}$, so that $\chi$ is the law of $\rho_{N(v)}$. Given
$\rho_Z$, the Gibbs weight factorizes over $w\in N(v)$, so the spins $\rho_w$
are conditionally independent with means $m_w=\tanh(\beta h_w)$, where
$h_w=\sum_{z\sim w,\,z\ne v}s_{wz}\rho_z$; thus $|m_w|\le\vartheta_d$. Since
$e^{\beta y}=\cosh\beta\,(1+ty)$ for $y=\pm1$, for $\iota=\pm1$
\[
 \E\Bigl[e^{\iota\beta\sum_w\varsigma_w\rho_w}\Bigm|\rho_Z\Bigr]
 =\cosh^{2d}\beta\,\prod_{w\in N(v)}\bigl(1+\iota t\varsigma_wm_w\bigr).
\]
Put $t_{\beta'}=\tanh\beta'\ge t\vartheta_d\ge|tm_w|$. Given $\rho_Z$, let
$(\varepsilon_w)_{w\in N(v)}$ be independent signs with
$\E[\varepsilon_w\mid\rho_Z]=tm_w/t_{\beta'}\in[-1,1]$. The map
$\varepsilon\mapsto1+\iota t_{\beta'}\varsigma_w\varepsilon$ is affine, and
$1+t_{\beta'}y=e^{\beta'y}/\cosh\beta'$ for $y=\pm1$, so conditional
independence gives
\[
 \prod_w\bigl(1+\iota t\varsigma_wm_w\bigr)
 =\E\Bigl[\prod_w\bigl(1+\iota t_{\beta'}\varsigma_w\varepsilon_w\bigr)\Bigm|\rho_Z\Bigr]
 =\frac{\E\bigl[e^{\iota\beta'\sum_w\varsigma_w\varepsilon_w}\bigm|\rho_Z\bigr]}
 {\cosh^{2d}\beta'}.
\]
Average over $\rho_Z$ and add the cases $\iota=\pm1$. With $\hat\chi$ the law of
$\varepsilon$, this gives
\[
 R_\chi(\varsigma)=c_0\,\E_{\hat\chi}\Bigl[2\cosh\Bigl(\beta'\sum_w\varsigma_w\varepsilon_w\Bigr)\Bigr].
\]
Replace $\hat\chi$ by its flip symmetrization
$\tilde\chi=\frac12(\hat\chi+\hat\chi\circ(-\mathrm{id})^{-1})\in\mathcal N_{2d}$,
which does not change the right side because $\cosh$ is even.
\end{proof}

The proof represents a neighbor of bias at most $\vartheta_d$, seen through an
edge of second eigenvalue $t$, as a sign of bias at most one seen through an
edge of second eigenvalue $\tanh\beta'\ge t\vartheta_d$; in the language of
\cite[Section~3]{Mossel2004}, this degrades one binary symmetric channel into
another. With $\beta'=\beta$ one may take $\tilde\chi=\chi$. Only the geometry
of the star uses $L\ge4$.

Define $r_{\beta'}(\chi;\mathsf T)$ as $r(\chi;\mathsf T)$ in
\eqref{eq:ov-pB}, with $R^{\beta'}_\chi$ in place of $R_\chi$, and keep $a$,
$C_\beta$ and $p_A$ at the physical $\beta$.

\begin{corollary}[Transfer to the noisy class]\label{cor:ov-noisy}
Let $L\ge4$ and $\beta'\in[\beta_c,\beta]$.
\begin{enumerate}\renewcommand{\labelenumi}{(\roman{enumi})}
\item Proposition~\ref{prop:ov-star-law} holds with $R_\chi$ replaced by
$R^{\beta'}_{\tilde\chi}$, where $\tilde\chi\in\mathcal N_{2d}$ is the law of
Lemma~\ref{lem:ov-noisy} for $\chi=\chi_{v,s_{\rm off}}$. In particular, every
conditional opening probability \eqref{eq:ov-exploration-identity} is at
least $p_A\inf_{\chi}r_{\beta'}(\chi;N(v))$, the infimum being over
$\chi\in\mathcal N_{2d}$.
\item Proposition~\ref{prop:ov-odds} holds with $R_\chi$ replaced by
$R^{\beta'}_{\tilde\chi}$. In particular, every lower bound on
\[
 C_\beta^S\,\frac{\E_{\pi^{N_+}}\bigl[(R^{\beta'}_\chi)^{-2}\bigr]}
 {\E_{\pi^{N_-}}\bigl[(R^{\beta'}_\chi)^{-2}\bigr]}
\]
that holds for all $\chi\in\mathcal N_{2d}$ is a lower bound on
$\Upsilon_+/\Upsilon_-$.
\end{enumerate}
\end{corollary}

\begin{proof}
In Propositions~\ref{prop:ov-star-law} and~\ref{prop:ov-odds}, $R_\chi$ enters
only through the factor $R_\chi(\varsigma)^{-2}$ of a weight that is normalized
over $\varsigma$, or of a ratio of two such sums. The atom $g$ of $\mathcal G_v$
in Proposition~\ref{prop:ov-star-law}, and the conditioning in
Proposition~\ref{prop:ov-odds}, fix $s_{\rm off}$ and hence $\chi$; by
Lemma~\ref{lem:ov-noisy} the factor equals
$c_0^{-2}R^{\beta'}_{\tilde\chi}(\varsigma)^{-2}$, and $c_0$ cancels. The tilts,
the constant $C_\beta$, the activation probability and the failed-edge factor
\eqref{eq:ov-failed-fair} do not involve $R_\chi$ and stay at $\beta$. For the
bound in (i), the proofs of Lemma~\ref{lem:ov-mediant} and
Theorem~\ref{thm:ov-extremal} use only linearity of $\chi\mapsto R_\chi$, the
identity $R_\chi(\varsigma_u,-\varsigma_w,\varsigma_{\rm r})
=R_{\chi^{(w)}}(\varsigma_u,\varsigma_w,\varsigma_{\rm r})$, closure of
$\mathcal N_{2d}$ under coordinate flips, and the product structure of
$\pi^{\mathsf T}$. All four hold for $R^{\beta'}$, since flipping $\varsigma_w$
has the same effect as flipping $\rho_w$ inside the hyperbolic cosine. Hence
$\inf_{\mathsf T\ni u}\inf_\chi r_{\beta'}(\chi;\mathsf T)=
\inf_\chi r_{\beta'}(\chi;N(v))$, and $\tilde\chi\in\mathcal N_{2d}$.
\end{proof}

\begin{remark}[The noisy class is smaller]\label{rem:ov-nesting}
For $0<\beta_1\le\beta_2$ and every law $\chi$ on $\{\pm1\}^{N(v)}$ there is a
law $\chi''$, flip-invariant if $\chi$ is, with
$R^{\beta_2}_{\chi''}=(\cosh\beta_2/\cosh\beta_1)^{2d}R^{\beta_1}_\chi$:
multiply each coordinate of $\rho\sim\chi$ by an independent sign of mean
$\tanh\beta_1/\tanh\beta_2$ and repeat the computation in the proof of
Lemma~\ref{lem:ov-noisy}. With $\beta_1=\beta'$ and $\beta_2=\beta$, every
function $R^{\beta'}_\chi$ with $\chi\in\mathcal N_{2d}$ is thus a constant
multiple of some $R_{\chi''}$ with $\chi''\in\mathcal N_{2d}$. The quantities
certified below are invariant under $R\mapsto cR$ for constants $c>0$, so
their infimum over the noisy class is at least their infimum over
$\mathcal N_{2d}$; by the same argument, the noisy classes shrink as $\beta'$
decreases. At every point where the sharpened inputs are used, in
Table~\ref{tab:ov-noisy} and in Section~\ref{sec:macrostep}, the inclusion is
strict: the certified floor $p$ exceeds the aligned frozen value
$p_Ar(\frac12(\delta_{\mathbf1}+\delta_{-\mathbf1});N(v))$, which bounds $p_B$
from above. At $(d,t)=(9,3/20)$, for instance, $p=0.2689$, while the aligned
frozen value is at most $0.268766$ with $\beta$ inside $R$ and at least
$0.269100$ with $\beta_c$ inside $R$; Section~\ref{sec:ms-certificates} gives
the values in dimensions $7$ and $8$. So the sharpened floors lie above every
floor that is uniform over $\mathcal N_{2d}$.
\end{remark}

\paragraph{An exact pair certificate for the floor.}
By Corollary~\ref{cor:ov-noisy}(i) it suffices to bound
$r_{\beta'}(\chi;N(v))$ from below over $\chi\in\mathcal N_{2d}$. With
$\bar n=2d-1$ as in \eqref{eq:ov-XY}, put $\mathsf X=\{-\bar n,-\bar n+2,\ldots,\bar n\}$,
$f_{\beta'}(y)=2\cosh(\beta'y)$ and
\[
 \mathbf p_x=\bigl(f_{\beta'}(x-1),\,f_{\beta'}(x+1)\bigr)\in(0,\infty)^2
 \qquad(x\in\mathsf X).
\]
Let $p_A/2<p<\tanh2\beta$, and put $\underline r=p/p_A$ and
$\bar\lambda=e^{4\beta}(1-\underline r)/\underline r$ as in
Proposition~\ref{prop:ov-vertex}. For odd $U\in\{1,3,\ldots,\bar n\}$ put
\[
 \alpha_U=e^{2\beta U}-\bar\lambda e^{-2\beta U},\qquad
 \gamma_U=\bar\lambda e^{2\beta U}-e^{-2\beta U},\qquad
 \varphi_U(A,B)=\alpha_UA^{-2}-\gamma_UB^{-2}.
\]
Both coefficients are positive. Since $\underline r>\frac12$, we have
$\bar\lambda<e^{4\beta}\le e^{4\beta U}$, so $\alpha_U>0$; since
$\underline r<a$ and $(1-a)/a=e^{-4\beta}$, we have $\bar\lambda>1$, so
$\gamma_U>e^{2\beta U}-e^{-2\beta U}>0$.

\begin{lemma}[Chords suffice]\label{lem:ov-chords}
Fix $U$, and let rationals $\psi(x)$, $x\in\mathsf X$, satisfy
\begin{equation}\label{eq:ov-chord}
 \vartheta\psi(x)+(1-\vartheta)\psi(x')\ \ge\
 \varphi_U\bigl(\vartheta\mathbf p_x+(1-\vartheta)\mathbf p_{x'}\bigr)
 \qquad(x,x'\in\mathsf X,\ \vartheta\in[0,1]).
\end{equation}
Then $\sum_x\pi(x)\psi(x)\ge\varphi_U\bigl(\sum_x\pi(x)\mathbf p_x\bigr)$ for
every probability law $\pi$ on $\mathsf X$.
\end{lemma}

\begin{proof}
Let $\zeta=\sum_x\pi(x)\mathbf p_x$. Minimize $\sum_x\pi'(x)\psi(x)$ over the
probability laws $\pi'$ on $\mathsf X$ with
$\sum_x\pi'(x)\mathbf p_x=\zeta$. This linear program has three equality
constraints and a nonempty compact feasible set, so the minimum is attained at
a vertex $\pi^*$. The columns
$(1,\mathbf p_x)$ on the support of a vertex are linearly independent, so
$\pi^*$ has at most three support points, with affinely independent
$\mathbf p_x$. With one or two support points, the claim for $\pi^*$ is
\eqref{eq:ov-chord}. With three, let $T$ be the triangle with vertices
$\mathbf p_{x_i}$, and let $l$ be the affine function with
$l(\mathbf p_{x_i})=\psi(x_i)$, so that $\sum_x\pi^*(x)\psi(x)=l(\zeta)$. On each
edge of $T$, $l$ is a chord, so $\varphi_U-l\le0$ there by \eqref{eq:ov-chord}.
Since $\partial_A^2\varphi_U=6\alpha_UA^{-4}>0$, the function $\varphi_U-l$ is
convex along every horizontal line $\{B=\mathrm{const}\}$, and each horizontal
section of $T$ is a segment with endpoints on its edges, or a single point of an
edge. Hence $\varphi_U(\zeta)\le l(\zeta)=\sum\pi^*\psi\le\sum\pi\psi$.
\end{proof}

Conversely, \eqref{eq:ov-chord} is the case of two-point laws, so the lemma
loses nothing. It is the two-dimensional case of the fact that a lower convex
envelope is determined by simplices (Carath\'eodory's theorem), sharpened by
convexity in one direction.

\begin{lemma}[Pair certificate for the floor]\label{lem:ov-pair}
Let $\beta'\in[\beta_c,\beta]$ and $p_A/2<p<\tanh2\beta$, and suppose that
rationals $\psi_U(x)$, for odd $U\in\{1,\ldots,\bar n\}$ and $x\in\mathsf X$,
satisfy \eqref{eq:ov-chord} for every $U$ and
\begin{equation}\label{eq:ov-class-sums}
 \Delta^{\rm P}_j:=\sum_{U}\ \sum_{h}\binom{\bar n-j}{h}\binom{j}{k_U-h}\,
 \psi_U\bigl(U+2(\bar n-j)-4h\bigr)\ \le\ 0\qquad(j=0,\ldots,\bar n),
\end{equation}
where $k_U=(\bar n+U)/2$ and $h$ ranges over
$\max(0,k_U-j)\le h\le\min(k_U,\bar n-j)$. Then
$r_{\beta'}(\chi;N(v))\ge p/p_A$ for every $\chi\in\mathcal N_{2d}$.
Consequently, on every torus $\mathbb T_L$ with $L\ge4$, every conditional
opening probability \eqref{eq:ov-exploration-identity} is at least $p$.
\end{lemma}

\begin{proof}
Write $R=R^{\beta'}_\chi$. As in \eqref{eq:ov-XY},
$r_{\beta'}(\chi;N(v))=aX/(aX+(1-a)Y)$ with
$X=\E_{\pi'}[R(1,\varsigma')^{-2}]$ and $Y=\E_{\pi'}[R(-1,\varsigma')^{-2}]$;
since $a/(1-a)=e^{4\beta}$, $r_{\beta'}\ge\underline r$ if and only if
$\mathcal Y:=Y-\bar\lambda X\le0$. The tilted law is
$\pi'(\varsigma')=(a(1-a))^{\bar n/2}e^{2\beta U(\varsigma')}$ with
$U(\varsigma')=\sum_x\varsigma'_x$, which is odd because $\bar n$ is odd. Use
$R(-1,\varsigma')=R(1,-\varsigma')$, substitute $\varsigma'\to-\varsigma'$ in
$Y$, and pair $\varsigma'$ with $-\varsigma'$:
\[
 \mathcal Y=(a(1-a))^{\bar n/2}\sum_{U\ge1}\ \sum_{\varsigma':\,U(\varsigma')=U}
 \varphi_U\bigl(R(1,-\varsigma'),\,R(1,\varsigma')\bigr).
\]
Let $\chi'$ be the law of $\rho'=(\rho_x)_{x\ne u}$ given $\rho_u=1$ under
$\chi$. By flip invariance and since $f_{\beta'}$ is even,
$R(1,\varsigma')=\E_{\chi'}f_{\beta'}(\varsigma'\cdot\rho'+1)$ and
$R(1,-\varsigma')=\E_{\chi'}f_{\beta'}(\varsigma'\cdot\rho'-1)$. Both
coordinates are averages over the same scalar, so
$(R(1,-\varsigma'),R(1,\varsigma'))=\sum_x\pi_{\varsigma'}(x)\mathbf p_x$, where
$\pi_{\varsigma'}$ is the law of $\varsigma'\cdot\rho'\in\mathsf X$ under
$\chi'$. Lemma~\ref{lem:ov-chords} bounds each term by
$\E_{\chi'}\psi_U(\varsigma'\cdot\rho')$, and summing,
\[
 \mathcal Y\le(a(1-a))^{\bar n/2}\,\E_{\chi'}\Sigma_\psi(\rho'),\qquad
 \Sigma_\psi(\rho')=\sum_{U\ge1}\ \sum_{U(\varsigma')=U}\psi_U(\varsigma'\cdot\rho').
\]
The function $\Sigma_\psi$ is invariant under permutations of the coordinates,
so it depends only on the number $j$ of plus coordinates of $\rho'$. If
$\varsigma'$ has $k_U$ plus signs, $h$ of them on the $\bar n-j$ minus
coordinates of $\rho'$, then $\varsigma'\cdot\rho'=U+2(\bar n-j)-4h$, and there
are $\binom{\bar n-j}{h}\binom{j}{k_U-h}$ such $\varsigma'$. Hence
$\Sigma_\psi(\rho')=\Delta^{\rm P}_j\le0$ and $\mathcal Y\le0$. The last assertion
follows from Corollary~\ref{cor:ov-noisy}(i).
\end{proof}

Lemma~\ref{lem:ov-torus-floor} uses $p\le p_B$ only through the bound
$\nu_L(B_{uv}=1\mid g\cap\mathsf F_v(F'))\ge p$. Under the hypotheses of
Lemma~\ref{lem:ov-pair} this bound holds by Corollary~\ref{cor:ov-noisy}(i), so
Lemma~\ref{lem:ov-torus-floor} holds verbatim on every torus with $L\ge4$ for
such $p$, although $p$ may exceed $p_B$. At the level $p_A/2$ the floor is the
lower product bound of \cite[Lemma~4.2]{MNS2008}, as in
\eqref{eq:ov-closed-floor}. Lemma~\ref{lem:ov-pair} loses only in two places:
each pair of opposite sign patterns is relaxed separately, and the class sums
are taken over frozen patterns $\rho'$. At the points of
Table~\ref{tab:ov-noisy} the certified floor is within $0.11\%$ of the aligned
frozen value with $\beta_c$ inside $R$, which is an upper bound on every floor
that is uniform over the noisy class. In the certificates, \eqref{eq:ov-chord}
is decided for each pair $x,x'$ as positivity on $[0,1]$ of the polynomial
$l(\vartheta)A(\vartheta)^2B(\vartheta)^2-\alpha_UB(\vartheta)^2+\gamma_UA(\vartheta)^2$
of degree $5$, where $A$, $B$ and $l$ interpolate $f_{\beta'}(x-1)$,
$f_{\beta'}(x+1)$ and $\psi_U$ linearly. Its coefficients are rational when
$e^{2\beta'}$ is, because $x\pm1$ is even and
$f_{\beta'}(2k)=e^{2\beta'k}+e^{-2\beta'k}$.

\paragraph{A symmetrized certificate for the Holley line.}
By Corollary~\ref{cor:ov-noisy}(ii) it suffices to bound
$C_\beta^S\E_{\pi^{N_+}}[R^{-2}]/\E_{\pi^{N_-}}[R^{-2}]$ from below, for
$R=R^{\beta'}_\chi$ and uniformly over $\chi\in\mathcal N_{2d}$.
Lemmas~\ref{lem:ov-numerator} and~\ref{lem:ov-denominator} bound the numerator
and the denominator separately, each at its own worst frozen law. The following
certificate bounds the signed combination at one law, after cancelling mass
between $\varsigma$ and $-\varsigma$.

Fix an environment $k\in\{0,\ldots,2d\}$, a partition $N(v)=N_+\cup N_-$ with
$|N_+|=k$ and $S=2k-2d$, and a rational $\Lambda'_k>0$. For $\varsigma$ let
$i_\pm(\varsigma)$ be the number of $w\in N_\pm$ with $\varsigma_w=1$, and put
\[
 \delta\pi(\varsigma)=\pi^{N_+}(\varsigma)-\Lambda'_k\,\pi^{N_-}(\varsigma),\qquad
 \delta\pi^{\rm s}(\varsigma)=\tfrac12\bigl(\delta\pi(\varsigma)+\delta\pi(-\varsigma)\bigr);
\]
both depend on $\varsigma$ only through $(i_+,i_-)$. Given tangent points
$c(i_+,i_-)>0$, define for $\rho\in\{\pm1\}^{N(v)}$
\[
 \Gamma_k(\rho)=\sum_\varsigma\delta\pi^{\rm s}(\varsigma)\,g_\varsigma(\rho),\qquad
 g_\varsigma(\rho)=\begin{cases}
 3c^{-2}-2f_{\beta'}(\varsigma\cdot\rho)\,c^{-3},&\delta\pi^{\rm s}(\varsigma)>0,\\
 f_{\beta'}(\varsigma\cdot\rho)^{-2},&\delta\pi^{\rm s}(\varsigma)<0,\\
 0,&\delta\pi^{\rm s}(\varsigma)=0,
 \end{cases}
\]
with $c=c(i_+(\varsigma),i_-(\varsigma))$.

\begin{lemma}[Symmetrized odds certificate]\label{lem:ov-sym-holley}
$\Gamma_k(\rho)$ depends only on the numbers $l_\pm$ of plus coordinates of
$\rho$ in $N_\pm$. If $\Gamma_k(l_+,l_-)\ge0$ for all $0\le l_\pm\le|N_\pm|$,
then
\[
 \E_{\pi^{N_+}}\bigl[(R^{\beta'}_\chi)^{-2}\bigr]\ \ge\
 \Lambda'_k\,\E_{\pi^{N_-}}\bigl[(R^{\beta'}_\chi)^{-2}\bigr]
 \qquad\text{for every }\chi\in\mathcal N_{2d}.
\]
\end{lemma}

\begin{proof}
A permutation of $N(v)$ that preserves $N_+$ and $N_-$ leaves
$\delta\pi^{\rm s}$, the counts $i_\pm$ and the tangent points unchanged and
preserves $\varsigma\cdot\rho$ when applied to both; reindexing the sum shows
that $\Gamma_k$ is invariant under these permutations, which act transitively
on the $\rho$ with given $(l_+,l_-)$. Write $R=R^{\beta'}_\chi$. Since
$R(-\varsigma)=R(\varsigma)$,
\[
 \sum_\varsigma\delta\pi(\varsigma)R(\varsigma)^{-2}
 =\sum_\varsigma\delta\pi^{\rm s}(\varsigma)R(\varsigma)^{-2}.
\]
Put $W_\rho(\varsigma)=f_{\beta'}(\varsigma\cdot\rho)$, so that
$R=\E_\chi W_\rho$. Where $\delta\pi^{\rm s}(\varsigma)>0$, use the tangent-line
bound $y^{-2}\ge3c^{-2}-2yc^{-3}$ for $y,c>0$ at $y=R(\varsigma)$; it is linear
in $y$, so $R(\varsigma)^{-2}\ge\E_\chi[3c^{-2}-2W_\rho(\varsigma)c^{-3}]$.
Where $\delta\pi^{\rm s}(\varsigma)<0$, use Jensen's inequality
$R(\varsigma)^{-2}\le\E_\chi W_\rho(\varsigma)^{-2}$. Multiplying by
$\delta\pi^{\rm s}(\varsigma)$ and summing gives
\[
 \E_{\pi^{N_+}}R^{-2}-\Lambda'_k\E_{\pi^{N_-}}R^{-2}
 =\sum_\varsigma\delta\pi^{\rm s}(\varsigma)R(\varsigma)^{-2}
 \ \ge\ \E_\chi\Gamma_k(\rho)\ \ge\ \min_{l_+,l_-}\Gamma_k(l_+,l_-)\ \ge\ 0.
\]
\end{proof}

The symmetrization cancels mass between $\varsigma$ and $-\varsigma$ before the
one-sided bounds are applied; this, and the use of a single law in both
expectations, is the gain over Lemmas~\ref{lem:ov-numerator}
and~\ref{lem:ov-denominator}. Lemma~\ref{lem:ov-sym-holley} is a symmetrized
form of the single-site odds certificate for Holley's criterion
\cite{Holley1974}, \cite[Theorem~4.8]{GHM2001}.

\begin{corollary}[Sharpened Holley line]\label{cor:ov-sym-holley}
Let rationals $g_K\ge1>g_h$ and $\beta'\in[\beta_c,\beta]$ be given, and suppose
that for every $k=0,\ldots,2d$ the hypothesis of Lemma~\ref{lem:ov-sym-holley}
holds with $\Lambda'_k=g_hg_K^SC_\beta^{-S}$ and some tangent points. Then on
every torus $\mathbb T_L$ with $L\ge4$, for every vertex $v$, every
$\eta\in\{\pm1\}^{V\setminus v}$ and every $s_{\rm off}$,
\[
 \nu_L\bigl(q_v=1\bigm| q_{V\setminus v}=\eta,\ s_{\rm off},\
 (A_e)_{e\notin\starv(v)}\bigr)\ \ge\ \gamma_v(1\mid\eta),
\]
and Lemma~\ref{lem:ov-torus-odds} holds on $\mathbb T_L$.
\end{corollary}

\begin{proof}
By Corollary~\ref{cor:ov-noisy}(ii), applied with the law
$\tilde\chi\in\mathcal N_{2d}$, and Lemma~\ref{lem:ov-sym-holley},
$\Upsilon_+/\Upsilon_-\ge C_\beta^S\Lambda'_k=g_hg_K^S=e^{2h+2K'S}$. Conclude as
in the proof of Proposition~\ref{prop:ov-holley-line}; the proof of
Lemma~\ref{lem:ov-torus-odds} then applies verbatim.
\end{proof}

Where the sharpened inputs are used, Lemmas~\ref{lem:ov-torus-floor}
and~\ref{lem:ov-torus-odds} hold on tori with $L\ge4$ rather than $L\ge3$.
Every later use has $L\ge8$, or lets $L\to\infty$.

\subsection{The plus-density of the weak Ising field}\label{sec:ov-density}

The plus-density of the Ising field enters the second moment below only through
the one-site marginal $\rho_L=\mu_L(X_o=1)$. Domination by independent sites at
the worst conditional plus probability, the principle of \cite[Eq.~(8)]{ABL1987}
(see also \cite[Proposition~4.16]{GHM2001}) used in Appendix~\ref{sec:mns}, would give density only about $0.31$ at $d=12$,
$t=3/25$. That probability never enters here, because the second-moment argument
conditions only on increasing events. For translation-invariant laws with the
downward FKG property, the product measures that they dominate are
characterized in \cite[Theorems~1.2 and~4.1]{LiggettSteif2006}, and Jiang and Lang apply
this to the plus spins of the Ising minus phase \cite[Proposition~2]{JiangLang2026}.
Any product density obtained in this way is at most the one-site density,
whereas the pair-ratio bound of Section~\ref{sec:ov-second-moment} keeps a lower
bound on $\rho_L$ itself, at the price of the boost term; we have not evaluated
the product route at the points of Table~\ref{tab:overlap-constants}. The
following bound is the classical mean-field bound on the magnetization; see for
instance \cite{Pearce1981} and \cite[Theorem~3.53(2)]{FriedliVelenik2017}. We
include the short proof in the form used here.

\begin{lemma}[Mean-field lower bound on the plus-density]\label{lem:ov-meanfield}
Let $h<0$, $0\le K'\le|h|$ and $2dK'<1$, and let $m_*\in(0,1)$ be the unique
root of $m=\tanh(|h|+2dK'm)$. Then $\rho_L\ge(1-m_*)/2$ for every $L\ge3$. In
particular, if a rational $\bar m$ satisfies $\tanh(|h|+2dK'\bar m)\le\bar m$,
then $\rho_L\ge\rho_-:=(1-\bar m)/2$ for every $L\ge3$.
\end{lemma}

\begin{proof}
Put $\bar h=|h|>0$. Flipping all spins gives $\rho_L=(1-m_L)/2$, where $m_L$ is
the magnetization at $o$ of the Ising model with coupling $K'$ and field
$\bar h$; by the first Griffiths inequality
\cite{Griffiths1967,KellySherman1968}, $m_L\ge0$. For $s\in\mathbb R$ let
$\nu_s$ be the Ising law on $V\setminus\{o\}$ with coupling $K'$ on the edges
not touching $o$, field $\bar h+s$ at the neighbors of $o$ and field $\bar h$
elsewhere, and let $Z_s$ be its partition function. Summing over the spin at
$o$,
\[
 m_L=\tanh\Bigl(\bar h+\tfrac12\log\frac{Z_{K'}}{Z_{-K'}}\Bigr),\qquad
 \log\frac{Z_{K'}}{Z_{-K'}}=\int_{-K'}^{K'}\mathfrak f(s)\,ds,\qquad
 \mathfrak f(s)=\sum_{y\in N(o)}\langle\sigma_y\rangle_{\nu_s}.
\]
For $s\ge-\bar h$ all fields are nonnegative, so by the GHS inequality
\cite{GHS1970} $\mathfrak f$ is concave on $[-\bar h,\infty)\supseteq[-K',K']$;
this is the only use of $K'\le|h|$. Hence
$\mathfrak f(s)+\mathfrak f(-s)\le2\mathfrak f(0)$ and
$\frac12\int_{-K'}^{K'}\mathfrak f\le K'\mathfrak f(0)$. Give the edges at $o$
the couplings $\lambda K'$, $\lambda\in[0,1]$, with field $\bar h$ everywhere.
At $\lambda=0$ the marginal off $o$ is $\nu_0$, and by the second Griffiths
inequality \cite{Griffiths1967,KellySherman1968} each $\langle\sigma_y\rangle$
is nondecreasing in $\lambda$; at $\lambda=1$ it equals $m_L$ by translation
invariance. Thus $\mathfrak f(0)\le2dm_L$ and $m_L\le\tanh(\bar h+2dK'm_L)$.
The function $m\mapsto\tanh(\bar h+2dK'm)-m$ has derivative at most
$2dK'-1<0$, is positive at $0$ and negative at $1$, so $m_L\le m_*$; and
$\tanh(\bar h+2dK'\bar m)\le\bar m$ forces $\bar m\ge m_*$.
\end{proof}

With $2d\bar m=N_1/N_2$ in lowest terms, the test
$\tanh(|h|+2dK'\bar m)\le\bar m$ is equivalent to the inequality
$g_h^{-N_2}g_K^{N_1}\le((1+\bar m)/(1-\bar m))^{N_2}$ between rationals.
At $d=12$, $t=3/25$ the value $\bar m=1297431/10^7$ passes it, so
$\rho_L\ge0.4351284$ on every torus. For the other rows of
Table~\ref{tab:overlap-constants}, in the order listed, the second program of
Section~\ref{sec:ov-constants} uses
$\bar m=114873/1250000$, $546199/5000000$, $759111/5000000$, $444177/2500000$,
$1495251/10^7$, $136251/1250000$, $1440853/10^7$ and $1240161/10^7$; each
passes the test and gives at least the tabulated $\rho_-$. For the rows of
Table~\ref{tab:ov-noisy}, in the order listed, one of the two programs uses
$\bar m=1037761/10^7$, $1201617/10^7$, $1387291/10^7$, $199063/1250000$ and
$228091/1250000$, and Table~\ref{tab:ms-local} lists the values used in
Section~\ref{sec:macrostep}.

\subsection{A weighted second moment for Ising sites and Bernoulli bonds}\label{sec:ov-second-moment}

This subsection concerns only the auxiliary law
\[
 P^{\rm aux}_L=\mu_L\otimes\mathrm{Ber}(p)^{\otimes E}
\]
of $(X,Y)\in\{\pm1\}^V\times\{0,1\}^E$ on $\mathbb T_L$, with $\mu_L$ from
\eqref{eq:ov-ising} and $0<p\le1$. For $\Lambda\subseteq V$ write
$\{\Lambda\subseteq+\}=\{X_x=1\text{ for all }x\in\Lambda\}$. Let
$\mathcal O_n$ be the event that some path from $o$ to $\partial B_n$ inside
$B_n$ has $X=1$ at all its vertices and $Y=1$ on all its edges. Assume
\begin{equation}\label{eq:ov-H}
 d\ge6,\qquad h<0,\qquad 0\le K'\le|h|,\qquad 2dK'<1,
\end{equation}
and put $t'=\tanh K'$ and $\alpha_{\rm D}=2dt'<1$. Let $\rho_-$ be as in
Lemma~\ref{lem:ov-meanfield} and $\kappa=(1-\rho_-)/\rho_-$.

\paragraph{Oriented paths and the weighted count.}
An oriented path of length $n$ is a word
$\gamma=(\gamma_0,\ldots,\gamma_{n-1})\in[d]^n$, with lifts
$\tilde x_k(\gamma)=\sum_{i<k}e_{\gamma_i}\in\Z^d$ and positions
$x_k(\gamma)=\tilde x_k(\gamma)\bmod L$ for $0\le k\le n$. Put
$V_\gamma=\{x_0(\gamma),\ldots,x_n(\gamma)\}$ and let $E_\gamma$ be the set of
its $n$ steps. If $L\ge2n+1$, then (i) the $x_k(\gamma)$ are distinct; (ii)
$x_k(\gamma)=x_j(\gamma')$ forces $k=j$; and (iii) $x_k(\gamma)$ is at graph
distance $k$ from $o$. Indeed, the lifts lie in $\{0,\ldots,n\}^d$ and have
coordinate sum $k$; equality modulo $L$ forces equal coordinate sums since
$n<L$, and then equal lifts since $L>2n$; and a lift difference
$v\in[-n,n]^d$ has torus length $|v|_1$ since $n<L/2$. Call $\gamma$ open if
$V_\gamma\subseteq+$ and $Y_e=1$ on $E_\gamma$. Then
$\pi_L(\gamma):=P^{\rm aux}_L(\gamma\text{ open})=p^n\mu_L(V_\gamma\subseteq+)>0$.
Put
\[
 W_n=\sum_{\gamma\in[d]^n}d^{-n}\,\frac{\mathbf1\{\gamma\text{ open}\}}{\pi_L(\gamma)}.
\]

\begin{lemma}[Weighted Paley--Zygmund bound]\label{lem:ov-pz}
Let $L\ge2n+1$, and let $\E_{\gamma,\gamma'}$ denote the average over
independent uniform $\gamma,\gamma'\in[d]^n$. For vertex sets
$\Lambda,\Lambda'$ of $\mathbb T_L$ put
$\mathcal R_L(\Lambda,\Lambda')=\mu_L(\Lambda\cup\Lambda'\subseteq+)/
[\mu_L(\Lambda\subseteq+)\,\mu_L(\Lambda'\subseteq+)]$. Then $\E W_n=1$,
\[
 \E W_n^2=\E_{\gamma,\gamma'}\Bigl[p^{-|E_\gamma\cap E_{\gamma'}|}\,
 \mathcal R_L(V_\gamma,V_{\gamma'})\Bigr],
\]
and $P^{\rm aux}_L(\mathcal O_n)\ge P^{\rm aux}_L(W_n>0)\ge1/\E W_n^2$.
\end{lemma}

\begin{proof}
Clearly $\E W_n=1$. Since $X$ and $Y$ are independent,
\[
 P^{\rm aux}_L(\gamma,\gamma'\text{ open})
 =p^{|E_\gamma\cup E_{\gamma'}|}\,\mu_L(V_\gamma\cup V_{\gamma'}\subseteq+),
\]
and $|E_\gamma\cup E_{\gamma'}|=2n-|E_\gamma\cap E_{\gamma'}|$ by (i); this
gives the second moment. The Cauchy--Schwarz inequality
\[
 \E W_n=\E\bigl[W_n\mathbf1\{W_n>0\}\bigr]
 \le(\E W_n^2)^{1/2}\,P^{\rm aux}_L(W_n>0)^{1/2}
\]
gives the last bound. If $W_n>0$, some $\gamma$ is open, and by (iii) it runs
inside $B_n$ from $o$ to $\partial B_n$.
\end{proof}

Normalizing each path by its own probability makes the pair law in
Lemma~\ref{lem:ov-pz} exactly uniform, although $\mu_L(V_\gamma\subseteq+)$
depends on the shape of $\gamma$. The Ising law now enters only through
$\mathcal R_L$, which is bounded uniformly over all pairs of vertex sets.

\paragraph{The Ising pair ratio.}
We use Dobrushin's comparison theorem \cite{Dobrushin1970,Follmer1982} in an
averaged, finite-volume form, and include its proof. Let $\Lambda$ be a finite
set and $\Omega=\{\pm1\}^\Lambda$. For $x\in\Lambda$ let
$\gamma_x(\cdot\mid\omega)$ and $\tilde\gamma_x(\cdot\mid\omega)$ be probability
kernels on $\{\pm1\}$ that do not depend on $\omega_x$, and write
$(\gamma_xf)(\omega)=\sum_{s=\pm1}\gamma_x(s\mid\omega)f(\omega^{x,s})$, where
$\omega^{x,s}$ sets the coordinate $x$ to $s$. Let
$\delta_y(f)=\max\{|f(\omega)-f(\omega')|:\omega=\omega'\text{ off }y\}$, and
let $\mathsf C_{xy}=\max\{|\gamma_x(1\mid\omega)-\gamma_x(1\mid\omega')|:
\omega=\omega'\text{ off }y\}$ for $y\ne x$, with $\mathsf C_{xx}=0$.

\begin{lemma}[Dobrushin comparison, averaged form]\label{lem:ov-dobrushin}
Suppose $\alpha_{\mathsf C}:=\max_x\sum_y\mathsf C_{xy}<1$, put
$\mathsf D=\sum_{k\ge0}\mathsf C^k$, and let $\mu$, $\tilde\mu$ be laws on
$\Omega$ with $\mu(\gamma_xf)=\mu(f)$ and $\tilde\mu(\tilde\gamma_xf)=\tilde\mu(f)$
for all $x$ and $f$. Then for every $f:\Omega\to\mathbb R$
\[
 |\tilde\mu(f)-\mu(f)|\ \le\ \sum_{x,y\in\Lambda}\delta_y(f)\,\mathsf D_{yx}\,
 \bar b_x,\qquad
 \bar b_x=\int\bigl|\tilde\gamma_x(1\mid\omega)-\gamma_x(1\mid\omega)\bigr|\,
 \tilde\mu(d\omega).
\]
\end{lemma}

\begin{proof}
Enumerate $\Lambda=\{x_1,\ldots,x_N\}$ and sweep periodically: $g_0=f$ and
$g_k=\gamma_{x_k}g_{k-1}$, with indices modulo $N$. Then $\mu(g_k)=\mu(f)$, and by
$\tilde\gamma$-invariance
$\tilde\mu(g_{k-1})-\tilde\mu(g_k)=\tilde\mu(\tilde\gamma_{x_k}g_{k-1}-\gamma_{x_k}g_{k-1})$.
Since the two kernels have equal total mass,
$|(\tilde\gamma_x-\gamma_x)g(\omega)|\le\delta_x(g)|\tilde\gamma_x(1\mid\omega)
-\gamma_x(1\mid\omega)|$, and so
$|\tilde\mu(g_{k-1})-\tilde\mu(g_k)|\le\delta_{x_k}(g_{k-1})\bar b_{x_k}$.
Moreover $\delta_x(\gamma_xg)=0$ and
$\delta_y(\gamma_xg)\le\delta_y(g)+\mathsf C_{xy}\delta_x(g)$ for $y\ne x$.
Define vectors $\Delta^{(0)}=\delta(f)$ and, at step $k$ with $x=x_k$ and
$r_k=\Delta^{(k-1)}_x$, put $\Delta^{(k)}_x=0$ and
$\Delta^{(k)}_y=\Delta^{(k-1)}_y+\mathsf C_{xy}r_k$ for $y\ne x$. By induction
$\delta(g_k)\le\Delta^{(k)}$, and the row vector
$R^{(k)}=\sum_{i\le k}r_ie_{x_i}$ satisfies
$\Delta^{(k)}+R^{(k)}=\delta(f)+R^{(k)}\mathsf C$. Hence
$R^{(k)}\le\delta(f)+R^{(k)}\mathsf C$, and iterating, with
$\|u\mathsf C^J\|_1\le\alpha_{\mathsf C}^J\|u\|_1$, gives
$R^{(k)}\le\delta(f)\mathsf D$. Since $\delta_{x_k}(g_{k-1})\le r_k$, summing the
one-step bounds over $k\le M$ gives
\[
 |\tilde\mu(f)-\mu(f)|\le\sum_xR^{(M)}_x\bar b_x+\sum_x\delta_x(g_M),
\]
where the last term bounds $|\tilde\mu(g_M)-\mu(g_M)|$ by the oscillation of
$g_M$.
Between its own updates each $\Delta_x$ only increases, so at the end of a
sweep it is at most the amount removed at $x$ during the next sweep. These
amounts are summable, since $\sum_kr_k\le\|\delta(f)\mathsf D\|_1$, so
$\delta(g_M)\le\Delta^{(M)}\to0$ along $M=mN$. Together with
$R^{(M)}\le\delta(f)\mathsf D$ this proves the claim.
\end{proof}

Let $\mathrm{Adj}_L$ be the adjacency matrix of $\mathbb T_L$ and
$D^{\mathbb T}=\sum_{k\ge0}(t'\mathrm{Adj}_L)^k$, which is finite because
$\alpha_{\rm D}<1$; $D^{\mathbb T}_{xy}$ is the $t'$-weighted number of walks from $x$ to
$y$. We use the FKG inequality \cite{FKG1971} for $\mu_L$ in the following
forms: increasing events are positively correlated; conditioning on an
increasing event of positive probability raises the expectation of increasing
functions and lowers that of decreasing ones; and $\mu_L$ conditioned on
$\{\Lambda\subseteq+\}$ is again a ferromagnetic Ising law, with plus boundary
spins on $\Lambda$.

\begin{lemma}[Ising pair ratio]\label{lem:ov-pair-ratio}
Let $L\ge3$, $K'\ge0$, $h\in\mathbb R$ and $\alpha_{\rm D}<1$, and put
$\kappa_L=(1-\rho_L)/\rho_L$. For pairwise disjoint
$\Lambda_0,\Lambda_1,\Lambda_2\subseteq V$,
\[
 \frac{\mu_L(\Lambda_0\cup\Lambda_1\cup\Lambda_2\subseteq+)}
 {\mu_L(\Lambda_0\cup\Lambda_1\subseteq+)\,\mu_L(\Lambda_0\cup\Lambda_2\subseteq+)}
 \ \le\ \frac1{\mu_L(\Lambda_0\subseteq+)}\exp\Bigl(\kappa_L
 \sum_{x\in\Lambda_2,\,y\in\Lambda_1}D^{\mathbb T}_{xy}\Bigr),
\]
and $\mu_L(\Lambda_0\subseteq+)\ge\rho_L^{|\Lambda_0|}$. In particular, with
$\Lambda_0=V_\gamma\cap V_{\gamma'}$, $\Lambda_1=V_\gamma\setminus V_{\gamma'}$
and $\Lambda_2=V_{\gamma'}\setminus V_\gamma$,
\[
 \mathcal R_L(V_\gamma,V_{\gamma'})\ \le\
 \frac1{\mu_L(V_\gamma\cap V_{\gamma'}\subseteq+)}
 \exp\Bigl(\kappa_L\sum_{x\in V_{\gamma'}\setminus V_\gamma,\,
 y\in V_\gamma\setminus V_{\gamma'}}D^{\mathbb T}_{xy}\Bigr),
\]
and $\kappa_L$ is nonincreasing in $\rho_L$.
\end{lemma}

\begin{proof}
Write $P(\Lambda)=\mu_L(\Lambda\subseteq+)$. The ratio equals
$\mu_L(\Lambda_2\subseteq+\mid\Lambda_0\cup\Lambda_1\subseteq+)/
[P(\Lambda_0)\,\mu_L(\Lambda_2\subseteq+\mid\Lambda_0\subseteq+)]$, and FKG gives
$P(\Lambda_0)\ge\rho_L^{|\Lambda_0|}$. Order $\Lambda_2=\{c_1,\ldots,c_k\}$ and
let $C_{<j}=\{c_i:i<j\}$. The remaining quotient is
$\prod_j\mu_1^{(j)}(X_{c_j}=1)/\mu_2^{(j)}(X_{c_j}=1)$, with
$\mu_1^{(j)}=\mu_L(\cdot\mid\Lambda_0\cup\Lambda_1\cup C_{<j}\subseteq+)$ and
$\mu_2^{(j)}=\mu_L(\cdot\mid\Lambda_0\cup C_{<j}\subseteq+)$. Fix $j$ and apply
Lemma~\ref{lem:ov-dobrushin} on $\Lambda'=V\setminus(\Lambda_0\cup C_{<j})$ with
$\mu=\mu_2^{(j)}$, its heat-bath kernels $\gamma_x$, $\tilde\mu=\mu_1^{(j)}$,
$\tilde\gamma_x=\gamma_x$ for $x\notin\Lambda_1$ and $\tilde\gamma_x=\delta_1$
for $x\in\Lambda_1$. Both invariances hold: $\mu_2^{(j)}$ is an Ising law on
$\Lambda'$ with plus boundary spins, and $\mu_1^{(j)}$ is $\mu_2^{(j)}$
conditioned on $\{\Lambda_1\subseteq+\}$, so it is almost surely plus on
$\Lambda_1$ and has the single-site conditional laws $\gamma_x$ off $\Lambda_1$.
Here $\gamma_x(1\mid\omega)=\frac12\bigl(1+\tanh(h+K'S_x(\omega))\bigr)$, with
$S_x(\omega)$ the sum of the neighboring spins, including boundary spins.
Changing one neighbor moves $h+K'S_x$ by $2K'$, and
$\sup_u\frac12[\tanh(u+K')-\tanh(u-K')]=\tanh K'$, so
$\mathsf C_{xy}\le t'\mathbf1\{x\sim y\}$, $\alpha_{\mathsf C}\le\alpha_{\rm D}$
and $\mathsf D\le D^{\mathbb T}$ entrywise. Also $\bar b_x=0$ for $x\notin\Lambda_1$,
while for $y\in\Lambda_1$, by FKG and the DLR equation,
$\bar b_y=\mu_1^{(j)}[\gamma_y(-1\mid\cdot)]\le\mu_L[\gamma_y(-1\mid\cdot)]=1-\rho_L$.
With $f=\mathbf1\{X_{c_j}=1\}$ this gives
$\mu_1^{(j)}(X_{c_j}=1)\le\mu_2^{(j)}(X_{c_j}=1)+(1-\rho_L)\sum_{y\in\Lambda_1}
D^{\mathbb T}_{c_jy}$. Since $\mu_2^{(j)}(X_{c_j}=1)\ge\rho_L$ by FKG, each factor is at
most $1+\kappa_L\sum_{y\in\Lambda_1}D^{\mathbb T}_{c_jy}\le\exp(\kappa_L\sum_{y\in\Lambda_1}
D^{\mathbb T}_{c_jy})$.
\end{proof}

By Lemma~\ref{lem:ov-meanfield}, $\kappa_L\le\kappa$ and
$\rho_L\ge\rho_-$ uniformly in $L\ge3$. The density enters in three places, and
each time through the unconditioned marginal $\rho_L$:
$P(\Lambda_0)\ge\rho_L^{|\Lambda_0|}$, $\mu_2^{(j)}(X_{c_j}=1)\ge\rho_L$, and
$\bar b_y\le1-\rho_L$.

\paragraph{Walk sums and shared edges.}
For $0<t_0<1/(2d)$ and $z\in\Z^d$ let
$\mathcal D_{t_0}(z)=\sum_{k\ge0}t_0^k\,n_k(z)$, where $n_k(z)$ is the number
of nearest-neighbor walks of length $k$ from $0$ to $z$ in $\Z^d$.

\begin{lemma}[Walk-sum bounds]\label{lem:ov-resolvent}
Let $r=|z|_1$ and $2dt_0<1$.
\begin{enumerate}\renewcommand{\labelenumi}{(\roman{enumi})}
\item $D^{\mathbb T}_{xy}=\sum_{w\in\Z^d}\mathcal D_{t'}(\tilde y-\tilde x+Lw)$ for any
lifts $\tilde x,\tilde y$ of $x,y$.
\item $\mathcal D_{t_0}(z)\le t_0^r\,\dfrac{r!}{\prod_i|z_i|!}\,\Gamma_r(t_0)$, where
$\Gamma_r(t_0)=2^{-d}\sum_{j=0}^d\binom dj\bigl(1-2t_0(d-2j)\bigr)^{-(r+1)}$.
\item $\sum_{|u|_1\ge R}\mathcal D_{t_0}(u)\le(2dt_0)^R/(1-2dt_0)$.
\item Put $\mathcal B(m)=t_0^m\,\mathfrak M_m\,\Gamma_m(t_0)$, where
$\mathfrak M_m=m!/((\ell+1)!^{j}\,\ell!^{\,d-j})$ for $m=\ell d+j$, $0\le j<d$, is
the largest multinomial coefficient with $d$ parts. Then
$\sup_{|u|_1=m}\mathcal D_{t_0}(u)\le\mathcal B(m)$, and if
$\lambda_*:=dt_0/(1-2dt_0)<1$, then $\mathcal B(m+1)\le\lambda_*\mathcal B(m)$.
\end{enumerate}
\end{lemma}

\begin{proof}
(i) Torus walks from $x$ lift bijectively to walks from $\tilde x$. (ii) In one
dimension $\sum_kn^{(1)}_k(\ell)u^k/k!=I_{|\ell|}(2u)$, a modified Bessel
function, and a walk in $\Z^d$ is a shuffle of one-dimensional walks, so
$\sum_kn_k(z)u^k/k!=\prod_iI_{|z_i|}(2u)$. With $k!=\int_0^\infty s^ke^{-s}ds$,
$\mathcal D_{t_0}(z)=\int_0^\infty e^{-s}\prod_iI_{|z_i|}(2t_0s)\,ds$. Termwise,
$I_\nu(x)\le(x/2)^\nu I_0(x)/\nu!$ and $I_0(x)\le\cosh x$, so the integrand is at
most $e^{-s}(t_0s)^r\cosh(2t_0s)^d/\prod_i|z_i|!$, whose integral is the stated
bound. (iii) Walks shorter than $R$ cannot reach $|u|_1\ge R$, and there are
$(2d)^k$ walks of length $k$. (iv) The first claim follows from (ii). Removing
one step from a maximizing composition gives
$\mathfrak M_{m+1}\le\min(m+1,d)\,\mathfrak M_m$, and
$\Gamma_{m+1}(t_0)\le\Gamma_m(t_0)/(1-2dt_0)$.
\end{proof}

\begin{lemma}[Correlation gain along shared edges]\label{lem:ov-shared}
Let $L\ge3$, $K'\ge0$ and $h<0$. Put $M=|h|+(2d-1)K'$, $M'=(2d-1)K'$ and
\[
 c_{\rm cov}=\frac{2\sinh(2K')}{\bigl(e^{K'}\cosh(2M)+e^{-K'}\cosh(2M')\bigr)^2}.
\]
Suppose $\rho_-\le\rho_L$ and $\rho_-^2\ge c_{\rm cov}/4$, and put
$\rho_c=\rho_-+c_{\rm cov}/(4\rho_-)$. Then $\mu_L(X_y=1\mid X_x=1)\ge\rho_c$
for adjacent $x,y$. Consequently, if $\Lambda\subseteq V$ and $E_\Lambda$ is a
set of edges of $\mathbb T_L$ with endpoints in $\Lambda$ such that $(\Lambda,E_\Lambda)$
is a disjoint union of paths, then
$\mu_L(\Lambda\subseteq+)\ge\rho_-^{|\Lambda|-|E_\Lambda|}\rho_c^{|E_\Lambda|}$.
\end{lemma}

\begin{proof}
Given the spins $\xi$ off $\{x,y\}$, the pair $(X_x,X_y)$ is a two-spin
Ising system with coupling $K'$ and fields
$b_x=h+K'\sum_{z\in N(x)\setminus y}X_z$ and $b_y=h+K'\sum_{z\in N(y)\setminus x}X_z$.
Differentiating the logarithm of its partition function gives
$\mathrm{Cov}(X_x,X_y\mid\xi)=2\sinh(2K')/(e^{K'}\cosh(b_x+b_y)
+e^{-K'}\cosh(b_x-b_y))^2\ge c_{\rm cov}$, since $|b_x+b_y|\le2M$ and
$|b_x-b_y|\le2M'$. By the law of total covariance and FKG applied to the two
increasing conditional means, $\mathrm{Cov}(X_x,X_y)\ge c_{\rm cov}$. Hence
$\mu_L(X_y=1\mid X_x=1)=\rho_L+\mathrm{Cov}(X_x,X_y)/(4\rho_L)\ge\rho_c$,
because $r\mapsto r+c_{\rm cov}/(4r)$ is increasing for $r\ge\sqrt{c_{\rm cov}}/2$.
Along a path segment $x_0,\ldots,x_\ell$, FKG for
$\mu_L(\cdot\mid X_{x_i}=1)$ gives $\mu_L(X_{x_{i+1}}=1\mid X_{x_0}=\cdots=X_{x_i}=1)
\ge\rho_c$, so the segment is plus with probability at least $\rho_-\rho_c^\ell$.
FKG across the $|\Lambda|-|E_\Lambda|$ segments completes the proof.
\end{proof}

The proof uses $c_{\rm cov}$ only as a lower bound on
$\mathrm{Cov}(X_x,X_y)$ for adjacent $x,y$. The lemma therefore holds with
$c_{\rm cov}$ replaced, in the hypothesis and in $\rho_c$, by any $c>0$ that
bounds these covariances from below; the certificates use a rational lower
bound for $c_{\rm cov}$. The next lemma gives a larger admissible $c$ when
$\alpha_{\rm D}$ is small, by bounding the correlations among the neighbors
of the edge.

\begin{lemma}[A sharper covariance bound along an edge]\label{lem:ov-covariance}
Assume \eqref{eq:ov-H}, $L\ge40$ and $\alpha_{\rm D}\le0.35$. Let $m_L\ge0$
be the magnetization of Lemma~\ref{lem:ov-meanfield}, so that
$\mu_L(X_x=1)=(1-m_L)/2$, let $\bar m\ge m_L$ and $\bar t\ge t'$ with
$2d\bar t<1$, and for $r=1,2,3$ put
$\delta_r=\mathsf N_r\bar t^{\,r}\Gamma_r(\bar t)+10^{-15}$, with
$(\mathsf N_1,\mathsf N_2,\mathsf N_3)=(1,2,6)$ and $\Gamma_r$ as in
Lemma~\ref{lem:ov-resolvent}(ii).
\begin{enumerate}\renewcommand{\labelenumi}{(\roman{enumi})}
\item For $x\ne y$, $0\le\mathrm{Cov}_{\mu_L}(X_x,X_y)\le D^{\mathbb T}_{xy}$.
\item For adjacent $x,y$, $\mathrm{Cov}_{\mu_L}(X_x,X_y)\ge c'_{\rm cov}$, where,
with $M$ and $M'$ as in Lemma~\ref{lem:ov-shared},
\[
 c'_{\rm cov}=\frac{2\sinh(2K')}{\bigl(e^{K'}(1+E_+\,\mathsf c(2M))
 +e^{-K'}(1+E_-\,\mathsf c(2M'))\bigr)^2},\qquad
 \mathsf c(U)=\frac{\cosh U-1}{U^2},
\]
$E_+=4h^2+4|h|K'(4d-2)\bar m+K'^2E_T$, $E_-=2K'^2E_1$,
$E_1=(2d-1)+(2d-1)(2d-2)(\bar m^2+\delta_2)$ and
$E_T=2E_1+2\bigl[(2d-2)(\bar m^2+\delta_1)+((2d-1)^2-(2d-2))(\bar m^2+\delta_3)\bigr]$.
\item Lemma~\ref{lem:ov-shared} holds with $\rho'_c=\rho_-+c'_{\rm cov}/(4\rho_-)$
in place of $\rho_c$, provided $\rho_-^2\ge c'_{\rm cov}/4$.
\end{enumerate}
The bound $c'_{\rm cov}$ only decreases if $K'$ is replaced by a lower bound in
$\sinh(2K')$ and in $e^{-K'}$ and by an upper bound elsewhere, and $|h|$ by an
upper bound, so rational enclosures may be used.
\end{lemma}

\begin{proof}
(i) Nonnegativity is FKG. Write $\mu_\pm=\mu_L(\cdot\mid X_y=\pm1)$. Then
$\mathrm{Cov}(X_x,X_y)=2\mu_L(X_y=1)\mu_L(X_y=-1)[\mu_+(X_x)-\mu_-(X_x)]$, and
$2\mu_L(X_y=1)\mu_L(X_y=-1)\le\frac12$. Apply Lemma~\ref{lem:ov-dobrushin} on
$V\setminus\{y\}$ with $\mu=\mu_-$ and $\tilde\mu=\mu_+$, the Ising laws there
with boundary spin $-1$, respectively $+1$, at $y$, and their heat-bath kernels.
As in the proof of Lemma~\ref{lem:ov-pair-ratio}, $\mathsf C_{zz'}\le t'\mathbf1\{z\sim z'\}$
and $\mathsf D\le D^{\mathbb T}$ entrywise. The two kernels differ only at
neighbors $z$ of $y$, by the flip of one neighboring spin, so
$\bar b_z\le t'\mathbf1\{z\sim y\}$. With $f=X_x$, only $\delta_x(f)=2$ is
nonzero, and
\[
 |\mu_+(X_x)-\mu_-(X_x)|\le2\sum_{z\sim y}D^{\mathbb T}_{xz}\,t'
 =2\bigl(D^{\mathbb T}t'\mathrm{Adj}_L\bigr)_{xy}=2\bigl(D^{\mathbb T}-I\bigr)_{xy}
 =2D^{\mathbb T}_{xy}.
\]

(ii) In the notation of the proof of Lemma~\ref{lem:ov-shared}, that proof gives
$\mathrm{Cov}(X_x,X_y)\ge\E[2\sinh(2K')/\mathcal P^2]$ with
$\mathcal P=e^{K'}\cosh(b_x+b_y)+e^{-K'}\cosh(b_x-b_y)$, and Jensen's
inequality for $u\mapsto u^{-2}$ gives
$\mathrm{Cov}(X_x,X_y)\ge2\sinh(2K')/(\E\mathcal P)^2$. Since
$\mathsf c(U)=\sum_{k\ge1}U^{2k-2}/(2k)!$ is increasing,
$\cosh u\le1+u^2\mathsf c(U)$ for $|u|\le U$; we use this with
$U=2M\ge|b_x+b_y|$ and $U=2M'\ge|b_x-b_y|$. It remains to show
$\E(b_x+b_y)^2\le E_+$ and $\E(b_x-b_y)^2\le E_-$. Put
$T_x=\sum_{z\in N(x)\setminus y}X_z$ and $T_y=\sum_{z\in N(y)\setminus x}X_z$,
so that $b_x+b_y=2h+K'(T_x+T_y)$ and $b_x-b_y=K'(T_x-T_y)$. By
Lemma~\ref{lem:ov-meanfield} and its proof, $\E X_z=-m_L$ with
$0\le m_L\le\bar m$, and by (i),
$\E[X_zX_{z'}]=m_L^2+\mathrm{Cov}(X_z,X_{z'})\in[0,\bar m^2+D^{\mathbb T}_{zz'}]$.
Since $L\ge40$, the sets $N(x)\setminus y$ and $N(y)\setminus x$ are disjoint,
with $2d-1$ sites each; two sites in the same set are at $\ell^1$ distance $2$,
and across the sets exactly $2d-2$ pairs are at distance $1$ and all others at
distance $3$. By parts (i)--(iii) of Lemma~\ref{lem:ov-resolvent}, for
$z\ne z'$ at distance $r\le3$,
$D^{\mathbb T}_{zz'}\le\mathcal D_{\bar t}(\tilde z'-\tilde z)
+\alpha_{\rm D}^{L-3}/(1-\alpha_{\rm D})\le\delta_r$: the largest multinomial
coefficient $r!/\prod_i|u_i|!$ with $|u|_1=r$ is $\mathsf N_r$, every image
$\tilde z'-\tilde z+Lw$ with $w\ne0$ has $\ell^1$ norm at least $L-3$, and
$0.35^{37}/0.65<10^{-15}$. Hence $\E T_x^2\le E_1$, $\E(T_x+T_y)^2\le E_T$ and
$\E(T_x+T_y)=-(4d-2)m_L$. Since $h<0$, this gives
$\E(b_x+b_y)^2=4h^2+4|h|K'(4d-2)m_L+K'^2\E(T_x+T_y)^2\le E_+$, and
$\E(b_x-b_y)^2\le K'^2(\E T_x^2+\E T_y^2)\le E_-$, because
$\E[T_xT_y]\ge0$. Each of these bounds is monotone in the parameters in the
direction stated after the lemma.

(iii) Repeat the proof of Lemma~\ref{lem:ov-shared} with $c'_{\rm cov}$ in place
of $c_{\rm cov}$.
\end{proof}

\paragraph{The difference walk.}
For independent uniform $\gamma,\gamma'\in[d]^n$ put
$Z_k=\tilde x_k(\gamma)-\tilde x_k(\gamma')$. Then $Z$ is a random walk from $0$
on $\mathbb A=\{z\in\Z^d:\sum_iz_i=0\}$ with the symmetric step law
$\mathbf P(0)=1/d$, $\mathbf P(e_i-e_j)=1/d^2$ for $i\ne j$; write
$\mathbf P_y$ and $\mathbf E_y$ for the walk started at $y$ and
$\operatorname{rad}(z)=|z|_1/2$. For $L\ge2n+1$, property (ii) of oriented paths
shows that the two paths meet only at equal times:
\begin{itemize}
\item $V_\gamma\cap V_{\gamma'}=\{x_k(\gamma):Z_k=0\}$, so
$|V_\gamma\cap V_{\gamma'}|=K_V:=\#\{0\le k\le n:Z_k=0\}$;
\item a step is shared if and only if $Z_k=Z_{k+1}=0$, so
$|E_\gamma\cap E_{\gamma'}|=K_E:=\#\{0\le k<n:Z_k=Z_{k+1}=0\}$, and
$(V_\gamma\cap V_{\gamma'},E_\gamma\cap E_{\gamma'})$ is a disjoint union of
$K_V-K_E$ path segments;
\item with $\mathcal J=\{k\le n:Z_k\ne0\}$, $V_\gamma\setminus V_{\gamma'}=\{x_k(\gamma):k\in\mathcal J\}$
and $V_{\gamma'}\setminus V_\gamma=\{x_k(\gamma'):k\in\mathcal J\}$, each indexed
injectively.
\end{itemize}
Put $u_k=\mathbf P_0(Z_k=0)=\sum_{k_1+\cdots+k_d=k}\bigl(k!/(d^k\prod_ik_i!)\bigr)^2$,
\[
 G_d=\sum_{k\ge0}u_k,\qquad G^{(2)}_d=\sum_{k\ge0}(k+1)u_k,\qquad
 F_d=1-1/G_d,
\]
and let $\mathfrak t=\inf\{k\ge1:Z_k=0\}$, $\mathfrak t_0=\inf\{k\ge0:Z_k=0\}$,
$\varphi(y)=\mathbf P_y(\mathfrak t_0<\infty)$,
$g(0,y)=\sum_k\mathbf P_0(Z_k=y)$ and
$g_0(y,y')=\mathbf E_y\#\{0\le k<\mathfrak t_0:Z_k=y'\}$.

\begin{lemma}[Collision Green function]\label{lem:ov-green}
\begin{enumerate}\renewcommand{\labelenumi}{(\roman{enumi})}
\item $\mathbf P_y(Z_k=z)\le u_k$ for all $y,z$.
\item For $k\ge d$, $u_k\le\sqrt d\,(2\pi\lfloor k/d\rfloor)^{-(d-1)/2}$.
Hence $G_d<\infty$ for $d\ge4$ and $G^{(2)}_d<\infty$ for $d\ge6$.
\item $F_d=\mathbf P_0(\mathfrak t<\infty)$ and $\mathbf P_0(Z_1\ne0,\mathfrak t<\infty)=F_d-1/d$;
almost surely $Z$ visits $0$ finitely often.
\item $\varphi(y)=g(0,y)/G_d$, and $\sum_{y\ne0}\varphi(y)^2=G^{(2)}_d/G_d^2-1$.
\item $\mathbf E_0\#\{1\le k<\mathfrak t:Z_k=y\}=\varphi(y)$ for $y\ne0$.
\item $g_0(y,y')=g_0(y',y)$.
\end{enumerate}
\end{lemma}

\begin{proof}
(i) By the Cauchy--Schwarz inequality,
$\mathbf P_0(Z_k=w)=\sum_x\mathbf P(\tilde x_k=x)\mathbf P(\tilde x'_k=x-w)\le u_k$.
(ii) $u_k\le\max_x\mathbf P(\tilde x_k=x)$, and the largest atom is
nonincreasing in $k$, since the law at time $k+1$ is a mixture of shifts of the
law at time $k$. At time $da$ the largest atom is $(da)!/(d^{da}a!^d)$, and
Robbins' form of Stirling's formula \cite{Robbins1955},
$\sqrt{2\pi}\,n^{n+1/2}e^{-n+1/(12n+1)}<n!<\sqrt{2\pi}\,n^{n+1/2}e^{-n+1/(12n)}$,
bounds it by $\sqrt d\,(2\pi a)^{-(d-1)/2}e^{1/(12da)-d/(12a+1)}$ for $a\ge1$.
The last factor is at most $1$ for $d\ge2$, since then $12a+1\le12d^2a$.
(iii) The number of visits to $0$ is geometric with mean $G_d$, and
$\mathbf P_0(Z_1=0)=1/d$; finiteness follows from $G_d<\infty$ and Borel--Cantelli.
(iv) The first-passage decomposition gives $g(y,0)=\varphi(y)G_d$, and
$g(y,0)=g(0,y)$ by symmetry. Also
$\sum_yg(0,y)^2=\sum_{k,l}\sum_y\mathbf P_0(Z_k=y)\mathbf P_0(Z_l=-y)
=\sum_{k,l}u_{k+l}=G^{(2)}_d$; remove $y=0$. (v) A path from $0$ to $y$ avoiding
$0$ after time $0$ has the probability of its reversal, which runs from $y$ and
first hits $0$ at its last step; sum over lengths. (vi) The same reversal.
\end{proof}

\paragraph{The boost and the renewal.}
Fix a rational $\bar t\ge t'$ with $\lambda_*=d\bar t/(1-2d\bar t)<1$, let
$\mathcal B$ be as in Lemma~\ref{lem:ov-resolvent}(iv) at $t_0=\bar t$, and put
\[
 \bar\psi(r)=\sum_{s\in\Z}\mathcal B\bigl(\max(|s|,2r-|s|)\bigr),\qquad
 \varepsilon_{n,L}=\frac{n^2\alpha_{\rm D}^{L-n}}{1-\alpha_{\rm D}}.
\]
Then $\bar\psi$ is nonincreasing and
$\bar\psi(r)\le(2r+1+2\lambda_*/(1-\lambda_*))\mathcal B(r)$.

\begin{lemma}[The boost is carried by the difference walk]\label{lem:ov-boost}
For $L\ge2n+1$,
\[
 \sum_{x\in V_{\gamma'}\setminus V_\gamma,\ y\in V_\gamma\setminus V_{\gamma'}}D^{\mathbb T}_{xy}
 \ \le\ \sum_{k\in\mathcal J}\bar\psi(\operatorname{rad}Z_k)+\varepsilon_{n,L}.
\]
\end{lemma}

\begin{proof}
Take $k,j\in\mathcal J$ and $s=j-k$. Coordinate sums give
$|\tilde x_k-\tilde x'_j|_1\ge|s|$, and since
$\tilde x_k-\tilde x'_j=Z_k-(\tilde x'_j-\tilde x'_k)$ with
$|\tilde x'_j-\tilde x'_k|_1=|s|$, also
$|\tilde x_k-\tilde x'_j|_1\ge2\operatorname{rad}(Z_k)-|s|$. By
Lemma~\ref{lem:ov-resolvent}(i), $D^{\mathbb T}_{x'_jx_k}$ is the sum of
$\mathcal D_{t'}(\tilde x_k-\tilde x'_j)\le\mathcal B(\max(|s|,2\operatorname{rad}(Z_k)-|s|))$,
by Lemma~\ref{lem:ov-resolvent}(iv) since $t'\le\bar t$ and $\mathcal B$ is
nonincreasing, and of wrap-around terms at $\ell^1$ distance at least $L-n$;
summing the first
terms over $j$ gives at most $\bar\psi(\operatorname{rad}Z_k)$, and by
Lemma~\ref{lem:ov-resolvent}(iii) the wrap-around terms total at most
$\alpha_{\rm D}^{L-n}/(1-\alpha_{\rm D})$ for each of the at most $n^2$ pairs.
The bound on $\bar\psi$ uses $2r-|s|\ge r$ for $|s|\le r$ and
$\sum_{s>r}\mathcal B(s)\le\mathcal B(r)\lambda_*/(1-\lambda_*)$.
\end{proof}

Let $\rho_e\in\{\rho_-,\rho_c,\rho'_c\}$, where $\rho_e=\rho_c$ requires the
hypotheses of Lemma~\ref{lem:ov-shared}, and $\rho_e=\rho'_c$ those of
Lemma~\ref{lem:ov-covariance}(iii), including $L\ge40$. Define weights along
$Z$:
\[
 \varpi_k=\begin{cases}
 \rho_-^{-1},& Z_k=0\ \text{and}\ (k=0\ \text{or}\ Z_{k-1}\ne0),\\
 (\rho_ep)^{-1},& k\ge1\ \text{and}\ Z_k=Z_{k-1}=0,\\
 \exp\bigl(\kappa\bar\psi(\operatorname{rad}Z_k)\bigr),& Z_k\ne0.
 \end{cases}
\]
Then
\[
 \prod_{k=0}^n\varpi_k=p^{-K_E}\rho_-^{-(K_V-K_E)}\rho_e^{-K_E}
 \exp\Bigl(\kappa\sum_{k\in\mathcal J}\bar\psi(\operatorname{rad}Z_k)\Bigr).
\]
The factor
$\rho_-^{-(K_V-K_E)}\rho_e^{-K_E}$ bounds
$1/\mu_L(V_\gamma\cap V_{\gamma'}\subseteq+)$, by FKG if $\rho_e=\rho_-$, by
Lemma~\ref{lem:ov-shared} if $\rho_e=\rho_c$ and by
Lemma~\ref{lem:ov-covariance}(iii) if $\rho_e=\rho'_c$, since
$(V_\gamma\cap V_{\gamma'},E_\gamma\cap E_{\gamma'})$ is a disjoint union of
$K_V-K_E$ path segments with $K_E$ edges. Lemmas~\ref{lem:ov-pz}, \ref{lem:ov-pair-ratio} and \ref{lem:ov-boost},
with $\kappa_L\le\kappa$, therefore give, for $n\ge1$ and $L\ge2n+1$ (and
$L\ge40$ if $\rho_e=\rho'_c$),
\begin{equation}\label{eq:ov-walk-weights}
 \E W_n^2\ \le\ e^{\kappa\varepsilon_{n,L}}\ \mathbf E_0\Bigl[\prod_{k=0}^n\varpi_k\Bigr].
\end{equation}
All weights are at least $1$. We cut the path of $Z$ at its visits to $0$. The
weight of an excursion from $0$, including the weight at its return, and the
weight of a final infinite excursion are
\begin{gather*}
 \Phi=\mathbf1\{\mathfrak t=1\}(\rho_ep)^{-1}+\mathbf1\{2\le\mathfrak t<\infty\}\,\rho_-^{-1}
 \prod_{1\le k<\mathfrak t}e^{\kappa\bar\psi(\operatorname{rad}Z_k)},\\
 \Phi_\infty=\mathbf1\{\mathfrak t=\infty\}\prod_{k\ge1}e^{\kappa\bar\psi(\operatorname{rad}Z_k)}.
\end{gather*}

\begin{proposition}[Renewal]\label{prop:ov-renewal}
If $\mathbf E_0\Phi<1$, then
$\sup_n\mathbf E_0\prod_{k=0}^n\varpi_k\le\rho_-^{-1}\mathbf E_0\Phi_\infty/(1-\mathbf E_0\Phi)$.
\end{proposition}

\begin{proof}
The finite products increase to $\Pi=\prod_{k\ge0}\varpi_k$. Let
$0=T_0<\cdots<T_N$ be the visits to $0$, finitely many by
Lemma~\ref{lem:ov-green}(iii). Then $\Pi=\rho_-^{-1}\prod_{i=1}^N\Phi^{(i)}\cdot
\Phi_\infty^{(N+1)}$, where $\Phi^{(i)}$ is the weight of $(T_{i-1},T_i]$ and
$\Phi_\infty^{(N+1)}$ that of the final excursion. The strong Markov property at
the $T_i$ and Tonelli give
$\mathbf E_0\Pi=\rho_-^{-1}\sum_{N\ge0}(\mathbf E_0\Phi)^N\,\mathbf E_0\Phi_\infty$.
\end{proof}

If $K'=0$, then $D^{\mathbb T}$ is the identity and the left side of
Lemma~\ref{lem:ov-boost} vanishes; with $\bar\psi$ replaced by $0$, all the
steps above are equalities, and
$\mathbf E_0\Phi=1/(d\rho p)+(F_d-1/d)/\rho$ for iid sites of density $\rho$.
The supremum is then infinite when $\mathbf E_0\Phi\ge1$, so in the iid case the
criterion below is necessary and sufficient for a bounded second moment. This
is the renewal form of the oriented second moment due to Kesten, presented for
oriented bond percolation in \cite[Section~2]{CoxDurrett1983}. The weights of
Lemma~\ref{lem:ov-pz}, which normalize each path by its own probability, are
those of \cite[Eq.~(5.16)]{LyonsPeres2016} (see also \cite{BPP1998}), and the
pair ratio $\mathcal R_L$ plays the role of the quasi-independence constant of
\cite[Theorem~5.19]{LyonsPeres2016}. A Cox--Durrett second moment for a
dependent site field, with corrections from nearby vertices of the two paths,
appears in \cite[proof of Theorem~1.2]{BDNS2020}. We are not aware of an earlier
bound of this kind for an Ising site field. For $K'>0$ the
Ising correlations add the boost $\bar\psi$, which acts only between meetings
of the two paths. The next lemma shows that its cost in $\mathbf E_0\Phi$ is of
second order in the re-meeting probability.

\begin{lemma}[Cycle weights]\label{lem:ov-cycle}
Put $U(z)=e^{\kappa\bar\psi(\operatorname{rad}z)}-1$ for $z\ne0$,
$U_1=e^{\kappa\bar\psi(1)}-1\ge U(z)$, and
\begin{gather*}
 \Sigma U=\sum_{r\ge1}|\mathbb A_r|\bigl(e^{\kappa\bar\psi(r)}-1\bigr),\\
 |\mathbb A_r|=\#\{z\in\mathbb A:|z|_1=2r\}
 =\sum_{\substack{k,l\ge1\\k+l\le d}}\binom dk\binom{d-k}l\binom{r-1}{k-1}\binom{r-1}{l-1},
\end{gather*}
\[
 \eta'=\sum_{k\ge0}\min\bigl(U_1,\ u_k\,\Sigma U\bigr),\qquad
 \Xi=\mathbf E_0\Bigl[\mathbf1\{Z_1\ne0,\ \mathfrak t<\infty\}\prod_{1\le k<\mathfrak t}
 e^{\kappa\bar\psi(\operatorname{rad}Z_k)}\Bigr].
\]
\begin{enumerate}\renewcommand{\labelenumi}{(\roman{enumi})}
\item $\sup_{y\ne0}\mathbf E_y\sum_{0\le k<\mathfrak t_0}U(Z_k)\le\eta'$.
\item If $\eta'<1$, then $\Xi\le F_d-1/d+U_1(G^{(2)}_d/G_d^2-1)/(1-\eta')$.
\item If $\eta'<1$, then $\mathbf E_0\Phi_\infty\le(1-1/d)/(1-\eta')$.
\item $\mathbf E_0\Phi=1/(d\rho_ep)+\Xi/\rho_-$.
\end{enumerate}
\end{lemma}

\begin{proof}
(i) By Lemma~\ref{lem:ov-green}(i),
$\mathbf E_y[U(Z_k)\mathbf1\{Z_k\ne0\}]=\sum_{z\ne0}\mathbf P_y(Z_k=z)U(z)
\le\min(U_1,u_k\Sigma U)$; drop the killing and sum over $k$.

(ii) Expand $\prod_{1\le k<\mathfrak t}(1+U(Z_k))$ as a sum over
$1\le k_1<\cdots<k_m<\mathfrak t$ of $\prod_iU(Z_{k_i})$. The term $m=0$ contributes
$F_d-1/d$ by Lemma~\ref{lem:ov-green}(iii). For $m\ge1$ use the Markov property
at $k_1,\ldots,k_m$: the first visit contributes $\varphi(y_1)$ by
Lemma~\ref{lem:ov-green}(v), each later visit at most $g_0(y_{i-1},y_i)$, and
the return at most $\varphi(y_m)$. Hence the $m$-th term is at most
$\langle v,\mathcal A^{m-1}v\rangle$ on $\ell^2(\mathbb A\setminus\{0\})$, with
$v(y)=U(y)^{1/2}\varphi(y)$ and
$\mathcal A(y,y')=U(y)^{1/2}g_0(y,y')U(y')^{1/2}$, a symmetric kernel by
Lemma~\ref{lem:ov-green}(vi). By (i),
$\sum_{y'}\mathcal A(y,y')U(y')^{1/2}\le\eta'U(y)^{1/2}$, and the Schur test
\cite{Schur1911} gives $\|\mathcal A\|\le\eta'$: sum the inequality
$2|f(y)g(y')|\le|f(y)|^2U(y')^{1/2}/U(y)^{1/2}+|g(y')|^2U(y)^{1/2}/U(y')^{1/2}$
against the symmetric kernel $\mathcal A$. Finally
$\|v\|^2\le U_1\sum_{y\ne0}\varphi(y)^2=U_1(G^{(2)}_d/G_d^2-1)$ by
Lemma~\ref{lem:ov-green}(iv); sum the geometric series.

(iii) $\Phi_\infty\le\mathbf1\{Z_1\ne0\}\prod_{1\le k<\mathfrak t}(1+U(Z_k))$. Apply the
Markov property at time $1$, where $\mathbf P_0(Z_1\ne0)=1-1/d$, and then (i)
successively at $k_1<\cdots<k_m$:
$\mathbf E_y\sum_{k_1<\cdots<k_m<\mathfrak t_0}\prod_iU(Z_{k_i})\le\eta'^m$. This is
Khas'minskii's lemma \cite{Khasminskii1959}.

(iv) $\mathbf P_0(\mathfrak t=1)=1/d$.
\end{proof}

The first-order correction to $\Xi$ carries two factors of the re-meeting
probability $\varphi$, one for reaching $y$ within a returning excursion and one
for returning afterwards. This is why the boost costs so little: at most
$0.0019$ in the score in dimension $12$ ($0.0013$ at $t=3/25$), at most
$0.0034$ in dimension $10$, and $0.0067$ at $(d,t)=(9,3/20)$.

\begin{theorem}[Second-moment criterion]\label{thm:ov-criterion}
Assume \eqref{eq:ov-H}, let $\rho_-$, $\kappa$, $\bar t$, $\lambda_*<1$ and
$\rho_e$ be as above, and let $0<p\le1$. Define
\[
 \mathcal S=\frac1{d\rho_ep}+\frac{F_d-1/d}{\rho_-}
 +\frac{U_1\,(G^{(2)}_d/G_d^2-1)}{(1-\eta')\,\rho_-}.
\]
If $\eta'<1$ and $\mathcal S<1$, then for all $n\ge1$ and $L\ge2n+1$ (and
$L\ge40$ if $\rho_e=\rho'_c$),
\[
 P^{\rm aux}_L(\mathcal O_n)\ \ge\ e^{-\kappa\varepsilon_{n,L}}\,
 \frac{\rho_-(1-\eta')(1-\mathcal S)}{1-1/d}.
\]
The conclusion persists if $F_d$, $G^{(2)}_d$, $U_1$, $\Sigma U$ and $\eta'$ are
replaced by upper bounds and $G_d$ in $G^{(2)}_d/G_d^2$ by a lower bound. We
write $\mathcal S$, $\mathcal S_c$ and $\mathcal S'_c$ for the score with
$\rho_e=\rho_-$, $\rho_c$ and $\rho'_c$, respectively.
\end{theorem}

\begin{proof}
Lemma~\ref{lem:ov-cycle}(ii) and (iv) give $\mathbf E_0\Phi\le\mathcal S<1$, so
Proposition~\ref{prop:ov-renewal}, Lemma~\ref{lem:ov-cycle}(iii) and
\eqref{eq:ov-walk-weights} give
$\E W_n^2\le e^{\kappa\varepsilon_{n,L}}(1-1/d)/(\rho_-(1-\eta')(1-\mathcal S))$.
Apply Lemma~\ref{lem:ov-pz}. Each term of $\mathcal S$ and $\eta'$ is monotone in
the quantities named.
\end{proof}

\begin{corollary}[Uniformity in the torus size]\label{cor:ov-uniform}
Under the hypotheses of Theorem~\ref{thm:ov-criterion}, with $\mathcal S_\bullet$
the score used and $\theta_*$ as in \eqref{eq:ov-theta}, put
$L_0(n)=\max\{40,\ 2n+1+\lceil\log(100\kappa n^2/(1-\alpha_{\rm D}))/
\log(1/\alpha_{\rm D})\rceil^+\}$. Then $P^{\rm aux}_L(\mathcal O_n)\ge\theta_*$
for every $n\ge1$ and every $L\ge L_0(n)$.
\end{corollary}

\begin{proof}
For $L\ge L_0(n)$ we have $L\ge2n+1$, $L\ge40$ and
$\alpha_{\rm D}^{L-n}\le(1-\alpha_{\rm D})/(100\kappa n^2)$, that is
$\kappa\varepsilon_{n,L}\le1/100$.
\end{proof}

\subsection{The certified constants}\label{sec:ov-constants-summary}

At each row of Tables~\ref{tab:overlap-constants} and~\ref{tab:ov-noisy},
finitely many inequalities are decided in exact rational arithmetic, and
rigorous rational enclosures are used for finitely many constants. They are
exactly the hypotheses of Propositions~\ref{prop:ov-vertex}
and~\ref{prop:ov-holley-line}, Lemma~\ref{lem:ov-meanfield} and
Theorem~\ref{thm:ov-criterion}, with the sharpened local inputs of
Section~\ref{sec:ov-noisy} in dimension $9$. Two independently written
programs certify every row. The list of conditions, the tables and the way
the constants are computed are in Section~\ref{sec:ov-constants} of
Appendix~\ref{app:certificates}.

\begin{remark}[Dimension nine and below; numerical observations]\label{rem:ov-dim9}
The first term of $\mathcal S$ requires $d\rho_ep>1$: two oriented paths that
meet share their next step with probability $1/d$, and each shared step costs
$(\rho_ep)^{-1}$. When $K'=0$ the criterion is also necessary for a bounded
second moment, as noted after Proposition~\ref{prop:ov-renewal}. With the
certificate families of Sections~\ref{sec:ov-floor}--\ref{sec:ov-holley},
floating-point scans at $d=9$ and $t\in\{0.13,0.14,\ldots,0.17\}$ give
$9\rho_cp$ between $0.905$ and $0.985$ and values of $\mathcal S_c$ of at least
$1.059$. The sharpened inputs of Section~\ref{sec:ov-noisy} raise $9\rho_cp$ to
at least $1.058$ at $t=3/20$, an exact bound, and certify the rows of
Table~\ref{tab:ov-noisy}. Their floors exceed the aligned frozen value with
$\beta$ inside $R$, which bounds every floor that is uniform over
$\mathcal N_{2d}$ (Remark~\ref{rem:ov-nesting}); no such floor could have given
them. In dimension $8$, the local inputs certified in
Section~\ref{sec:macrostep} give $8\rho_cp<0.99$, so the oriented criterion with
$\rho_e=\rho_c$ fails there. Section~\ref{sec:macrostep} changes the path
family instead.
\end{remark}

\subsection{Assembly and the passage to periodic limits}\label{sec:ov-assembly}

The quantile coupling below turns the conditional bond floor into an independent
bond field. It is stated for a general law, so that it applies both on the torus
and in a conditioned limit law. For $\Lambda\subseteq V(G)$ let
$E_+^\Lambda(q)$ be the set of edges $xy$ with $x,y\in\Lambda$ and $q_x=q_y=1$.

\begin{lemma}[Overlap-revealed exploration]\label{lem:ov-exploration}
Let $G$ be $\mathbb T_L$ with $L\ge2n+2$, or $\Z^d$, and let $P$ be a law of
$(q,B)\in\{\pm1\}^{V(G)}\times\{0,1\}^{E(G)}$ satisfying
\eqref{eq:ov-floor-undivided} with $P$ in place of $\nu_L$ for every ordered
pair $(u,v)$ of neighbors in $B_n$, every $F\subseteq N(v)\setminus\{u\}$ and
every $H\in\mathcal E_v^+$ with $H\subseteq\{q_u=q_v=1\}$. If the law of
$q|_{B_n}$ dominates a law $\mu$ on $\{\pm1\}^{B_n}$, in the sense that some
coupling of $X\sim\mu$ with $q|_{B_n}$ has $X\le q|_{B_n}$ almost surely, then
\[
 P\bigl(q_o=1,\ o\blue\partial B_n\bigr)\ \ge\
 \bigl(\mu\otimes\mathrm{Ber}(p)^{\otimes E(B_n)}\bigr)(\mathcal O_n).
\]
Here $\mathcal O_n$ is the event of Section~\ref{sec:ov-second-moment}, which
depends only on $(X|_{B_n},Y|_{E(B_n)})$.
\end{lemma}

\begin{proof}
Everything happens in the finite vector $(q|_{B_n},B|_{E(B_n)})$ and independent
auxiliary randomness. On $\{q_o=1\}$, explore the blue cluster of $o$ in
$E_+^{B_n}(q)$. Repeatedly pick, by a fixed rule depending on the transcript, an
unqueried edge $e=uv\in E_+^{B_n}(q)$ with $u$ reached and $v$ unreached; let
$r_e$ be the conditional probability of $\{B_e=1\}$ given $q|_{B_n}$ and the
indicators queried so far; reveal $B_e$; and draw $\tilde U_e$ uniformly from
$[0,r_e]$ if $B_e=1$ and from $(r_e,1]$ if $B_e=0$. If $B_e=1$, add $v$. When no
such edge remains, and on $\{q_o=-1\}$, give the unqueried edges fresh
independent uniforms.

Every earlier query in $\starv(v)$ failed, for otherwise $v$ would be reached.
The conditioning event is therefore $H\cap\mathsf F_v(F)$, where $F$ is the set
of earlier-queried neighbors of $v$ and $H\in\mathcal E_v^+$ fixes $q|_{B_n}$ and
the off-star queried indicators, with $H\subseteq\{q_u=q_v=1\}$; the earlier
auxiliary uniforms carry no further information. The hypothesis gives $r_e\ge p$.
Given the transcript, $B_e$ is Bernoulli($r_e$), so $\tilde U_e$ is uniform on
$[0,1]$ and independent of the transcript. Since
$B_e=\mathbf1\{\tilde U_e\le r_e\}$ on queried edges, the choice of the next
edge is a function of $q|_{B_n}$ and of the uniforms already revealed, so for
every $\eta$
the conditional law of $\tilde U$ given $q|_{B_n}=\eta$ has density one on
$[0,1]^{E(B_n)}$: $\tilde U$ is iid uniform and independent of $q|_{B_n}$. Put
$Y_e=\mathbf1\{\tilde U_e\le p\}$. On queried edges $Y_e=1$ forces
$\tilde U_e\le p\le r_e$, hence $B_e=1$. If a $Y$-open path in $E_+^{B_n}(q)$
from $o$ left the explored cluster through an edge $xy$, the exploration would
have queried $xy$ and found $B_{xy}=0$, hence $Y_{xy}=0$. So the $Y$-cluster of
$o$ in $E_+^{B_n}(q)$ is contained in the blue cluster of $o$ inside $B_n$.

Finally let $X$ be drawn from the conditional kernel of a monotone coupling of
$\mu$ and the law of $q|_{B_n}$, given $q|_{B_n}$, with independent randomness.
Then $X\le q|_{B_n}$, and since $\tilde U$ is independent of $(q|_{B_n},X)$,
the pair $(X,Y)$ has law $\mu\otimes\mathrm{Ber}(p)^{\otimes E(B_n)}$. On
$\mathcal O_n(X,Y)$ we have $X_o=1$, hence $q_o=1$. The witnessing path has
$q=1$ at its vertices and is $Y$-open, so it lies in the $Y$-cluster of $o$ in
$E_+^{B_n}(q)$, and its endpoint on $\partial B_n$ is joined to $o$ by a blue
path inside $B_n$.
\end{proof}

\paragraph{Rows and their three inputs.}
From now on a \emph{row} is a row of Table~\ref{tab:overlap-constants} or
Table~\ref{tab:ov-noisy}, or one of the rows of Section~\ref{sec:macrostep}
(Table~\ref{tab:ms-local}). Besides (C3)--(C4), which hold for every row, the
rest of this section uses a row only through the following three inputs, with
the row's own constants $\theta_*$ and $L_0$.
\begin{itemize}
\item[(R1)] \emph{Bond floor.} On every torus $\mathbb T_L$ with $L\ge4$, the
bound \eqref{eq:ov-floor-undivided} holds for all $(u,v)$, $F$ and $H$ as in
Lemma~\ref{lem:ov-torus-floor}. This is Lemma~\ref{lem:ov-torus-floor} with
(C1), or with (N1) and Lemma~\ref{lem:ov-pair}.
\item[(R2)] \emph{Holley line.} On every torus $\mathbb T_L$ with $L\ge4$, the
conclusion of Lemma~\ref{lem:ov-torus-odds} holds. This is
Proposition~\ref{prop:ov-holley-line} with (C2), or
Corollary~\ref{cor:ov-sym-holley} with (N2).
\item[(R3)] \emph{Uniform torus bound.} $P^{\rm aux}_L(\mathcal O_n)\ge\theta_*$
for every $n\ge1$ and every $L\ge L_0(n)$. This is
Corollary~\ref{cor:ov-uniform} with (C3)--(C6) and
Lemma~\ref{lem:ov-meanfield}, or Theorem~\ref{thm:ms-engine} for the rows of
Section~\ref{sec:macrostep}, where $L_0=L_0^{\rm M}$.
\end{itemize}
In every case $L_1(n)=\max\{L_0(n),2n+4\}$ and $L_1(n)\ge8$; for the rows of
this section this is \eqref{eq:ov-L1}.

\begin{proof}[Proof of Proposition~\ref{prop:overlap-central}]
Fix a row and let $n\ge1$ and $L\ge L_1(n)$; then $L\ge2n+4$, $L\ge4$ and
$L\ge L_0(n)$. By (R1), the law $\nu_L$ satisfies the bond floor
\eqref{eq:ov-floor-undivided}. By (R2), the conclusion of
Lemma~\ref{lem:ov-torus-odds} with $H=\{q_{V\setminus v}=\xi\}$ gives the
hypothesis of Lemma~\ref{lem:ov-holley} for the law of $q$ under $\nu_L$, with
$\Lambda=\bar\Lambda=V$. Hence $q$ dominates $\mu_L$ on $\mathbb T_L$, and restricting
the coupling, $q|_{B_n}$ dominates the $B_n$-marginal of $\mu_L$.
Lemma~\ref{lem:ov-exploration} gives
\[
 \nu_L\bigl(q_o=1,\ o\blue\partial B_n\bigr)\ \ge\ P^{\rm aux}_L(\mathcal O_n),
\]
and by (R3) the right side is at least $\theta_*$. No
conditioning on $q_o$ and no factor $\frac12$ is involved: the event
$\mathcal O_n$ already forces $X_o=1$, and $\rho_-<\frac12$ is the price of that
requirement.

For $s=-1$, the map $\Theta(J,\sigma,\tau,A)=(J,\sigma,-\tau,A)$ preserves
$\nu_L$, because $\mu_{\mathbb T_L,J,\beta}(-\tau)=\mu_{\mathbb T_L,J,\beta}(\tau)$ in zero
field. It leaves every blue indicator unchanged and maps $q$ to $-q$, as in the
proof of Corollary~\ref{cor:direct-cmr}. Hence
$\nu_L(q_o=-1,\ o\blue\partial B_n)=\nu_L(q_o=1,\ o\blue\partial B_n)\ge\theta_*$.
\end{proof}

For $L\ge2n+2$, reduction modulo $L$ is an isomorphism from the induced subgraph
of $\Z^d$ on $B_n$ onto the induced subgraph of $\mathbb T_L$ on its ball of radius
$n$. It is injective because $|x-y|_1\le2n<L$, onto because every point of the
torus ball has a lift of minimal $\ell^1$ norm, and torus edges come from
lattice edges because $x-y=\pm e_i+Lw$ with $w\ne0$ would force
$|x-y|_1\ge L-1>2n$. The event in Proposition~\ref{prop:overlap-central} is
therefore the same function of the same finitely many coordinates on $\mathbb T_L$ and
on $\Z^d$.

For almost-sure coexistence we must repeat the exploration under a limit law
conditioned on an exterior event. Proposition~\ref{prop:as-existence} does not
apply verbatim: its opening bound would have to hold given the full label
exterior at every edge, whereas here the bond floor holds only on
$\{q_u=q_v=1\}$, and the site layer is a dependent Ising field rather than
independent sites. We follow its exterior-conditioning idea and treat the two
layers separately.

\begin{lemma}[Local bounds in every selected limit]\label{lem:ov-limit-local}
Fix a row and let $\nu\in\mathcal L_{d,\beta}$. For every ordered pair $(u,v)$ of
neighbors in $\Z^d$, every $F\subseteq N(v)\setminus\{u\}$ and every
$H\in\mathcal E_v^+$ with $H\subseteq\{q_u=q_v=1\}$, the bond floor
\eqref{eq:ov-floor-undivided} holds with $\nu$ in place of $\nu_L$. For every
$v$ and every $H\in\mathcal E_v$,
$\nu(H\cap\{q_v=1\})\ge\E_\nu[\mathbf1_H\gamma_v(1\mid q_{N(v)})]$. Both
statements remain true for $\nu(\cdot\mid T)$ whenever $\nu(T)>0$ and
$T\in\mathcal E_v$ for every $v$.
\end{lemma}

\begin{proof}
Write the bounds as $\E[\mathbf1_H\mathcal Z]\ge0$ with the bounded cylinder
functions $\mathcal Z=\mathbf1\{q_u=q_v=1\}\mathbf1_{\mathsf F_v(F)}(B_{uv}-p)$
and $\mathcal Z=\mathbf1\{q_v=1\}-\gamma_v(1\mid q_{N(v)})$. For a cylinder event
$H$, both sides are expectations of cylinder functions of $(\sigma,\tau,B)$ in a
fixed ball; the torus inequalities (R1) and (R2) hold once the ball is
identified with a torus ball and $L\ge4$, and they pass to $\nu$ along its
subsequence. By dominated convergence, the
events $H$ with $\E_\nu[\mathbf1_H\mathcal Z]\ge0$ form a monotone class. It
contains the algebra of cylinder events of $\mathcal E_v^+$, respectively
$\mathcal E_v$, and therefore the generated $\sigma$-field. This is the argument
of Proposition~\ref{prop:local-transfer}. For the conditioned law apply the
bounds to $H\cap T$ and divide by $\nu(T)$.
\end{proof}

The site layer is compared, inside a finite box with minus boundary, with the
torus Ising law of (R3).

\begin{lemma}[Box versus torus]\label{lem:ov-box-torus}
Let $1\le n<R$, $L\ge2R+4$, $\alpha_{\rm D}<1$ and
$f:\{\pm1\}^{B_n}\to[0,1]$. Then
$|\mu^-_{B_R}(f)-\mu_L(f)|\le|B_n|\,\alpha_{\rm D}^{R-n}/(1-\alpha_{\rm D})$.
\end{lemma}

\begin{proof}
Apply Lemma~\ref{lem:ov-dobrushin} on $\Lambda=B_R$ with $\mu=\mu^-_{B_R}$ and its
heat-bath kernels, and with $\tilde\mu$ the $B_R$-marginal of $\mu_L$ and
$\tilde\gamma_x(\cdot\mid\omega)=\mu_L(X_x\in\cdot\mid X_{B_R\setminus x}=\omega)$.
As before $\mathsf C_{xy}\le t'\mathbf1\{x\sim y\}$ and $\mathsf D\le D^{\mathbb T}$. For
$x\in B_{R-1}$ all neighbors lie in $B_R$, so the Markov property gives
$\tilde\gamma_x=\gamma_x$ and $\bar b_x=0$; for $x\in\partial B_R$,
$\bar b_x\le1$. Also $\delta_y(f)\le\mathbf1\{y\in B_n\}$. For $y\in B_n$ and
$x\in\partial B_R$, every vector $\tilde x-\tilde y+Lw$ in
Lemma~\ref{lem:ov-resolvent}(i) has $\ell^1$ norm at least $R-n$, since
$L-(R+n)\ge R-n$, and distinct $(x,w)$ give distinct vectors. By
Lemma~\ref{lem:ov-resolvent}(iii),
$\sum_{x\in\partial B_R}D^{\mathbb T}_{yx}\le\alpha_{\rm D}^{R-n}/(1-\alpha_{\rm D})$; sum
over $y\in B_n$.
\end{proof}

For $s=\pm1$ let $G_s$ be the graph with vertex set $\{x:q_x=s\}$ and edge set
$\{xy:B_{xy}=1,\ q_x=q_y=s\}$, and let $T_s$ be the event that $G_s$ has no
infinite component.

\begin{proposition}[Almost-sure coexistence]\label{prop:ov-as}
For each row and every $\nu\in\mathcal L_{d,\beta}$, $\nu(T_+)=\nu(T_-)=0$.
\end{proposition}

\begin{proof}
Deleting a vertex and its star cannot change whether $G_s$ has an infinite
component, so $T_s\in\mathcal E_v$ for every $v$, exactly as in the proof of
Proposition~\ref{prop:as-existence}. Moreover the cylinder event
$\{B_e=1,\ q_x\ne q_y\}$ is empty on every torus by
Lemma~\ref{lem:ov-gauge}(iv), so $\nu$-almost surely every blue edge joins
equal overlaps.

Suppose $\nu(T_+)>0$ and put $\nu'=\nu(\cdot\mid T_+)$. By
Lemma~\ref{lem:ov-limit-local}, $\nu'$ satisfies both local bounds. Fix $n<R$
and let $\pi$ be the law of $q|_{B_{R+1}}$ under $\nu'$. The site bound with
$H=\{q_{B_{R+1}\setminus v}=\xi\}$ gives
$\pi(q_v=1\mid\xi)\ge\gamma_v(1\mid\xi_{N(v)})$ for $v\in B_R$, so
Lemma~\ref{lem:ov-holley} with $\Lambda=B_R$ and $\bar\Lambda=B_{R+1}$ shows that
$q|_{B_R}$ dominates $\mu^-_{B_R}$; restrict the coupling to $B_n$.
Lemma~\ref{lem:ov-exploration}, applied to $\nu'$, gives
$\nu'(q_o=1,\ o\blue\partial B_n)\ge\mu^-_{B_R}(f_n)$, where
$f_n(X)=\mathrm{Ber}(p)^{\otimes E(B_n)}(\mathcal O_n\mid X)\in[0,1]$. By
Lemma~\ref{lem:ov-box-torus}, which applies since $\alpha_{\rm D}<1$ by (C3),
and by (R3), with $L\ge\max(L_0(n),2R+4)$,
\[
 \mu^-_{B_R}(f_n)\ \ge\ \mu_L(f_n)-|B_n|\frac{\alpha_{\rm D}^{R-n}}{1-\alpha_{\rm D}}
 \ \ge\ \theta_*-|B_n|\frac{\alpha_{\rm D}^{R-n}}{1-\alpha_{\rm D}}.
\]
Letting $R\to\infty$ gives $\nu'(q_o=1,\ o\blue\partial B_n)\ge\theta_*$ for every
$n$. These events decrease in $n$, and their intersection is
$\{q_o=1,\ |C_{\rm b}(o)|=\infty\}$, so this event has $\nu'$-probability at least
$\theta_*>0$. But $\nu'$-almost surely the blue cluster of $o$ on $\{q_o=1\}$ is
a connected subgraph of $G_+$, which has no infinite component on $T_+$. This
contradiction shows $\nu(T_+)=0$.

Each $\nu_L$ is invariant under $\Theta$, which maps cylinder events to cylinder
events. Hence $\nu\circ\Theta^{-1}=\nu$ on cylinder events, which form a
$\pi$-system, and therefore on the generated $\sigma$-field. Since $\Theta$
fixes $B$ and maps $q$ to $-q$, $\Theta^{-1}(T_+)=T_-$, and
$\nu(T_-)=\nu(T_+)=0$.
\end{proof}

The site layer is thus handled by Holley domination inside a finite box, where
the minus boundary absorbs the unknown exterior, and by the Dobrushin comparison
of that box with the torus. No ergodicity, no uniqueness of the Ising Gibbs
measure and no exploration of an infinite volume is used.

\begin{remark}[Scope and consequences]\label{rem:ov-scope}
Theorem~\ref{thm:explicit} concerns blue geometry in the selected periodic
limits at the listed temperatures; no monotonicity in $\beta$ is available to
fill in intervals. It asserts neither spin-glass order nor an imbalance of
densities. At the same temperatures, Corollary~\ref{cor:explicit-balance} adds
a susceptibility bound and exact balance of the two sectors. The
$\pm1$ law has no atom at zero, so the discussion at the beginning of
Section~\ref{sec:discussion} applies: combined with the at-most-two bound of
\cite[Theorem~1.9(b)]{Pei2026}, there are $\nu$-almost surely exactly two
infinite blue components, one of each overlap sign. The proof below does not
use that upper bound. The root bound of Theorem~\ref{thm:explicit} also gives
positive mean densities. For $\Lambda\subseteq\Z^d$ finite let
$D_s(\Lambda)$ be the fraction of $x\in\Lambda$ with $q_x=s$ and
$|C_{\rm b}(x)|=\infty$. Translations of the torus preserve $\nu_L$, so every
$\nu\in\mathcal L_{d,\beta}$ is invariant under translations of cylinder events;
applying continuity from above at each $x$ gives
$\E_\nu D_s(\Lambda)=\nu(q_o=s,\ |C_{\rm b}(o)|=\infty)\ge\theta_*$, and Fatou's
lemma applied to $1-D_s$ gives
$\E_\nu[\limsup_ND_s(\{-N,\ldots,N\}^d)]\ge\theta_*$ for $s=\pm1$.
\end{remark}

\begin{proof}[Proof of Theorem~\ref{thm:explicit}]
Every pair of Table~\ref{tab:explicit} belongs to a row in the above sense,
and $(9,\frac7{50})$ and $(8,\frac3{20})$ to two rows each. Fix a row,
$\nu\in\mathcal L_{d,\beta}$ along $L_k\to\infty$, and $s=\pm1$. For
fixed $n$, the event $\mathcal A_n^s=\{q_o=s,\ o\blue\partial B_n\}$ depends only
on $(\sigma_o,\tau_o)$ and the blue indicators of $B_n$, and is the same
cylinder event on $\mathbb T_{L_k}$ and on $\Z^d$ once $L_k\ge2n+2$. By
Proposition~\ref{prop:overlap-central}, or by Proposition~\ref{prop:ms-central}
for the rows of Section~\ref{sec:macrostep},
$\nu(\mathcal A_n^s)=\lim_k\nu_{L_k}(\mathcal A_n^s)\ge\theta_*$. This is the order
of limits of Figure~\ref{fig:local-transfer}: a fixed observation first, then
$L\to\infty$, and the radius last. A blue path from $o$ to $\partial B_{n+1}$
inside $B_{n+1}$, stopped at its first visit to $\partial B_n$, stays in $B_n$, so
$\mathcal A_{n+1}^s\subseteq\mathcal A_n^s$, and
$\bigcap_n\mathcal A_n^s=\{q_o=s,\ |C_{\rm b}(o)|=\infty\}$. Continuity from
above gives $\nu(o\blue\infty,\ q_o=s)\ge\theta_*$.

By Proposition~\ref{prop:ov-as}, $\nu$-almost surely each $G_s$ has an infinite
component. Since $\nu$-almost surely every blue edge joins equal overlaps, the
components of $G_s$ are exactly the blue clusters of the vertices with $q=s$.
Hence $\nu$-almost surely there is an infinite blue cluster on which $q\equiv1$
and one on which $q\equiv-1$. For every version of $\nu(T_+\cup T_-\mid J)$,
this conditional probability is nonnegative with $\nu$-mean zero, hence vanishes
for $\nu$-almost every $J$; in particular every regular conditional law of
$(\sigma,\tau,B)$ given $J$ assigns $T_+\cup T_-$ probability zero for
$\nu$-almost every $J$. This is the reading of probability-one statements in
Section~\ref{sec:model}.
\end{proof}
\FloatBarrier

\section{A macrostep engine and dimensions seven to nine}\label{sec:macrostep}

Section~\ref{sec:overlap-route} uses its global step only through the uniform
torus bound (R3): for the auxiliary law
$P^{\rm aux}_L=\mu_L\otimes\mathrm{Ber}(p)^{\otimes E}$ of
Section~\ref{sec:ov-second-moment},
\begin{equation}\label{eq:ms-star}
 P^{\rm aux}_L(\mathcal O_n)\ \ge\ \theta_*\qquad
 \text{for all }n\ge1\text{ and }L\ge L_0(n).
\end{equation}
The local inputs, the quantile coupling and the passage to periodic limits use
nothing else about the global step. With the local inputs of this paper, the
oriented criterion gives \eqref{eq:ms-star} down to dimension $9$ but does not
certify dimension $8$: after two oriented paths meet, they share their next
step with probability $1/d$, and each shared step costs $(\rho_ep)^{-1}$; at
the local inputs certified in dimension $8$ below, $8\rho_cp<0.99$
(Remark~\ref{rem:ov-dim9}).

This section replaces oriented paths by macrostep paths. A macrostep makes a
short excursion in three lateral coordinates, along a word in which no
coordinate occurs with both signs, and then one step in one of the remaining
$f=d-3$ forward coordinates. Paths made of such blocks go back to Kesten's
block paths \cite[Section~2, Eqs.~(2.1)--(2.3)]{Kesten1990} and, in lazy form,
to the macrosteps of \cite[Section~4]{JiangLang2026}; we use monotone lateral
words in three coordinates, which keep the paths self-avoiding. Two macrostep
paths meet only at equal forward levels, as block paths do
\cite[p.~226]{Kesten1990}. Each level offers many more choices than an oriented
step, and two paths share a forward step only if their excursions end at the
same point and their forward directions agree.

The price is that the second moment is no longer a renewal series in one
scalar. The relative position of two paths after $j$ macrosteps is a Markov
chain, and the weight of a pair of paths is bounded by a product of
per-macrostep weights along it (Proposition~\ref{prop:ms-product}). Shared
vertices and edges are counted step by step, and the Ising correlations between
the unshared vertices of the two paths are charged, level by level, to single
macrosteps (Lemma~\ref{lem:ms-boost}). The product is bounded by expanding it in
the excess of the weights over $1$, as in the interaction-matrix method of
\cite[Section~4]{JiangLang2026}. There the excess kernel is supported on a
finite set; here it is not, because the Ising correlations decay but never
vanish. Theorem~\ref{thm:ms-criterion} handles the far region with a
supersolution and Khas'minskii's lemma \cite{Khasminskii1959} for the forward
walk, which is transient because $f\ge4$. A certificate for this criterion
gives \eqref{eq:ms-star} (Theorem~\ref{thm:ms-engine}), and the local inputs and
the assembly of Section~\ref{sec:overlap-route} then give
Theorem~\ref{thm:explicit} in dimensions $7$, $8$ and $9$. Apart from the
oriented criterion of Section~\ref{sec:ov-second-moment}, we are not aware of
an earlier second-moment criterion for site--bond percolation with an Ising
site field on $\Z^d$ that is certified at an explicit finite dimension.

\paragraph{The statement proved.}
Table~\ref{tab:ms-local} lists five rows. Each gives a dimension $d=f+3$, a
temperature $t=\tanh\beta$, a bond parameter $p$, an Ising line
$g_K=e^{2K'}$, $g_h=e^{2h}$ and a rational $\bar m$, all exact, together with
local inputs from Sections~\ref{sec:ov-floor}--\ref{sec:ov-holley} or from
Section~\ref{sec:ov-noisy}. Table~\ref{tab:ms-certificates} lists, for each row,
one or more near sets with two certificates each. For each row let $\theta_*$
be the constant \eqref{eq:ms-theta} of the certificate with the largest
$\theta_*$, and $L_0^{\rm M}$ the function \eqref{eq:ms-L0}.

\begin{proposition}[Uniform torus bound in dimensions seven to nine]\label{prop:ms-central}
For each row of Table~\ref{tab:ms-local}, every $n\ge1$, every
$L\ge L_0^{\rm M}(n)$ and $s=\pm1$,
\[
 \nu_{\mathbb T_L^d,\beta}\bigl(q_o=s,\ o\blue\partial B_n\bigr)\ \ge\ \theta_*,
\]
where the blue connection uses paths inside $B_n$.
\end{proposition}

As in Section~\ref{sec:overlap-route}, the row fixes $(d,\beta)$, the rational
inputs and the certificate, hence $\theta_*$ and $L_0^{\rm M}$; then $n$, then
$L\ge L_0^{\rm M}(n)$, then $s$ are arbitrary. Theorem~\ref{thm:explicit}
follows for these rows in Section~\ref{sec:ov-assembly}. The proposition is
proved in Section~\ref{sec:ms-proofs}, after the geometry of macrostep paths
(Section~\ref{sec:ms-paths}), the product bound (Section~\ref{sec:ms-product})
and the criterion (Section~\ref{sec:ms-criterion}).

\subsection{Macrostep paths and their geometry}\label{sec:ms-paths}

Fix $d=f+3$ with $f\ge4$, a word-length cap $c\ge1$ and a word weight
$y_{\rm w}\in(0,1)$; the certificates use $c=3$. Write
$x=(\operatorname{lat}(x),\operatorname{fwd}(x))\in\Z^3\times\Z^f$, call the
first three coordinates \emph{lateral} and the other $f$ \emph{forward}, and put
$\operatorname{lev}(x)=\sum_{i=1}^f\operatorname{fwd}(x)_i$, the \emph{level} of
$x$. Let $e_1,e_2,e_3$ be the lateral and $\epsilon_1,\ldots,\epsilon_f$ the
forward unit vectors.

A \emph{word} of length $l\le c$ is a sequence $\xi=(\xi_1,\ldots,\xi_l)$ of
lateral unit vectors $\pm e_q$ in which no coordinate occurs with both signs;
$\mathcal W_c$ is the set of words, including the empty word. Put
$P_k(\xi)=\xi_1+\cdots+\xi_k$, with $P_0(\xi)=0$, and $e(\xi)=P_{|\xi|}(\xi)$.
The word law is $\pi_{\rm w}(\xi)=y_{\rm w}^{|\xi|}/Z_{\rm w}$ with
$Z_{\rm w}=\sum_{\xi\in\mathcal W_c}y_{\rm w}^{|\xi|}$. For $c=3$ there are
$1+6+30+126=163$ words, and $Z_{\rm w}=1+6y_{\rm w}+30y_{\rm w}^2+126y_{\rm w}^3$.

A \emph{macrostep} is a pair $\omega=(\xi,i)\in\mathcal W_c\times[f]$, with law
$\pi_{\rm w}\otimes\mathrm{Unif}[f]$. For $\omega_{1:n}=(\omega_1,\ldots,\omega_n)$
put $x_0=o$ and $x_j=x_{j-1}+e(\xi_j)+\epsilon_{i_j}$. The \emph{excursion} of
macrostep $j$ is $\zeta_j(k)=x_{j-1}+P_k(\xi_j)$, $0\le k\le|\xi_j|$, and the
path $\gamma(\omega_{1:n})$ visits
\[
 \zeta_1(0),\ldots,\zeta_1(|\xi_1|),\ \zeta_2(0),\ldots,\zeta_2(|\xi_2|),\ \ldots,\
 \zeta_n(0),\ldots,\zeta_n(|\xi_n|),\ x_n
\]
in this order. Consecutive points are nearest neighbors: a lateral step within
an excursion, or the forward step from $\zeta_j(|\xi_j|)$ to $x_j$, which is
$\zeta_{j+1}(0)$ for $j<n$.
Write $V(\gamma)$ and $E(\gamma)$ for its vertex and edge sets, and let
$V^{(a)}(\gamma)=\{\zeta_{a+1}(k)\}_k$ for $0\le a<n$ and
$V^{(n)}(\gamma)=\{x_n\}$ be its level sets. Figure~\ref{fig:ms-paths} shows two
such paths.

\begin{figure}[!htbp]
\centering
\begin{tikzpicture}[x=1.9cm,y=0.6cm,>=Latex,
 ptA/.style={circle,fill=white,draw=blue!60!black,inner sep=1.3pt},
 ptB/.style={circle,fill=white,draw=red!60!black,inner sep=1.3pt},
 shared/.style={circle,fill=black,inner sep=1.7pt},
 pathA/.style={thick,blue!60!black},
 pathB/.style={thick,red!60!black,dashed},
 pair/.style={densely dotted,thick,gray!60!black}]
\foreach \a in {0,...,4} {
 \draw[gray!25] (\a,-1.7)--(\a,3.7);
 \node[font=\scriptsize,text=gray!60!black] at (\a,-2.3) {level $\a$};}
\draw[pathA] (0,0)--(0,2)--(1,2)--(1,3)--(2,3)--(2,1)--(3,1)--(4,1);
\draw[pathB] (0,0)--(0,-1)--(1,-1)--(1,1)--(2,1)--(2,0)--(3,0)--(3,-1)--(4,-1);
\foreach \p in {(0,1),(0,2),(1,2),(1,3),(2,3),(2,2),(3,1),(4,1)}
 \node[ptA] at \p {};
\foreach \p in {(0,-1),(1,-1),(1,0),(1,1),(2,0),(3,0),(3,-1),(4,-1)}
 \node[ptB] at \p {};
\node[shared] at (0,0) {};
\node[shared] at (2,1) {};
\node[font=\scriptsize,left] at (0,0) {$o$};
\draw[pair] (1,2) to[bend left=35] node[font=\scriptsize,right] {gap $0$} (1,1);
\draw[pair] (3,-1) -- node[font=\scriptsize,right,pos=0.45] {gap $1$} (4,1);
\draw[pair] (0,-1) to[bend right=25] node[font=\scriptsize,below,pos=0.55] {gap $3$} (3,1);
\node[font=\scriptsize,text=blue!60!black,anchor=west] at (4.1,1) {$\gamma$};
\node[font=\scriptsize,text=red!60!black,anchor=west] at (4.1,-1) {$\gamma'$};
\end{tikzpicture}
\caption{Two macrostep paths with $n=4$, projected onto the level (horizontal)
and one lateral coordinate (vertical). Each macrostep is a lateral excursion
along a monotone word (a vertical segment) followed by one forward step (a
horizontal segment). Points on different levels are distinct, and the two paths
can share a point only on a common level where their forward coordinates agree
(black dots, where $D_0=D_2=0$ is assumed). In Lemma~\ref{lem:ms-boost} the
correlation between two unshared points is charged to a single macrostep,
according to the gap between their levels: a pair at gap $0$ to the macrostep
of the common level, a pair at gap $1$ or $2$ to the macrostep of the earlier
level, and a pair at gap at least $3$ to that of the later level.}
\label{fig:ms-paths}
\end{figure}
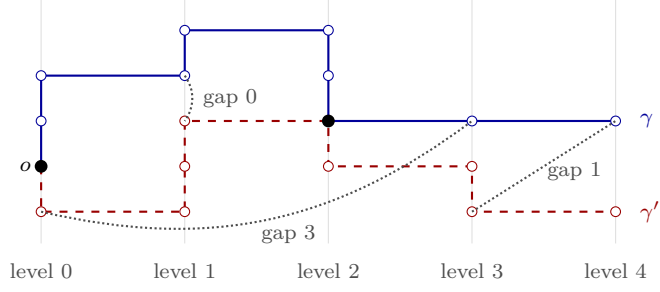

For independent $\omega_{1:n}$ and $\omega'_{1:n}$ with the macrostep law, write
primes for the quantities of $\gamma'=\gamma(\omega'_{1:n})$. Let
$\mathbb A^{(f)}=\{D\in\Z^f:\sum_iD_i=0\}$, and put
\[
 Z_j=(A_j,D_j)=\bigl(\operatorname{lat}(x_j-x'_j),\,\operatorname{fwd}(x_j-x'_j)\bigr)
 \in\Z^3\times\mathbb A^{(f)}
\]
and $\boldsymbol\omega_j=(\omega_j,\omega'_j)$. For a state $z=(A,D)$ and a pair
$\boldsymbol\omega=((\xi,i),(\xi',i'))$ put
\begin{gather*}
 N_V(z;\boldsymbol\omega)=\#\bigl\{(k,k'):D=0,\ A+P_k(\xi)=P_{k'}(\xi')\bigr\},\\
 N_{\rm f}(z;\boldsymbol\omega)=\mathbf1\{D=0,\ A+e(\xi)=e(\xi'),\ i=i'\},
\end{gather*}
and let $N_E(z;\boldsymbol\omega)$ be the number of $1\le k\le|\xi|$ such that
$k-1$ and $k$ are matched, in the sense of $N_V$, to indices $k'_1$ and $k'_2$
with $|k'_1-k'_2|=1$.

\begin{lemma}[Monotone words]\label{lem:ms-words}
\begin{enumerate}\renewcommand{\labelenumi}{(\alph{enumi})}
\item $|P_k(\xi)-P_{k'}(\xi)|_1=|k-k'|$ for all $0\le k,k'\le|\xi|$.
\item $\operatorname{lev}(\zeta_j(k))=j-1$ and $\operatorname{lev}(x_n)=n$, and the
points of $\gamma$ listed above are pairwise distinct.
\item Every vertex $v$ of $\gamma$ has $|\operatorname{lat}(v)_q|\le cn$ for
$q=1,2,3$ and $\operatorname{fwd}(v)\in\{0,\ldots,n\}^f$.
\end{enumerate}
\end{lemma}

\begin{proof}
(a) Within a word, all letters acting on a coordinate have the same sign, so
between positions $k<k'$ each coordinate moves monotonically by the number of
its letters, and the $\ell^1$ distance is the total number of letters, $k'-k$.
(b) Lateral steps do not change $\operatorname{fwd}$, and each forward step
raises the level by one, so points on different levels are distinct. On level
$a<n$ the points are the $\zeta_{a+1}(k)$, which are distinct by (a), and on
level $n$ there is only $x_n$. (c) Macrostep $j$ changes each lateral
coordinate by at most $|\xi_j|\le c$ and each forward coordinate by $0$ or $1$.
\end{proof}

\begin{lemma}[Meetings only on equal levels]\label{lem:ms-meetings}
For two macrostep paths $\gamma=\gamma(\omega_{1:n})$ and $\gamma'=\gamma(\omega'_{1:n})$
in $\Z^d$:
\begin{enumerate}\renewcommand{\labelenumi}{(\alph{enumi})}
\item $V^{(a)}(\gamma)\cap V^{(a')}(\gamma')=\varnothing$ if $a\ne a'$. For $a<n$,
$\zeta_{a+1}(k)=\zeta'_{a+1}(k')$ if and only if $D_a=0$ and
$A_a+P_k(\xi_{a+1})=P_{k'}(\xi'_{a+1})$, and for each $k$ there is at most one
such $k'$.
\item $|V(\gamma)\cap V(\gamma')|=\sum_{j=1}^nN_V(Z_{j-1};\boldsymbol\omega_j)+\mathbf1\{Z_n=0\}$.
\item A lateral edge of $\gamma$ on level $a$ can coincide only with a lateral
edge of $\gamma'$ on level $a$, and a forward edge from level $a$ to $a+1$ only
with the forward edge of $\gamma'$ between the same levels. Consequently
$|E(\gamma)\cap E(\gamma')|=\sum_{j=1}^n(N_E+N_{\rm f})(Z_{j-1};\boldsymbol\omega_j)$.
\item $(V(\gamma)\cap V(\gamma'),E(\gamma)\cap E(\gamma'))$ is a disjoint union
of vertex-disjoint simple paths, isolated vertices included, with exactly
$|V(\gamma)\cap V(\gamma')|-|E(\gamma)\cap E(\gamma')|$ components.
\item The increments
$Z_j-Z_{j-1}=(e(\xi_j)-e(\xi'_j),\ \epsilon_{i_j}-\epsilon_{i'_j})$ are iid and
independent of $Z_{j-1}$. The components $(A_j)$ and $(D_j)$ are independent
random walks, and $(D_j)$ is the difference of two independent uniform oriented
walks in $\Z^f$, as in Section~\ref{sec:ov-second-moment} with $f$ in place of
$d$.
\end{enumerate}
\end{lemma}

\begin{proof}
(a) Levels differ by Lemma~\ref{lem:ms-words}(b). On level $a$ both excursions
have the forward coordinates of $x_a$ and $x'_a$, so equality needs $D_a=0$ and
equality of the lateral parts, $\operatorname{lat}(x_a)+P_k=\operatorname{lat}(x'_a)+P_{k'}$;
uniqueness of $k'$ follows from Lemma~\ref{lem:ms-words}(a).
(b) This follows from (a); the level-$n$ term is $\mathbf1\{x_n=x'_n\}$.
(c) Lateral edges join two points of one level, and forward edges join levels
$a$ and $a+1$. A lateral edge $\{\zeta(k-1),\zeta(k)\}$ of $\gamma$ is an edge of
$\gamma'$ if and only if both endpoints are matched to points that are
consecutive on the simple path $\gamma'$. The forward edges from level $a$
coincide if and only if their lower endpoints coincide, that is, $D_a=0$ and
$A_a+e(\xi_{a+1})=e(\xi'_{a+1})$, and in addition $x_{a+1}=x'_{a+1}$, that is,
$i_{a+1}=i'_{a+1}$.
(d) The edges $E(\gamma)\cap E(\gamma')$ are edges of the simple path $\gamma$,
and their endpoints lie in $V(\gamma)\cap V(\gamma')$.
(e) The macrosteps are iid, the increment is a function of $\boldsymbol\omega_j$
alone, and the lateral and forward parts of $\boldsymbol\omega_j$ are
independent.
\end{proof}

\begin{lemma}[No wrapping on the torus]\label{lem:ms-torus}
Let $L\ge2cn+1$ and reduce $n$-macrostep paths from $o$ modulo $L$.
\begin{enumerate}\renewcommand{\labelenumi}{(\alph{enumi})}
\item Reduction modulo $L$ is injective on $V(\gamma)\cup V(\gamma')$ for any two
such paths. Hence the torus paths are self-avoiding, and their shared vertices,
shared edges and all counts of Lemma~\ref{lem:ms-meetings} are the same on
$\mathbb T_L$ as in $\Z^d$.
\item The torus distance from $o$ to $x_n$ is at least $n$. Hence, if a torus
macrostep path is open, its initial segment up to its first vertex at distance
$n$ from $o$ lies in $B_n$ and witnesses $\mathcal O_n$.
\end{enumerate}
\end{lemma}

\begin{proof}
(a) For vertices $u,v$ of the two paths in $\Z^d$,
$|\operatorname{lat}(u-v)_q|\le2cn<L$ and $|\operatorname{fwd}(u-v)_i|\le n<L$
by Lemma~\ref{lem:ms-words}(c), so $u\equiv v$ modulo $L$ only if $u=v$.
(b) The forward coordinates of $x_n$ lie in $[0,n]$ with $n<L/2$, so they
contribute $\sum_i\operatorname{fwd}(x_n)_i=n$ to the torus distance. Along the
path the distance to $o$ changes by at most one per step, so it takes the value
$n$; the vertices before its first occurrence are at distance less than $n$.
\end{proof}

The condition $L\ge2cn+1$ cannot be dropped: for $L=3$, $c=3$, $n=1$ and two
copies of the word $(e_1,e_1,e_1)$ with the same forward step, the excursion
returns to $o$ modulo $3$, and the shared edges contain a triangle.

\subsection{The weighted second moment for macrostep paths}\label{sec:ms-product}

Let $\nu_n$ be the law of $\omega_{1:n}$. On $\mathbb T_L$ with $L\ge2cn+1$,
write $\gamma=\gamma(\omega_{1:n})$ for the reduced path, put
$\pi_L(\gamma)=P^{\rm aux}_L(\gamma\text{ open})=p^{|E(\gamma)|}\mu_L(V(\gamma)\subseteq+)$,
and
\[
 W_n=\sum_{\omega_{1:n}}\nu_n(\omega_{1:n})\,
 \frac{\mathbf1\{\gamma(\omega_{1:n})\text{ open}\}}{\pi_L(\gamma(\omega_{1:n}))}.
\]

\begin{lemma}[Weighted Paley--Zygmund bound for macrostep paths]\label{lem:ms-pz}
For $L\ge2cn+1$, $\E W_n=1$,
\[
 \E W_n^2=\E_{\nu_n\otimes\nu_n}\Bigl[p^{-|E(\gamma)\cap E(\gamma')|}\,
 \mathcal R_L\bigl(V(\gamma),V(\gamma')\bigr)\Bigr],
\]
and $P^{\rm aux}_L(\mathcal O_n)\ge P^{\rm aux}_L(W_n>0)\ge1/\E W_n^2$.
\end{lemma}

\begin{proof}
The proof of Lemma~\ref{lem:ov-pz} applies. It uses only the independence of
$X$ and $Y$ and the identity
$|E(\gamma)\cup E(\gamma')|=|E(\gamma)|+|E(\gamma')|-|E(\gamma)\cap E(\gamma')|$,
where the edges of each path are distinct by Lemma~\ref{lem:ms-torus}(a). By
Lemma~\ref{lem:ms-torus}(b), $W_n>0$ implies $\mathcal O_n$.
\end{proof}

This is the self-normalized path weighting of \cite[Eq.~(5.16)]{LyonsPeres2016};
for block paths and independent sites, \cite[Lemma~1, Eq.~(2.9)]{Kesten1990}
has the same form. The next lemma plays the role of the quasi-independence
bound of \cite[Theorem~5.19]{LyonsPeres2016} for a dependent site field.

\begin{lemma}[Pair ratio for macrostep paths]\label{lem:ms-pair}
Assume \eqref{eq:ov-H} and the hypotheses of Lemma~\ref{lem:ov-shared}. For
$n\ge1$, $L\ge2cn+1$ and any two $n$-macrostep torus paths from $o$,
\[
 \mathcal R_L\bigl(V(\gamma),V(\gamma')\bigr)\ \le\
 \rho_-^{-(|V\cap V'|-|E\cap E'|)}\rho_c^{-|E\cap E'|}
 \exp\Bigl(\kappa\sum_{x\in V'\setminus V,\,y\in V\setminus V'}D^{\mathbb T}_{xy}\Bigr),
\]
where $V=V(\gamma)$, $V'=V(\gamma')$, $E=E(\gamma)$ and $E'=E(\gamma')$.
\end{lemma}

\begin{proof}
Lemma~\ref{lem:ov-pair-ratio} with $\Lambda_0=V\cap V'$, $\Lambda_1=V\setminus V'$
and $\Lambda_2=V'\setminus V$, and $\kappa_L\le\kappa$, gives
$\mathcal R_L\le\mu_L(V\cap V'\subseteq+)^{-1}\exp(\kappa\sum D^{\mathbb T}_{xy})$.
By Lemma~\ref{lem:ms-torus}(a), the graph $(V\cap V',E\cap E')$ on $\mathbb T_L$
is isomorphic to the corresponding graph in $\Z^d$, which is a disjoint union
of paths by Lemma~\ref{lem:ms-meetings}(d). Lemma~\ref{lem:ov-shared} gives
$\mu_L(V\cap V'\subseteq+)\ge\rho_-^{|V\cap V'|-|E\cap E'|}\rho_c^{|E\cap E'|}$.
\end{proof}

\paragraph{The boost tables.}
Let $t'=\tanh K'$ and assume $\lambda_*=dt'/(1-2dt')<1$. Let $\mathcal B$ be as in
Lemma~\ref{lem:ov-resolvent}(iv), in dimension $d$ and with $t_0=t'$, so that
$\sup_{|u|_1=r}\mathcal D_{t'}(u)\le\mathcal B(r)$ and
$\mathcal B(r+1)\le\lambda_*\mathcal B(r)$; in particular $\mathcal B$ is
nonincreasing. For integers $F\ge1$ and $r,r_0,D_1\ge0$ put
\begin{gather*}
 \mathcal T_1(F,r)=\sum_{k=0}^c\mathcal B\bigl(F+|r-k|\bigr),\qquad
 \mathcal T_2(r,F)=\max_{\varrho\ge r}\sum_{k=0}^c\mathcal B\bigl(F+|\varrho-k|\bigr),\\
 \mathcal T_3(r_0,D_1)=\sum_{s\ge3}\mathcal T_2\bigl(\max(0,r_0-c(s-1)),\
 \max(s,D_1-s)\bigr).
\end{gather*}
The maximum in $\mathcal T_2$ is attained for $\varrho\in[r,\max(r,c)]$, since
for $\varrho\ge c$ every summand is nonincreasing in $\varrho$. Both
$\mathcal T_2$ and $\mathcal T_3$ are nonincreasing in each argument, and
$\mathcal T_1$ is nonincreasing in $F$, and in $r$ for $r\ge c$. With
$g=1+2\lambda_*/(1-\lambda_*)$ we have $\mathcal T_2(0,F)\le g\,\mathcal B(F)$,
because $\mathcal B(F+j)\le\lambda_*^j\mathcal B(F)$ and at most two summands
have $|\varrho-k|=j$ for each $j\ge1$; in particular the series
$\mathcal T_3$ converges.

\paragraph{The boost decomposition.}
For a state $z=(A,D)$ and $\boldsymbol\omega=((\xi,i),(\xi',i'))$, let
$\mathsf U\subseteq\{0,\ldots,|\xi|\}$ and $\mathsf U'\subseteq\{0,\ldots,|\xi'|\}$
be the indices of the points of the two excursions that are not matched in the
sense of Lemma~\ref{lem:ms-meetings}(a). Put $D_1=|D|_1$,
\[
 F_1=|D-\epsilon_{i'}|_1,\quad F_2=\min_a|D-\epsilon_{i'}-\epsilon_a|_1,\quad
 F'_1=|D+\epsilon_i|_1,\quad F'_2=\min_a|D+\epsilon_i+\epsilon_a|_1,
\]
and for $k\in\mathsf U$ and $k'\in\mathsf U'$
\begin{gather*}
 r_1(k)=|A+P_k(\xi)-e(\xi')|_1,\qquad r_0(k)=|A+P_k(\xi)|_1,\\
 r'_1(k')=|P_{k'}(\xi')-A-e(\xi)|_1,\qquad r'_0(k')=|P_{k'}(\xi')-A|_1.
\end{gather*}
Define $\mathfrak b=\mathfrak b_0+\mathfrak b_1+\mathfrak b_2+\mathfrak b_3$, a
function of $(z;\boldsymbol\omega)$, by
\begin{gather*}
 \mathfrak b_0=\sum_{k\in\mathsf U,\,k'\in\mathsf U'}\mathcal B\bigl(D_1+|A+P_k(\xi)-P_{k'}(\xi')|_1\bigr),\\
 \mathfrak b_1=\sum_{k\in\mathsf U}\mathcal T_1\bigl(F_1,r_1(k)\bigr)
 +\sum_{k'\in\mathsf U'}\mathcal T_1\bigl(F'_1,r'_1(k')\bigr),\\
 \mathfrak b_2=\sum_{k\in\mathsf U}\mathcal T_2\bigl((r_1(k)-c)^+,F_2\bigr)
 +\sum_{k'\in\mathsf U'}\mathcal T_2\bigl((r'_1(k')-c)^+,F'_2\bigr),\\
 \mathfrak b_3=\sum_{k\in\mathsf U}\mathcal T_3\bigl(r_0(k),D_1\bigr)
 +\sum_{k'\in\mathsf U'}\mathcal T_3\bigl(r'_0(k'),D_1\bigr).
\end{gather*}
Every argument of $\mathcal B$ is at least $1$: $F_1,F'_1\ge1$ because the
coordinate sum of $D\mp\epsilon$ is $\mp1$, $F_2,F'_2\ge2$, the two points in
$\mathfrak b_0$ are distinct, and $\max(s,\cdot)\ge3$ in $\mathcal T_3$.

\begin{lemma}[The Ising boost along the pair chain]\label{lem:ms-boost}
For $L\ge2cn+1$ and any two $n$-macrostep paths from $o$,
\[
 \sum_{x\in V'\setminus V,\ y\in V\setminus V'}D^{\mathbb T}_{xy}\ \le\
 \sum_{j=1}^n\mathfrak b(Z_{j-1};\boldsymbol\omega_j)+\mathcal B(1)
 +2\mathcal T_3(0,0)+\varepsilon'_{n,L},
\]
where $\varepsilon'_{n,L}=((c+1)n+1)^2\,\alpha_{\rm D}^{L-2cn}/(1-\alpha_{\rm D})$.
\end{lemma}

\begin{proof}
\emph{Step 1 (lifts).} Let $\tilde x,\tilde y$ be the points of the paths in
$\Z^d$. By Lemma~\ref{lem:ov-resolvent}(i),
$D^{\mathbb T}_{xy}=\sum_{w\in\Z^d}\mathcal D_{t'}(\tilde y-\tilde x+Lw)$. The
term $w=0$ is at most $\mathcal B(|\tilde y-\tilde x|_1)$, since
$\tilde x\ne\tilde y$. For $w\ne0$ some coordinate of $\tilde y-\tilde x+Lw$ has
absolute value at least $L-2cn$, by Lemma~\ref{lem:ms-words}(c); distinct $w$
give distinct points, so by Lemma~\ref{lem:ov-resolvent}(iii) these terms sum
to at most $\alpha_{\rm D}^{L-2cn}/(1-\alpha_{\rm D})$. There are at most
$|V||V'|\le((c+1)n+1)^2$ pairs. It remains to show that
$\Sigma=\sum_{x\in V'\setminus V,\,y\in V\setminus V'}\mathcal B(|\tilde x-\tilde y|_1)$
is at most $\sum_j\mathfrak b(Z_{j-1};\boldsymbol\omega_j)+\mathcal B(1)+2\mathcal T_3(0,0)$.
From now on we work in $\Z^d$ and drop the tildes.

\emph{Step 2 (charging by levels).} Put
$U_a=V^{(a)}(\gamma)\setminus V^{(a)}(\gamma')$ and
$U'_a=V^{(a)}(\gamma')\setminus V^{(a)}(\gamma)$. By
Lemma~\ref{lem:ms-meetings}(a), $V\setminus V'$ and $V'\setminus V$ are the
disjoint unions of the $U_a$ and of the $U'_a$. Split $\Sigma$ by the levels of
$y$ and $x$ and charge each pair once:
\begin{itemize}
\item level gap $0$ below level $n$: to macrostep $a+1$ at the common level
$a<n$, through $\mathfrak b_0$;
\item level gap $0$ on level $n$: to the end, through $\mathcal B(1)$;
\item level gap $1$ or $2$: to the vertex on the earlier level $e$, at
macrostep $e+1$, through its $\mathcal T_1$ and $\mathcal T_2$ terms;
\item level gap at least $3$: to the vertex on the later level $g$, at
macrostep $g+1$, or to the end if $g=n$, through its $\mathcal T_3$ term.
\end{itemize}
Every unordered cross pair falls in exactly one case, and the charged vertex is
unshared, because it lies in some $U_a$ or $U'_a$.

\emph{Step 3 (gap $0$).} For $a<n$, $y=\zeta_{a+1}(k)$ with $k\in\mathsf U$ and
$x=\zeta'_{a+1}(k')$ with $k'\in\mathsf U'$; the forward part of $y-x$ is $D_a$
and the lateral part is $A_a+P_k(\xi_{a+1})-P_{k'}(\xi'_{a+1})$, so these pairs
sum to exactly $\mathfrak b_0(Z_a;\boldsymbol\omega_{a+1})$. On level $n$ the only
possible pair is $(x'_n,x_n)$ with $x_n\ne x'_n$, which contributes at most
$\mathcal B(1)$.

\emph{Step 4 (gaps $1$ and $2$).} Let $y=\zeta_{e+1}(k)\in U_e$ be a vertex of
$\gamma$ on the earlier level $e<n$; a vertex of $\gamma'$ on the earlier level
is treated in the same way, with primed quantities. The vertices of $\gamma'$ on
level $e+1$ are points $x'_{e+1}+u$ with $u$ lateral and
$|u|_1=k''\in\{0,\ldots,c\}$, at most one for each $k''$, by
Lemma~\ref{lem:ms-words}(a) applied to $\xi'_{e+2}$ (for $e+1=n$ only $u=0$
occurs). Now $\operatorname{fwd}(y)-\operatorname{fwd}(x'_{e+1})=D_e-\epsilon_{i'_{e+1}}$,
with $\ell^1$ norm $F_1$, and
$|\operatorname{lat}(y)-\operatorname{lat}(x'_{e+1})|_1=r_1(k)$, so the lateral
distance to $x'_{e+1}+u$ is at least $|r_1(k)-k''|$. Since $\mathcal B$ is
nonincreasing, these pairs contribute at most $\mathcal T_1(F_1,r_1(k))$. The
vertices of $\gamma'$ on level $e+2\le n$ are points $x'_{e+2}+u$ with
$|u|_1=k''\le c$, where $x'_{e+2}=x'_{e+1}+e(\xi'_{e+2})+\epsilon_{i'_{e+2}}$. Their
forward distance to $y$ is
$|D_e-\epsilon_{i'_{e+1}}-\epsilon_{i'_{e+2}}|_1\ge F_2$, and the lateral distance
from $y$ to $x'_{e+2}$ is
$\varrho\ge r_1(k)-|e(\xi'_{e+2})|_1\ge(r_1(k)-c)^+$. These pairs contribute at
most $\sum_{k''}\mathcal B(F_2+|\varrho-k''|)\le\mathcal T_2((r_1(k)-c)^+,F_2)$.
The terms charged to macrostep $e+1$ depend only on $Z_e$ and
$\boldsymbol\omega_{e+1}$: the unknown future $\xi'_{e+2}$, $i'_{e+2}$ is absorbed
by the sum over all $k''\le c$, the maximum in $\mathcal T_2$ and the minimum
in $F_2$.

\emph{Step 5 (gaps $s\ge3$).} Let $v$ be an unshared vertex on the later level
$g$: $v=\zeta_{g+1}(k)$ with $g<n$, or $v=x_n$ with $g=n$. The points of
$\gamma'$ on level $g-s\ge0$ are $E'-(e(\xi')-P_{k'}(\xi'))$, where
$\xi'=\xi'_{g-s+1}$ and $E'=x'_{g-s}+e(\xi')$ is the end of that excursion, with
$|e(\xi')-P_{k'}(\xi')|_1=|\xi'|-k'\in\{0,\ldots,c\}$. Laterally,
$\operatorname{lat}(x'_g)-\operatorname{lat}(E')=\sum_{j=g-s+2}^ge(\xi'_j)$ is a
sum of $s-1$ vectors of norm at most $c$, so
$\varrho=|\operatorname{lat}(v)-\operatorname{lat}(E')|_1\ge(r_0-c(s-1))^+$ with
$r_0=|\operatorname{lat}(v)-\operatorname{lat}(x'_g)|_1$; this is $r_0(k)$ for
$v=\zeta_{g+1}(k)$ and $|A_n|_1$ for $v=x_n$. In the forward coordinates,
$\operatorname{fwd}(v)-\operatorname{fwd}(E')=D_g+O'_s$, where
$O'_s=\sum_{j=g-s+1}^g\epsilon_{i'_j}$ has nonnegative entries summing to $s$;
hence $|D_g+O'_s|_1\ge\max(s,|D_g|_1-s)$, by the coordinate sum and the triangle
inequality. The pairs of $v$ with level $g-s$ therefore contribute at most
$\mathcal T_2((r_0-c(s-1))^+,\max(s,|D_g|_1-s))$, and summing over $s\ge3$
gives $\mathcal T_3(r_0,|D_g|_1)$. For $g=n$ this is at most
$\mathcal T_3(0,0)$ by monotonicity, and the two end vertices $x_n$, $x'_n$
give $2\mathcal T_3(0,0)$.

\emph{Step 6.} By Step 2 every cross pair is charged exactly once, and the
charges at macrostep $j$ are exactly the four groups defining
$\mathfrak b(Z_{j-1};\boldsymbol\omega_j)$.
\end{proof}

Two facts carry the proof: the $k$-th point of a monotone word is at $\ell^1$
distance exactly $k$ from its start and $|\xi|-k$ from its end, and oriented
increments satisfy $|D+O_s|_1\ge\max(s,|D|_1-s)$. The gap-$\ge3$ terms, a
supremum over the past of the other path, are positive at every state and
decay only slowly, roughly like $\mathcal B(r_0/c)$ in the lateral distance
$r_0$, so the kernel below is not finitely supported. Lemma~\ref{lem:ms-boost}
is the macrostep counterpart of Lemma~\ref{lem:ov-boost}. There the charge at
time $k$ depends only on $Z_k$ and vanishes when the two paths meet; here it
depends on the pair of macrosteps and need not vanish on a level where the
paths meet, since two excursions can share some of their points and not
others. For a dependent site
field of finite range, corrections from nearby vertices of the two paths
appear in \cite[proof of Theorem~1.2]{BDNS2020}.

\begin{proposition}[Product of per-macrostep weights]\label{prop:ms-product}
Assume \eqref{eq:ov-H}, the hypotheses of Lemma~\ref{lem:ov-shared},
$\lambda_*<1$ and $\rho_cp<\rho_-$. Put
\[
 \mathsf W(z;\boldsymbol\omega)=\rho_-^{-N_V(z;\boldsymbol\omega)}
 \Bigl(\frac{\rho_-}{\rho_cp}\Bigr)^{(N_E+N_{\rm f})(z;\boldsymbol\omega)}
 e^{\kappa\mathfrak b(z;\boldsymbol\omega)}.
\]
Then $\mathsf W\ge1$, and for $n\ge1$ and $L\ge2cn+1$,
\[
 \E W_n^2\ \le\ e^{\kappa\varepsilon'_{n,L}}\,C_{\rm fin}\,
 \mathbf E_0\Bigl[\prod_{j=1}^n\mathsf W(Z_{j-1};\boldsymbol\omega_j)\Bigr],\qquad
 C_{\rm fin}=\rho_-^{-1}e^{\kappa(\mathcal B(1)+2\mathcal T_3(0,0))},
\]
where $\mathbf E_0$ is the expectation for the pair chain started at $Z_0=0$.
\end{proposition}

\begin{proof}
$\mathsf W\ge1$ because $\rho_-\le1$, $\rho_cp<\rho_-$ and $\mathfrak b\ge0$.
Combine Lemmas~\ref{lem:ms-pz}, \ref{lem:ms-pair} and~\ref{lem:ms-boost}. By
Lemma~\ref{lem:ms-meetings}(b),(c),
\[
 \prod_{j=1}^n\rho_-^{-N_V}\Bigl(\frac{\rho_-}{\rho_cp}\Bigr)^{N_E+N_{\rm f}}
 =\rho_-^{-(|V\cap V'|-\mathbf1\{Z_n=0\})}\Bigl(\frac{\rho_-}{\rho_cp}\Bigr)^{|E\cap E'|},
\]
and multiplying by $\rho_-^{-\mathbf1\{Z_n=0\}}\le\rho_-^{-1}$ bounds
$p^{-|E\cap E'|}\rho_-^{-(|V\cap V'|-|E\cap E'|)}\rho_c^{-|E\cap E'|}$. The end
terms of Lemma~\ref{lem:ms-boost} give $e^{\kappa(\mathcal B(1)+2\mathcal T_3(0,0))}$,
and the wrap-around term gives $e^{\kappa\varepsilon'_{n,L}}$.
\end{proof}

\subsection{An interaction-matrix criterion with a far region}\label{sec:ms-criterion}

Let $\mathbb S=\Z^3\times\mathbb A^{(f)}$, and let $\mathsf P$ be the transition
kernel of the pair chain. Since $f\ge4$, the forward difference walk is
transient by Lemma~\ref{lem:ov-green}(ii) with $f$ in place of $d$, so
$\mathsf G=\sum_{k\ge0}\mathsf P^k$ has finite entries. Define
\[
 \mathsf Q(z,z'')=\E\bigl[(\mathsf W(z;\boldsymbol\omega)-1)\,
 \mathbf1\{z+\Delta(\boldsymbol\omega)=z''\}\bigr]\ \ge0,\qquad
 \eta^{\mathsf Q}=\mathsf Q\mathbf1,
\]
where $\boldsymbol\omega$ is a pair of independent macrosteps and
$\Delta(\boldsymbol\omega)$ its increment. The group $\mathfrak G$ generated by
signed permutations of the lateral coordinates and permutations of the forward
coordinates preserves the macrostep law and the monotone words, and preserves
$\ell^1$ distances. Hence
$\mathsf W(\mathfrak gz;\mathfrak g\boldsymbol\omega)=\mathsf W(z;\boldsymbol\omega)$,
$\mathsf Q(\mathfrak gz,\mathfrak gz'')=\mathsf Q(z,z'')$ and
$\mathsf G(\mathfrak gz,\mathfrak gz')=\mathsf G(z,z')$ for
$\mathfrak g\in\mathfrak G$, and the computations below reduce to orbits.

\begin{theorem}[Interaction-matrix criterion with a far region]\label{thm:ms-criterion}
Let $C'\subset\mathbb S$ be finite and $\mathfrak G$-invariant, with
$C'\supseteq\{D=0,\ |A|_1\le2c\}$. Let
\begin{itemize}
\item $w:C'\to(0,\infty)$ be constant on $\mathfrak G$-orbits, and
$\bar\eta\ge\eta^{\mathsf Q}$ on $C'$;
\item $\bar{\mathsf M}$ satisfy
$(\bar{\mathsf M}w)(z)\ge\sum_{z''}\mathsf Q(z,z'')\sum_{z'\in C'}\mathsf G(z'',z')w(z')$
for $z\in C'$;
\item $\Gamma_w\ge\max_{z\in C'}\sum_{z'\in C'}\mathsf G(z,z')w(z')$;
\item $\upsilon:\mathbb S\setminus C'\to[0,\infty)$ satisfy
$\upsilon\ge\eta^{\mathsf Q}$ on $\mathbb S\setminus C'$, and
$\eta_F\ge\sup_{z''\in\mathbb S}\sum_{z'\notin C'}\mathsf G(z'',z')\upsilon(z')$.
\end{itemize}
If $\lambda\in(0,1)$ and $\mathsf H>0$ satisfy
\[
 \text{(i)}\ \ \bar{\mathsf M}w+\mathsf H\eta_F\bar\eta\ \le\ \lambda w\ \text{ on }C',
 \qquad
 \text{(ii)}\ \ \Gamma_w+\mathsf H\eta_F\ \le\ \lambda\mathsf H,
\]
then for every $n\ge1$
\[
 \mathbf E_0\Bigl[\prod_{j=1}^n\mathsf W(Z_{j-1};\boldsymbol\omega_j)\Bigr]
 \ \le\ 1+\frac{\Lambda_*\lambda\mathsf H}{1-\lambda},\qquad
 \Lambda_*=\max\Bigl(\max_{C'}\frac{\bar\eta}{w},\frac1{\mathsf H}\Bigr).
\]
\end{theorem}

We call $(C',w,\lambda,\mathsf H,\bar\eta,\bar{\mathsf M},\Gamma_w,\upsilon,\eta_F)$
a \emph{certificate}, and $C'$ its \emph{near set}.

\begin{proof}
\emph{Expansion.} Write $\mathsf W(Z_{j-1};\boldsymbol\omega_j)=1+\mathsf x_j$
with $\mathsf x_j\ge0$. Then
$\prod_{j\le n}(1+\mathsf x_j)=\sum_{J\subseteq[n]}\prod_{j\in J}\mathsf x_j$.
For $J=\{j_1<\cdots<j_r\}$, the Markov property
(Lemma~\ref{lem:ms-meetings}(e): $\boldsymbol\omega_j$ is independent of
$Z_{j-1}$, and $\mathsf x_j$ and $Z_j$ are functions of
$(Z_{j-1},\boldsymbol\omega_j)$) gives
\[
 \mathbf E_0\prod_{j\in J}\mathsf x_j
 =\bigl(\mathsf P^{j_1-1}\mathsf Q\,\mathsf P^{j_2-j_1-1}\mathsf Q\cdots\mathsf Q\,
 \mathsf P^{n-j_r}\mathbf1\bigr)(0),
\]
with $\mathsf P^{n-j_r}\mathbf1=\mathbf1$. All entries are nonnegative, and
$J\mapsto(j_1-1,j_2-j_1-1,\ldots)$ is injective. Summing over $J$ and replacing
each power of $\mathsf P$ by $\mathsf G$ gives
\[
 \mathbf E_0\prod_{j\le n}\mathsf W(Z_{j-1};\boldsymbol\omega_j)
 \ \le\ \sum_{r\ge0}\bigl((\mathsf G\mathsf Q)^r\mathbf1\bigr)(0)
 =1+\sum_{r\ge1}\bigl(\mathsf G(\mathsf Q\mathsf G)^{r-1}\eta^{\mathsf Q}\bigr)(0).
\]

\emph{The near part from any start.} For $z''\notin C'$ let $\tau$ be the
hitting time of $C'$ by the pair chain started at $z''$. Since the chain is
outside $C'$ before $\tau$, the strong Markov property gives
\[
 \sum_{z'\in C'}\mathsf G(z'',z')w(z')
 =\mathbf E_{z''}\Bigl[\mathbf1\{\tau<\infty\}\,
 \bigl(\mathsf G(w\mathbf1_{C'})\bigr)(Z_\tau)\Bigr]\ \le\ \Gamma_w;
\]
for $z''\in C'$ this is the definition of $\Gamma_w$.

\emph{A supersolution.} Put $\varphi=w$ on $C'$ and $\varphi=\mathsf H\upsilon$ off
$C'$. For every $z''$, by the previous step and (ii),
\begin{equation}\label{eq:ms-Gphi}
 (\mathsf G\varphi)(z'')=\sum_{z'\in C'}\mathsf G(z'',z')w(z')
 +\mathsf H\sum_{z'\notin C'}\mathsf G(z'',z')\upsilon(z')
 \ \le\ \Gamma_w+\mathsf H\eta_F\ \le\ \lambda\mathsf H.
\end{equation}
On $C'$, the definition of $\bar{\mathsf M}$, the bound
$\eta^{\mathsf Q}\le\bar\eta$ and (i) give
$(\mathsf Q\mathsf G\varphi)(z)\le(\bar{\mathsf M}w)(z)+\mathsf H\eta_F\eta^{\mathsf Q}(z)
\le\lambda w(z)=\lambda\varphi(z)$. Off $C'$, \eqref{eq:ms-Gphi} gives
$(\mathsf Q\mathsf G\varphi)(z)\le\eta^{\mathsf Q}(z)\sup\mathsf G\varphi
\le\upsilon(z)\lambda\mathsf H=\lambda\varphi(z)$. So
$\mathsf Q\mathsf G\varphi\le\lambda\varphi$ on $\mathbb S$. Also
$\eta^{\mathsf Q}\le\Lambda_*\varphi$: on $C'$ because $\eta^{\mathsf Q}\le\bar\eta$,
and off $C'$ because $\eta^{\mathsf Q}\le\upsilon=\varphi/\mathsf H$.

\emph{Conclusion.} Since $\mathsf Q\mathsf G$ is a positive operator,
$(\mathsf Q\mathsf G)^{r-1}\eta^{\mathsf Q}\le\Lambda_*\lambda^{r-1}\varphi$ by
induction, and \eqref{eq:ms-Gphi} at $z''=0$ gives
\[
 \sum_{r\ge1}\bigl(\mathsf G(\mathsf Q\mathsf G)^{r-1}\eta^{\mathsf Q}\bigr)(0)
 \ \le\ \Lambda_*\sum_{r\ge1}\lambda^{r-1}(\mathsf G\varphi)(0)
 \ \le\ \frac{\Lambda_*\lambda\mathsf H}{1-\lambda}.
\]
\end{proof}

The expansion is the one in the proof of \cite[Theorem~5]{JiangLang2026}, where
the excess kernel vanishes outside a finite set \cite[p.~9]{JiangLang2026}, so
that only powers of a finite matrix occur. Theorem~\ref{thm:ms-criterion}
extends it to the kernel $\mathsf Q$, which is not finitely supported, through
the far-region supersolution $\varphi=\mathsf H\upsilon$ and the hitting-time
bound for the near part, together with the forward bound of
Lemma~\ref{lem:ms-khas} below. Condition (i) implies
$\bar{\mathsf M}w\le\lambda w$, so $\lambda$ is at least the Perron root of
$\bar{\mathsf M}$, by the Collatz--Wielandt bound in the form used in
\cite[Eq.~(23)]{JiangLang2026}.

\begin{lemma}[Bounds in the far region]\label{lem:ms-far}
Let $z=(A,D)\notin C'$. Then $N_V=N_E=N_{\rm f}=0$ at $z$ for every
$\boldsymbol\omega$, so
$\eta^{\mathsf Q}(z)=\E[e^{\kappa\mathfrak b(z;\boldsymbol\omega)}-1]
\le\kappa\,\E[\mathfrak b(z;\boldsymbol\omega)]\,e^{\kappa\sup\mathfrak b}$.
\begin{enumerate}\renewcommand{\labelenumi}{(\alph{enumi})}
\item For finitely many states, $\eta^{\mathsf Q}(z)$ is a finite sum over
$\boldsymbol\omega$.
\item Fix $D$ and $r\ge0$, with $r\ge2c+1$ if $D=0$. For every $A$ with
$|A|_1\ge r$ and every $\boldsymbol\omega$ with word lengths $(l,l')$ and
forward steps $(i,i')$,
$\mathfrak b(z;\boldsymbol\omega)\le\hat{\mathfrak b}_r(D;l,l',i,i')$, where
\begin{multline*}
 \hat{\mathfrak b}_r=\sum_{k=0}^{l}\Bigl[\mathfrak s_k+\mathcal T_2\bigl((r-k-l')^+,F_1\bigr)
 +\mathcal T_2\bigl((r-k-l'-c)^+,F_2\bigr)+\mathcal T_3\bigl((r-k)^+,D_1\bigr)\Bigr]\\
 +\sum_{k'=0}^{l'}\Bigl[\mathcal T_2\bigl((r-k'-l)^+,F'_1\bigr)
 +\mathcal T_2\bigl((r-k'-l-c)^+,F'_2\bigr)+\mathcal T_3\bigl((r-k')^+,D_1\bigr)\Bigr],
\end{multline*}
with $\mathfrak s_k=\sum_{k'=0}^{l'}\mathcal B(\max(1,r-k-k'))$ if $D=0$ and
$\mathfrak s_k=\mathcal T_2((r-k)^+,D_1)$ if $D\ne0$. Consequently
$\eta^{\mathsf Q}(z)\le\kappa\,\E[\hat{\mathfrak b}_r]\exp(\kappa\max\hat{\mathfrak b}_r)$
for all $A$ with $|A|_1\ge r$, where the expectation and the maximum are over
$(l,l',i,i')$.
\item If $D_1=|D|_1\ge4$, then for all $A$ and $\boldsymbol\omega$
\[
 \mathfrak b\ \le\ 2(c+1)\bigl[\mathcal T_2(0,D_1)+\mathcal T_2(0,D_1-1)
 +\mathcal T_2(0,D_1-2)+\mathcal T_3(0,D_1)\bigr]=:\hat{\mathfrak b}^{\rm sh}(D_1),
\]
and $\eta^{\mathsf Q}(z)\le xe^x$ with $x=\kappa\hat{\mathfrak b}^{\rm sh}(D_1)$.
\item For $D_1=2R$ with $R\ge2$,
$\kappa\hat{\mathfrak b}^{\rm sh}(2R)\le2\kappa(c+1)g\,(3+2/(1-\lambda_*))\,\mathcal B(R)$.
\end{enumerate}
\end{lemma}

\begin{proof}
A match on level $a$ needs $D=0$ and $|A|_1=|P_{k'}(\xi')-P_k(\xi)|_1\le2c$, and
$N_{\rm f}=1$ needs $D=0$ and $|A|_1=|e(\xi')-e(\xi)|_1\le2c$; both are excluded
off $C'$. The exponential bound is $e^x-1\le xe^x$.
(b) All terms are monotone, so it suffices to bound each distance from below.
On the common level, the point $k$ of the first excursion is at lateral distance
$\varrho\ge|A|_1-k$ from the start of the second, whose points are at lateral
distance $k'$ from its start by Lemma~\ref{lem:ms-words}(a). If $D=0$ the
distance is at least $\max(1,r-k-k')$, because the points are unmatched and
$r\ge2c+1$; if $D\ne0$ the same-level sum is at most
$\sum_{k'}\mathcal B(D_1+|\varrho-k'|)\le\mathcal T_2((r-k)^+,D_1)$. For gaps
$1$ and $2$, $r_1(k)\ge|A|_1-k-|e(\xi')|_1\ge r-k-l'$, and
$\mathcal T_1(F,r_1)\le\mathcal T_2(\varrho,F)$ for every $\varrho\le r_1$,
because $\mathcal T_1(F,r_1)$ is one of the sums over which $\mathcal T_2(\varrho,F)$
maximizes. For gaps at least $3$, $r_0(k)\ge r-k$. The points of the second
excursion are treated in the same way. The exponential bound follows from
$e^{\kappa\mathfrak b}-1\le\kappa\mathfrak be^{\kappa\mathfrak b}
\le\kappa\hat{\mathfrak b}_re^{\kappa\max\hat{\mathfrak b}_r}$.
(c) Take $r=0$ in (b), use $F_1,F'_1\ge D_1-1$ and $F_2,F'_2\ge D_1-2$, and bound
each of the at most $2(c+1)$ points by all four terms.
(d) First, $\mathcal T_2(0,F)\le g\mathcal B(F)$. Second,
$\mathcal T_3(0,2R)=\sum_{s\ge3}\mathcal T_2(0,\max(s,2R-s))
\le g[\sum_{m\ge R}\mathcal B(m)+\sum_{s>R}\mathcal B(s)]
\le2g\mathcal B(R)/(1-\lambda_*)$. Third, $\mathcal B(2R-2)\le\mathcal B(R)$ for
$R\ge2$.
\end{proof}

\begin{lemma}[A forward Khas'minskii bound for the far region]\label{lem:ms-khas}
Let $R_S\ge2$ and $\mathcal D_{\rm in}=\{D\in\mathbb A^{(f)}:|D|_1\le2R_S\}$. For
$D\in\mathcal D_{\rm in}$ let $\bar\upsilon(D)$ satisfy
$\eta^{\mathsf Q}(z)\le\bar\upsilon(D)$ for all $z=(A,D)\notin C'$, for instance
from Lemma~\ref{lem:ms-far}(a),(b). For $D\notin\mathcal D_{\rm in}$ put
$\bar\upsilon(D)=xe^x$ with $x=\kappa\hat{\mathfrak b}^{\rm sh}(|D|_1)$. Let
$\upsilon(z)=\bar\upsilon(D)$ for $z=(A,D)\notin C'$, so that
$\upsilon\ge\eta^{\mathsf Q}$ off $C'$ by Lemma~\ref{lem:ms-far}. Let $u_n(\delta)$
be the probability that the forward difference walk started at $0$ is at
$\delta$ at time $n$, and $T_u(N)=\sum_{n>N}u_n(0)$. Then, for every
$N_F\ge0$, $\eta_F=h_{\rm in}+\eta_{\rm out}$ is admissible in
Theorem~\ref{thm:ms-criterion}, where
\begin{gather*}
 h_{\rm in}=\max_{D_a\in\mathcal D_{\rm in}}\sum_{D'\in\mathcal D_{\rm in}}
 \bar\upsilon(D')\Bigl[\sum_{n\le N_F}u_n(D'-D_a)+T_u(N_F)\Bigr],\\
 \eta_{\rm out}=\sum_{n\ge0}\min\bigl(\bar\upsilon_{\rm out},\,u_n(0)\,S_{\rm out}\bigr),
\end{gather*}
with $S_{\rm out}=\sum_{D'\notin\mathcal D_{\rm in}}\bar\upsilon(D')$ and
$\bar\upsilon_{\rm out}=\sup_{D'\notin\mathcal D_{\rm in}}\bar\upsilon(D')$.
Moreover $\bar\upsilon_{\rm out}$ is the value of $\bar\upsilon$ at
$|D|_1=2R_S+2$.
\end{lemma}

\begin{proof}
Summing $\mathsf G((A'',D''),(A',D'))$ over $A'$ gives the Green function
$g_F(D'',D')$ of the forward difference walk, because the lateral walk has a
probability law (Lemma~\ref{lem:ms-meetings}(e)). Since $\upsilon$ depends only
on $D$,
$\sum_{z'\notin C'}\mathsf G(z'',z')\upsilon(z')\le\sum_{D'}g_F(D'',D')\bar\upsilon(D')$.
For the inner part, the strong Markov property at the hitting time of
$\mathcal D_{\rm in}$ bounds it by its maximum over starts
$D_a\in\mathcal D_{\rm in}$, where
$g_F(D_a,D')=\sum_nu_n(D'-D_a)$ and $u_n(\delta)\le u_n(0)$ by
Lemma~\ref{lem:ov-green}(i) with $f$ in place of $d$. For the outer part,
$\sum_{D'\notin\mathcal D_{\rm in}}\Pp_{D''}(D_n=D')\bar\upsilon(D')
\le\min(\bar\upsilon_{\rm out},u_n(0)S_{\rm out})$ for each $n$. Finally,
$\hat{\mathfrak b}^{\rm sh}(D_1)$ is nonincreasing in $D_1$ because
$\mathcal T_2(0,\cdot)$ and $\mathcal T_3(0,\cdot)$ are, $x\mapsto xe^x$ is
increasing, and $|D|_1$ is even on $\mathbb A^{(f)}$, so the supremum is attained
on the innermost outer shell $|D|_1=2R_S+2$.
\end{proof}

This is Khas'minskii's argument \cite{Khasminskii1959} for the forward walk,
which is transient because $f\ge4$; the proof of Lemma~\ref{lem:ov-cycle}(iii)
uses it in the same way. $S_{\rm out}$ is finite by Lemma~\ref{lem:ms-far}(d),
since there are at most $\binom{R+f-1}{f-1}^2$ vectors $D\in\mathbb A^{(f)}$ with
$|D|_1=2R$ and $\mathcal B$ decays geometrically.

\subsection{The engine theorem}\label{sec:ms-proofs}

\begin{theorem}[The macrostep engine]\label{thm:ms-engine}
Let $d=f+3$ with $f\ge4$, $c\ge1$ and $y_{\rm w}\in(0,1)$, and let rationals
$g_K=e^{2K'}$, $g_h=e^{2h}$, $0<p\le1$ and $\bar m$ be given. Put
$\rho_-=(1-\bar m)/2$, $\kappa=(1-\rho_-)/\rho_-$, $t'=\tanh K'$,
$\alpha_{\rm D}=2dt'$ and $\lambda_*=dt'/(1-2dt')$, and let $\rho_c$ be as in
Lemma~\ref{lem:ov-shared}. Assume:
\begin{itemize}
\item[(M1)] $h<0<K'\le|h|$ and $2dK'<1$, and $\bar m$ passes the test of
Lemma~\ref{lem:ov-meanfield};
\item[(M2)] $\lambda_*<1$, $\rho_-^2\ge c_{\rm cov}/4$ and $\rho_cp<\rho_-$;
\item[(M3)] there is a certificate for Theorem~\ref{thm:ms-criterion} with the
weight $\mathsf W$ of Proposition~\ref{prop:ms-product}.
\end{itemize}
With $\Lambda_*$ as in Theorem~\ref{thm:ms-criterion}, put
\begin{equation}\label{eq:ms-theta}
 C_{\rm fin}=\rho_-^{-1}e^{\kappa(\mathcal B(1)+2\mathcal T_3(0,0))},\qquad
 \theta_*=\frac{0.99}{C_{\rm fin}\bigl(1+\Lambda_*\lambda\mathsf H/(1-\lambda)\bigr)},
\end{equation}
\begin{equation}\label{eq:ms-L0}
 L_0^{\rm M}(n)=2cn+1+\Bigl\lceil\frac{\log\bigl(100\kappa((c+1)n+1)^2/(1-\alpha_{\rm D})\bigr)}
 {\log(1/\alpha_{\rm D})}\Bigr\rceil.
\end{equation}
Then $P^{\rm aux}_L(\mathcal O_n)\ge\theta_*$ for every $n\ge1$ and every
$L\ge L_0^{\rm M}(n)$.
\end{theorem}

\begin{proof}
Let $n\ge1$ and $L\ge L_0^{\rm M}(n)$. Since $|h|>0$, the test of
Lemma~\ref{lem:ov-meanfield} forces $\bar m\ge m_*>0$, so $\rho_-<\frac12$ and
$\kappa>1$. Since $0<K'$ and $2dK'<1$, also $0<\alpha_{\rm D}<1$. Hence the
logarithm in the numerator of \eqref{eq:ms-L0} is positive, the denominator is
positive and finite, and $L\ge2cn+2$. So Lemmas~\ref{lem:ms-torus}--\ref{lem:ms-boost} and
Proposition~\ref{prop:ms-product} apply; their hypotheses on the Ising law are
(M1) and (M2), with $\rho_L\ge\rho_-$ for every $L\ge3$ by
Lemma~\ref{lem:ov-meanfield}. Moreover
$\alpha_{\rm D}^{L-2cn}\le\alpha_{\rm D}(1-\alpha_{\rm D})/(100\kappa((c+1)n+1)^2)$,
that is, $\kappa\varepsilon'_{n,L}\le\alpha_{\rm D}/100\le1/100$.
Proposition~\ref{prop:ms-product} and Theorem~\ref{thm:ms-criterion} give
$\E W_n^2\le e^{1/100}C_{\rm fin}(1+\Lambda_*\lambda\mathsf H/(1-\lambda))$, and
Lemma~\ref{lem:ms-pz} gives
$P^{\rm aux}_L(\mathcal O_n)\ge e^{-1/100}/(C_{\rm fin}(1+\Lambda_*\lambda\mathsf H/(1-\lambda)))
\ge\theta_*$, since $e^{-1/100}>0.99$.
\end{proof}

As after Lemma~\ref{lem:ov-shared}, $c_{\rm cov}$ in (M2) and in $\rho_c$ may be
replaced by a rational lower bound; the certificates do so.

\begin{proof}[Proof of Proposition~\ref{prop:ms-central}]
Fix a row of Table~\ref{tab:ms-local} and its $\theta_*$-maximizing
certificate. By (V0)--(V1) of Section~\ref{sec:ms-certificates}, the row
satisfies (C3)--(C4) and has the local inputs (R1) and (R2) of
Section~\ref{sec:ov-assembly}: through (C1)--(C2),
Lemma~\ref{lem:ov-torus-floor} and Proposition~\ref{prop:ov-holley-line}, or
through (N1)--(N2), Lemma~\ref{lem:ov-pair} and
Corollary~\ref{cor:ov-sym-holley}. By (V1)--(V3), the row satisfies hypotheses
(M1)--(M3) of Theorem~\ref{thm:ms-engine}, which therefore gives (R3), with
$\theta_*$ from \eqref{eq:ms-theta} and $L_0=L_0^{\rm M}$. Since
$L_0^{\rm M}(n)\ge2cn+2\ge2n+4$, we have $L_1=L_0^{\rm M}$, and the proof of
Proposition~\ref{prop:overlap-central} applies verbatim.
\end{proof}

The minimal-$\lambda$ certificate of each near set gives the smaller constants
listed in Table~\ref{tab:ms-certificates}; either may be used. The torus sizes
are moderate: at $d=9$, $L_0^{\rm M}(1)=15$, $L_0^{\rm M}(10)=73$ and
$L_0^{\rm M}(100)=617$, and at $(d,t)=(7,31/200)$ they are $14$, $72$ and $615$.

\subsection{The certified instances}\label{sec:ms-certificates-summary}

For each row of Table~\ref{tab:ms-local}, the local inputs and hypotheses
(M1)--(M3) of Theorem~\ref{thm:ms-engine} are decided in exact rational
arithmetic, for near sets and certificates listed in
Table~\ref{tab:ms-certificates}. Section~\ref{sec:ms-certificates} of
Appendix~\ref{app:certificates} lists the conditions, what is computed, the
sizes of the computations, and which local inputs each row needs;
dimension $7$ is certified only with the sharpened inputs of
Section~\ref{sec:ov-noisy}.

\begin{remark}[Dimension six and below]\label{rem:ms-limits}
With three lateral coordinates, $d\le6$ leaves $f\le3$ forward coordinates. The
forward difference walk is then recurrent, since its return probabilities decay
like $n^{-(f-1)/2}$, so its Green function $g_F$ is infinite and the far-region
bound of Lemma~\ref{lem:ms-khas} is unavailable. This is a limitation of that
bound, not of the criterion: the relative position of the two paths is still a
transient chain. Floating-point evaluations of the same engine in dimension $6$,
at the best local inputs of the kind of Section~\ref{sec:ov-noisy} that we found,
give near-matrix Perron roots above $1$ once the Ising correlation factor of
Lemma~\ref{lem:ms-boost} is included (about $1.02$ with lateral words of length
up to $6$), although the collision part alone falls below $1$ for long enough
words. We make no claim below dimension $7$. For Bernoulli
percolation, Jiang and Lang remark that a more sophisticated version of their
method may extend their coexistence theorem to all $d\ge6$
\cite[Remark~2]{JiangLang2026}.
\end{remark}
\FloatBarrier

\section{Blue percolation and physical response}\label{sec:response}

The blue susceptibility in Corollary~\ref{cor:chi} counts connected vertices.
The ordinary spin-glass susceptibility instead measures squared spin
correlations. The following estimate distinguishes them at the same temperatures.
It averages the fair coupling signs of a whole vertex star at once, in the
mean-square spirit of Fr\"ohlich and Zegarlinski
\cite[Section 2 and Appendix A]{FZ1987}.
For a finite nonempty set $A$, write
$q_{A,x}=\prod_{v\in A}\sigma_{x+v}\tau_{x+v}$ and let $Q_A$ be its
cube-average limit, which exists by stationarity. In particular
$Q=Q_{\{0\}}$ is the ordinary spatial overlap. Put $F(u)=\E\cosh(uW)$ and
$s_2=\E W^2$. An edge-by-edge variant that restores the star one edge at a
time (Proposition~\ref{prop:sg-bound} in
Appendix~\ref{app:separation-certificate}) needs roughly $3ae^{5a}<1$ at the
mean-field scale $\beta^2=a/(\Delta s_2)$. Removing the whole star at once and
keeping every odd part of it gives a bound that stays finite up to
$a<\tfrac12\log3$ uniformly in the dimension. The price is a family of
multi-point row sums, which
close under the same recursion.

\begin{proposition}[A dimension-uniform susceptibility bound]\label{prop:sg-uniform}
Let $G$ be a finite graph, parallel edges allowed, of maximum degree at most
$\Delta\ge2$, with zero field and symmetric iid couplings. If
$F(2\beta)^{\Delta-1}<3$, then
\begin{equation}\label{eq:sg-uniform}
 \sup_x\sum_y\E\langle\sigma_x\sigma_y\rangle_{G,J}^2
 \le\chi_*(\Delta,\beta):=1+\frac{F(2\beta)^{\Delta}-1}{3-F(2\beta)^{\Delta-1}}.
\end{equation}
At $\Delta=2d$, every selected periodic joint limit satisfies
\begin{equation}\label{eq:sg-limit-uniform}
 1\le\chi_{\rm SG}(\nu):=\sum_x\E_\nu[q_0q_x]\le\chi_*(2d,\beta),\qquad
 \E_\nu\Bigl(\frac1{|\Lambda|}\sum_{x\in\Lambda}q_x\Bigr)^2
 \le\frac{\chi_*(2d,\beta)}{|\Lambda|}
\end{equation}
for every finite box $\Lambda$. Consequently $Q=0$ almost surely, and $D_+=D_-$
almost surely by Corollary~\ref{cor:imbalance-overlap}.
\end{proposition}

\begin{proof}
For an induced subgraph $H$ of $G$ and $S\subseteq V(H)$ write
$\sigma_S=\prod_{v\in S}\sigma_v$, $c_H(S)=\E\langle\sigma_S\rangle_{H,J}^2$
and $c_H(x,y)=c_H(\{x,y\})$.

\emph{Cavity step.} Let $u$ be a vertex of $H$ with star edges
$e_i=\{u,y_i\}$, $1\le i\le D$, where the $y_i$ need not be distinct; write
$W_i=W_{e_i}$ and $\epsilon_i=\epsilon_{e_i}$. Let $B\subseteq V(H)\setminus\{u\}$
have odd size and let $\mu$ be the Gibbs measure of $H-u$. With
$h=\beta\sum_iW_i\epsilon_i\sigma_{y_i}$, summing out $\sigma_u$ gives
$\langle\sigma_u\sigma_B\rangle_H=\mu(\sigma_B\sinh h)/\mu(\cosh h)$. Expanding
$e^{\pm h}=\prod_i\bigl(\cosh(\beta W_i)\pm\epsilon_i\sigma_{y_i}\sinh(\beta W_i)\bigr)$,
\[
 \mu(\sigma_B\sinh h)=\sum_{|A|\ {\rm odd}}s_A\,\epsilon^A\,
   \mu(\sigma_{B\triangle\partial A}),\qquad
 s_A=\prod_{i\in A}\sinh(\beta W_i)\prod_{i\notin A}\cosh(\beta W_i),
\]
where $A$ ranges over subsets of $\{1,\ldots,D\}$,
$\epsilon^A=\prod_{i\in A}\epsilon_i$, and $\partial A$ is the set of
vertices occurring an odd number of times among $(y_i)_{i\in A}$. The
characters $\epsilon^A$ are orthonormal under the fair signs, $\mu$ does not
depend on the star couplings, and $\mu(\cosh h)\ge1$. Averaging first over the
star signs and then over the star magnitudes, which are independent of the
couplings of $H-u$, yields
\begin{equation}\label{eq:cavity-l2}
 c_H(\{u\}\cup B)\le\sum_{|A|\ {\rm odd}}\omega(A)\,c_{H-u}(B\triangle\partial A),
 \qquad
 \omega(A)=\bigl(\E\sinh^2\!\beta W\bigr)^{|A|}\bigl(\E\cosh^2\!\beta W\bigr)^{D-|A|}.
\end{equation}
Since $\cosh^2+\sinh^2=\cosh2(\cdot)$ and $\cosh^2-\sinh^2=1$, the weights sum
to $\mathsf w_D:=\tfrac12\bigl(F(2\beta)^D-1\bigr)$, which increases with $D$.

\emph{Summed recursion.} For $H=G-X$ with $X\ne\varnothing$ and an odd set
$S\subseteq V(H)$, put $\mathcal R_H(S)=\sum_{v\in V(H)}c_H(S\triangle\{v\})$.
Fix $w\in S$. The term $v=w$ is at most one. For $v\ne w$ apply
\eqref{eq:cavity-l2} at $u=w$ with the odd set $B=(S\triangle\{v\})\setminus\{w\}$,
and sum over $v\in V(H-w)$:
\begin{equation}\label{eq:row-recursion}
 \mathcal R_H(S)\le1+\sum_{|A|\ {\rm odd}}\omega(A)\,
   \mathcal R_{H-w}\bigl((S\setminus\{w\})\triangle\partial A\bigr).
\end{equation}
Let $\mathcal C$ be the finite family of such pairs $(H,S)$ in which every
vertex of $S$ has a $G$-neighbour in $X$. Each $w\in S$ then has at most
$\Delta-1$ edges in $H$, so the weights in \eqref{eq:row-recursion} sum to at
most $\mathsf w_{\Delta-1}$. The new pairs again lie in $\mathcal C$: the new
sets are odd, and $\partial A$ consists of neighbours of the deleted vertex $w$.
With $M=\max_{\mathcal C}\mathcal R_H(S)$, \eqref{eq:row-recursion} gives
$M\le1+\mathsf w_{\Delta-1}M$, and $\mathsf w_{\Delta-1}<1$ is the hypothesis.
Hence $M\le1/(1-\mathsf w_{\Delta-1})$. Finally, for $x\in V(G)$,
\eqref{eq:cavity-l2} at $u=x$ gives
\[
 \sum_yc_G(x,y)\le1+\sum_{|A|\ {\rm odd}}\omega(A)\,\mathcal R_{G-x}(\partial A)
 \le1+\frac{\mathsf w_\Delta}{1-\mathsf w_{\Delta-1}},
\]
which is \eqref{eq:sg-uniform}.

The tori $\mathbb T_L^d$, $L\ge3$, have maximum degree $2d$. On each torus,
conditional replica independence gives $\E q_xq_y=c_G(x,y)\ge0$. Pass this
bounded local observable and every finite partial row sum to the selected
limit, then increase the finite sum. Stationarity gives the block-variance
bound in \eqref{eq:sg-limit-uniform}, and the cube ergodic theorem forces
$Q=0$. No independence of limiting replicas is used.
\end{proof}

Put $a=(\Delta-1)s_2\beta^2$. If the law is strictly sub-Gaussian, meaning
$\E\cosh(\lambda W)\le e^{\lambda^2s_2/2}$ for all real $\lambda$, as for $\pm1$,
Gaussian and uniform couplings, then $F(2\beta)^{\Delta-1}\le e^{2a}$ and
$F(2\beta)^{\Delta}\le e^{4a}$ in every dimension, so $\chi_*$ is bounded
uniformly in $\Delta$ for each fixed $0<a<\tfrac12\log3$. With only an
exponential moment, $F(2\beta)=1+2s_2\beta^2+O(\beta^4)$, so at fixed $a$ both
powers tend to $e^{2a}$ as $\Delta\to\infty$, and the same holds for all
sufficiently large $\Delta$. This is the range $\beta\approx0.741\,(2ds_2)^{-1/2}$.

\begin{corollary}[Percolation with finite overlap susceptibility]\label{cor:separation}
For a fixed symmetric exponential-moment law, put
$\Delta=2d$. At $\beta=c/(\Delta m)$ with any fixed $c>1$, all sufficiently
large finite dimensions have blue percolation in both overlap signs in every
selected periodic limit, while
\[
 \chi_{\rm SG}=1+O_c(d^{-1}),\qquad Q=0\quad\hbox{almost surely}.
\]
Moreover $\chi_{\rm SG}\le2$ throughout
$0\le\beta\le(3\Delta s_2)^{-1/2}$ for all sufficiently large $d$.
For unit couplings, every temperature of Theorem~\ref{thm:explicit} is a
common point (Corollary~\ref{cor:explicit-balance}).
\end{corollary}

\begin{proof}
With an exponential moment, $F(2\beta)=1+2s_2\beta^2+O(\beta^4)$ as
$\beta\to0$. At $\beta=c/(\Delta m)$ this gives $F(2\beta)^\Delta-1=O_c(\Delta^{-1})$
and $F(2\beta)^{\Delta-1}\to1$, so $\chi_*=1+O_c(d^{-1})$. Combine
Proposition~\ref{prop:sg-uniform} with Theorem~\ref{thm:headline}. At
$\beta=(3\Delta s_2)^{-1/2}$ both $F(2\beta)^{\Delta-1}$ and $F(2\beta)^\Delta$
tend to $e^{2/3}$, so $\chi_*\to2/(3-e^{2/3})<2$; and $\chi_*$ increases with
$\beta$, giving the whole interval.
\end{proof}

\begin{remark}[Ordering of the first onsets]\label{rem:onset-ordering}
Under the assumptions of Corollary~\ref{cor:separation}, let $\beta_\chi(d)$ be
the infimum of $\beta$ for which some $\nu\in\mathcal L_{d,\beta}$ has
$\chi_{\rm SG}(\nu)=\infty$, with $\inf\varnothing=\infty$.
Equation~\eqref{eq:onset} and Proposition~\ref{prop:sg-uniform} give, for
every fixed $0<a<\tfrac12\log3$ and all sufficiently large $d$,
\[
 \beta_{\mathrm b}(d)<\bigl(a/((2d-1)s_2)\bigr)^{1/2}\le\beta_\chi(d).
\]
Thus $T_\chi<T_{\mathrm b}$ whenever $\beta_\chi<\infty$, where
$T_i=1/\beta_i$. This compares first onsets without assuming monotonicity;
existence of a susceptibility transition remains unproved.
\end{remark}

The percolation bound holds on a whole interval:
Corollary~\ref{cor:window} in Appendix~\ref{sec:mns} proves coexistence of both
blue sectors for
$ (1+\varepsilon)/(2dm)\le\beta\le\sqrt{\log(2d)/(24ds_2)} $
for fixed $\varepsilon>0$ and sufficiently large $d$.
Figure~\ref{fig:response-window} compares this interval with the response bound.

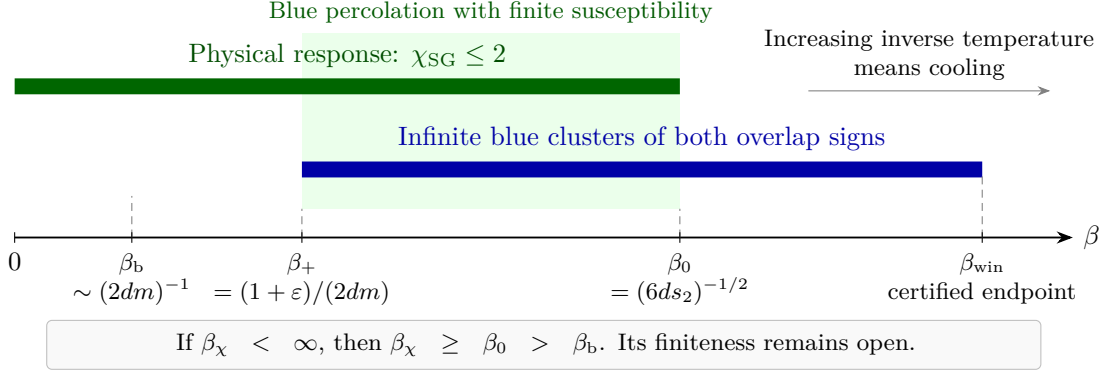
\begin{figure}[!htbp]
\centering
\begin{adjustbox}{max width=\linewidth}
\begin{tikzpicture}[x=1cm,y=1cm,font=\small,>=Stealth]
\fill[green!8] (3.8,.38) rectangle (8.8,2.7);
\node[font=\footnotesize,text=green!35!black] at (6.3,2.98)
 {Blue percolation with finite susceptibility};
\draw[green!40!black,line width=6pt,line cap=butt] (0,2)--(8.8,2);
\node[above,text=green!30!black] at (4.4,2.12)
 {Physical response: $\chi_{\rm SG}\le2$};
\draw[blue!65!black,line width=6pt,line cap=butt] (3.8,.9)--(12.8,.9);
\node[above,text=blue!65!black] at (8.3,1.02)
 {Infinite blue clusters of both overlap signs};
\draw[->,line width=.7pt] (0,0)--(14,0)
 node[right] {$\beta$};
\node[align=center,font=\footnotesize] at (12.1,2.42)
 {Increasing inverse temperature\\means cooling};
\draw[->,gray] (10.5,1.94)--(13.7,1.94);
\draw[densely dashed,gray] (1.55,0)--(1.55,.55);
\draw[densely dashed,gray] (3.8,0)--(3.8,.38);
\draw[densely dashed,gray] (8.8,0)--(8.8,.38);
\draw[densely dashed,gray] (12.8,0)--(12.8,.82);
\foreach \x in {0,1.55,3.8,8.8,12.8}
 \draw (\x,.07)--(\x,-.07);
\node[below] at (0,-.07) {$0$};
\node[below,align=center,font=\footnotesize] at (1.55,-.07)
 {$\beta_{\mathrm b}$\\$\sim(2dm)^{-1}$};
\node[below,align=center,font=\footnotesize] at (3.8,-.07)
 {$\beta_+$\\$=(1+\varepsilon)/(2dm)$};
\node[below,align=center,font=\footnotesize] at (8.8,-.07)
 {$\beta_0$\\$=(6ds_2)^{-1/2}$};
\node[below,align=center,font=\footnotesize] at (12.8,-.07)
 {$\beta_{\rm win}$\\certified endpoint};
\node[draw=gray!50,rounded corners=2pt,fill=gray!5,
 text width=12.8cm,align=center,inner sep=5pt,font=\footnotesize]
 at (7,-1.42)
 {If $\beta_\chi<\infty$, then $\beta_\chi\ge\beta_0>\beta_{\mathrm b}$.
 Its finiteness remains open.};
\end{tikzpicture}
\end{adjustbox}
\caption{Certified parameter intervals for a fixed symmetric exponential-moment
 coupling law in sufficiently high dimension, in the selected periodic ensemble.
 The axis is schematic, not to scale. Here $m=\E|J|$, $s_2=\E|J|^2$,
 $\varepsilon>0$ is fixed, and
 $\beta_{\rm win}=\sqrt{\log(2d)/(24ds_2)}$.
 The blue interval is Corollary~\ref{cor:window}; the response interval is
 Corollary~\ref{cor:separation}. Only $\beta_{\mathrm b}$ denotes an established
 onset: $\beta_+$, $\beta_0$, and $\beta_{\rm win}$ are bounds supplied by the
 proofs, not asserted transition points. Their overlap gives percolation with
 bounded susceptibility, without any monotonicity assumption on the blue law.}
\label{fig:response-window}
\end{figure}

\subsection{Macroscopic blue density and finite-cluster mass}

Let $D_s$ and $F_s$ be the cube densities of vertices of overlap $s$ belonging,
respectively, to infinite and finite blue components. These exist almost surely
by stationarity, and $D_s+F_s=(1+sQ)/2$. They are densities in the infinite-volume
law; no finite-torus largest-cluster conclusion is implicit.

\begin{corollary}[Exact balance]\label{cor:exact-balance}
Whenever $F(2\beta)^{2d-1}<3$, every selected periodic joint limit satisfies
$D_+=D_-$ and $F_+=F_-$ almost surely.
\end{corollary}

\begin{proof}
Proposition~\ref{prop:sg-uniform} gives $Q=0$ almost surely, and $Q=0$ forces
$D_+=D_-$ by Corollary~\ref{cor:imbalance-overlap}. Then $D_s+F_s=1/2$ gives
equality of the finite-sector densities.
\end{proof}

\begin{corollary}[Explicit points with finite susceptibility]\label{cor:explicit-balance}
Let the couplings be iid and uniform on $\{-1,1\}$, let $(d,t)$ be any pair of
Table~\ref{tab:explicit}, and put $\beta=\operatorname{atanh}t$. Every
$\nu\in\mathcal L_{d,\beta}$ has infinite blue clusters of both overlap signs,
$\chi_{\rm SG}(\nu)<7/2$, and $Q=0$, $D_+=D_-$ almost surely. At
$(d,t)=(7,3/20)$, moreover $\chi_{\rm SG}(\nu)<7/4$.
\end{corollary}

\begin{proof}
For unit couplings $F(2\beta)=\cosh2\beta=(1+t^2)/(1-t^2)$ is rational. At
the seventeen pairs, exact evaluation gives $F(2\beta)^{2d-1}<12/5$ and
$\chi_*(2d,\beta)<7/2$, and $\chi_*(14,\operatorname{atanh}\tfrac3{20})<7/4$
(Appendix~\ref{app:separation-certificate}). Combine
Theorem~\ref{thm:explicit}, Proposition~\ref{prop:sg-uniform} and
Corollary~\ref{cor:exact-balance}.
\end{proof}

\begin{corollary}[A balanced percolating regime]\label{cor:balanced}
For fixed symmetric exponential-moment disorder, suppose
\[
 \Delta m\beta_d\longrightarrow\infty,\qquad
 \Delta s_2\beta_d^2\longrightarrow0.
\]
For all sufficiently large $d$, every selected limit has $Q=0$ and
$D_+=D_-$ almost surely,
both signs percolate, and, uniformly over those limits,
\[
 \chi_{\rm SG}\longrightarrow1,\qquad
 \nu(|C_{\rm b}(0)|<\infty)\longrightarrow0,\qquad
 \E_\nu|D_s-1/2|\longrightarrow0.
\]
The last two conclusions, $Q=0$ and exact balance also hold at
$\beta_d=(10\Delta s_2)^{-1/2}$, where $\chi_{\rm SG}\le2$ eventually.
If the law has no atom at zero, the at-most-two input used in
Section~\ref{sec:discussion} identifies $D_+$ and $D_-$ with the densities of
the two individual infinite blue clusters.
\end{corollary}

\begin{proof}
Use the direct site parameter $p_d=b(\ell/M)^{\Delta-1}$ from
Corollary~\ref{cor:direct-cmr}, evaluated at $\beta_d$.
Since $b\sim m\beta_d$ and
$\log(\ell/M)=-9s_2\beta_d^2+O(\beta_d^4)$,
$\Delta p_d\to\infty$ in both regimes. In the second regime the restoration
factor tends to $e^{-9/10}>0$.
Kesten
\cite[Theorem 2S, p.~221]{Kesten1990} proves
\begin{equation}\label{eq:kesten-survival}
 \frac{\theta_d^{\rm site}(c/\Delta)}{c/\Delta}\longrightarrow y(c),
 \qquad y(c)=1-e^{-cy(c)},\quad c>1.
\end{equation}
Root-forced iid survival is increasing in its parameter. Compare $p_d$ with
$c/\Delta$ for each fixed $c$, then let $c\to\infty$; the direct rooted
comparison gives $\nu(|C_{\rm b}(0)|=\infty)\to1$ uniformly.
Proposition~\ref{prop:sg-uniform} supplies $Q=0$ and the stated susceptibility
bounds. Hence $D_s=1/2-F_s$ exactly. Global reversal of one replica preserves
blue bonds and exchanges $F_+$ and $F_-$, so
\begin{equation}\label{eq:balanced-error}
 \E_\nu|D_s-1/2|=\E_\nu F_s
 =\tfrac12\nu(|C_{\rm b}(0)|<\infty).
\end{equation}
This proves the density statements; Corollary~\ref{cor:exact-balance} gives
exact equality of the two densities in every sufficiently large fixed dimension.
The signed finite-blue density also vanishes there.
The conclusion concerns the periodic ensemble and does not assert
quenched Gibbs uniqueness.
\end{proof}

Under finite mean alone, the cutoff comparison gives the complementary bound
\begin{equation}\label{eq:finite-mass-profile}
 D_s\ge\tfrac12 y(c)-o(1),\qquad
 F_++F_-\le1-y(c)+o(1),\qquad |Q|=o(1)
\end{equation}
at $\beta=c/(\Delta m)$, fixed $c>1$. The errors are deterministic and uniform
over selected limits, almost surely in each. Appendix~\ref{app:density}
justifies conditioning on these invariant densities. Thus a large fixed $c$
already makes the possible finite-blue contribution small, without proving it
vanishes at any fixed dimension.

\section{Density imbalance and Gibbs states}\label{sec:discussion}

Theorems~\ref{thm:headline} and~\ref{thm:explicit} establish percolation in
both overlap sectors, but do not compare their infinite-cluster densities. In SK, the theorem of
Machta, Newman and Stein \cite{MNS2008} determines those densities from sector
sizes because the underlying sector graphs are complete. The corresponding
lattice question requires information beyond the lower percolation comparisons
proved here.

If the coupling law additionally has no atom at zero, the upper bound of
Pei \cite[Theorem~1.9(b)]{Pei2026} complements our existence result: wherever
our almost-sure coexistence conclusions hold, there are exactly two infinite
blue components, one in each overlap sector. To match the cited theorem's
full CMR setting, augment the finite-torus laws defining a selected
$\nu\in\mathcal L_{d,\beta}$ with the standard red bonds and extract a further
jointly convergent subsequence. Its blue marginal remains $\nu$, so the
at-most-two conclusion descends to that selected limit. The upper bound uses
a separate finite-box coalescence argument; the lower comparisons proved here
do not bound the number of infinite components within a sector.

\subsection{Imbalance suffices for Gibbs-state nonuniqueness}

The next implication is the qualitative criterion of
Machta, Newman and Stein \cite[Section II C, pp.~4--5, footnote 32]{MNS2007}.
We formulate its tail argument explicitly because it requires less than an
identity between density imbalance and spatial overlap. The relevant Gibbs
facts are tail conditioning and conditional disintegration, not a
nonpercolation-to-uniqueness theorem. Such a theorem is due to Gandolfi: if,
for almost every disorder and every sequence of boundary conditions, blue bonds
do not percolate in the disagreement sector, then the Gibbs state is almost
surely unique \cite[Theorem~4.3]{Gandolfi2018}. Theorems~\ref{thm:headline}
and~\ref{thm:explicit} give blue percolation in the disagreement sector in
periodic joint limits, not along boundary-condition sequences, so the two
statements concern different limits; neither is used below.

\begin{proposition}[Persistent imbalance and spin-flip symmetry]\label{prop:imbalance}
Fix $d$ and $\beta>0$. Let a stationary two-replica CMR law have, conditional
on almost every $J$, a spin marginal that is Gibbs for the sum of the two
uncoupled zero-field Ising Hamiltonians. Given $J$ and the spins, assume the
standard independent blue-activation kernel. Define $D_s$ as above. If
\[
 \nu(D_+\ne D_-)>0,
\]
then for a positive-probability set of disorders there exist distinct
spin-flip-related single-replica Gibbs states. For iid disorder, such states
exist for almost every disorder. Every selected periodic joint limit in this
paper has the stipulated Gibbs and activation properties.

Neither at-most-two nor finite-blue cancellation is needed for this implication.
Under our coexistence conclusion and at-most-two, the premise is precisely inequality of the two infinite
cluster densities. It is an infinite-volume density premise, not merely a
finite-volume observation about the two largest clusters.
\end{proposition}

\begin{proof}
For all configurations define $D_s$ first by the limsup of the cube averages;
stationarity supplies actual limits almost surely. A finite change of edges
changes infinite-cluster membership only at finitely many vertices. Indeed,
deleting finitely many edges splits each affected component into finitely many
pieces, whose finite pieces have finite total size; adding finitely many edges
can absorb only finitely many finite components. Finite changes of overlap
labels likewise do not affect their densities.

Realize blue bonds as a local function of $(J,\sigma,\tau,U)$ with iid
activation uniforms $U$. For fixed $J$, put
\[
 D=D_+-D_-,\qquad d_J(\sigma,\tau)=\int D(J,\sigma,\tau,U)\,dU.
\]
For fixed spins, $D$ is invariant under finite changes of $U$, so the
Kolmogorov zero-one law gives $D=d_J$ almost surely. A finite change of spins,
with the same uniforms, alters only finitely many labels and edges. Thus
$d_J$ is pointwise invariant under every finite spin change: it is a
two-replica spin-tail observable. Flipping the first replica globally reverses
its sign and leaves the blue graph unchanged.

Fix a $J$ for which $d_J$ is nonzero with positive probability. Condition the
two-replica spin law on whichever of $\{d_J>0\}$ or $\{d_J<0\}$ has positive
probability. This tail conditioning preserves the product Gibbs specification
\cite[Proposition 6.61(1)]{FriedliVelenik2017}. Flip its first replica to obtain
a second Gibbs law. The two laws have disjoint signed-tail supports and the
same second-replica marginal. Disintegrating over that common marginal gives
single-replica Gibbs states for the first replica: the local DLR kernel for
the first replica is independent of the second, so its DLR identities remain
valid under this disintegration. The conditional states are spin-flip related
and must differ on a positive-measure set, since otherwise the two joint laws
would coincide. Existence of such states is a measurable translation-invariant
property of $J$; iid-disorder ergodicity upgrades positive probability to one
\cite[Proposition 4.4]{Newman1997}. For the spin-flip refinement, measurability
follows by testing a countable family of local odd spin observables against
the compact set of Gibbs states.

For the last assertion, pass the finite-torus DLR equations for the product
spin specification against bounded disorder/exterior-spin cylinders to the
selected joint limit, then disintegrate over $J$. The local Ising kernels are
continuous and depend on finitely many couplings. The independent activation
kernel passes in the same way; its blue probability is continuous even at
$J_e=0$. This standard local-limit argument is also described in
\cite[Proposition 4.12, Remark 4.13 and Lemmas B.3--B.4]{Newman1997}.
It asserts the product specification, not conditional independence of the
limiting replicas.
\end{proof}

For the selected periodic ensemble, imbalance also has a physical consequence
in a finite local spin channel. The observable need not be a single spin.

\begin{proposition}[Imbalance forces an odd local order channel]\label{prop:odd-order}
Let $\nu\in\mathcal L_{d,\beta}$ and suppose $\nu(D_+\ne D_-)>0$.
There is a fixed finite set $A$ of odd cardinality for which
\begin{equation}\label{eq:odd-order}
 \rho_A:=\E_\nu Q_A^2>0,\qquad
 \lim_{n\to\infty}\frac1{|\Lambda_n|}
       \sum_{x\in\Lambda_n}\E_\nu[q_{A,0}q_{A,x}]=\rho_A,
\end{equation}
where $\Lambda_n=\{-n,\ldots,n\}^d$.
Every summand is nonnegative, so $\chi_A(\nu)=\infty$.

Under the standing finite-mean assumption, let $P_A(h)$ be the quenched
pressure per site for two replicas with the additional, dimensionless exponent
$h\sum_xq_{A,x}$, using only translates supported in the volume. Then
\begin{equation}\label{eq:composite-cusp}
 P_A(h)-P_A(0)\ge \sqrt{\rho_A}\,|h|.
\end{equation}
Thus this finite-range composite replica coupling has a pressure cusp at zero.
Corollary~\ref{cor:imbalance-overlap} below upgrades the conclusion to the
single-site replica coupling.
\end{proposition}

\begin{proof}[Proof outline]
Approximate the signed infinite-blue density by spatial averages of finite-radius
signed blue-arm events, with activation coins integrated out. Expand each local
observable in the spin monomials of the two replicas. Its oddness under either
replica flip leaves only odd monomials. Finite-volume replica independence and
Cauchy--Schwarz bound every mixed monomial by the corresponding same-observable
overlap second moments. These inequalities pass locally before taking spatial
averages. Nonzero imbalance therefore forces one fixed odd channel to survive.
Tail conditioning on a positive value of that channel preserves both the
Gibbs specification and the disorder marginal. A finite-boundary entropy estimate
then gives the pressure inequality. Appendix~\ref{app:odd-order} gives the
details, including the order of limits.
\end{proof}

The pressure step is a Gibbs variational argument in the Griffiths framework.
Itoi, Mukaida and Tasaki \cite[Theorem 3.1, Lemma 4.1 and Appendix C]{IMT2024}
relate overlap fluctuations to replica-coupling singularities and discuss general
local observables. Our selected-local-limit formulation requires the explicit
bridge in Appendix~\ref{app:odd-order}; their full-volume overlap premise cannot
simply be substituted.
The reduction from blue-density imbalance identifies the finite odd local
overlap to which this pressure argument applies.

The odd set $A$ supplied by Proposition~\ref{prop:odd-order} can be traded for a
single site. The comparison is local: it reweights finitely many couplings in
finite volume and needs no property of the limiting replicas.

\begin{proposition}[Odd local order forces ordinary overlap order]\label{prop:odd-to-site}
Let $\beta>0$ and $\nu\in\mathcal L_{d,\beta}$. Fix a finite set $A\subset\Z^d$
of odd cardinality and a finite edge set $F$ whose odd-degree vertices are
exactly those of $A\triangle\{0\}$; put $k=|F|$. Choose $0<w_-\le w_+<\infty$
with $\bar p:=\Pp(w_-\le W\le w_+)>0$. For every $\overline W>0$ and every
$x\in\Z^d$,
\begin{equation}\label{eq:odd-to-site}
 \E_\nu[q_{A,0}q_{A,x}]
 \le C\,\E_\nu[q_0q_x]+2k\,\Pp(W>\overline W),
 \qquad
 C=\max\Bigl(1,\frac{16}{\bar p\,\delta^2}\Bigr)^{2k},
\end{equation}
where $\delta=\sinh(2\beta w_-)/\cosh^2\!\bigl(\beta(w_++\overline W)\bigr)$.
Consequently $\rho_A\le C\,\E_\nu Q^2+2k\,\Pp(W>\overline W)$, and
$\E_\nu Q^2=0$ forces $\rho_A=0$. If $W$ is bounded, choosing $\overline W$
above its support gives $\chi_A(\nu)\le C\,\chi_{\rm SG}(\nu)$.
\end{proposition}

The standing assumption $m>0$ supplies an interval $[w_-,w_+]$ with $\bar p>0$,
and an atom of $W$ at zero is allowed. The proof, in
Appendix~\ref{app:odd-order}, replaces the couplings on a finite edge set and
extracts the full multiaffine coefficient of the reweighted correlation. Local
reweighting and finite-difference interpolation are standard; we use them only
in this comparison.

\begin{corollary}[Imbalance forces ordinary overlap order]\label{cor:imbalance-overlap}
Let $\beta>0$, $\nu\in\mathcal L_{d,\beta}$, and suppose $\nu(D_+\ne D_-)>0$.
Then $\E_\nu Q^2>0$ and $\chi_{\rm SG}(\nu)=\infty$. Under the standing
finite-mean assumption, the quenched two-replica pressure with the ordinary
site coupling $h\sum_xq_x$ satisfies
\[
 P_{\{0\}}(h)-P_{\{0\}}(0)\ge\sqrt{\E_\nu Q^2}\,|h|.
\]
In contrapositive form: if $Q=0$ almost surely, then $D_+=D_-$ almost surely.
This is trivial at $\beta=0$, where there are no blue bonds.
\end{corollary}

\begin{proof}
Proposition~\ref{prop:odd-order} gives an odd $A$ with $\rho_A>0$; taking
$\overline W$ large in Proposition~\ref{prop:odd-to-site} gives
$\E_\nu Q^2>0$. The terms $\E_\nu[q_0q_x]$ are nonnegative and their spatial
Ces\`aro averages converge to $\E_\nu Q^2$, as in \eqref{eq:odd-order}, so
$\chi_{\rm SG}(\nu)=\infty$. The pressure bound is the argument of
Appendix~\ref{app:odd-order} with $A=\{0\}$.
\end{proof}

These are statements about the selected periodic ensemble. They do not assert
that $Q\ne0$ on the event $\{D_+\ne D_-\}$, nor the statewise identity
\eqref{eq:overlapdecomposition} with $F_{\mathrm{fin}}=0$.

Proving persistent imbalance would therefore close this route to
Gibbs-state nonuniqueness, to ordinary spatial-overlap order and to a
singularity of the ordinary replica-coupling pressure. Its proof still
requires a new estimate that distinguishes the two macroscopic sectors; the
present lower comparisons treat them symmetrically.
Figure~\ref{fig:physics-gap} summarizes these implications.

\begin{figure}[!b]
\centering
\begin{adjustbox}{max width=\linewidth}
\begin{tikzpicture}[>=Latex,
 box/.style={draw,rounded corners,align=center,font=\small,text width=3.55cm,minimum height=1.15cm}]
\node[box,fill=blue!5] (exist) at (0,0)
 {Two infinite blue clusters\\One in each overlap sector};
\node[box,fill=orange!8] (imb) at (4.65,0)
 {Persistent imbalance\\$\nu(D_+\ne D_-)>0$\\Open estimate};
\node[box,fill=gray!8] (gibbs) at (9.3,0)
 {Distinct Gibbs states\\related by spin flip\\Proposition~\ref{prop:imbalance}};
\node[box,fill=gray!8] (overlap) at (9.3,-2.8)
 {Ordinary overlap order\\$\E_\nu Q^2>0$, $\chi_{\rm SG}=\infty$\\Corollary~\ref{cor:imbalance-overlap}};
\draw[->,thick,dashed] (exist)--(imb);
\draw[->,thick] (imb)--(gibbs);
\draw[->,thick] (imb)--(overlap);
\node[align=center,font=\small,text width=3.7cm] at (0,-2.8)
 {Percolation with $Q=0$\\and finite $\chi_{\rm SG}$\\Corollary~\ref{cor:separation}};
\end{tikzpicture}
\end{adjustbox}
\caption{The two remaining physical questions have different hypotheses.
The exact-two interpretation assumes no atom at zero and uses
Pei \cite[Theorem 1.9(b)]{Pei2026}.
Dashed arrows require unproved input; solid arrows are established implications.
Imbalance alone suffices for Gibbs-state nonuniqueness and, in the selected
periodic ensemble, for ordinary overlap order with a pressure cusp
(Propositions~\ref{prop:odd-order} and~\ref{prop:odd-to-site}). Neither
implication needs control of the signed finite-blue contribution; that
control is needed only to turn \eqref{eq:overlapdecomposition} into a
statewise conclusion, and for the converse direction.}
\label{fig:physics-gap}
\end{figure}
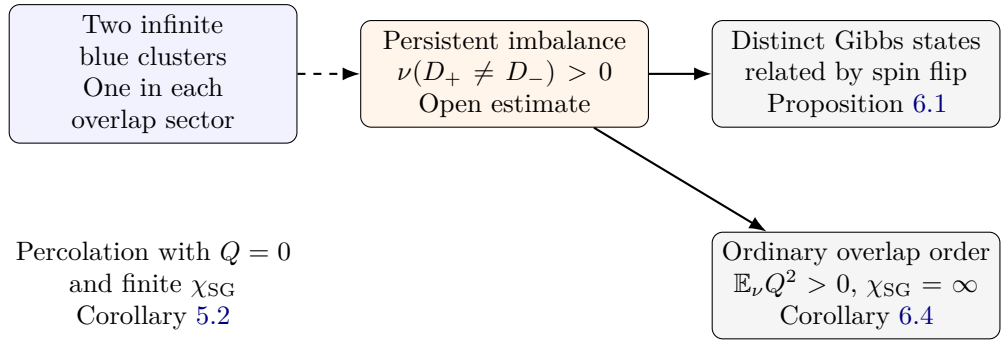

\subsection{Overlap order and the statewise density identity}

Corollary~\ref{cor:imbalance-overlap} gives ordinary overlap order from
imbalance at the level of the selected ensemble. A statewise identification
of the overlap with the blue imbalance needs more. In a translation-ergodic
component of a stationary joint law, define
\[
 Q=\E[q_0],\qquad
 \theta_s=\Pp(q_0=s,\ |C_{\mathrm b}(0)|=\infty),\qquad
 F_{\mathrm{fin}}=\E[q_0\mathbf1_{\{|C_{\mathrm b}(0)|<\infty\}}].
\]
These expectations are taken in that component. The spatial overlap mean and
the densities of vertices in infinite blue components satisfy the exact identity
\begin{equation}\label{eq:overlapdecomposition}
 Q=\theta_+-\theta_-+F_{\mathrm{fin}}.
\end{equation}
Under the preceding exactly-two conclusion, $\theta_s$ is the density of the
unique infinite blue component in sector $s$.
A strict inequality $|F_{\mathrm{fin}}|<|\theta_+-\theta_-|$ would give
$Q\ne0$ in that component, and a bound preventing $F_{\mathrm{fin}}$ from
carrying the whole overlap would give the converse implication, from overlap
order to imbalance. The density estimates above bound the possible finite
contribution but supply neither inequality.

The term $F_{\mathrm{fin}}$ cannot be discarded by assigning independent signs
to finite blue clusters. In the full CMR representation, red bonds join
opposite-overlap blue components. Overlap signs are constrained by red-bond parity
within each connected grey component (blue union red), as in the signed
connectivity formula of \cite[Section II]{MNS2007}. Finite blue components can
belong to an infinite grey component and have correlated contributions.

Finite grey components, by contrast, cancel. Work in a selected limit
carrying the standard red bonds, as in the discussion of
Proposition~\ref{prop:imbalance}, and write $C_{\rm g}(x)$ for the grey
component of $x$. For each finite $K\ni x$, flipping one replica on $K$ on
the cylinder event $\{C_{\rm g}(x)=K\}$ preserves every torus law containing
$K$: the edges leaving $K$ are closed, a closed edge has the same weight in
every spin configuration, and edges inside $K$ keep their spin products. This
map acts on finitely many spins on a clopen event, so it also preserves the
selected limit, and flipping one replica on $C_{\rm g}(x)$ whenever it is finite
is the countable disjoint union of these maps. It reverses
$g_x=q_x\mathbf1\{|C_{\rm g}(x)|<\infty\}$ and leaves $g_y$ unchanged for
$y\notin C_{\rm g}(x)$, so $\E[g_xg_y\mathbf1\{y\notin C_{\rm g}(x)\}]=0$. Hence
\[
 \E\Bigl(\frac1{|\Lambda|}\sum_{x\in\Lambda}g_x\Bigr)^2
 \le\frac1{|\Lambda|}\E\bigl[\min(|C_{\rm g}(0)|,|\Lambda|);\,
      |C_{\rm g}(0)|<\infty\bigr]\longrightarrow0,
\]
and the spatial density of $g$ vanishes almost surely. Put
$S_I=\E\bigl[q_0;\,|C_{\rm b}(0)|<\infty,\ |C_{\rm g}(0)|=\infty\bigr]$, the
signed density of finite blue components lying in infinite grey components.
Therefore, in almost every translation-ergodic component of this
red-augmented limit, $F_{\mathrm{fin}}=S_I$ and \eqref{eq:overlapdecomposition}
sharpens to $Q=\theta_+-\theta_-+S_I$.
A useful question is therefore whether the iid Edwards--Anderson law admits a
bound preventing this signed finite-cluster contribution from cancelling an
infinite-cluster imbalance. Exact vanishing is a stronger requirement than
\eqref{eq:overlapdecomposition} needs.

Appendix~\ref{app:limitations} records the signed-grey response identities,
alternative representations and the limitations of the current comparison in
low dimension. Those calculations identify what additional estimates would be
needed, but do not supply the missing imbalance or noncancellation bound.

\subsection{The remaining transition target}\label{sec:future}

The existence theorem resolves the blue-percolation part of this program in
high finite dimensions. The physical estimates make its relation to ordering
more precise: an interval above the first blue onset has exactly balanced
infinite sectors and finite susceptibility (Corollaries~\ref{cor:separation}
and~\ref{cor:exact-balance}), and near the onset every fixed odd local channel
has finite susceptibility as well (Proposition~\ref{prop:sg-bound}).
Persistent imbalance would break this balance, give spin-flip-related Gibbs
states, and force ordinary overlap order with a cusp of the ordinary
replica-coupling pressure.
These are conditional consequences, not an established ordered phase.

\begin{hypothesis}[Persistent imbalance at a fixed positive temperature]\label{hyp:imbalance}
For unit symmetric iid couplings there exist a finite $d$, a finite
$\beta>0$, and $\nu\in\mathcal L_{d,\beta}$ such that
$\E_\nu(D_+-D_-)^2>0$.
\end{hypothesis}

Proposition~\ref{prop:imbalance} and Corollary~\ref{cor:imbalance-overlap}
state exactly what this hypothesis would settle. Neither suppression of a
second infinite cluster nor a separate cancellation theorem for finite blue
clusters is required.

The present method also locates a region where the hypothesis is false:
Corollary~\ref{cor:exact-balance} excludes every parameter with
$F(2\beta)^{2d-1}<3$.
In high dimensions this includes a nontrivial interval on the
$\beta\asymp d^{-1/2}$ scale, well above the blue onset at order $d^{-1}$.
Failure of this sufficient balance criterion does not establish imbalance.
A promising new input would control normalized overlap-interface weights or
retain a nonvanishing odd boundary signal at fixed temperature and arbitrarily
large scales. The opening estimates of this paper do neither: they guarantee
growth symmetrically in the two sectors. This is the principal unresolved ordering
step, separate from improving the explicit dimension of geometric percolation.

\appendix
\section{Certified constants for the explicit dimensions}\label{app:certificates}

This appendix collects the finite checks behind Theorem~\ref{thm:explicit}.

\subsection{Dimensions nine to twelve: the oriented engine}\label{sec:ov-constants}

Table~\ref{tab:overlap-constants} lists the nine certified rows in dimensions
$10$ to $12$, and Table~\ref{tab:ov-noisy} the five rows in dimension $9$. For
each row of the first table the following finitely many inequalities were
decided exactly.
\begin{itemize}
\item[(C1)] $\Delta_j\le0$ for $j=0,\ldots,2d-1$ in Proposition~\ref{prop:ov-vertex},
with $\underline r=p/p_A$; hence $p\le p_B$.
\item[(C2)] The Holley line \eqref{eq:ov-holley-line} for $k=0,\ldots,2d$.
\item[(C3)] $g_K\ge1$, $g_h<1$, $g_Kg_h\le1$ and $g_K^{2d}<2718/1000<e$; that
is, $K'\ge0$, $h<0$, $K'\le|h|$, and $4dK'<1$, which is stronger than $2dK'<1$.
\item[(C4)] The mean-field test of Lemma~\ref{lem:ov-meanfield} for a rational
$\bar m$, giving $\rho_-=(1-\bar m)/2$.
\item[(C5)] $\lambda_*<1$ for a rational $\bar t\ge\tanh K'=(g_K-1)/(g_K+1)$.
One program takes $\bar t=(g_K-1)/(g_K+1)$ exactly, the other a rational upper
bound for $K'$.
\item[(C6)] $\eta'<1$; moreover $\mathcal S<1$ in dimensions $11$ and $12$, and
$\rho_-^2\ge c_{\rm cov}/4$ and $\mathcal S_c<1$ in dimension $10$.
\end{itemize}
Together with $d\ge6$, these are exactly the hypotheses of
Propositions~\ref{prop:ov-vertex} and~\ref{prop:ov-holley-line},
Lemma~\ref{lem:ov-meanfield} and Theorem~\ref{thm:ov-criterion}.

The rows of Table~\ref{tab:ov-noisy} replace (C1) and (C2) by the sharpened
local inputs of Section~\ref{sec:ov-noisy}.
\begin{itemize}
\item[(N1)] A rational $w'_c$ with $e^{2\beta_c}\le w'_c\le e^{2\beta}$, which
fixes $\beta'=\frac12\log w'_c\in[\beta_c,\beta]$; $p_A/2<p<\tanh2\beta$; and
rationals $\psi_U(x)$ satisfying \eqref{eq:ov-chord} and
\eqref{eq:ov-class-sums}. By Lemma~\ref{lem:ov-pair}, every conditional
opening probability is then at least $p$.
\item[(N2)] For $k=0,\ldots,2d$, tangent points with $\Gamma_k(l_+,l_-)\ge0$
in every class, where $\Lambda'_k=g_hg_K^SC_\beta^{-S}$; this is the hypothesis
of Corollary~\ref{cor:ov-sym-holley}.
\end{itemize}
Conditions (C3)--(C5) are as above, and (C6) becomes: $\eta'<1$; at
$t\ne7/50$, $\rho_-^2\ge c_{\rm cov}/4$ and $\mathcal S_c<1$; at $t=7/50$,
$2dK'\le0.35$, so that $\alpha_{\rm D}\le0.35$, together with
$\rho_-^2\ge c'_{\rm cov}/4$ and $\mathcal S'_c<1$. With $d\ge6$
and $L\ge40$, which Corollary~\ref{cor:ov-uniform} provides, these are the
hypotheses of Lemmas~\ref{lem:ov-pair}, \ref{lem:ov-meanfield}
and~\ref{lem:ov-covariance}, Corollary~\ref{cor:ov-sym-holley} and
Theorem~\ref{thm:ov-criterion}.

\begin{table}[!htbp]
\centering
\begin{adjustbox}{max width=\linewidth}
\begin{tabular}{cc|ccc|ccccc}
\hline
$d$ & $t$ & $p$ & $g_K$ & $g_h$ & $\rho_-\ge$ & $\eta'\le$ &
$\mathcal S\le$ & $\mathcal S_c\le$ & $\theta_*\ge$\\ \hline
$12$ & $3/25$   & $0.2161$ & $1.02634$  & $0.83527$  & $0.4351284$ & $0.168984$ & $0.907068$ & $0.897470$ & $0.036292$\\
$12$ & $11/100$ & $0.2023$ & $1.022308$ & $0.873143$ & $0.4540507$ & $0.127091$ & $0.926888$ & $0.918156$ & $0.031296$\\
$12$ & $23/200$ & $0.2094$ & $1.024283$ & $0.855173$ & $0.4453800$ & $0.146071$ & $0.913716$ & $0.904575$ & $0.035441$\\
$12$ & $1/8$    & $0.2224$ & $1.028469$ & $0.815686$ & $0.4240888$ & $0.196297$ & $0.905149$ & $0.895074$ & $0.034915$\\
$12$ & $13/100$ & $0.2283$ & $1.030678$ & $0.794291$ & $0.4111645$ & $0.230516$ & $0.910357$ & $0.899716$ & $0.030630$\\ \hline
$11$ & $13/100$ & $0.2310$ & $1.03088$  & $0.817696$ & $0.4252374$ & $0.215864$ & $0.952104$ & $0.940625$ & $0.017392$\\
$11$ & $3/25$   & $0.2180$ & $1.026502$ & $0.855434$ & $0.4454995$ & $0.162436$ & $0.960915$ & $0.950461$ & $0.015882$\\ \hline
$10$ & $27/200$ & $0.2399$ & $1.033426$ & $0.822472$ & $0.4279573$ & $0.236224$ & $1.007739$ & $0.994576$ & $0.001950$\\
$10$ & $13/100$ & $0.2336$ & $1.031101$ & $0.840843$ & $0.4379919$ & $0.205637$ & $1.009832$ & $0.997261$ & $0.001048$\\ \hline
\end{tabular}
\end{adjustbox}
\caption{Certified constants for Theorem~\ref{thm:explicit}, with
$\beta=\operatorname{atanh}t$. The inputs $p\le p_B$, $g_K=e^{2K'}$ and
$g_h=e^{2h}$ are exact rationals, written here as terminating decimals. The
columns $\rho_-$, $\eta'$, $\mathcal S$ and $\mathcal S_c$ are the density bound
of Lemma~\ref{lem:ov-meanfield}, the constant of Lemma~\ref{lem:ov-cycle}, and
the scores of Theorem~\ref{thm:ov-criterion} with $\rho_e=\rho_-$ and
$\rho_e=\rho_c$. The last column is \eqref{eq:ov-theta}, computed with
$\mathcal S$ in dimensions $12$ and $11$ and with $\mathcal S_c$ in
dimension~$10$, where the bound on $\mathcal S$ exceeds $1$. Each computed
entry is the less favorable of two independent evaluations; lower bounds are
rounded down and upper bounds up.}
\label{tab:overlap-constants}
\end{table}

\begin{table}[!htbp]
\centering
\begin{adjustbox}{max width=\linewidth}
\begin{tabular}{cc|ccc|ccccc}
\hline
$d$ & $t$ & $p$ & $g_K$ & $g_h$ &
$\rho_-\ge$ & $\eta'\le$ & $\mathcal S_c\le$ & $\mathcal S'_c\le$ & $\theta_*\ge$\\ \hline
$9$ & $7/50$   & $0.2544$ & $1.034171$ & $0.864557$ & $0.4481119$ & $0.2244019$ & $1.0033446$ & $0.9975456$ & $0.000950$\\
$9$ & $29/200$ & $0.2618$ & $1.0365$   & $0.848784$ & $0.4399191$ & $0.2538402$ & $0.9945159$ & $0.9875702$ & $0.002004$\\
$9$ & $3/20$   & $0.2689$ & $1.038887$ & $0.831941$ & $0.4306354$ & $0.2886861$ & $0.9906982$ & $0.9824125$ & $0.003173$\\
$9$ & $31/200$ & $0.2759$ & $1.041383$ & $0.814648$ & $0.4203748$ & $0.3309161$ & $0.9910479$ & --- & $0.002804$\\
$9$ & $4/25$   & $0.2828$ & $1.043938$ & $0.796248$ & $0.4087636$ & $0.3827874$ & $0.9967978$ & --- & $0.000899$\\ \hline
\end{tabular}
\end{adjustbox}
\caption{Certified constants for Theorem~\ref{thm:explicit} in dimension $9$,
with the sharpened local inputs of Section~\ref{sec:ov-noisy}. The inputs $p$,
$g_K$ and $g_h$ are exact rationals, written as terminating decimals; $p$ is a
lower bound for every conditional opening probability
(Lemma~\ref{lem:ov-pair}), although it exceeds $p_B$. The other columns are as in Table~\ref{tab:overlap-constants}, with $\mathcal S'_c$ the score for
$\rho_e=\rho'_c$ of Lemma~\ref{lem:ov-covariance}, which needs
$\alpha_{\rm D}\le0.35$ and does not apply in the last two rows. The last column
is \eqref{eq:ov-theta}, computed with $\mathcal S'_c$ at $t=7/50$ and with
$\mathcal S_c$ otherwise. The computed upper bounds on the basic score
$\mathcal S$ exceed $1$ in every row. Each global entry is the less favorable of
two independent evaluations, rounded in the safe direction.}
\label{tab:ov-noisy}
\end{table}

\paragraph{How the constants are computed.}
All algebraic quantities of the local certificates are exact rationals when
$t$ is rational, and (C1)--(C4) are decided in exact rational or integer
arithmetic. The global criterion also involves $K'=\frac12\log g_K$,
$|h|=-\frac12\log g_h$, exponentials, $\sqrt d$ and $\pi$. These finitely many
irrational numbers are replaced by rigorous rational enclosures. The collision
probabilities $u_k$ are computed exactly for $k<4d$, and the tail
$\sum_{k\ge4d}u_k$ is bounded through Lemma~\ref{lem:ov-green}(ii) and
$\sum_{a\ge A}a^{-s}\le A^{-s}+A^{1-s}/(s-1)$; the tail of $G^{(2)}_d$ is
bounded in the same way. For instance,
\[
 G_{12}\in[1.1011537095,\,1.1011552715],\qquad G^{(2)}_{12}\le1.2272959943,
\]
and $F_{12}-\frac1{12}\le0.0085295258$, $F_{11}-\frac1{11}\le0.0104053880$ and
$F_{10}-\frac1{10}\le0.0130075189$. The values $\mathcal B(m)$ and
$|\mathbb A_r|$ are exact rationals. The sums
$\bar\psi(r)$ and $\Sigma U$ are bounded by exact finite sums, with $m$ up to
$86$--$124$ and $r$ up to $40$--$60$ depending on the program, plus explicit
geometric tails from Lemma~\ref{lem:ov-resolvent}(iv) and
$|\mathbb A_r|\le\binom{r+d-1}{d-1}^2$.

The computations are small. For dimension $d$, (C1) consists of $2d$ double sums
of at most $d(d+1)$ rational terms. For (C2), each of the $2d+1$ environments
needs the $k+1$ double sums $\widetilde X_k(j)$ of Lemma~\ref{lem:ov-numerator}
and one denominator sum, followed by one comparison. (C4) is a single
inequality between rationals. At $d=12$ this amounts to $24$ sums of at most
$156$ terms, $325$ numerator sums of at most $169$ terms and $25$ denominator
sums; each row takes between one and thirty seconds. These are closed-form
finite sums, not a search over configurations.

Every inequality was checked by two independent programs. The second was
written from the mathematical statements alone, without sharing code with the
first. It organizes the floor sums as convolutions of binomial laws, computes
the collision probabilities from a partition sum rather than from polynomial
powers, and chooses its own tangent points and line weights. Both programs
certify every row, and their score bounds agree to about $10^{-6}$.
The programs are described under Computational reproducibility.

For the rows of Table~\ref{tab:ov-noisy}, (N1) consists of $162$ point
conditions, $1377$ chord polynomials of degree $5$ (nine values of $U$ and
$153$ pairs $x,x'$) and $18$ class sums, and (N2) of $1330$ two-block classes
in the $19$ environments, at most $100$ per environment. Two independently
written programs decide (N1) and (N2) exactly. The first decides the chord
conditions by Bernstein coefficients with subdivision and Sturm sequences, and
computes the class sums and the classes of (N2) from binomial counts. The
second, written from the statements alone, uses its own Sturm sequences,
obtains the class sums by enumerating all $2^{17}$ sign vectors, and counts the
classes of (N2) by enumerating the sign patterns on $N_+$ and on $N_-$. On a
laptop each takes under two minutes for all five temperatures. In the order of Table~\ref{tab:ov-noisy},
the largest class sums $\max_j\Delta^{\rm P}_j$ are at most $-54.61$, $-30.89$,
$-41.55$, $-43.99$ and $-28.83$, and the smallest class values
$\min\Gamma_k$ over all environments are at least $9.09\times10^{-5}$,
$7.40\times10^{-5}$, $8.60\times10^{-5}$, $8.28\times10^{-5}$ and
$8.21\times10^{-5}$; these Holley margins are exact but small. The global
constants are evaluated by the two programs above, and $c'_{\rm cov}$ and
$\mathcal S'_c$ by two independent implementations of
Lemma~\ref{lem:ov-covariance}, which agree to $4\times10^{-6}$. At $t=3/20$, for instance, $w'_c=674054518843/(5\cdot10^{11})$.
The data $w'_c$, $\psi_U(x)$, tangent points and line weights were chosen by
floating-point linear programs, which play no role in the verification.

At $d=12$, $t=3/25$ the three terms of $\mathcal S$ are at most $0.886231$,
$0.019603$ and $0.001235$, while $\kappa\le1.29818$ and
$\alpha_{\rm D}\le0.311972$. The certified margins $1-\mathcal S$ range from
$7.31\%$ to $9.48\%$ in dimension $12$, and are $4.78\%$ and $3.90\%$ in
dimension $11$. In dimension $10$ the basic criterion certifies neither row:
its upper bounds are
$1.0078$ and $1.0099$. The shared-edge gain $\rho_c-\rho_-\approx0.006$ of
Lemma~\ref{lem:ov-shared} brings $\mathcal S_c$ below $1$, with margins $0.54\%$
at $t=27/200$ and $0.27\%$ at $t=13/100$. These are exact inequalities. At
$d=10$, $t=3/25$ the computed upper bound on $\mathcal S_c$ is $1.0193>1$, and
that point is not certified.

In dimension $9$ the basic criterion certifies no row of
Table~\ref{tab:ov-noisy}. The shared-edge refinement certifies four rows, with
margins $1-\mathcal S_c$ between $0.32\%$ and $0.93\%$. At $t=3/20$ the three
terms of $\mathcal S_c$ are at most $0.9450117$, $0.0390500$ and $0.0066367$,
and $9\rho_cp\ge1.0581879$. At $t=7/50$ only the covariance refinement of
Lemma~\ref{lem:ov-covariance} certifies, with margin $0.24\%$; the same
refinement gives the larger constants $\theta_*\ge0.004544$ at $t=29/200$ and
$\theta_*\ge0.006000$ at $t=3/20$. At $t=3/25$, $13/100$ and $17/100$ the
computed upper bounds on $\mathcal S_c$ exceed $1.03$, and these points are not
certified.

\subsection{Dimensions seven to nine: the macrostep engine}\label{sec:ms-certificates}

For each row of Table~\ref{tab:ms-local} the following were decided in exact
rational arithmetic.
\begin{itemize}
\item[(V0)] The local inputs: (C1)--(C2) for the rows with the inputs of
Sections~\ref{sec:ov-floor}--\ref{sec:ov-holley}, and (N1)--(N2) for the rows
with the inputs of Section~\ref{sec:ov-noisy}, by the program that decides them
for Table~\ref{tab:overlap-constants}, respectively the first program of
Section~\ref{sec:ov-constants}; the second program of Section~\ref{sec:ov-constants}
is written for $d=9$ and is not run at these rows.
\item[(V1)] Conditions (C3)--(C4) and $g_K>1$, which give hypothesis (M1) of
Theorem~\ref{thm:ms-engine}.
\item[(V2)] $\lambda_*<1$ with $t'=(g_K-1)/(g_K+1)$,
$\rho_-^2\ge c_{\rm cov}/4$ and $\rho_cp<\rho_-$, with a rational lower bound for
$c_{\rm cov}$; this is hypothesis (M2).
\item[(V3)] For each near set of Table~\ref{tab:ms-certificates}, conditions
(i) and (ii) of Theorem~\ref{thm:ms-criterion}, with $\lambda<1$, for two
certificates: one with the smallest $\lambda$ found, and one maximizing
$\theta_*$. This gives hypothesis (M3).
\end{itemize}

\begin{table}[!htbp]
\centering
\begin{adjustbox}{max width=\linewidth}
\begin{tabular}{cc|c|cccc|ccc}
\hline
$d$ & $t$ & Section & $p$ & $g_K$ & $g_h$ & $\bar m$ & $\rho_-\ge$ &
$\rho_c\ge$ & $C_{\rm fin}\le$\\ \hline
$9$ & $7/50$   & \ref{sec:ov-floor}--\ref{sec:ov-holley} & $0.2492$ & $1.036137$ & $0.83033$  & $0.1353842$ & $0.4323079$ & $0.4387813$ & $2.3700$\\
$8$ & $3/20$   & \ref{sec:ov-floor}--\ref{sec:ov-holley} & $0.2651$ & $1.041622$ & $0.824956$ & $0.1413829$ & $0.4293085$ & $0.4366859$ & $2.3966$\\
$8$ & $3/20$   & \ref{sec:ov-noisy} & $0.2707$ & $1.039407$ & $0.860671$ & $0.1079894$ & $0.4460053$ & $0.4532042$ & $2.2991$\\
$7$ & $31/200$ & \ref{sec:ov-noisy} & $0.2802$ & $1.04253$  & $0.876345$ & $0.0927806$ & $0.4536097$ & $0.4616630$ & $2.2632$\\
$7$ & $3/20$   & \ref{sec:ov-noisy} & $0.273$  & $1.039978$ & $0.88898$  & $0.0808479$ & $0.4595760$ & $0.4673814$ & $2.2288$\\ \hline
\end{tabular}
\end{adjustbox}
\caption{The rows of Section~\ref{sec:macrostep}, with $d=f+3$ and
$\beta=\operatorname{atanh}t$. The inputs $p$, $g_K=e^{2K'}$, $g_h=e^{2h}$ and
$\bar m$ are exact rationals, written as terminating decimals; the local inputs
come from the sections named in the third column. The last three columns are
rigorous bounds on the density $\rho_-=(1-\bar m)/2$, the shared-edge density
$\rho_c$ of Lemma~\ref{lem:ov-shared} and $C_{\rm fin}$ of \eqref{eq:ms-theta},
rounded in the safe direction. For instance
$\kappa=(1-\rho_-)/\rho_-\le1.313167$ and
$\alpha_{\rm D}=2d\tanh K'\le0.31947$ at $d=9$, and $\kappa\le1.204539$ and
$\alpha_{\rm D}\le0.29152$ at $(7,31/200)$.}
\label{tab:ms-local}
\end{table}

\begin{table}[!htbp]
\centering
\begin{adjustbox}{max width=\linewidth}
\begin{tabular}{cc|c|cc|c|cc|cc}
\hline
$d$ & $t$ & Section & $(2R_D,R_A)$ & orbits &
$\eta_F\le$ & $\lambda\le$ & $\theta_*\ge$ & $\lambda\le$ & $\theta_*\ge$\\
 & & & & & & \multicolumn{2}{c|}{smallest $\lambda$} &
 \multicolumn{2}{c}{largest $\theta_*$}\\ \hline
$9$ & $7/50$   & \ref{sec:ov-floor}--\ref{sec:ov-holley} & $(4,8)$ & $246$ & $0.003728$ & $0.779215$ & $5.67\times10^{-4}$ & $0.867129$ & $0.001908$\\
$8$ & $3/20$   & \ref{sec:ov-floor}--\ref{sec:ov-holley} & $(4,8)$ & $246$ & $0.009564$ & $0.892512$ & $5.39\times10^{-4}$ & $0.924459$ & $6.60\times10^{-4}$\\
$8$ & $3/20$   & \ref{sec:ov-floor}--\ref{sec:ov-holley} & $(2,8)$ & $82$  & $0.020519$ & $0.934027$ & $6.44\times10^{-4}$ & \multicolumn{2}{c}{the same}\\
$8$ & $3/20$   & \ref{sec:ov-floor}--\ref{sec:ov-holley} & $(2,6)$ & $46$  & $0.028310$ & $0.962154$ & $4.83\times10^{-4}$ & \multicolumn{2}{c}{the same}\\
$8$ & $3/20$   & \ref{sec:ov-noisy} & $(2,6)$ & $46$ & $0.022690$ & $0.858694$ & $1.884\times10^{-3}$ & $0.872724$ & $0.001937$\\
$7$ & $31/200$ & \ref{sec:ov-noisy} & $(6,10)$ & $804$ & $0.005789$ & $0.924672$ & $2.65\times10^{-4}$ & $0.946970$ & $3.26\times10^{-4}$\\
$7$ & $3/20$   & \ref{sec:ov-noisy} & $(6,10)$ & $804$ & $0.004217$ & $0.935276$ & $1.71\times10^{-4}$ & $0.954393$ & $2.27\times10^{-4}$\\ \hline
\end{tabular}
\end{adjustbox}
\caption{Certificates for Theorem~\ref{thm:ms-criterion}, with $c=3$, $f=d-3$,
word weight $y_{\rm w}=1/10$ for $d=9,8$ and $y_{\rm w}=3/25$ for $d=7$, and
near sets $C'=\{|D|_1\le2R_D,\ |A|_1\le R_A\}$, whose numbers of
$\mathfrak G$-orbits are listed. For each near set two certificates were
decided: the one with the smallest $\lambda$ found, and the one maximizing
$\theta_*$ of \eqref{eq:ms-theta}, which has larger $\lambda$ and smaller
$\mathsf H$; in two rows they coincide. Upper bounds are rounded up and lower
bounds down, and each entry is the less favorable of two independent programs.
The $\theta_*$ of Table~\ref{tab:explicit} is the largest one of each row of
Table~\ref{tab:ms-local}.}
\label{tab:ms-certificates}
\end{table}

\paragraph{What is computed.}
For a near set $C'=\{|D|_1\le2R_D,\ |A|_1\le R_A\}$, reduced to
$\mathfrak G$-orbits, the programs compute the following.
\begin{enumerate}
\item \emph{Near kernel.} For each orbit representative $z\in C'$ and each of
the $163^2f^2$ pairs $\boldsymbol\omega$, the counts $N_V$, $N_E$, $N_{\rm f}$, the
functional $\mathfrak b$ and an upper bound on $\mathsf W-1$; the masses
$\pi_{\rm w}(\xi)\pi_{\rm w}(\xi')f^{-2}(\mathsf W-1)$ are aggregated by the
orbit of the next state. This gives upper bounds on $\mathsf Q(z,\cdot)$ and
$\bar\eta(z)$.
\item \emph{Green bounds.} $\mathsf G(z'',z')\le\sum_{n\ge0}u_n(D'-D'')l_n(A'-A'')$,
where $l_n$ is the law of the lateral difference walk at time $n$. The forward
$u_n$ are exact rationals for $n\le N_0=24$, the $l_n$ are upper bounds from a
direct convolution, and the tail $n>N_0$ is at most $l_{N_0}(0)T_u(N_0)$ per
pair of states, since $l_n(a)\le l_n(0)\le l_{N_0}(0)$ for $n\ge N_0$ by the
Cauchy--Schwarz and Young inequalities.
\item \emph{Matrix.} $\bar{\mathsf M}$ on $C'$-orbits from the near kernel and
the Green bounds, and $\Gamma_w$ from the near rows of the Green bounds.
\item \emph{Far region.} The bound $\bar\upsilon$ from
Lemma~\ref{lem:ms-far}: exact kernel averages on a band of lateral distances
beyond $C'$, part (b) beyond the band, and parts (c)--(d) for
$|D|_1>2R_S=8$. Then $\eta_F$ from Lemma~\ref{lem:ms-khas}, with exact $u_n$ up to
$N_F=60$. The tails $T_u(N)$ are exact partial sums up to $40f-1$ plus the
bound of Lemma~\ref{lem:ov-green}(ii) with $f$ in place of $d$, summed as in
Section~\ref{sec:ov-constants}.
\item \emph{Certificate.} A floating-point search takes
$w=(\lambda_0-\bar{\mathsf M})^{-1}(\bar\eta+\delta\mathbf1)$ for small $\delta>0$
and $\mathsf H=\Gamma_w/(\lambda-\eta_F)$ and bisects in $\lambda_0$; then (i) and
(ii) are decided in exact rational arithmetic, reading every binary64 number
as the dyadic rational it represents.
\end{enumerate}
With this choice of $w$, condition (i) roughly requires $\mathsf H\eta_F\le1$,
while (ii) forces $\mathsf H\ge\Gamma_w/(\lambda-\eta_F)$. So $\lambda$ must
exceed both the Perron root of $\bar{\mathsf M}$ and, roughly,
$\eta_F(1+\Gamma_w)$: the far region binds, and larger near sets lower
$\eta_F$.

\paragraph{Sizes and arithmetic.}
Per near state there are $163^2f^2$ pairs of macrosteps: $956{,}484$ at $d=9$,
$664{,}225$ at $d=8$ and $425{,}104$ at $d=7$. The near kernels therefore sum
about $2.35\times10^8$ terms at $d=9$ ($246$ orbits), $3.06\times10^7$ at $d=8$
($46$ orbits) and $3.42\times10^8$ at $d=7$ ($804$ orbits), and the far bands
about $10^8$ to $6\times10^8$ more. Only the $23$ orbits with $D=0$ and
$|A|_1\le6$ carry intersections of the two paths; the other rows carry only the
Ising boost. The tables $\mathcal B$ and $\mathcal T_1,\mathcal T_2,\mathcal T_3$
are exact rationals or exact dyadic upper bounds. The kernel sums are computed
in binary64 arithmetic by two independently written programs. The
\emph{directed-rounding program} rounds every operation upward. The \emph{a
priori error program} uses round-to-nearest arithmetic with an a priori forward
error analysis and inflates each aggregated quantity by a factor exceeding the
accumulated relative error. Exponentials are bounded by Taylor polynomials with
explicit remainders; the exponents $\kappa\mathfrak b$ that occur are at most
$0.44$ in every certified run, inside the range $[0,1]$ of these enclosures.
Both programs certify every row. At the near sets with the first local inputs,
their near matrices, $\bar\eta$ and near Green rows agree entrywise within a
relative $3\times10^{-8}$, those of the a priori error program being slightly
larger, as its error margins require. On a workstation a run takes between
half a minute and fifteen minutes. This is a moderate computer-assisted
certificate, not a criterion in a few scalars.

The local inputs of the rows are small by comparison. At $(9,7/50)$ and
$(8,3/20)$ with the inputs of Sections~\ref{sec:ov-floor}--\ref{sec:ov-holley},
the largest $\Delta_j$ of Proposition~\ref{prop:ov-vertex} are at most
$-2.306\times10^{-4}$ and $-1.760\times10^{-4}$, and the logarithmic slack in
\eqref{eq:ov-holley-line} is at least $0.001005$ and $0.000980$ in all
environments. With the sharpened inputs, (N1) involves $960$ chord polynomials
at $d=8$ and $637$ at $d=7$, with largest class sums at most $-13.94$ at
$(8,3/20)$, $-3.08$ at $(7,31/200)$ and $-2.80$ at $(7,3/20)$, and the smallest
class values $\Gamma_k$ are at least $7.49\times10^{-5}$, $9.33\times10^{-5}$
and $9.31\times10^{-5}$. At these three points the floors $p=0.2707$, $0.2802$
and $0.273$ exceed the aligned frozen values with $\beta$ inside $R$, which are
at most $0.270378$, $0.279462$ and $0.272179$, so no floor that is uniform over
$\mathcal N_{2d}$ could replace them (Remark~\ref{rem:ov-nesting}).

The margins $1-\lambda$ of the minimal certificates are about $22\%$ at
$(9,7/50)$, $10.7\%$ at $(8,3/20)$ with the first inputs and $14\%$ with the
sharpened ones, and $7.5\%$ and $6.4\%$ at $(7,31/200)$ and $(7,3/20)$. The
constants $\theta_*$ are small; their only role is to be positive.

\begin{remark}[Which inputs each row needs]\label{rem:ms-inputs}
The rows with the inputs of Sections~\ref{sec:ov-floor}--\ref{sec:ov-holley} do
not use Section~\ref{sec:ov-noisy}. They also certify with $\rho_c$ replaced by
$\rho_-$ in $\mathsf W$, with $\lambda\le0.809296$ at $d=9$ and
$\lambda\le0.916404$ at $d=8$ on the $246$-orbit near sets, and with the smaller
floors $p=0.22$ at $d=9$ ($82$ orbits, $\lambda\le0.981408$) and $p=0.255$ at
$d=8$ ($246$ orbits, $\lambda\le0.948297$). Dimension $7$ is certified only
with the sharpened inputs. With local inputs of the kind of
Sections~\ref{sec:ov-floor}--\ref{sec:ov-holley} at $t=3/20$, the computed near
matrix $\bar{\mathsf M}$ of the $804$-orbit near set has Perron root at least
$1.004280$, by an exact Collatz--Wielandt enclosure, so condition (i) fails for
every $w>0$ and $\lambda<1$. At $t=4/25$ its root is below $1$, but the far region is too large
($\eta_F\le0.0113$) and the search found no certificate. Both statements
concern the computed upper bounds, not the exact kernel of the path family.
\end{remark}

\FloatBarrier

\section{Application to finite-mean Edwards--Anderson disorder}\label{sec:finite-mean}

Assume only $0<\E W<\infty$. The exponential restoration moment may now
be infinite. We avoid this cost by skipping targets incident to a very large
coupling, while leaving the physical couplings and the Gibbs law unchanged.
The lower comparison loses a vanishing fraction of targets; a separate capped
path estimate supplies the matching upper bound.

\subsection{Local bounds at good targets}

Fix a cutoff $K<\infty$ with $\alpha:=\Pp(W\le K)>0$. Reveal the iid marks
$C_e=\mathbf1_{\{W_e\le K\}}$, and call a vertex \emph{good} if all its
incident marks are one. Under the conditional magnitude law given $W\le K$, put
\begin{equation}\label{eq:cutoffparameters}
\begin{aligned}
 b_K&=\E[\tanh(\beta W)\mid W\le K],\\
 \ell_K&=\E[\sech^2(\beta W)\mid W\le K],\\
 M_K&=\E[\cosh(4\beta W)\mid W\le K],
\end{aligned}
\end{equation}
and, for a degree bound $\Delta\ge2$, define
\begin{equation}\label{eq:pi}
 r_K=\ell_K/M_K,\qquad \pi_K=b_K r_K^{\Delta-1}.
\end{equation}

\begin{lemma}[Good-target bounds]\label{lem:goodtarget}
Conditional on the marks, let $v$ be good and let $H\in\mathcal F_v^{\mathrm{ext}}$
have positive conditional probability. For $e=uv$ and $s=\pm1$,
\begin{equation}\label{eq:goodtarget}
 \nu_{G,\beta}(B_e=1\mid C,H)\ge\pi_K,
 \qquad
 \nu_{G,\beta}(q_v=s\mid C,H)\ge\tfrac12 r_K^\Delta.
\end{equation}
\end{lemma}

\begin{proof}
Use Lemma~\ref{lem:restoration} on the family of graph laws conditional on
the fixed ambient marks. Removing an edge does not remove its mark from this
conditioning parameter. Present magnitudes have their independent conditional
laws given those marks; the cavity law is independent of the magnitude and fair
sign of the edge to be restored. Lemma~\ref{lem:cmr-leaf} supplies leaf success
$b_K$ after averaging the leaf magnitude, and isolated overlap probability $1/2$.
Figure~\ref{fig:leaf} shows the permitted deletion.

Conditional on the marks, all incident magnitudes have the bounded conditional
law used in \eqref{eq:cutoffparameters}. Averaging Lemma~\ref{lem:insertion}
over one such inserted magnitude multiplies any retained expectation by a factor
between $\ell_K$ and $M_K$. The observables $\mathbf1_H$, $\mathbf1_HB_e$ and, for full-star
restoration, $\mathbf1_H\mathbf1_{\{q_v=s\}}$ all qualify. The two
conclusions are therefore precisely the leaf and isolated-label conclusions of
Lemma~\ref{lem:restoration}. Retained disorder outside the restored star may
be arbitrarily large; the pointwise bounds are uniform in it.
\end{proof}

The success/failure refinement in Corollary~\ref{cor:odds} also applies here;
the simpler parameter $\pi_K$ suffices for the theorem.

\subsection{Discarding rare large-coupling stars}

\begin{lemma}[A site field with edge defects]\label{lem:defects}
On a finite or countable simple graph of maximum degree $\Delta$, let $C_e$ be
iid Bernoulli($\alpha$), let $X_v$ be iid Bernoulli($\pi$), independently,
and suppose $0\le\pi<1$. Then
\begin{equation}\label{eq:defectfield}
 Y_v=X_v\prod_{e\ni v}C_e
 \quad\text{stochastically dominates iid sites of parameter}\quad
 \varpi=\pi\left[\frac{\alpha(1-\pi)}{1-\alpha\pi}\right]^\Delta.
\end{equation}
\end{lemma}

\begin{proof}
Fix $v$ and a finite history $H$ specifying $Y$ at vertices other than $v$.
For $e=vw$, fix all marks except $C_e$ and integrate the independent $X$ variables.
Changing $C_e$ affects the history only through $Y_w$. If $w$ is not observed,
the good/bad likelihood ratio is one. If $Y_w=0$ is observed, it is either one
or $1-\pi$. If $Y_w=1$, the good mark is forced. Bayes' formula therefore gives
\[
 \Pp(C_e=1\mid C_{\setminus e},H)
 \ge\zeta:=\frac{\alpha(1-\pi)}{1-\alpha\pi}.
\]
Average this inequality conditional on $H$ and any previously fixed incident
marks, then multiply along the star to obtain
$\Pp(v\text{ good}\mid H)\ge\zeta^{\deg(v)}$.
The variable $X_v$ is independent of this history and the marks, giving
$\Pp(Y_v=1\mid H)\ge\pi\zeta^\Delta$.
Sequential conditional sampling proves the domination; on a countable graph
apply the same construction to its finite-dimensional laws.
\end{proof}

Apply this lemma with $\pi=\pi_K$ from \eqref{eq:pi}, and denote the resulting
parameter by $\varpi_K$. Conditional on the marks, perform the once-per-target
exploration but skip every bad target. Proposition~\ref{prop:exploration}, using Lemma~\ref{lem:goodtarget}, couples its
successes to iid labels $X_v$ of parameter $\pi_K$, independent of the marks.
The blue root cluster contains the cluster of $Y$ in \eqref{eq:defectfield}
with its root forced occupied. By Lemma~\ref{lem:defects}, it consequently
contains an iid site root cluster of parameter $\varpi_K$, with the same forcing.
For precision, couple $Y\ge Z$ with $Z$ iid, and apply the increasing map that
sets each field's root coordinate to one. This preserves containment and does
not condition on the root's original mark.

Proposition~\ref{prop:local-transfer} passes the good-target inequalities to
every selected periodic joint limit. Choose $K$ outside the atom set of $W$,
so finite mark cylinders are continuity events. Include nonnegative cylinder
tests in $C$ and the good-target indicator before taking the limit; then
disintegrate. There is one common full-measure mark set for all vertices, edges,
signs and generating exterior histories. The monotone-class extension and
blue/overlap compatibility hold there. Alternatively, restrict exploration to
fixed balls while retaining the full ambient torus law and full ambient stars,
then pass to the limit and exhaust the lattice. This also gives the rooted
cluster-expectation comparison.

Here is the probability-one conclusion, including both overlap signs. For almost
every mark configuration $C$, iid sites of parameter $\pi_K$ on the deterministic
good-vertex graph have an infinite component almost surely whenever
$\varpi_K>\pc(\Z^d)$. Indeed, their joint field $Y$ dominates supercritical iid
percolation, and Fubini gives the assertion after fixing $C$.
For each such $C$, some good vertex $v$ has positive probability
$\vartheta_C(v)>0$ of belonging to that independent infinite component.

Apply Proposition~\ref{prop:as-existence} under the physical law conditional
on this fixed $C$, taking $A$ to be its deterministic set of good vertices.
Lemma~\ref{lem:goodtarget} supplies the exterior edge bound and the positive
root-label bound. The proposition proves almost-sure existence for each overlap
sign. Here the marks are fixed to work on a deterministic good-vertex graph.

\subsection{A capped upper bound}

For $x\ge0$, define bounded functions and moments
\begin{equation}\label{eq:cap}
\begin{gathered}
 Q(x)=\min\{4,\cosh(4x)\cosh^2x\},\\
 A_\beta=\E[\tanh(\beta W)Q(\beta W)],\qquad
 D_\beta=\E[Q(\beta W)^2].
\end{gathered}
\end{equation}
No moment assumption is needed for their finiteness. Set
\begin{equation}\label{eq:capparameter}
 \widehat u_\Delta=A_\beta D_\beta^{(\Delta-1)/2}.
\end{equation}

Fix all magnitudes and average only signs, spins and coins. The leaf comparison
at a fresh target $v$, with candidate $e=uv$, gives
\[
 \Pp(B_e=1\mid H,W)
 \le t_e\prod_{f\ni v,\,f\ne e}
             \cosh(4\beta W_f)\cosh^2(\beta W_f),
 \qquad t_e=\tanh(\beta W_e).
\]
The candidate's unobserved activation coin also bounds this probability by
$1-e^{-4\beta W_e}\le4t_e$. Since all factors in the product are at least one,
taking the smaller bound yields
\begin{equation}\label{eq:cappedtarget}
 \Pp(B_e=1\mid H,W)
 \le t_e\prod_{f\ni v,\,f\ne e}Q(\beta W_f).
\end{equation}

Apply \eqref{eq:cappedtarget} successively along a prescribed self-avoiding path
of $n$ edges. A path edge contributes its own $t_e$ and at most one $Q$ factor;
an edge outside the path contributes at most two $Q$ factors, one at each endpoint.
There are at most $n(\Delta-1)$ incidences of the latter kind before excluding
path-edge incidences. Independence of the magnitudes and
$\E Q\le\sqrt{\E Q^2}$ therefore give
\begin{equation}\label{eq:cappedpath}
 \nu_{G,\beta}(\text{the path is blue})
 \le A_\beta^nD_\beta^{n(\Delta-1)/2}
 =\widehat u_\Delta^n.
\end{equation}
Overcounting is harmless because $D_\beta\ge1$. Proposition~\ref{prop:path-bound} gives
\eqref{eq:general-path-count} and \eqref{eq:general-susceptibility} with $u$ replaced by
$\widehat u_\Delta$, uniformly in volume and hence in every periodic local limit.

\subsection{Completion under the finite-mean assumption}

\begin{proof}[Proof of Theorem~\ref{thm:headline} under \eqref{eq:moment}]
Take $\beta=c/(\Delta m)$ with fixed $c>0$. Choose a continuity point
$K\in[1/\beta,2/\beta]$. Finite mean gives
\[
 \Delta\Pp(W>K)\to0,\qquad b_K/\beta\to m,
 \qquad \beta\E[W^2;W\le K]\to0.
\]
For the last assertion, the integrand
$W(\beta W)\mathbf1_{\{W\le K\}}$ is bounded by $2W$ and tends pointwise to
zero. Taylor bounds on $[0,2]$ then show
\[
 \Delta(1-\ell_K)\to0,\qquad
 \Delta(M_K-1)\to0,\qquad \Delta\pi_K\to c.
\]
With $\zeta=\alpha(1-\pi_K)/(1-\alpha\pi_K)$ as above,
$1-\zeta=(1-\alpha)/(1-\alpha\pi_K)$, so
$\zeta^\Delta\to1$ and $\Delta\varpi_K\to c$.
For $c>1$, Kesten's threshold \eqref{eq:kesten} makes
$\varpi_K>\pc(\Z^d)$ for large $d$. The preceding good-target comparison proves
positive root percolation and probability-one percolation in both overlap signs.
Global flipping of one replica gives equal, positive probabilities for the two
root-sign events.

For the upper bound, $Q(x)\to1$ as $x\to0$ and
$0\le Q(x)^2-1\le C\min\{x^2,1\}$. Dominated convergence gives
\[
 A_\beta/\beta\to m,
 \qquad D_\beta-1=o(\beta).
\]
Here the dominating functions are $4W$ for the first ratio and $CW$ for the
second. For the latter, use
\[
 \frac{\min\{(\beta W)^2,1\}}{\beta}\le W.
\]
Consequently $(\Delta-1)\widehat u_\Delta\to c$.
For $0<c<1$, choose $\rho=(1+c)/2$. The functions $A_\beta$ and $D_\beta$ are
nondecreasing, so the path bound holds with rate $\rho$ throughout
$0\le\beta\le c/(\Delta m)$ for sufficiently large $d$.
This completes the theorem and the onset asymptotic \eqref{eq:onset}.
The same lower parameter $\varpi_K$ and upper parameter $\widehat u_\Delta$
supply the asymptotics used in Corollary~\ref{cor:chi} below.
\end{proof}

\begin{corollary}[Subcritical susceptibility]\label{cor:chi}
Under the finite-mean hypothesis \eqref{eq:moment}, for every fixed $0\le c<1$,
\[
 \sup_{\nu\in\mathcal L_{d,c/(2dm)}}
 \left|\E_\nu|C_{\mathrm b}(o)|-\frac1{1-c}\right|\longrightarrow0.
\]
\end{corollary}

\begin{proof}
The upper path bound with $(2d-1)\widehat u_{2d}\to c$ gives the upper
limit by summing over path lengths. For the lower limit use the root-forced
iid site comparison with $p=\varpi_K$, for which $2dp\to c$. Fix $k\ge1$. There are $2^k\binom dk$ vertices
obtained by moving once in each of $k$ distinct signed coordinate directions.
Each admits $k!$ monotone paths of length $k$. In root-forced site percolation
of parameter $p$, each such path is occupied with probability $p^k$, and two
distinct paths use at least $k+1$ nonroot vertices. Bonferroni therefore gives
\[
 \Pp(o\leftrightarrow x)
 \ge k!p^k-\binom{k!}{2}p^{k+1}.
\]
If $\Delta p\to c$, the expected number of connected vertices among these
$2^k\binom dk$ choices has lower limit at least $c^k$.
Summing over $0\le k\le k_0$ and then letting $k_0\to\infty$ gives $1/(1-c)$.
All bounds are uniform over the selected local limits. At $c=0$ all blue edges
are absent and the susceptibility equals one.
\end{proof}

\section{The fresh-star route: star estimates and a single-floor comparison}\label{app:fresh-star}

This appendix gives a second route to an explicit dimension for iid fair
couplings $J_e=\pm1$. It keeps the one-attempt exploration of
Section~\ref{sec:general} and sharpens both of its inputs. The leaf floor of
Corollary~\ref{cor:cmr-local} restores the target's other interactions one at a
time and pays a factor $\sech^2\beta/\cosh4\beta$ for each, although all of
them share the same two central spins. Summing over these spins jointly, with
both quenched normalizations retained, gives an explicit floor $p_0(2d,\beta)$
for the conditional blue-opening probability, uniform over the entire exterior
of the target (Proposition~\ref{prop:unit-star}). The exploration never reveals
a failed edge at a target before testing it, so this single floor is all it
needs. Globally, Kesten's second-moment argument for oriented paths
\cite[Section 2]{CoxDurrett1983} replaces the asymptotic threshold
\eqref{eq:kesten} by the explicit bound
$\pc(\Z^d)\le F_d=1-1/G_d$, where $G_d$ is the expected number of meetings of
two independent oriented random walks (Lemma~\ref{lem:oriented-pz}).

Theorem~\ref{thm:single-floor} combines the two inputs: if $p_0(2d,\beta)>F_d$,
both overlap signs percolate in every selected periodic joint limit. Exact
evaluation of the floor gives dimensions $16$ to $20$ and $22$
(Corollary~\ref{cor:single-floor}); dimension $22$ was proved in \cite{PeiV1}
by a longer argument (Remark~\ref{rem:v1-certificate}).
A cruder floor, free of finite sums, gives dimensions $25$ and $26$
(Corollary~\ref{cor:fs-closed-form}). Here the target's spins are never revealed,
so the floor averages over the target's overlap. Section~\ref{sec:overlap-route}
instead reveals the whole overlap field, queries blue bonds only between
agreeing sites, where failed queries are harmless, and compares the overlap
field with a weak Ising field; with the sharpened local inputs and the macrostep
paths of Section~\ref{sec:macrostep}, that route reaches dimension $7$
(Theorem~\ref{thm:explicit}).

Section~\ref{sec:joint-star} states the star estimate, and
Section~\ref{sec:fs-single-floor} deduces the single-floor comparison from it,
using the oriented estimates of Section~\ref{sec:fs-oriented}.
Sections~\ref{app:star}--\ref{app:certificate} prove the star estimate and reduce
its row $h=0$ to a closed form and a few explicit sums.
Section~\ref{sec:fs-certificates} evaluates them at the certified points, and
Section~\ref{sec:fs-closed-form} gives the floor without finite sums.

\subsection{The CMR star after failed queries}\label{sec:joint-star}

For unit couplings the leaf floor of Corollary~\ref{cor:cmr-local} is
$\tanh\beta\,(\sech^2\beta/\cosh4\beta)^{\Delta-1}$: every incident interaction
other than the tested one is restored separately and costs its own factor. All
interactions at a target, however, share the same two central spins. Averaging
this star together preserves information that is lost when its edges are
restored separately. It also integrates earlier blue failures exactly. Write
$\eta=e^{-2\beta}$ and $\alpha=1-\eta^2$, the activation probability in
\eqref{eq:blue}; this is $p_A$ of Section~\ref{sec:overlap-route}. Consider a failed incident edge whose other endpoint has the
tested parent's overlap, and let $\epsilon$ be its sign after the gauge
transformation of Section~\ref{app:star}. In the sector where the target also
has the parent's overlap, the failed edge contributes
\begin{equation}\label{eq:failed-factor}
 e^{2\beta\sigma_v\epsilon}
   (1-\alpha\mathbf1_{\{\sigma_v\epsilon=1\}})=\eta,
 \qquad \sigma_v,\epsilon\in\{-1,1\},
\end{equation}
where $\sigma_v$ is the target's spin; in the opposite sector it contributes $1$.
This identity integrates a blue failure, allowing a red edge; it does not
remove the interaction from the partition function.

For a vertex $v$, let $\mathcal F_v^{\mathrm{ext},J}$ be the sigma-field
generated by $\mathcal F_v^{\rm ext}$ and the off-star couplings
$(J_e)_{e\notin\starv(v)}$.

\begin{proposition}[A star with incident failures]\label{prop:unit-star}
For iid fair unit couplings on a finite simple graph, let $v$ have degree
$\Delta$, and let $e=uv,f_1,\ldots,f_h$ be distinct incident edges, with
$0\le h<\Delta$. On retained configurations in which the other endpoints of
all $f_i$ have $u$'s overlap, put $F_h=\{B_{f_1}=\cdots=B_{f_h}=0\}$. Then
\begin{equation}\label{eq:failed-star-bound}
 \nu_{G,\beta}(B_e=1\mid\mathcal F_v^{\mathrm{ext},J},F_h)
 \ge p_h(\Delta,\beta),
\end{equation}
where the explicit finite expression $p_h$ is defined in
\eqref{eq:star-parameter}. The target spins, incident coupling signs,
incident red indicators and $B_e$ are not retained. The bound holds on
conditioning events of positive probability.
\end{proposition}

The floor is a minimum over \emph{overlap arrangements}. Let $k$ be the number
of neighbors of $v$ that have $u$'s overlap, counting $u$ but not the other
endpoints of the $h$ failed edges; $k$ is determined by the retained exterior
spins. Since $\mathcal F_v^{\rm ext}\subseteq\mathcal F_v^{\mathrm{ext},J}$, the
bound also holds with $\mathcal F_v^{\rm ext}$ in place of
$\mathcal F_v^{\mathrm{ext},J}$.
Only the row $h=0$, in which $F_0$ is the sure event, enters the single-floor
comparison below. The rows $h\ge1$ serve explorations that test a target more
than once.

The normalized-mass argument is short. Delete the star and let $R$ be the
\emph{original full-star} one-replica partition ratio. After integrating
central spins and the failed activations, write
\[
 A=\E_\epsilon[a(\epsilon)R(\epsilon)^{-2}],\qquad
 D=\E_\epsilon[w(\epsilon)R(\epsilon)^{-2}]
\]
for the blue-event and total masses. Every incident sign, including those on
failed edges, remains in $R$. Pairing opposite central-spin orientations in a
weighted Jensen estimate gives $A\ge N$. Choose $\lambda\ge0$ with
$w-\lambda a\ge0$ pointwise. A bound $D-\lambda A\le M$, with $M\ge0$, gives
\begin{equation}\label{eq:linear-star}
 \frac AD\ge\frac{N}{\lambda N+M}.
\end{equation}
Section~\ref{app:star} proves these bounds and evaluates the finite expression
uniformly over all neighbor-overlap arrangements. De Santis and Gandolfi
\cite[Theorem 3.2]{DeSantisGandolfi1999} similarly retain common closed mass
in a small FK star; the weights here are those of two quenched CMR replicas.

For the smallest overlap groups, a signed residual gives a stronger bound.
For any $\lambda,c>0$, write $H=w-\lambda a$ and
$H_+=\max(H,0)$, $H_-=\max(-H,0)$. Apply inverse Jensen to $H_+/R^2$
and the tangent inequality $R^{-2}\ge3/c^2-2R/c^3$ to $-H_-/R^2$.
Keeping the same cavity orientation in both terms gives an upper bound
$\mathcal M_{h,k}(\lambda,c)$ on $D-\lambda A$, defined in
\eqref{eq:signed-residual}. Consequently
\begin{equation}\label{eq:signed-ratio}
 \mathcal M_{h,k}(\lambda,c)\le0
 \quad\Longrightarrow\quad \frac AD\ge\frac1\lambda.
\end{equation}
This criterion does not require a separate lower bound on $A$.
Section~\ref{app:signed-residual} proves it in the same exterior conditioning
class. Each overlap arrangement may use either \eqref{eq:linear-star} or
\eqref{eq:signed-ratio}.

\subsection{A single-floor oriented comparison}\label{sec:fs-single-floor}

Proposition~\ref{prop:unit-star} holds on every atom of the exterior
sigma-field. It therefore needs no information about how an exploration reached
the target, and its row $h=0$ can replace the leaf floor in the argument of
Section~\ref{sec:general}. The remaining input is an explicit percolation
threshold. Let $Y$ and $Y'$ be independent oriented random walks on $\Z^d$
started at $o$, each step uniform on $\{e_1,\ldots,e_d\}$, and put
\begin{equation}\label{eq:fs-collision-green}
\begin{gathered}
 u_n(d)=\Pp(Y_n=Y'_n)
 =\sum_{k_1+\cdots+k_d=n}\biggl(\frac{n!}{k_1!\cdots k_d!\,d^{n}}\biggr)^2,\\
 G_d=\sum_{n\ge0}u_n(d),\qquad F_d=1-\frac1{G_d},
\end{gathered}
\end{equation}
with $F_d=1$ if $G_d=\infty$. Thus $G_d$ is the expected number of meetings of
the two walks, and $F_d$ is the probability that they meet again after time $0$
(Lemma~\ref{lem:oriented-pz}). The same quantities enter the second moment of
Section~\ref{sec:overlap-route}.

Call $\underline p\in(0,1]$ a \emph{uniform fresh-star floor} at $(d,\beta)$ if,
for every $L\ge3$, every oriented edge $uv$ of $\mathbb T_L^d$ and every
$H\in\mathcal F_v^{\rm ext}$ with $\nu_{\mathbb T_L^d,\beta}(H)>0$,
\begin{equation}\label{eq:fs-uniform-floor}
 \nu_{\mathbb T_L^d,\beta}(B_{uv}=1\mid H)\ge\underline p.
\end{equation}
For $L\ge3$ the torus is a simple $2d$-regular graph, so
Proposition~\ref{prop:unit-star} with $h=0$ shows that $p_0(2d,\beta)$ is a
uniform fresh-star floor. So is any $\underline p$ that \eqref{eq:linear-star}
or \eqref{eq:signed-ratio} certifies separately in every overlap arrangement.

\begin{theorem}[Single-floor oriented comparison]\label{thm:single-floor}
Fix $d\ge1$, $\beta>0$ and iid fair couplings $J_e=\pm1$, and let
$\underline p$ be a uniform fresh-star floor at $(d,\beta)$. If
$\underline p>F_d$, then every $\nu\in\mathcal L_{d,\beta}$ satisfies
\begin{equation}\label{eq:fs-root-bound}
 \nu(o\blue\infty,\ q_o=s)
 \ge\frac12\cdot\frac{\underline p-F_d}{\underline p\,(1-F_d)}>0
 \qquad(s=\pm1),
\end{equation}
and, $\nu$-almost surely, each overlap sign contains an infinite blue component.
In particular, this holds when $p_0(2d,\beta)>F_d$.
\end{theorem}

The theorem is an explicit finite-dimensional version of
Remark~\ref{rem:short-route}. Kesten's asymptotic site threshold is replaced by
the bound $\pc(\Z^d)\le F_d$ of Lemma~\ref{lem:oriented-pz}, and the
leaf-restoration floor by the star floor. The oriented comparison loses a
factor of about two: $2dF_d=2+O(1/d)$, whereas $2d\pc(\Z^d)\to1$.

\begin{proof}
Let $\nu\in\mathcal L_{d,\beta}$ and $\underline p>F_d$.

\emph{Transfer.} Written without division, \eqref{eq:fs-uniform-floor} says
that $\E_{\nu_L}[hB_{uv}]\ge\underline p\,\E_{\nu_L}[h]$ for every bounded
nonnegative $\mathcal F_v^{\rm ext}$-measurable cylinder test $h$, uniformly
in $L$. For all sufficiently large $L$, the window of $h$ together with
$\starv(v)$ is identified with a box of $\Z^d$, so the inequality concerns the
same coordinates in every such torus and in the limit. All coordinates take
finitely many values, so Proposition~\ref{prop:local-transfer} gives
\begin{equation}\label{eq:fs-floor-limit}
 \nu(B_{uv}=1\mid\mathcal F_v^{\rm ext})\ge\underline p
 \quad\nu\text{-a.s., for every oriented edge $uv$ of $\Z^d$.}
\end{equation}
Neither consistency of the finite-volume laws nor a boundary condition is used.

\emph{One attempt per target.} Proposition~\ref{prop:exploration} needs its bound
only given off-star histories at the target, and these belong to
$\mathcal F_v^{\rm ext}$. By \eqref{eq:fs-floor-limit} and the tower property,
$\nu(B_{uv}=1\mid H)\ge\underline p$ for every such $H$ with $\nu(H)>0$. The
exploration tests $uv$ only when $v$ has never been attempted, and then retires
$v$, so no failed edge at $v$ is revealed before its test; this is why the row
$h=0$ suffices. Proposition~\ref{prop:exploration}, applied under $\nu$ with
$p=\underline p$, gives iid Bernoulli$(\underline p)$ sites $X$ with $X_o=1$
and $C_X(o)\subseteq C_{\rm b}(o)$.

\emph{Oriented paths.} An infinite occupied oriented path from $o$ lies in
$C_X(o)$, so Lemma~\ref{lem:oriented-pz} gives
$\nu(o\blue\infty)\ge(\underline p-F_d)/(\underline p(1-F_d))$.

\emph{Both signs.} In zero field, reversing the second replica preserves every
finite-volume law and every blue indicator, and it exchanges the events
$\{q_o=1\}$ and $\{q_o=-1\}$. It is continuous for local convergence, so $\nu$
is invariant as well, as in the proof of Corollary~\ref{cor:direct-cmr}. This
proves \eqref{eq:fs-root-bound}.

\emph{Almost-sure existence.} Apply Proposition~\ref{prop:as-existence} with
$G=\Z^d$, $A$ its whole vertex set, root $r=o$, labels $\xi=q$ and edge
indicators $\omega=B$.
Blue edges join equal overlaps on every torus, and this closed local constraint
passes to $\nu$ (Appendix~\ref{app:probability}). The sigma-field
$\mathcal E_v$ of that proposition is contained in $\mathcal F_v^{\rm ext}$, so
\eqref{eq:fs-floor-limit} supplies its edge bound, and
Lemma~\ref{lem:oriented-pz} shows that the site model percolates. For the label
bound, Corollary~\ref{cor:cmr-local} with $W\equiv1$ gives
$\nu_{G,\beta}(q_v=s\mid H)\ge\frac12(\sech^2\beta/\cosh4\beta)^{2d}$ for every
positive-probability $H\in\mathcal F_v^{\rm ext}$ on every torus, and
Proposition~\ref{prop:local-transfer} transfers this bound to $\nu$. Hence both
overlap signs contain infinite blue components $\nu$-almost surely.
\end{proof}

Positive root probability comes from a finite-volume inequality, its limit,
and an exploration in the limit law. Almost-sure existence is a separate
exterior-conditioning argument and uses no ergodicity. As in
Theorem~\ref{thm:headline}, the conclusions concern every selected periodic
joint limit and its own conditional kernel; nothing is asserted about other
selections of quenched Gibbs states.

\begin{corollary}[Dimensions 16 to 20 and 22]\label{cor:single-floor}
For iid fair couplings $J_e=\pm1$, let
\[
\begin{gathered}
 (d,t)\in\bigl\{(16,\tfrac{13}{125}),\ (16,\tfrac{51}{500}),\
 (17,\tfrac1{10}),\ (18,\tfrac{12}{125}),\\
 (19,\tfrac7{75}),\ (20,\tfrac1{11}),\ (22,\tfrac{11}{125})\bigr\},
\end{gathered}
\]
and $\beta=\operatorname{atanh}t$. Then $p_0(2d,\beta)>F_d$. Consequently every
$\nu\in\mathcal L_{d,\beta}$ satisfies $\nu(o\blue\infty,q_o=s)\ge\theta_{\rm fs}$
for $s=\pm1$, with the constant $\theta_{\rm fs}>0$ of Table~\ref{tab:fs-floor}, and
both overlap signs contain infinite blue components $\nu$-almost surely.
\end{corollary}

Section~\ref{sec:fs-certificates} proves the corollary. There the minimum over
arrangements is shown to occur at $k=1$, where the floor has the closed form
\eqref{eq:fs-k1-floor}, and exact arithmetic compares that closed form with an
explicit upper bound on $F_d$. At $d=22$ the floor exceeds $F_{22}$ by a factor of at least $1.2077$.
Since the joint law depends on the couplings and on $\beta$ only through
$\beta J_e$, the same conclusions hold for couplings $\pm J_0$ at
$\beta J_0=\operatorname{atanh}t$. Each conclusion concerns the stated dimension
and temperature; no monotonicity in $d$ or in $\beta$ is asserted.

\subsection{Oriented paths and the collision Green function}\label{sec:fs-oriented}

Call a path in $\Z^d$ \emph{oriented} if its steps lie in
$\{e_1,\ldots,e_d\}$.

\begin{lemma}[Oriented site second moment]\label{lem:oriented-pz}
Let $0<p\le1$, let $X_o=1$, and let $(X_x)_{x\ne o}$ be iid Bernoulli$(p)$ on
$\Z^d$. The difference walk $Y-Y'$ of \eqref{eq:fs-collision-green} returns to
the origin with probability $F_d$. If $p>F_d$, then
\[
 \Pp\bigl(\text{some infinite oriented path } o=\gamma_0,\gamma_1,\ldots
   \text{ has } X_{\gamma_m}=1\text{ for all }m\bigr)
 \ge\frac{p-F_d}{p(1-F_d)}.
\]
In particular $\pc(\Z^d)\le F_d$, and $\thsite(p)\ge(p-F_d)/(1-F_d)$ when the
root is not forced.
\end{lemma}

\begin{proof}
The coordinate sum of an oriented path equals its length. Hence oriented paths
are self-avoiding, and two oriented paths from $o$ can share a vertex only at
equal times. There are $d^n$ oriented paths of length $n$ from $o$.

Let $Z=Y-Y'$, a random walk with iid increments, and put
$K=\#\{m\ge1:Z_m=0\}$ and $F'=\Pp(Z_m=0\text{ for some }m\ge1)$. By the strong
Markov property at successive returns, $\Pp(K\ge j)=F'^{\,j}$, so
$\E K=F'/(1-F')$, which is infinite when $F'=1$. By Tonelli's theorem,
$\E K=\sum_{m\ge1}u_m(d)=G_d-1$. Therefore $F'=1-1/G_d=F_d$, also when
$G_d=\infty$. For $p>F_d$,
\[
 \E p^{-K}=\sum_{j\ge0}(1-F_d)F_d^{\,j}p^{-j}=\frac{1-F_d}{1-F_d/p}<\infty.
\]

For $n\ge1$ put
\[
 W_n=\sum_{|\gamma|=n}(dp)^{-n}\,
 \mathbf1\{X_{\gamma_1}=\cdots=X_{\gamma_n}=1\},
\]
the sum running over oriented paths of length $n$ from $o$; thus $\E W_n=1$.
For two such paths $\gamma,\gamma'$, the union of their non-root vertex sets has
$2n-K_n(\gamma,\gamma')$ elements, where
$K_n(\gamma,\gamma')=\#\{1\le m\le n:\gamma_m=\gamma'_m\}$. Hence
\[
 \E W_n^2=\sum_{\gamma,\gamma'}d^{-2n}p^{-K_n(\gamma,\gamma')}
 =\E p^{-K_n}\le\E p^{-K},
\]
where in the middle expression $\gamma,\gamma'$ are the paths of the two
independent walks, $K_n\le K$, and $p\le1$. The Cauchy--Schwarz inequality
$1=(\E W_n\mathbf1_{\{W_n>0\}})^2\le\E W_n^2\,\Pp(W_n>0)$ gives
\[
 \Pp(W_n>0)\ge\frac{1-F_d/p}{1-F_d}=\frac{p-F_d}{p(1-F_d)}
 \qquad(n\ge1).
\]
The event $\{W_n>0\}$ says that an occupied oriented path of length $n$ starts
at $o$. These events decrease in $n$, because prefixes of occupied paths are
occupied, so their intersection has at least the same probability. On the
intersection, the occupied oriented paths from $o$ form a tree in which every
vertex has at most $d$ children and which contains paths of every length. By
K\"onig's lemma it contains an infinite path. For an unforced root, multiply
by $\Pp(X_o=1)=p$.
\end{proof}

This is Kesten's second-moment and renewal argument as presented by Cox and
Durrett \cite[Section 2, pp.~154--155, Eqs.~(2.3)--(2.8)]{CoxDurrett1983}, who
record that the bound was first proved by Kesten. Their version concerns
oriented bond percolation and counts the edges shared by the two walks. The
site version is implicit in the proof of \cite[Proposition~5.22]{LyonsPeres2016},
which bounds the shared edges by meetings of the two walks, and a version for a
dependent site field is used in \cite[proof of Theorem~1.2]{BDNS2020}. With
lazy-walk paths, concurrent work obtains sharper site bounds in intermediate
dimensions, for instance $\pc(\Z^{15})\le0.0691<F_{15}$
\cite[Section~4 and Table~2]{JiangLang2026}.

To evaluate $F_d$, write
$m_n(\mathbf k)=\frac{n!}{k_1!\cdots k_d!}d^{-n}$ for compositions
$\mathbf k=(k_1,\ldots,k_d)$ of $n$, and $\bar m_n=\max_{\mathbf k}m_n(\mathbf k)$
for the largest multinomial atom.

\begin{lemma}[Collision Green function bounds]\label{lem:collision-green}
The following hold.
\begin{enumerate}
\item[(a)] $u_n=\sum_{\mathbf k}m_n(\mathbf k)^2$, with $u_0=1$, $u_1=1/d$,
$u_2=(2d-1)/d^3$ and $u_3=(6d^2-9d+4)/d^5$.
\item[(b)] If $n=da+j$ with $0\le j<d$, then
$\bar m_n=n!/\bigl(d^n(a+1)!^{\,j}a!^{\,d-j}\bigr)$. In particular
$\bar m_n=n!/d^n$ for $n\le d$.
\item[(c)] $u_n\le\bar m_n$ and $\bar m_{n+1}\le\bar m_n$.
\item[(d)] $\bar m_{d(a+1)}\le\bar m_{da}\exp\bigl(-\frac{d-1}{2(a+1)}\bigr)$.
Hence $\bar m_{da}\le\bar m_{dA}\bigl(\frac{A+1}{a+1}\bigr)^{(d-1)/2}$ for
$a\ge A\ge0$.
\item[(e)] If $d\ge4$ and $A\ge1$, then
$\sum_{n\ge dA}u_n\le d\,\bar m_{dA}\bigl(1+\frac{2(A+1)}{d-3}\bigr)$.
\item[(f)] If $d\ge6$, then $G_d\le\widehat G_d$, where
\[
 \widehat G_d=1+\frac1d+\frac{2d-1}{d^3}+\frac{6d^2-9d+4}{d^5}
 +\frac{24}{d^4}+\frac{120}{d^5}+\frac{720(d-6)}{d^6}
 +d\,\frac{d!}{d^d}\Bigl(1+\frac4{d-3}\Bigr).
\]
\item[(g)] $G_{22}\le1+\frac1{22}+\frac2{22^2}+\frac{19\cdot6}{22^3}
+22\,\bar m_{22}\bigl(1+\frac4{19}\bigr)<1.0602932$.
\end{enumerate}
\end{lemma}

\begin{proof}
(a) Identify the position $Y_n$ with its vector of step counts. Then
$u_n=\sum_{\mathbf k}\Pp(Y_n=\mathbf k)^2$ with $\Pp(Y_n=\mathbf k)=m_n(\mathbf k)$.
For $u_2$, the count vectors are $2e_i$, $d$ of them with probability $d^{-2}$
each, and $e_i+e_j$ with $i<j$, $\binom d2$ of them with probability $2d^{-2}$
each. For $u_3$, the types $3e_i$, $2e_i+e_j$ and $e_i+e_j+e_l$ have multinomial
coefficients $1,3,6$, so
$u_3=[d+9d(d-1)+6d(d-1)(d-2)]d^{-6}=(6d^2-9d+4)/d^5$.

(b) If $k_i\ge k_j+2$, moving one unit from $i$ to $j$ multiplies
$n!/\prod_lk_l!$ by $k_i/(k_j+1)>1$. The maximum is therefore attained at a
composition whose parts differ by at most one.

(c) First, $u_n\le\bar m_n\sum_{\mathbf k}m_n(\mathbf k)=\bar m_n$. Second,
$m_{n+1}(\mathbf k)=d^{-1}\sum_{i:k_i\ge1}m_n(\mathbf k-e_i)\le\bar m_n$.

(d) By (b),
$\bar m_{d(a+1)}/\bar m_{da}=\prod_{i=1}^d\frac{da+i}{d(a+1)}
=\prod_{l=0}^{d-1}\bigl(1-\frac l{d(a+1)}\bigr)$. Since
$\log(1-x)\le-x$, its logarithm is at most $-\frac{d-1}{2(a+1)}$. Iterate,
and use $\sum_{b=A+1}^a\frac1b\ge\log\frac{a+1}{A+1}$.

(e) Group the indices $n\in[da,da+d)$; there $u_n\le\bar m_n\le\bar m_{da}$ by
(c). With $\sigma=(d-1)/2>1$, part (d) gives
\begin{align*}
 \sum_{n\ge dA}u_n
 &\le d\,\bar m_{dA}\sum_{a\ge A}\Bigl(\frac{A+1}{a+1}\Bigr)^\sigma\\
 &\le d\,\bar m_{dA}(A+1)^\sigma\Bigl[(A+1)^{-\sigma}
   +\int_{A+1}^\infty x^{-\sigma}\,dx\Bigr]
 =d\,\bar m_{dA}\Bigl(1+\frac{A+1}{\sigma-1}\Bigr).
\end{align*}

(f) Use (a) for $n\le3$. For $4\le n\le d-1$, (b) and (c) give
$u_n\le\bar m_n=n!/d^n$, and for $6\le n\le d-1$ these $d-6$ terms are each at
most $\bar m_6=720/d^6$ by (c). Apply (e) with $A=1$, where
$\bar m_d=d!/d^d$.

(g) Use $u_2\le\bar m_2=2/d^2$, then $u_n\le\bar m_3=6/d^3$ for
$3\le n\le21$ by (c), and (e) with $A=1$. Here
$22\,\bar m_{22}(1+4/19)<8.8\cdot10^{-8}$.
\end{proof}

Cox and Durrett \cite[p.~155]{CoxDurrett1983} use the comparison
$\Pp(S_k=S'_k)\le\max_x\Pp(S_k=x)$ for their walks $S,S'$ and the balanced-atom
bound $\Pp(S_k=x)\le d^{-jd}(jd)!/(j!)^d$ for $k\ge jd$, followed by Stirling's
formula. The elementary ratio bound (d) replaces Stirling's formula here.

\subsection{Proof of the star estimate}\label{app:star}

We prove Proposition~\ref{prop:unit-star}. Here $\Delta$ is the actual target
degree, equal to $2d$ on the periodic lattices. For $\beta>0$, put
\[
 t=\tanh\beta,\quad c_0=\cosh\beta,\quad C=\cosh(2\beta),\quad
 \vartheta=\tanh(2\beta),\quad \eta=e^{-2\beta},\quad\alpha=1-\eta^2.
\]
All are rational functions of $t$; for instance $\eta=(1-t)/(1+t)$,
$C=(1+t^2)/(1-t^2)$, $\vartheta=2t/(1+t^2)$ and $c_0^{-2}=1-t^2$. In the
notation of Section~\ref{sec:overlap-route}, $C=C_\beta$ and $\alpha=p_A$. Label the
neighbors of $v$ as $y_0=u,y_1,\ldots,y_{\Delta-1}$, so that the proposed edge
has index $0$. There are $h$ specified failed edges whose other endpoints have
the parent's overlap. Let $U$ be the remaining matching-overlap indices,
including $0$, and let $V$ be the opposite-overlap indices. Write $k=|U|\ge1$
and $s=|V|=\Delta-h-k$. Expectations over $\epsilon$ below use independent fair
signs as a reference measure; the quenched normalization is retained in every
conditional probability.

\paragraph{Normalized masses.}
Delete the entire target star, keeping its center isolated, and write $G^\circ$
for the resulting graph. For fixed exterior disorder and star signs
$\epsilon_i=J_{vy_i}$, the original one-replica partition ratio is
\begin{equation}\label{eq:full-star-ratio}
 R(\epsilon)=\frac{Z_G}{Z_{G^\circ}}
 =\E_{\rho\sim\mu_{G^\circ,J,\beta}}
       \cosh\!\left(\beta\sum_i\epsilon_i\rho_{y_i}\right)\ge1.
\end{equation}
The auxiliary $\rho$ is sampled from the cavity Gibbs law defining the
partition function, not conditioned on the observed history. Relative to
the augmented cavity law, with fair deleted signs and uniform central spins,
the full two-replica density is
\[
 R(\epsilon)^{-2}
 \exp\!\left(\beta\sum_i\epsilon_i
       (\sigma_v\sigma_{y_i}+\tau_v\tau_{y_i})\right).
\]
All other factors of the joint law, including every off-star blue likelihood,
are $\mathcal F_v^{\mathrm{ext},J}$-measurable. Hence, on each
positive-probability atom of $\mathcal F_v^{\mathrm{ext},J}$, the conditional law
of the star signs and central spins has a density proportional to this
expression with respect to fair signs and uniform central spins. Given these
variables, the star blue indicators are independent with the rule
\eqref{eq:blue}.
Fix the observed exterior replicas and gauge the star signs by
$\sigma_{y_i}$, transforming the cavity coordinates in $R$ accordingly; this
preserves the fair reference measure. The exponent becomes
$\beta\sum_i\epsilon_i(\sigma_v+\tau_vq_{y_i})$. In the sector $q_v=q_u$ it
equals $2\beta\sigma_v\sum_{i\notin V}\epsilon_i$, and in the sector $q_v=-q_u$
it equals $2\beta\sigma_v\sum_{i\in V}\epsilon_i$. Both replicas satisfy $vu$,
or a failed edge $vy_i$, exactly when $q_v=q_u$ and $\sigma_v$ equals the gauged
sign of that edge. Integrating the failed activations by
\eqref{eq:failed-factor} and averaging the two central spins gives
\begin{align}
 w_{h,k}(\epsilon)
 &=\tfrac12\left[\eta^h
       \cosh\!\left(2\beta\sum_{i\in U}\epsilon_i\right)
       +\cosh\!\left(2\beta\sum_{i\in V}\epsilon_i\right)\right],
 \label{eq:star-weights}\\
 a_{h,k}(\epsilon)
 &=\eta^h\frac{\alpha}{4}
       \exp\!\left(2\beta\epsilon_0\sum_{i\in U}\epsilon_i\right).
 \notag
\end{align}
The conditional opening probability is $A_{h,k}/D_{h,k}$, where
\[
 A_{h,k}=\E_\epsilon[a_{h,k}/R^2],\qquad
 D_{h,k}=\E_\epsilon[w_{h,k}/R^2].
\]
All $\Delta$ signs, including the $h$ failed-edge signs, still occur in
\eqref{eq:full-star-ratio}. Their disappearance from the spin weights in
\eqref{eq:star-weights} does not remove them from $R$. Off-star blue
observations remain retained exterior observables.

\paragraph{A paired numerator.}
Define
\begin{equation}\label{eq:star-numerator}
 H_\ell=\frac{(1+t)(1+t\vartheta)^\ell+(1-t)(1-t\vartheta)^\ell}{2},\qquad
 N_{h,k}=\frac{\eta^h t(1-t^2)^{\Delta-1}C^{k-1}}{H_{k-1}^2}.
\end{equation}
We claim $A_{h,k}\ge N_{h,k}$. Both $a_{h,k}$ and $R$ are invariant under
$\epsilon\mapsto-\epsilon$, so the tested sign may be fixed to $\epsilon_0=1$;
this pairs the opposite orientations. With $\ell=k-1$, index
$U\setminus\{0\}$ by $1,\ldots,\ell$ and put
$W=\exp(2\beta\sum_{i=1}^{\ell}\epsilon_i)$ and $Q=\E W=C^\ell$. Then
\[
 A_{h,k}=\eta^h\frac{\alpha e^{2\beta}}4\E(W/R^2),
\]
with the tested sign fixed as above. Weighted Jensen gives
\begin{equation}\label{eq:paired-jensen}
 \E(W/R^2)\ge\frac{Q^3}{(\E WR)^2}.
\end{equation}
The other $\Delta-k$ signs, including failed ones, are still fair in this
expectation. Expanding the cavity cosh and averaging all remaining signs gives
\[
 \frac{\E WR}{Q}
 =c_0^{\Delta-1}\E_\rho
  \frac{e^\beta\prod_{i=1}^{\ell}(1+t\vartheta y_i)
              +e^{-\beta}\prod_{i=1}^{\ell}(1-t\vartheta y_i)}2
 \le c_0^\Delta H_\ell,
\]
where $y_i$ are the relative cavity coordinates measured against the tested
neighbor; they can be dependent. The inequality is
pointwise: even and odd monomials have nonnegative coefficients, respectively
$\cosh\beta\,(t\vartheta)^{|S|}$ and $\sinh\beta\,(t\vartheta)^{|S|}$, so all
$y_i=1$ maximizes the expression. Finally $\alpha e^{2\beta}/4=t c_0^2$ and
\eqref{eq:paired-jensen} give $A_{h,k}\ge N_{h,k}$.

\paragraph{A nonnegative residual.}
Let $c(s)=1$ when $s$ is even and $c(s)=C$ otherwise, and put
\begin{equation}\label{eq:star-lambda}
 \lambda_{h,k}=\frac{1+\eta^{2k}+2c(s)\eta^{k-h}}{\alpha}.
\end{equation}
This is exactly $\min_\epsilon w_{h,k}/a_{h,k}$. First minimize the
$V$-cosh at $c(s)$, and then maximize
$x=\epsilon_0\sum_{i\in U}\epsilon_i\le k$ in the decreasing expression
\[
 \frac{1+e^{-4\beta x}+2c(s)\eta^{-h}e^{-2\beta x}}{\alpha}.
\]
Thus $w_{h,k}-\lambda_{h,k}a_{h,k}\ge0$ pointwise. Inverse Jensen gives
\[
 R(\epsilon)^{-2}\le
 \E_\rho\sech^2\!\left(\beta\sum_i\epsilon_i\rho_{y_i}\right).
\]
Apply this bound to the entire nonnegative residual before separating its
signed terms. The result is $D_{h,k}-\lambda_{h,k}A_{h,k}\le M_{h,k}$, where
\begin{equation}\label{eq:star-residual}
 M_{h,k}=\max_{x\in\{-1,1\}^{\Delta}}
 \E_\epsilon\!\left[(w_{h,k}-\lambda_{h,k}a_{h,k})
              \sech^2\!\left(\beta\sum_i\epsilon_i x_i\right)\right].
\end{equation}
Since $M_{h,k}\ge0$, the comparison \eqref{eq:linear-star} proves
\begin{equation}\label{eq:star-parameter}
 p_h(\Delta,\beta)
 :=\min_{1\le k\le\Delta-h}
       \frac{N_{h,k}}{\lambda_{h,k}N_{h,k}+M_{h,k}}
 \le\nu_{G,\beta}(B_e=1\mid\mathcal F_v^{\mathrm{ext},J},F_h).
\end{equation}
This proves Proposition~\ref{prop:unit-star}. Every estimate is pointwise in
retained exterior data. Multiplication by nonnegative retained tests and
integration therefore preserve the inequalities, including additional
allowed history events. Incident signs are averaged throughout and cannot
also be conditioned on. At $\beta=0$ set $p_h=0$; the positive-temperature
statement is the one used in this appendix.

\subsection{A signed residual}\label{app:signed-residual}

For any $\lambda,c>0$, put $H=w_{h,k}-\lambda a_{h,k}$,
$H_+=\max(H,0)$ and $H_-=\max(-H,0)$. Define
\begin{equation}\label{eq:signed-residual}
\begin{split}
 \mathcal M_{h,k}(\lambda,c)
 =\max_{x\in\{-1,1\}^{\Delta}}\E_\epsilon\bigg[
 &H_+\sech^2(\beta\epsilon\cdot x)\\
 &+H_-\left\{\frac{2\cosh(\beta\epsilon\cdot x)}{c^3}
                     -\frac3{c^2}\right\}\bigg].
\end{split}
\end{equation}
We claim $D_{h,k}-\lambda A_{h,k}\le\mathcal M_{h,k}(\lambda,c)$.
For $y>0$, the tangent lower bound is the identity
\[
 y^{-2}-3c^{-2}+2yc^{-3}
   =\frac{(y-c)^2(2y+c)}{c^3y^2}\ge0.
\]
Write $K_x(\epsilon)=\cosh(\beta\epsilon\cdot x)$, so
$R=\E_\rho K_\rho$. Inverse Jensen and the tangent bound give
\[
 \E_\epsilon\frac{H}{R^2}
 \le\E_\rho\E_\epsilon
       \left[H_+K_\rho^{-2}+2c^{-3}H_-K_\rho\right]
        -3c^{-2}\E_\epsilon H_-.
\]
Maximizing over the common cavity orientation proves the claim and hence
\eqref{eq:signed-ratio}. Neither nonnegativity of $H$ nor a lower bound on
$A_{h,k}$ is required for that implication. All failed signs remain in $R$;
the same retained-exterior and incident-failure hypotheses apply.
The truncations $H_+$ and $H_-$ couple the active and inactive sign groups,
so their orientation maxima cannot be separated.

\subsection{Finite sums and the floor in closed form}\label{app:certificate}

Let $S_{\ell,j}^{(n)}$ be a sum of $n$ independent signs: $j$ have mean
$\vartheta$, $\ell-j$ have mean $-\vartheta$, and $n-\ell$ are fair. Write
$\bar S_n=S^{(n)}_{0,0}$ for a sum of $n$ fair signs. Put
$\pi_+=(1+\vartheta)/2$, which equals $a$ of Section~\ref{sec:overlap-route},
and $\psi(x)=\sech^2(\beta x)$. Define
\[
 T_\ell=C^\ell\max_{0\le j\le\ell}\E\psi(S_{\ell,j}^{(\Delta)}),
 \qquad
 G_{\ell,j}^{\pm}=\E\psi(\pm1+S_{\ell,j}^{(\Delta-1)}).
\]
Then the exact residual maximum reduces to
\begin{equation}\label{eq:residual-sums}
 M_{h,k}=\frac{T_s}{2}+\frac{\eta^h C^k}{2}
 \max_{0\le j<k}
 \left[(\pi_+-\lambda_{h,k}\vartheta)G_{k-1,j}^+
                     +(1-\pi_+)G_{k-1,j}^-\right].
\end{equation}
To verify this reduction, fix $x_0=1$ by global sign symmetry and gauge
$\epsilon_i$ by $x_i$. In the $V$-cosh term every sign outside $V$ is fair,
so those orientations disappear. In the $U$-cosh and blue terms every sign
outside $U$ is fair. Exponential tilting gives the stated means $\pm\vartheta$
and the factors $C^s$ and $C^k$; conditioning on the tested sign gives
$\pi_+G^++(1-\pi_+)G^-$, with blue subtraction $\lambda_{h,k}\vartheta G^+$
because $\pi_+\alpha=\vartheta$. Failed signs remain among the fair signs in
both kernels. The $U$ and $V$ maxima involve disjoint orientation coordinates,
so they add.

The probabilities of the positive-sign count in $S_{\ell,j}^{(n)}$ are the
coefficients of
\[
 \{(1-\pi_+)+\pi_+z\}^{j}\{\pi_++(1-\pi_+)z\}^{\ell-j}
                    \{(1+z)/2\}^{n-\ell}.
\]
Multiplication by the corresponding kernels evaluates every term
in \eqref{eq:residual-sums}. When $t$ is rational, so are $\eta$, $C$,
$\vartheta$, $\pi_+$, $c_0^2$ and $\psi(x)=4\eta^{|x|}/(1+\eta^{|x|})^2$ for
integers $x$; every such evaluation is then an exact rational.

For the signed cases, fix $\epsilon_0=x_0=1$ by the separate global-sign
symmetries in $\epsilon$ and $x$. Enumerate the other active signs
$e_1,\ldots,e_{k-1}$ and their cavity orientations $\xi_1,\ldots,\xi_{k-1}$.
Among the $s=\Delta-h-k$ inactive coordinates, let $j$ have cavity orientation
$+1$, with $n_+$ positive signs in that group and $n_-$ in the other group.
Absorb the failed coordinates' cavity orientations into their fair signs,
and let $n_f$ of those signs be positive. The inactive sum and full cavity
field are
\begin{align*}
 S&=2(n_++n_-)-s,\\
 L&=1+\sum_{i=1}^{k-1}\xi_i e_i
       +2n_+-2n_-+s-2j+2n_f-h.
\end{align*}
Use $1+\sum_i e_i$ and $S$ in \eqref{eq:star-weights}, and $L$ in both
kernels of \eqref{eq:signed-residual}. The multiplicity is
\[
 2^{-(\Delta-1)}\binom{j}{n_+}\binom{s-j}{n_-}\binom{h}{n_f}.
\]
Summing over all sign patterns and permitted counts, then maximizing over
$\xi$ and $0\le j\le s$, exhausts the signed residual; there are
$2^{k-1}(s+1)$ cavity orientations. In particular, the failed signs remain in
the full field $L$. The field $L$ has the parity of $h+k+s=\Delta$, so it is even
when $\Delta$ is even; then $\cosh(\beta L)=(\eta^{L/2}+\eta^{-L/2})/2$ and its
inverse square are rational as well.

\paragraph{The row \texorpdfstring{$h=0$}{h=0} as a closed form.}
From now on $h=0$ and $\Delta=2d$. Write $N_k$, $\lambda_k$ and $M_k$ for
$N_{0,k}$, $\lambda_{0,k}$ and $M_{0,k}$, put $s=\Delta-k$, and let
\[
 p(k)=\frac{N_k}{\lambda_kN_k+M_k},\qquad\text{so that}\qquad
 p_0(\Delta,\beta)=\min_{1\le k\le\Delta}p(k).
\]
Evaluating every arrangement requires of order $\Delta^2$ binomial sums. The
following lemmas give a closed form at $k=1$ and reduce the verification that
$k=1$ is the worst arrangement to a few sums. Let $\Pi_m$ be a sum of $m$
independent \emph{balanced pairs}, each the sum of one sign of mean $\vartheta$
and one of mean $-\vartheta$; a balanced pair takes the values $\pm2$ with
probability $\pi_+(1-\pi_+)=1/(4C^2)$ each and $0$ otherwise.

\begin{lemma}[Symmetric unimodal convolution]\label{lem:fs-unimodal}
Let $a,b\in\{0,1\}$. Suppose $f:a+2\Z\to[0,\infty)$ is finitely supported,
$g:b+2\Z\to[0,\infty)$ is bounded, and both are even and nonincreasing in
$|x|$. Then $(f*g)(z)=\sum_xf(x)g(z-x)$, on $a+b+2\Z$, is even and
nonincreasing in $|z|$.
\end{lemma}

\begin{proof}
Evenness is clear. By the layer-cake decomposition, $f$ is a finite
nonnegative combination of indicators of windows
$I=\{x\in a+2\Z:|x|\le\rho\}$. For such a window, $(\mathbf1_I*g)(z)$ is the sum
of $g$ over $\{y\in b+2\Z:z-\rho\le y\le z+\rho\}$. Replacing $z\ge0$ by $z+2$
removes the point $z-\rho$ and adds $z+\rho+2$. Since
$|z+\rho+2|\ge|z-\rho|$, the sum does not increase.
\end{proof}

\begin{lemma}[The balanced tilt maximizes]\label{lem:fs-balanced-tilt}
Let $n\ge\ell\ge1$ and $\vartheta\in[0,1)$, and let $f:\Z\to[0,\infty)$ be even
and nonincreasing in $|x|$. Put $g(j)=\E f(S^{(n)}_{\ell,j})$. Then
$g(j)=g(\ell-j)$, and $g(j+1)\ge g(j)$ whenever $2j+1\le\ell$. Hence
$\max_jg(j)=g(\lfloor\ell/2\rfloor)$. In particular $T_s=C^sE_k$ for
$1\le k\le\Delta$, where $s=\Delta-k$ and
$E_k=\E\psi(S^{(\Delta)}_{s,\lfloor s/2\rfloor})$.
\end{lemma}

\begin{proof}
Symmetry holds because $S_{\ell,\ell-j}^{(n)}$ has the law of
$-S_{\ell,j}^{(n)}$ and $f$ is even. For $2j+1\le\ell$, write
$S_{\ell,j}^{(n)}=Y+\xi^-$ and let $S_{\ell,j+1}^{(n)}$ have the law of
$Y+\xi^+$, where $\xi^\pm$ is a sign of mean $\pm\vartheta$ independent of $Y$.
Then
\[
 g(j+1)-g(j)=\vartheta\bigl[\E f(Y+1)-\E f(Y-1)\bigr].
\]
Pair the $j$ plus-tilted signs of $Y$ with $j$ of its $\ell-j-1\ge j$
minus-tilted signs, and write $Y=Y_{\rm s}+Y_{\rm t}$. The symmetric part
$Y_{\rm s}$, the sum of these balanced pairs and the fair signs, has an even pmf
$p_{\rm s}$ that is nonincreasing in $|x|$ by Lemma~\ref{lem:fs-unimodal}: the
pair pmf $(w,1-2w,w)$ with $w=\pi_+(1-\pi_+)\le\frac14$ and the fair pmf both
qualify. The remainder $Y_{\rm t}$ is a sum of $\ell-2j-1\ge0$ minus-tilted
signs, and $\delta(z):=\Pp(Y_{\rm t}=-z)-\Pp(Y_{\rm t}=z)\ge0$ for $z\ge0$,
because the ratio of these probabilities is $(\pi_+/(1-\pi_+))^z$. For $y>0$,
\[
 \Pp(Y=-y)-\Pp(Y=y)=\sum_xp_{\rm s}(x)\,\delta(y+x)\ge0,
\]
by pairing $x$ with $x'=-2y-x$: then $\delta(y+x')=-\delta(y+x)$, the points
$x,x'$ lie in the same coset, and $|x|<|x'|$ whenever $y+x>0$. Finally
\[
 \E f(Y+1)-\E f(Y-1)
 =\sum_{y>0}[f(y-1)-f(y+1)][\Pp(Y=-y)-\Pp(Y=y)]\ge0.
\]
The last assertion is the case $f=\psi$ and $n=\Delta$; for $k=\Delta$ it is
trivial.
\end{proof}

For $1\le k\le\Delta$, with $s=\Delta-k$ and $c(s)$ as in \eqref{eq:star-lambda},
put
\[
\begin{gathered}
 \tilde\lambda_k=\tfrac12(\eta^{-1}+\eta^{2k-1})+c(s)\eta^{k-1},\qquad
 \varkappa_k=\eta^{2k-1}+2c(s)\eta^{k-1},\\
 \widehat B_k=\max_{0\le j<k}\bigl[\eta G^-_{k-1,j}-\varkappa_kG^+_{k-1,j}\bigr].
\end{gathered}
\]

\begin{lemma}[The row $h=0$ as a single ratio]\label{lem:fs-floor-identity}
For $1\le k\le\Delta$,
\[
 \Gamma_k:=\frac{\sinh2\beta}{p(k)}
 =\tilde\lambda_k+c_0^{2\Delta}H_{k-1}^2
   \Bigl[C^{\Delta-2k+1}E_k+\tfrac12\widehat B_k\Bigr].
\]
In particular $\Gamma_1=2C+c_0^{2\Delta}(C^{\Delta-1}E_1-C\,E_\Delta)$, so
\begin{equation}\label{eq:fs-k1-floor}
 p(1)=p_*(d,t):=\frac{2t}{(1-t^2)\bigl[2C+(1-t^2)^{-2d}
   \bigl(C^{2d-1}E_1-C\,E_{2d}\bigr)\bigr]},
\end{equation}
where $E_1=\E\psi(\Pi_{d-1}+\bar S_2)$ and $E_{2d}=\E\psi(\bar S_{2d})$ are
finite sums over at most $2d+1$ values.
\end{lemma}

\begin{proof}
Since $1/p(k)=\lambda_k+M_k/N_k$ and $\sinh2\beta=2tc_0^2$, three identities
suffice. First, $\sinh(2\beta)/\alpha=1/(2\eta)$, so
$\sinh(2\beta)\lambda_k=\tilde\lambda_k$. Second,
$\sinh(2\beta)/N_k=2c_0^{2\Delta}H_{k-1}^2/C^{k-1}$. Third, from
$\pi_+=1/(1+\eta^2)$, $\lambda_k\vartheta=(1+\eta^{2k}+2c(s)\eta^k)/(1+\eta^2)$
and $C=(1+\eta^2)/(2\eta)$ one gets $C(\pi_+-\lambda_k\vartheta)=-\varkappa_k/2$
and $C(1-\pi_+)=\eta/2$. Hence \eqref{eq:residual-sums} reads
$M_k=\frac12T_s+\frac14C^{k-1}\widehat B_k$, and $T_s=C^sE_k$ by
Lemma~\ref{lem:fs-balanced-tilt}. Combining these gives the identity.

At $k=1$, $s=\Delta-1$ is odd, so $c(s)=C$, $\tilde\lambda_1=2C$ and
$\varkappa_1=\eta+2C$. Also $H_0=1$, and
$G^\pm_{0,0}=\E\psi(\pm1+\bar S_{\Delta-1})=E_\Delta$, so
$\widehat B_1=-2C\,E_\Delta$. The balanced sum defining $E_1$ is
$\Pi_{d-1}+\xi+\bar S_1$ with one minus-tilted sign $\xi$. Since
$\Pi_{d-1}+\bar S_1$ is symmetric and $\psi$ is even, replacing $\xi$ by a fair
sign does not change the expectation, so $E_1=\E\psi(\Pi_{d-1}+\bar S_2)$.
Finally $\sinh2\beta=2t/(1-t^2)$ and $c_0^{2\Delta}=(1-t^2)^{-2d}$.
\end{proof}

For $1\le k\le\Delta$ let $\varphi_k(y)=\E\psi(y+\bar S_{\Delta-k})$ for
$y\in k+2\Z$. By Lemma~\ref{lem:fs-unimodal} it is even and nonincreasing in
$|y|$, with maximum $\varphi^*_k=\varphi_k(k\bmod2)$. Put
$r_k=H_{k-1}^2/C^{2(k-1)}$.

\begin{lemma}[Monotone pieces]\label{lem:fs-monotone}
Let $\Delta$ be even. Then:
\begin{enumerate}
\item[(i)] $E_k\le E_1$ for $1\le k\le\Delta$;
\item[(ii)] $r_{k+1}<r_k$ and $H_k>H_{k-1}$ for $1\le k<\Delta$;
\item[(iii)] $\tilde\lambda_{k+1}<\tilde\lambda_k$ and
$\varkappa_{k+1}<\varkappa_k$ for $1\le k<\Delta$;
\item[(iv)] $\varphi^*_{k+1}\ge\varphi^*_k$ and
$\varphi_{k+1}(k+1)\le\varphi_k(k)$ for $1\le k<\Delta$;
\item[(v)] for $1\le k\le\Delta$ and $0\le j<k$,
$G^-_{k-1,j}\le\varphi^*_k$ and $G^+_{k-1,j}\ge\varphi_k(k)$; hence
$\widehat B_k\le\eta\varphi^*_k-\varkappa_k\varphi_k(k)$.
\end{enumerate}
\end{lemma}

\begin{proof}
(i) If $Y$ is symmetric and independent of a sign $\xi$ of any bias, then
$\E\psi(Y+\xi)=\E\psi(Y+1)$. For odd $k<\Delta$, $s$ is odd and the balanced
sum is $\Pi_{(s-1)/2}+\xi+\bar S_k$ with $Y=\Pi_{(s-1)/2}+\bar S_k$ symmetric,
so $E_k=E_{k+1}$. For even $k\le\Delta-2$, $E_k-E_{k+2}=\E g(\Pi_1)-\E g(\bar S_2)$,
where $g(z)=\E\psi(\Pi_{s/2-1}+\bar S_k+z)$ is even and nonincreasing in $|z|$
by Lemma~\ref{lem:fs-unimodal}. This difference equals
$(\frac12-\frac1{2C^2})(g(0)-g(2))\ge0$. Hence
$E_1=E_2\ge E_3=E_4\ge\cdots\ge E_{\Delta-1}=E_\Delta$.

(ii) $H_\ell/C^\ell=\frac12[(1+t)((1+t\vartheta)/C)^\ell
+(1-t)((1-t\vartheta)/C)^\ell]$, and $0<(1\pm t\vartheta)/C<1$ because
$(1+3t^2)(1-t^2)<(1+t^2)^2$. Also
$H_{\ell+1}-H_\ell=\frac{t\vartheta}2[(1+t)(1+t\vartheta)^\ell
-(1-t)(1-t\vartheta)^\ell]>0$.

(iii) The one-step decrease of $\tilde\lambda$ is
$\frac12\eta^{2k-1}(1-\eta^2)+(c(s)\eta^{k-1}-c(s-1)\eta^k)$, and that of
$\varkappa$ is $\eta^{2k-1}(1-\eta^2)+2(c(s)\eta^{k-1}-c(s-1)\eta^k)$. Since
$\Delta$ is even, $s$ has the parity of $k$. For odd $k$,
$c(s)\eta^{k-1}-c(s-1)\eta^k=\eta^{k-1}(C-\eta)>0$; for even $k$ it equals
$\eta^{k-1}(1-C\eta)=\eta^{k-1}(1-\eta^2)/2>0$.

(iv) $\varphi_k(y)=\frac12\varphi_{k+1}(y+1)+\frac12\varphi_{k+1}(y-1)$.
Maximizing gives the first claim. The second follows at $y=k$, because
$\varphi_{k+1}(k-1)\ge\varphi_{k+1}(k+1)$.

(v) Write $S^{(\Delta-1)}_{k-1,j}=Y_j+\bar S_{\Delta-k}$ with $Y_j$ the sum of
the $k-1$ tilted signs. Then $G^\pm_{k-1,j}=\E\varphi_k(Y_j\pm1)$, where
$Y_j\pm1$ has the parity of $k$ and $|Y_j\pm1|\le k$.
\end{proof}

\begin{proposition}[Reduction to $k=1$ by blocks]\label{prop:fs-block-reduction}
Let $\Delta=2d$. For $2\le k_a\le k_b\le\Delta$ put
\[
 \mathcal B(k_a,k_b)=\tilde\lambda_{k_a}+c_0^{2\Delta}\Bigl[r_{k_a}C^{\Delta-1}E_1
 +\tfrac12\bigl(H_{k_b-1}^2\,\eta\,\varphi^*_{k_b}
 -H_{k_a-1}^2\,\varkappa_{k_b}\,\varphi_{k_b}(k_b)\bigr)\Bigr].
\]
Then $\Gamma_k\le\mathcal B(k_a,k_b)$ for every $k\in[k_a,k_b]$. Consequently, if
$\{2,\ldots,\Delta\}$ is covered by blocks $[k_a,k_b]$ with
$\mathcal B(k_a,k_b)\le\Gamma_1$, then $p(k)\ge p(1)$ for every $k$, and
$p_0(2d,\beta)=p_*(d,t)$.
\end{proposition}

\begin{proof}
Lemma~\ref{lem:fs-floor-identity}, the identity
$H_{k-1}^2C^{\Delta-2k+1}=r_kC^{\Delta-1}$ and Lemma~\ref{lem:fs-monotone}(i),
(v) give
\[
 \Gamma_k\le\tilde\lambda_k+c_0^{2\Delta}\Bigl[r_kC^{\Delta-1}E_1
 +\tfrac12H_{k-1}^2\bigl(\eta\varphi^*_k-\varkappa_k\varphi_k(k)\bigr)\Bigr].
\]
Inside a block, bound $\tilde\lambda_k$ and $r_k$ by their values at $k_a$,
since they decrease, and $H_{k-1}^2\varphi^*_k$ by its value at $k_b$, since it
increases. Bound $H_{k-1}^2\varkappa_k\varphi_k(k)$ below by
$H_{k_a-1}^2\varkappa_{k_b}\varphi_{k_b}(k_b)$, since all three factors are
positive and monotone by Lemma~\ref{lem:fs-monotone}. Finally
$p(k)=\sinh(2\beta)/\Gamma_k$.
\end{proof}

A block certificate needs $E_1$, $E_\Delta$ and, for each block, the two
fair-binomial sums $\varphi^*_{k_b}$ and $\varphi_{k_b}(k_b)$. Such a finite
check at a given temperature cannot be replaced by a statement uniform in
$\beta$: exact evaluation of the full row gives $\min_kp(k)=p(\Delta)<p(1)$ at
$(\Delta,t)=(20,\frac13)$, $(32,\frac15)$ and $(44,\frac3{20})$. Any general
proof that $k=1$ is the worst arrangement must therefore restrict $\beta$. We
have no such proof; the corollaries below use only the finite block checks at
their own temperatures.

\subsection{The certified dimensions}\label{sec:fs-certificates}

\begin{proof}[Proof of Corollary~\ref{cor:single-floor}]
At each listed point $t$ is rational, so every quantity below is an exact
rational and every comparison is an exact rational inequality; decimals are
rounded in the safe direction. Table~\ref{tab:fs-blocks} lists, for each point,
a cover of $\{2,\ldots,2d\}$ by blocks with $\mathcal B(k_a,k_b)<\Gamma_1$.
Proposition~\ref{prop:fs-block-reduction} therefore gives
$p_0(2d,\beta)=p_*(d,t)$. The check uses $E_1$, $E_{2d}$ and two fair-binomial
sums per block, at most twenty explicit sums at each point.
Lemma~\ref{lem:collision-green}(f) gives $F_d\le1-1/\widehat G_d$.
Table~\ref{tab:fs-floor} shows $p_*(d,t)>F_d$ in every row, so
Theorem~\ref{thm:single-floor} applies. Its constant
$\frac12(\underline p-F_d)/(\underline p(1-F_d))$ increases in $\underline p$
and decreases in $F_d$; evaluating it at the tabulated bounds gives $\theta_{\rm fs}$.
\end{proof}

\begin{table}[!htbp]
\centering
\begin{tabular}{rr|rrrr}
\hline
$d$&$t$&$p_*(d,t)\ge$&$F_d\le$&$p_*/F_d\ge$&$\theta_{\rm fs}\ge$\\ \hline
16&$13/125$&$0.0679971076$&$0.067455$&$1.0080$&$0.0042$\\
16&$51/500$&$0.0680421165$&$0.067455$&$1.0087$&$0.0046$\\
17&$1/10$&$0.0659393703$&$0.063115$&$1.0447$&$0.0228$\\
18&$12/125$&$0.0640504639$&$0.059313$&$1.0798$&$0.0393$\\
19&$7/75$&$0.0623013944$&$0.055951$&$1.1135$&$0.0539$\\
20&$1/11$&$0.0606881634$&$0.052955$&$1.1460$&$0.0672$\\
22&$11/125$&$0.0577808419$&$0.047843$&$1.2077$&$0.0903$\\ \hline
\end{tabular}
\caption{Certified values for Corollary~\ref{cor:single-floor}. Lower bounds are
rounded down and upper bounds up; each entry summarizes an exact rational
comparison. The bound on $F_d$ is $1-1/\widehat G_d$ from
Lemma~\ref{lem:collision-green}(f), and
$\theta_{\rm fs}=\frac12(p_*-F_d)/(p_*(1-F_d))$.}
\label{tab:fs-floor}
\end{table}

\begin{table}[!htbp]
\centering
\begin{tabular}{rr|lrr}
\hline
$d$&$t$&Blocks $[k_a,k_b]$&Sums&Slack $\ge$\\ \hline
16&$13/125$&$2,3,4,5,[6,7],[8,13],[14,32]$&$16$&$0.0267$\\
16&$51/500$&$2,3,4,5,[6,7],[8,13],[14,32]$&$16$&$0.0065$\\
17&$1/10$&$2,3,4,5,[6,7],[8,12],[13,32],[33,34]$&$18$&$0.0073$\\
18&$12/125$&$2,3,4,5,[6,7],[8,11],[12,30],[31,36]$&$18$&$0.0052$\\
19&$7/75$&$2,3,4,5,[6,7],[8,11],[12,30],[31,38]$&$18$&$0.0064$\\
20&$1/11$&$2,3,4,5,[6,7],[8,10],[11,25],[26,40]$&$18$&$0.0044$\\
22&$11/125$&$2,3,4,5,6,[7,8],[9,13],[14,37],[38,44]$&$20$&$0.0099$\\ \hline
\end{tabular}
\caption{Block covers for Proposition~\ref{prop:fs-block-reduction}. A single
integer denotes a block with $k_a=k_b$. The sums are $E_1$, $E_{2d}$ and two
per block. The slack is the smallest value of $\Gamma_1-\mathcal B(k_a,k_b)$,
rounded down; at these points $\Gamma_1\in[3.020,3.093]$.}
\label{tab:fs-blocks}
\end{table}

As an independent check, the full row $p(1),\ldots,p(2d)$ was evaluated exactly
at every listed point, with the maximum in $T_s$ taken over all indices $j$;
this uses $\frac32\Delta(\Delta+1)$ binomial sums, $2970$ at $d=22$. In every case
the minimum occurs at $k=1$, and $p(2)/p(1)\in[1.0094,1.0149]$. Summing $u_n(d)$
exactly for $n<4d$ and bounding the tail by Lemma~\ref{lem:collision-green}(e)
with $A=4$ determines $F_d$ to within $10^{-8}$. The ratios in
Table~\ref{tab:fs-floor} then improve to $1.0146$, $1.0153$, $1.0502$, $1.0845$,
$1.1174$, $1.1494$ and $1.2103$. The programs that compute these values in exact
rational arithmetic are described in the reproducibility section. Their signed
cases agree exactly with the independent implementation of the dimension-$22$
certificate of \cite{PeiV1}.

\begin{remark}[Signed refinements]\label{rem:fs-signed}
The signed criterion \eqref{eq:signed-ratio} can replace the nonnegative
residual in individual arrangements. At $(d,t)=(22,\frac{11}{125})$, take
$\lambda=1/0.0586$ and $c=\frac{119}{100}$. Exact enumeration of
\eqref{eq:signed-residual} over $44$ and $86$ cavity orientations gives the
first two bounds below, and the nonnegative residual gives the third:
\[
\begin{gathered}
 \mathcal M_{0,1}(\lambda,c)\le-0.000374848573757,\qquad
 \mathcal M_{0,2}(\lambda,c)\le-0.019889903399663,\\
 \min_{3\le k\le44}p(k)\ge0.059010354650604.
\end{gathered}
\]
Thus $0.0586$ is a uniform fresh-star floor; it is the row $h=0$ of the
certificate of \cite{PeiV1} (Remark~\ref{rem:v1-certificate}). Four explicit terms
and a tail bound give $G_{22}<1/(1-0.0586)$ by
Lemma~\ref{lem:collision-green}(g), that is, $F_{22}<0.0568646<0.0586$, and
Theorem~\ref{thm:single-floor} applies with $\theta_{\rm fs}\ge0.0157$. The signed
case $k=2$ is needed for this variant, because the nonnegative residual gives
only $p(2)=0.0583248\ldots$. Similarly, at $(16,\frac{13}{125})$ the signed
criterion with $\lambda=10000/689$ and $c=\frac65$ gives
$\mathcal M_{0,1}(\lambda,c)<-0.00136$ over $32$ orientations. Together with
$\min_{k\ge2}p(k)\ge0.0690100$, this gives the floor $0.0689$, which exceeds
$1-1/\widehat G_{16}$ by a factor of at least $1.0214$.
\end{remark}

\begin{remark}[Where the route stops]\label{rem:fs-limits}
Dimension $16$ is the smallest dimension certified by this route. At $d=15$, the exact
sum of $u_n(15)$ over $n<60$ and
Lemma~\ref{lem:collision-green}(e) with $A=4$ give
$F_{15}\in[0.0718626355,0.0718626462]$. At
$t=\frac{27}{250}$ the arrangements $k\ge2$ still clear this threshold in exact
arithmetic: the signed criterion at $k=2$, with $\lambda=1/0.0725$ and
$c=\frac65$, gives $\mathcal M_{0,2}(\lambda,c)<-0.0105$ over $58$ cavity
orientations, and the nonnegative residual gives $\min_{k\ge3}p(k)\ge0.072807$.
For $k=1$, however, a floating-point scan over $t\in[0.095,0.12]$ and
$c\in[1,1.6]$ found no bound from \eqref{eq:linear-star} or
\eqref{eq:signed-ratio} above about $0.0713$, which is below $F_{15}$. With
the oriented bound $\pc(\Z^{15})\le F_{15}$, we therefore expect dimension $15$
to require a sharper estimate for the arrangement $k=1$. A sharper global
comparison is the other option: by the proof of Theorem~\ref{thm:single-floor}, a uniform
fresh-star floor above the concurrent bound $\pc(\Z^{15})\le0.0691$ of
\cite[Table~2]{JiangLang2026} would give percolation of both overlap signs at
$d=15$, without the explicit constant \eqref{eq:fs-root-bound}. We have not
certified such a floor for $k=1$.
\end{remark}

\subsection{A floor without finite sums}\label{sec:fs-closed-form}

Since $R\ge1$, bounding $R^{-2}\le1$ in the total mass of arrangement $k$ gives
$D\le\E_\epsilon w_{0,k}=\frac12(C^k+C^{\Delta-k})$. Together with $A\ge N_k$,
this gives the paired floor
\[
 p_{\rm pair}(k)=\frac{N_k}{\frac12(C^k+C^{\Delta-k})}
\]
in that arrangement. Its minimum over arrangements occurs at an endpoint.

\begin{proposition}[A closed-form star floor]\label{prop:fs-paired-floor}
Under the hypotheses of Proposition~\ref{prop:unit-star} with $h=0$,
\[
 \nu_{G,\beta}(B_e=1\mid\mathcal F_v^{\mathrm{ext},J})
 \ge\min\{p_{\rm pair}(1),p_{\rm pair}(\Delta)\},
\]
where
\[
 p_{\rm pair}(1)=\frac{2t(1-t^2)^{\Delta-1}}{C+C^{\Delta-1}},\qquad
 p_{\rm pair}(\Delta)=\frac{2t(1-t^2)^{\Delta-1}C^{\Delta-1}}
 {H_{\Delta-1}^2(C^\Delta+1)}.
\]
\end{proposition}

\begin{proof}
Extend $N_x$ and $H_{x-1}$ to real $x\in[1,\Delta]$ by the same formulas;
$H_{x-1}>0$ because $0<1-t\vartheta<1+t\vartheta$. Then $\log p_{\rm pair}(x)$ is
concave on $[1,\Delta]$. The term $(x-1)\log C$ is linear. The term
$-2\log H_{x-1}$ is concave, because $\log H_\ell$ is the logarithm of a positive
combination of exponentials of affine functions of $\ell$, hence convex. The
term $-\log(C^x+C^{\Delta-x})$ is concave for the same reason. A concave function
on an interval is bounded below by the smaller of its endpoint values. Since the
conditional opening probability in arrangement $k$ is at least
$p_{\rm pair}(k)$, this proves the proposition.
\end{proof}

\begin{corollary}[A closed-form certificate in dimensions 25 and 26]\label{cor:fs-closed-form}
For iid fair couplings $J_e=\pm1$, let
$(d,t,\theta_{\rm fs})=(25,\frac{27}{400},0.0020)$ or $(26,\frac{33}{500},0.0129)$, and
$\beta=\operatorname{atanh}t$. Then every $\nu\in\mathcal L_{d,\beta}$ satisfies
$\nu(o\blue\infty,q_o=s)\ge\theta_{\rm fs}$ for $s=\pm1$, and both overlap signs
contain infinite blue components $\nu$-almost surely.
\end{corollary}

\begin{proof}
Proposition~\ref{prop:fs-paired-floor} shows that
$\min\{p_{\rm pair}(1),p_{\rm pair}(2d)\}$ is a uniform fresh-star floor, and
Lemma~\ref{lem:collision-green}(f) gives $F_d\le1-1/\widehat G_d$. At $d=25$
these explicit rationals satisfy $p_{\rm pair}(1)\ge0.04196433$,
$p_{\rm pair}(50)\ge0.05133661$ and $F_{25}\le0.04180338$, so the floor exceeds
$F_{25}$ by a factor of at least $1.0038$. At $d=26$ the corresponding values
are $0.04113940$, $0.05037252$ and $0.04011778$, with factor at least $1.0254$.
Theorem~\ref{thm:single-floor} applies, and its constant, evaluated at these
bounds as in the proof of Corollary~\ref{cor:single-floor}, is at least
$\theta_{\rm fs}$.
\end{proof}

No finite sum enters this corollary: $p_{\rm pair}(1)$, $p_{\rm pair}(2d)$ and
$\widehat G_d$ are explicit rational expressions in $t$ and $d$, and the proof
compares a handful of them in exact arithmetic.

\begin{remark}[The earlier dimension-$22$ argument]\label{rem:v1-certificate}
The preprint \cite{PeiV1} proved the case $d=22$,
$\beta=\operatorname{atanh}(11/125)$ of Corollary~\ref{cor:single-floor} by a
different exploration. That exploration retried targets and therefore needed
opening bounds after up to four incident failures, supplied by
Proposition~\ref{prop:unit-star} and the signed criterion
\eqref{eq:signed-ratio}.
Repeated failures at one target were represented by a single random activation
threshold \cite[Lemma 3.1 and Eqs.~(3.1)--(3.2)]{KKT2015}. The coordinate
directions were grouped into three classes and projected onto the triangular
lattice \cite[Lemma 4.2, Proposition 3.2 and Section 4.1]{GPS2026}, and physical
successes rejected by thinning were retained as witnesses for later attempts.
The resulting comparison was with anisotropic independent bond percolation on
the triangular lattice. A certificate covering $210$ overlap and failure cases
placed its parameters above the exact critical surface
\cite[Theorem 11.116]{Grimmett1999}, with the critical polynomial exceeding one
by less than $0.001$.
Theorem~\ref{thm:single-floor} supersedes that argument: it needs only the row
$h=0$, which at this temperature exceeds $F_{22}$ by a factor of at least
$1.2077$. We do not reproduce the earlier argument.
\end{remark}

\section{Invariant conditioning and the finite-mean density profile}\label{app:density}

We prove \eqref{eq:finite-mass-profile}. Fix the cutoff quantities of
Appendix~\ref{sec:finite-mean} and write
\[
 Y_d=\frac{\theta_d^{\rm site}(\varpi_K)}{\varpi_K},\qquad
 \varepsilon_d=1-\alpha^\Delta,\qquad \delta_K=\tfrac12 r_K^\Delta.
\]
Here $\varpi_K>0$, and $\varepsilon_d$ is the probability that the root is not
good. The following finite-$d$ bounds imply the claim:
\begin{equation}\label{eq:density-finite}
 D_++D_-\ge Y_d,\qquad
 D_s\ge\delta_K(Y_d-\varepsilon_d),\qquad
 |Q|\le1-\alpha^\Delta r_K^\Delta
 \quad\hbox{almost surely}.
\end{equation}

Let $\mathcal I_0=\sigma(Q,D_+,D_-)$, defining these variables by limsup cube
averages on the full configuration space. Finite changes of spins and bonds do
not change them, by the finite-modification argument in
Proposition~\ref{prop:imbalance}. Hence $\mathcal I_0$ is exterior to every
target star. The iid cutoff marks $C$ retain their original distribution after
conditioning on $\mathcal I_0$. Indeed, for an invariant bounded variable $g$
and a cylinder $f(C)$, stationarity gives
\[
 \E[g f(C)]
 =\E\!\left[g\,\frac1{|\Lambda_n|}\sum_{x\in\Lambda_n}f(T_xC)\right]
 \longrightarrow\E[g]\E[f(C)]
\]
by the iid ergodic theorem and bounded convergence. A monotone class proves
the independence assertion. The same argument applies to any
translation-invariant event, including absence of an infinite blue component
of a specified overlap sign.

For fixed marks, let $h(C)$ be the root-forced survival probability of the iid
site field of parameter $\pi_K$ on good targets. The defect-field comparison
gives $\E_C h(C)\ge Y_d$. Good-target exploration remains valid conditional on
$C,\mathcal I_0$, and on $q_o=s$, since the root is never a target.
Without label conditioning it gives $D_++D_-\ge\E_Ch(C)\ge Y_d$.
When the root is good, the isolated-label estimate additionally yields
\[
 \nu(o\blue\infty,q_o=s\mid C,\mathcal I_0)
 \ge\delta_K h(C).
\]
Averaging gives $D_s\ge\delta_K\E_C[\mathbf1_{\{o\text{ good}\}}h(C)]
\ge\delta_K(Y_d-\varepsilon_d)$. The conditional root probabilities here equal
the corresponding empirical densities by the ergodic theorem and the tower
property. Finally, on a good root the conditional overlap mean is between
$-(1-r_K^\Delta)$ and $1-r_K^\Delta$; use the trivial bound on bad roots.
Averaging the unchanged mark law proves the last inequality in
\eqref{eq:density-finite}.

At $\beta=c/(\Delta m)$, finite mean gives
$\varepsilon_d\to0$, $r_K^\Delta\to1$, and $\Delta\varpi_K\to c$.
Equation~\eqref{eq:kesten-survival}, squeezed below at every fixed $1<c'<c$,
gives $\liminf Y_d\ge y(c)$. This proves
\eqref{eq:finite-mass-profile} with uniform deterministic errors.
More precisely, the lower bounds are limits of
$\inf_\nu\operatorname*{ess\,inf}_\nu D_s$ and the upper bounds use
$\sup_\nu\operatorname*{ess\,sup}_\nu$. Since $y(c)\uparrow1$, taking
$d\to\infty$ first and then $c\to\infty$ makes both infinite-sector densities
approach $1/2$ and the finite-blue mass vanish. No common probability space
across dimensions or finite-volume giant-component limit is asserted.

\section{Susceptibility certificates and an edge-by-edge bound}\label{app:separation-certificate}

The following bound restores the star of a vertex one edge at a time. For the
two-point susceptibility it is weaker than Proposition~\ref{prop:sg-uniform}:
its criterion $\kappa_\Delta<1$ implies $F(2\beta)^{\Delta-1}<3$. Indeed
$\cosh2z-\cosh z\ge3(\cosh z-1)$ gives $F(4\beta)-F(2\beta)\ge3(F(2\beta)-1)$, and
$\mathsf A_\beta\ge1$, so $\kappa_\Delta<1$ forces $F(2\beta)-1<2/(3\Delta)$ and
hence $F(2\beta)^{\Delta-1}<e^{2/3}<3$. It also gives exponential decay in the
distance and controls odd products of spins at fixed sites. It supplies the
certificate at $d=180$ below.

\begin{proposition}[A centered susceptibility bound]\label{prop:sg-bound}
Let $G$ be a finite simple graph of maximum degree at most $\Delta\ge1$,
with zero field and symmetric iid couplings. Assume $F(4\beta)<\infty$, with
$F(u)=\E\cosh(uW)$, and set
\begin{equation}\label{eq:sg-kappa}
 \mathsf A_\beta=\frac{F(4\beta)+F(2\beta)}2,\qquad
 \mathsf B_\beta=\frac{F(4\beta)-F(2\beta)}2,\qquad
 \kappa_\Delta=\Delta\mathsf B_\beta\mathsf A_\beta^{\Delta-1}.
\end{equation}
If $\kappa_\Delta<1$, then
\begin{equation}\label{eq:sg-finite}
 \sup_x\sum_y\E\langle\sigma_x\sigma_y\rangle_{G,J}^2
 \le\frac1{1-\kappa_\Delta},\qquad
 \E\langle\sigma_x\sigma_y\rangle_{G,J}^2
 \le\kappa_\Delta^{\operatorname{dist}_G(x,y)}.
\end{equation}
At $\Delta=2d$, every selected periodic joint limit satisfies
\begin{equation}\label{eq:sg-limit}
 1\le\chi_{\rm SG}(\nu):=\sum_x\E_\nu[q_0q_x]
 \le\frac1{1-\kappa_\Delta},\qquad
 \E_\nu\left(\frac1{|\Lambda|}\sum_{x\in\Lambda}q_x\right)^2
 \le\frac1{(1-\kappa_\Delta)|\Lambda|}.
\end{equation}
In particular its spatial overlap is zero almost surely. The analogous
full-volume overlap estimate holds on the defining finite tori.
More generally, for every finite odd set $A$,
\begin{equation}\label{eq:odd-susceptibility}
 0\le \E_\nu[q_{A,0}q_{A,x}]
 \le\kappa_\Delta^{\operatorname{dist}(A,A+x)},\qquad
 \chi_A(\nu):=\sum_x\E_\nu[q_{A,0}q_{A,x}]
 \le |A|^2\left(\frac{1+\kappa_\Delta}{1-\kappa_\Delta}\right)^d<\infty.
\end{equation}
Consequently $Q_A=0$ almost surely for every such $A$.
\end{proposition}

\begin{proof}
Write $c_H(x,y)=\E\langle\sigma_x\sigma_y\rangle_{H,J}^2$ for each
edge-deleted subgraph $H$ of a fixed finite ambient graph. Remove
$e=\{x,j\}$, with $x\ne y$, and put
\[
 a=\langle\sigma_x\sigma_y\rangle_0,\quad
 b=\langle\sigma_j\sigma_y\rangle_0,\quad
 c=\langle\sigma_x\sigma_j\rangle_0.
\]
The normalized correlation after restoring $e$ is $(a+ub)/(1+uc)$,
where $u=\epsilon\tanh(\beta W_e)$. For $t=\tanh(\beta W_e)$, its
fair-sign mean square is exactly
\[
 \frac{(a^2+t^2b^2)(1+t^2c^2)-4t^2abc}{(1-t^2c^2)^2}
 \le A(t)a^2+B(t)b^2,
\]
where
\[
 A(t)=\frac{1+3t^2}{(1-t^2)^2},\qquad
 B(t)=\frac{t^2(3+t^2)}{(1-t^2)^2}.
\]
Indeed $|c|\le1$ and $4|abc|\le2a^2+2b^2$.
The identities $A(\tanh z)=(\cosh4z+\cosh2z)/2$ and
$B(\tanh z)=(\cosh4z-\cosh2z)/2$ then give
\begin{equation}\label{eq:sg-edge}
 c_{H+e}(x,y)\le\mathsf A_\beta c_H(x,y)
                  +\mathsf B_\beta c_H(j,y).
\end{equation}
The missing sign and magnitude are independent of the cavity correlations;
their averaging therefore preserves the original quenched normalization.

Delete the whole $x$-star of $H$ and restore its $r\le\Delta$ edges in a
fixed order independent of $y$ and of the disorder, obtaining
$H_0,\ldots,H_r=H$ with restored edge $\{x,j_i\}$ at step $i$.
The isolated spin gives $c_{H_0}(x,y)=0$ for $y\ne x$. Iterating
\eqref{eq:sg-edge}, with $K=\mathsf B_\beta\mathsf A_\beta^{\Delta-1}$, yields
\begin{equation}\label{eq:sg-star}
 c_H(x,y)\le K\sum_{i=1}^r c_{H_{i-1}}(j_i,y),\qquad y\ne x.
\end{equation}
Now take the finite maximum
$M=\max_{H\subseteq G,x}\sum_y c_H(x,y)\le |V(G)|$.
The intermediate graphs do not depend on $y$, so summing first gives
$M\le1+\Delta K M$, proving the susceptibility bound.
For the distance estimate, maximize $c_H(x,y)$ over all deleted graphs and
all $x$ at ambient distance at least $n$ from a fixed $y$. Equation
\eqref{eq:sg-star} bounds this maximum by $\kappa_\Delta$ times the
corresponding maximum at distance $n-1$. Its initial value is at most one.
Disconnected pairs have zero correlation by component spin flips; the
$\beta=0$ case is immediate.

On each torus, conditional replica independence gives
$\E q_xq_y=c_G(x,y)\ge0$. Pass this bounded local observable and every finite
partial susceptibility sum to the selected limit, then increase the finite
sum. Stationarity gives the block-variance bound in \eqref{eq:sg-limit}.
The cube ergodic theorem and this $L^2$ bound force spatial overlap zero.
No independence of limiting replicas or assertion about all quenched Gibbs
states is used.

For the extension, fix an odd set $B$ in the finite ambient graph.
The same normalized edge identity applies to $\langle\sigma_A\sigma_B\rangle$.
If the odd set $A$ is at distance at least $n\ge1$ from $B$, choose $v\in A$
and delete its star; the isolated spin again makes the initial correlation zero.
Restoring $\{v,j\}$ replaces $A$ in the inhomogeneous term by
$A\mathbin{\triangle}\{v,j\}$, still odd and nonempty, at distance at least
$n-1$ from $B$. Maximizing over all such $A$ and all edge-deleted graphs
therefore gives $M_n\le\kappa_\Delta M_{n-1}$, with $M_0\le1$.
This proves the set-distance estimate. Pass to the periodic limit for each fixed
translate and use
\[
 \kappa^{\operatorname{dist}(A,A+x)}
 \le\sum_{a,b\in A}\kappa^{|x+b-a|_1}
\]
to obtain \eqref{eq:odd-susceptibility}. At $\kappa=0$ the conclusion follows
directly from independent spins. The block-variance argument then gives $Q_A=0$.
\end{proof}

At $d=180$, $t=\tanh\beta=1/56$, the unit-coupling MNS parameters are
\[
 L_\beta=\frac{112}{3249},\qquad
 R_\beta=\frac{9853313}{9834495},\qquad
 p_{\rm MNS}=\frac{L_\beta}{1+R_\beta^{360}}.
\]
Gomes, Pereira and Sanchis \cite[Theorem 2.3(2)]{GPS2026} give
$p_c^{\rm site}(\Z^{180})\le1-2^{-1/60}$. The exact integer inequalities
\[
 149273\cdot9834495^{360}>74727\cdot9853313^{360},\qquad
 2\cdot1977^{60}<2000^{60}
\]
therefore show $p_{\rm MNS}>23/2000>p_c^{\rm site}(\Z^{180})$.
For the physical susceptibility,
\[
 \mathsf A_\beta=\frac{9843904}{9828225},\qquad
 \mathsf B_\beta=\frac{9409}{9828225}.
\]
The further exact inequality
\[
 8\cdot360\cdot9409\cdot9843904^{359}<5\cdot9828225^{360}
\]
gives $\kappa_{360}<5/8$ and hence $\chi_{\rm SG}<8/3$.
The companion program \path{verify_physical_separation.py}
\cite{CMRCode} checks all three comparisons.
These constants supply
a common percolation and finite-susceptibility example that uses only
Propositions~\ref{prop:criteria} and~\ref{prop:sg-bound}.

For Corollary~\ref{cor:explicit-balance}, write $t=p/q$ in lowest terms, so
that $F(2\beta)=a/b$ with $a=q^2+p^2$ and $b=q^2-p^2$, and put $n=2d-1$. Then
\[
 F(2\beta)^{n}<\tfrac{12}5\iff5a^n<12b^n,\qquad
 \chi_*(2d,\beta)=1+\frac{a^{n+1}-b^{n+1}}{b\,(3b^n-a^n)},
\]
so each claim is a comparison of integers. Table~\ref{tab:explicit-chi} lists
the values rounded upward; two independent programs \cite{CMRCode},
\path{verify_explicit_balance.py} in exact rational arithmetic and
\path{verify_explicit_balance.mjs} in integer arithmetic, confirm all seventeen comparisons
$F(2\beta)^{2d-1}<12/5$ and $\chi_*<7/2$, and $\chi_*<7/4$ at $(7,3/20)$.

\begin{table}[!htbp]
\centering
\small
\begin{tabular}{cccc|cccc}
\hline
$d$ & $t$ & $F(2\beta)^{2d-1}\le$ & $\chi_*\le$ &
$d$ & $t$ & $F(2\beta)^{2d-1}\le$ & $\chi_*\le$\\ \hline
$12$ & $3/25$   & $1.9396$ & $1.9394$ & $10$ & $13/100$ & $1.9008$ & $1.8790$\\
$12$ & $11/100$ & $1.7448$ & $1.6274$ & $9$  & $7/50$   & $1.9474$ & $1.9740$\\
$12$ & $23/200$ & $1.8375$ & $1.7628$ & $9$  & $29/200$ & $2.0441$ & $2.1842$\\
$12$ & $1/8$    & $2.0520$ & $2.1784$ & $9$  & $3/20$   & $2.1493$ & $2.4673$\\
$12$ & $13/100$ & $2.1760$ & $2.5179$ & $9$  & $31/200$ & $2.2638$ & $2.8678$\\
$11$ & $13/100$ & $2.0338$ & $2.1422$ & $9$  & $4/25$   & $2.3884$ & $3.4749$\\
$11$ & $3/25$   & $1.8310$ & $1.7566$ & $8$  & $3/20$   & $1.9643$ & $2.0183$\\
$10$ & $27/200$ & $1.9990$ & $2.0721$ & $7$  & $31/200$ & $1.8679$ & $1.8478$\\
     &          &          &          & $7$  & $3/20$   & $1.7952$ & $1.7286$\\ \hline
\end{tabular}
\caption{The criterion of Proposition~\ref{prop:sg-uniform} at the explicit
pairs of Table~\ref{tab:explicit}, values rounded upward.}
\label{tab:explicit-chi}
\end{table}

\section{From imbalance to overlap order and a pressure cusp}\label{app:odd-order}

We prove Proposition~\ref{prop:odd-order}. Write $\mathcal I$ for the invariant
sigma-field of the stationary joint law. Let $f(J,\sigma,\tau)$ be a bounded
local observable, continuous in its finitely many coupling coordinates, and odd
under flipping either replica separately. On a finite support $S$ its spin
Fourier expansion has the form
\[
 f=\sum_{\substack{A,B\subseteq S\\ |A|,|B|\ {\rm odd}}}
       c_{AB}(J)\sigma_A\tau_B,
 \qquad \|c_{AB}\|_\infty\le\|f\|_\infty.
\]
The coefficients are bounded continuous local functions of the disorder.
For a finite cube $\Lambda$, put
\[
 M_{AB,\Lambda}=\frac1{|\Lambda|}\sum_{x\in\Lambda}
 c_{AB}(T_xJ)\sigma_{x+A}\tau_{x+B},\qquad
 Q_{A,\Lambda}=\frac1{|\Lambda|}\sum_{x\in\Lambda}q_{A,x}.
\]
In a sufficiently large finite torus, conditional replica independence expresses
$\E M_{AB,\Lambda}^2$ as the averaged double sum of
$c_{AB}(T_xJ)c_{AB}(T_yJ)
\langle\sigma_{x+A}\sigma_{y+A}\rangle_J
\langle\sigma_{x+B}\sigma_{y+B}\rangle_J$ divided by $|\Lambda|^2$.
Taking absolute values of the coefficients and applying Cauchy--Schwarz first
to the finite double sum and then to the disorder average gives
\begin{equation}\label{eq:mixed-overlap-bound}
 \E M_{AB,\Lambda}^2
 \le \|c_{AB}\|_\infty^2
       \sqrt{\E Q_{A,\Lambda}^2\,\E Q_{B,\Lambda}^2}.
\end{equation}
Every term in this inequality is a bounded local observable for fixed $\Lambda$.
Pass to the selected joint limit first, then let $\Lambda$ increase. The
ergodic theorem gives $L^2$ convergence of all these bounded spatial averages.
If every odd $Q_A$ vanished, \eqref{eq:mixed-overlap-bound} would make the
invariant average $\E_\nu[f\mid\mathcal I]$ vanish as well.

Apply this observation to
\[
 f_R(J,\sigma,\tau)=q_0\,
 \Pp_U(0\blue\partial B_R\mid J,\sigma,\tau),
\]
using paths within $B_R$. This is bounded, local, continuous in the couplings,
and odd under either replica flip. Its spatial average agrees with that of
$q_0\mathbf1_{\{0\blue\partial B_R\}}$: conditional on spins and disorder,
the centered activation variables at sufficiently separated translates depend
on disjoint sets of coins, so their average has variance $O_R(|\Lambda|^{-1})$.
Writing $F_R=\E_\nu[f_R\mid\mathcal I]$ and $D=D_+-D_-$, contraction of
conditional expectation yields
\[
 \|F_R-D\|_2^2
 \le \nu\bigl(|C_{\rm b}(0)|<\infty,\ 0\blue\partial B_R\bigr)
 \longrightarrow0.
\]
Thus $\E D^2>0$ forces a finite $R$ with $\|F_R\|_2>0$, and hence a
fixed odd $A$ with $\rho_A>0$. Bounded local passage also gives
\[
 \E_\nu[q_{A,0}q_{A,x}]
   =\lim_L\E_J\langle\sigma_A\sigma_{x+A}\rangle_{L,J}^2\ge0.
\]
The ergodic projection identity
$\E_\nu[q_{A,0}Q_A]=\E_\nu Q_A^2$ proves \eqref{eq:odd-order}.
This conclusion concerns spatial averages in the selected limit; it does not
interchange them with full-torus overlap averages.

For completeness we give the pressure bridge directly in this setting.
Let $Z_{\Lambda,J}^{A}(h)$ be the free-boundary two-replica partition sum with
exponent
\[
 \beta\sum_{\{x,y\}\subseteq\Lambda}J_{xy}
       (\sigma_x\sigma_y+\tau_x\tau_y)
       +h\sum_{x:\,x+A\subseteq\Lambda}q_{A,x},
 \qquad
 P_A(h)=\lim_{\Lambda\uparrow\Z^d}
          \frac{\E\log Z_{\Lambda,J}^{A}(h)}{|\Lambda|}.
\]
The first sum is over lattice edges, and the limit runs through cubes.
The finite-range boundary comparison
and $\E|J|<\infty$ give existence of this pressure and independence of free
versus periodic boundary conditions. The additional interaction is bounded
and has fixed finite range. Oddness of $A$ makes $P_A(h)=P_A(-h)$.

Put $a=\operatorname*{ess\,sup}_\nu Q_A>0$, and take $0<t<a$.
The event $E_t=\{Q_A>t\}$ is translation invariant and spin-tail measurable;
define the spatial average by a limsup on exceptional configurations.
Its conditional probability $\nu(E_t\mid J)$ is a translation-invariant
function of $J$. Ergodicity of iid disorder makes it the positive constant
$\nu(E_t)$ almost surely. Therefore $\nu_t=\nu(\,\cdot\mid E_t)$ retains
the original disorder marginal. Tail conditioning preserves the product
Gibbs specification, and stationarity gives
$m_t:=\E_{\nu_t}q_{A,0}=\E_{\nu_t}Q_A>t$.

For fixed $J$, let $\mu^0_{\Lambda,J}$ be the free, uncoupled two-replica
Gibbs measure. The relative entropy of the spin marginal of $\nu_t$ in
$\Lambda$ with respect to $\mu^0_{\Lambda,J}$ is at most
\[
 4\beta\sum_{e\ {\rm crossing}\ \partial\Lambda}|J_e|.
\]
Indeed, conditional on the exterior, the DLR density differs by a boundary
interaction whose oscillation has this bound. Convexity of relative entropy
gives the same estimate for the mixture over exterior spins.
The finite-volume Gibbs variational inequality consequently yields
\[
 \E\log Z_{\Lambda,J}^{A}(h)-\E\log Z_{\Lambda,J}^{A}(0)
 \ge h\,\#\{x:x+A\subseteq\Lambda\}\,m_t
       -4\beta\,\E\sum_{e\ {\rm crossing}\ \partial\Lambda}|J_e|.
\]
Divide by volume, take the thermodynamic limit, and let $t\uparrow a$.
For $h>0$ this gives $P_A(h)-P_A(0)\ge ah$; replica-flip symmetry gives
the bound $a|h|$ for both signs. The law of $Q_A$ is symmetric, so
$a\ge\sqrt{\rho_A}$. Convexity now implies
$P_A'(0+)\ge a$ and $P_A'(0-)\le-a$, proving the claimed cusp.
For a physical Hamiltonian coupling $\lambda$, the dimensionless parameter is
$h=\beta\lambda$.

We now prove Proposition~\ref{prop:odd-to-site}. The edge set $F$ exists:
pair the vertices of the even set $A\triangle\{0\}$ and take the symmetric
difference of the edge sets of lattice paths joining the pairs. Fix $x$, put
$H=F\triangle(x+F)$, so that $|H|\le2k$, and write $b_e=\sigma_u\sigma_v$ for
$e=uv$. Since $\prod_{e\in F}b_e=\sigma_A\sigma_0$,
$\prod_{e\in x+F}b_e=\sigma_{x+A}\sigma_x$ and $b_e^2=1$,
\begin{equation}\label{eq:path-parity}
 \sigma_A\sigma_{x+A}=O\prod_{e\in H}b_e,\qquad O=\sigma_0\sigma_x.
\end{equation}
The symmetric difference, not the union, is needed when $F$ and $x+F$ share
edges.

Work on a torus large enough to contain $A$, $x+A$ and $H$, with fixed
couplings $J$. Let $W'=(W'_e)_{e\in H}$ be independent with the law of $W$
conditioned on $[w_-,w_+]$. For $\zeta\in\{-1,1\}^H$ let $J^\zeta$ agree
with $J$ off $H$ and have $J^\zeta_e=\zeta_eW'_e$ on $H$, and put
$u^\pm_e=\tanh\bigl(\beta(\pm W'_e-J_e)\bigr)$. Since
$e^{\beta(J'_e-J_e)b_e}=\cosh(\beta(J'_e-J_e))\bigl(1+\tanh(\beta(J'_e-J_e))b_e\bigr)$
and the hyperbolic cosines cancel after normalization,
\[
 \mathcal P(u^\zeta)=R_\zeta\langle O\rangle_{J^\zeta},\qquad
 \mathcal P(u):=\Bigl\langle O\prod_{e\in H}(1+u_eb_e)\Bigr\rangle_J,\qquad
 R_\zeta:=\Bigl\langle\prod_{e\in H}(1+u^{\zeta_e}_eb_e)\Bigr\rangle_J,
\]
with $0<R_\zeta\le2^{|H|}$. The polynomial $\mathcal P$ is multiaffine, and
one divided difference in each coordinate extracts its top coefficient:
\begin{equation}\label{eq:top-coefficient}
 \Bigl\langle O\prod_{e\in H}b_e\Bigr\rangle_J
 =\frac{\sum_\zeta\bigl(\prod_{e\in H}\zeta_e\bigr)R_\zeta
        \langle O\rangle_{J^\zeta}}
       {\prod_{e\in H}(u^+_e-u^-_e)}.
\end{equation}
Each gap equals
$\sinh(2\beta W'_e)/[\cosh(\beta(W'_e-J_e))\cosh(\beta(W'_e+J_e))]$, which is at
least $\delta$ when $|J_e|\le\overline W$. On the event
$\mathcal G=\{\max_{e\in H}|J_e|\le\overline W\}$, Cauchy--Schwarz over the $2^{|H|}$
terms of \eqref{eq:top-coefficient} and $R_\zeta\le2^{|H|}$ give
\[
 \Bigl\langle O\prod_{e\in H}b_e\Bigr\rangle_J^2
 \le 8^{|H|}\delta^{-2|H|}\sum_\zeta\langle O\rangle_{J^\zeta}^2 .
\]
The right side does not depend on $J_H$. Averaging over $\zeta$, $W'$ and
the disorder, the couplings $J^\zeta$ have the iid law conditioned on
$\{w_-\le W_e\le w_+,\ e\in H\}$, whose density relative to the original law
is at most $\bar p^{-|H|}$. Hence
$\E\sum_\zeta\langle O\rangle_{J^\zeta}^2\le(2/\bar p)^{|H|}\E\langle O\rangle_J^2$.
The complement of $\mathcal G$ has probability at most $|H|\,\Pp(W>\overline W)$, and
squared correlations are at most one. With \eqref{eq:path-parity},
\[
 \E_J\langle\sigma_A\sigma_{x+A}\rangle_{L,J}^2
 \le\Bigl(\frac{16}{\bar p\,\delta^2}\Bigr)^{|H|}
    \E_J\langle\sigma_0\sigma_x\rangle_{L,J}^2
   +|H|\,\Pp(W>\overline W),
\]
and $|H|\le2k$ gives \eqref{eq:odd-to-site} on every sufficiently large torus.

Both sides are expectations of bounded local two-replica observables, by
finite-volume replica independence, so \eqref{eq:odd-to-site} passes to
$\nu$ along the defining subsequence; the constants depend neither on $L$ nor
on $x$. Spatial Ces\`aro averaging and the ergodic projection identity used
for \eqref{eq:odd-order}, applied to $A$ and to $\{0\}$, give the bound on
$\rho_A$. Letting $\overline W\to\infty$ proves that $\E_\nu Q^2=0$ forces
$\rho_A=0$. For bounded $W$ the error term vanishes and summing
\eqref{eq:odd-to-site} over $x$ gives $\chi_A\le C\chi_{\rm SG}$. No
conditional independence of the limiting replicas and no interchange of
full-torus and infinite-volume overlap averages is used.

\section{Reusable comparison and limit arguments}\label{app:probability}

These results isolate the probability statements used in the main CMR proof.
They apply to other graph-dependent laws whenever the stated retained-observable
and history assumptions can be verified.

\subsection{From a leaf law to conditional success}

Here is a model-independent way to verify the exploration hypothesis. The
deletion geometry is illustrated in Figure~\ref{fig:leaf}; its explicit leaf
probability is the CMR input from Section~\ref{sec:cmr-verification}. On a fixed
finite vertex set, let $P_A$ be a probability law for each edge set $A$ under
consideration. Configurations may have vertex variables and variables indexed by
edges in $A$. Deleting $f$ forgets its edge variables and retains all vertex
variables. A \emph{retained observable} is the same measurable function of these
surviving coordinates in both laws. In particular, an edge indicator $\omega_e$
for $e\ne f$ must be unchanged as a function of the surviving coordinates.
The marginal of $P_A$ after deletion need not equal $P_{A\setminus f}$.

For each $v$, choose a sigma-field $\mathcal H_v$ of allowed exterior
observables, retained upon deleting its whole star. In applications it may
include all vertex variables outside $v$ and all nonincident edge variables,
or a smaller sigma-field. It cannot include an observation of a deleted edge.

\begin{lemma}[Leaf and restoration comparison]\label{lem:restoration}
Fix $e=uv\in A$. Suppose every intermediate graph obtained by deleting edges
of $\starv_A(v)\setminus\{e\}$ has the following properties.
\begin{enumerate}
\item For each edge $f$ being restored, constants $0<\ell_f\le M_f<\infty$
give, for every bounded nonnegative retained observable $h$,
\begin{equation}\label{eq:abstract-restoration}
 \ell_f E_{A'\setminus f}h\le E_{A'}h\le M_f E_{A'\setminus f}h.
\end{equation}
It is enough that these inequalities hold for $h=\mathbf1_H$ and
$h=\mathbf1_H\omega_e$, for every $H\in\mathcal H_v$ used below.
\item In the graph $A_0=A\setminus(\starv_A(v)\setminus\{e\})$, where $v$
is a leaf, a constant $b\in[0,1]$ satisfies
\begin{equation}\label{eq:abstract-leaf}
 E_{A_0}[\mathbf1_H\omega_e]=bP_{A_0}(H),\qquad H\in\mathcal H_v.
\end{equation}
\end{enumerate}
Then, whenever $P_A(H)>0$,
\begin{equation}\label{eq:restoration-result}
 b\prod_{f\ni v,\,f\ne e}\frac{\ell_f}{M_f}
 \le P_A(\omega_e=1\mid H)
 \le b\prod_{f\ni v,\,f\ne e}\frac{M_f}{\ell_f}.
\end{equation}
For uniform constants $\ell\le M$ and degree at most $\Delta$, the lower
bound is $p=b(\ell/M)^{\Delta-1}$.
\end{lemma}

\begin{proof}
Restore the deleted edges in any fixed order. The lower inequality in
\eqref{eq:abstract-restoration}, applied to $\mathbf1_H\omega_e$, gives
\[
 E_A[\mathbf1_H\omega_e]
 \ge\Bigl(\prod_f\ell_f\Bigr)bP_{A_0}(H).
\]
The upper inequality, applied to $\mathbf1_H$, bounds the denominator by
$(\prod_f M_f)P_{A_0}(H)$. The latter cavity probability is positive whenever
$P_A(H)>0$, by the same upper bound. Division proves the lower estimate;
interchanging the bounds proves the upper estimate. The history is kept as a
fixed observable throughout. One does not rerun an exploration on the modified
graphs or compare their realized histories.
\end{proof}

The same argument supplies a label bound when an isolated vertex has a known
law, provided the restoration bounds also hold while restoring the whole star,
including $e$. If deleting that star makes a vertex label $\xi_v$ have conditional
probability $a_s$ of value $s$, independently of $\mathcal H_v$, restoration
gives $P_A(\xi_v=s\mid H)\ge a_s(\ell/M)^\Delta$. This is a separate
isolated-vertex assumption; it does not follow from the leaf-edge assumption.

\begin{corollary}[Odds refinement]\label{cor:odds}
Under Lemma~\ref{lem:restoration}, suppose the restoration bounds also apply
to the retained failure observable $\mathbf1_H(1-\omega_e)$. Put
$\mathcal R=\prod_{f\ni v,\,f\ne e}(M_f/\ell_f)$. Then
\begin{equation}\label{eq:odds-refinement}
 \frac{b}{b+\mathcal R(1-b)}
 \le P_A(\omega_e=1\mid H)
 \le\frac{\mathcal R b}{1-b+\mathcal R b}.
\end{equation}
\end{corollary}

\begin{proof}
For $0<b<1$, compare the success mass directly with the failure mass.
Restoration bounds their ratio between $b/[\mathcal R(1-b)]$ and
$\mathcal R b/(1-b)$.
Converting odds to probabilities proves the assertion. At $b=0$ or $b=1$,
the null cavity event remains null by the upper restoration bound.
\end{proof}

This elementary likelihood-ratio refinement is available under the full
retained-observable hypothesis, but not from its smaller named test class
without the extra failure test. The simpler bound $b/\mathcal R$ suffices for
our asymptotic arguments. Negative association or a marginal opening probability
alone does not provide either of these conditional bounds.

For example, the CMR failure observable $\mathbf1_H(1-B_e)$ is retained.
The good-target bound of Lemma~\ref{lem:goodtarget} therefore improves to
$b_K/[b_K+r_K^{-(\Delta-1)}(1-b_K)]$. The main asymptotic proof needs only
the simpler $\pi_K$.

\subsection{Passing to a selected infinite-volume law}

\begin{proposition}[Local transfer]\label{prop:local-transfer}
Let laws on finite ambient graphs converge locally to a law on a countable,
locally finite graph. Suppose the constants in conditional opening bounds
are uniform in volume, and each bound holds against every bounded nonnegative finite
exterior-cylinder test whose coordinates and full relevant stars lie in the
ambient graph. Suppose these tests, and the tested indicators, are continuity
observables for the local convergence. Then the same bounds hold in the selected
limit, first for cylinder histories and then for the generated exterior
sigma-fields. Proposition~\ref{prop:exploration} consequently applies to that
limit whenever its history assumptions are met.
\end{proposition}

\begin{proof}
Write a conditional lower bound without division as
$E[h\omega_e]\ge pE[h]$, for nonnegative exterior-cylinder tests $h$.
Both sides converge by hypothesis. A monotone-class argument extends the limit
inequality to bounded nonnegative exterior-measurable tests, which is equivalent
to the asserted conditional bound. Upper bounds work identically. For finite
alphabet coordinates every cylinder is a continuity set. If a random environment
is to be conditioned on, suppose its coordinates are standard Borel with a
countable generating family of continuity cylinders. Include those tests before taking the
limit and then disintegrate. Countability gives a common full-measure environment
set for the countably many vertices, edges, labels and generating histories.
\end{proof}

One can also transfer comparisons of clusters explored \emph{inside each fixed
ball}: keep the full ambient law, pass its finite-coordinate inequality to the
limit, and then exhaust the graph. This yields radius and cluster-expectation
bounds. Merely taking a weak limit of a global root-cluster containment does not
justify its preservation, since witnessing paths can escape to infinity.
Neither local transfer nor finite-ball exploration requires boundary-condition
comparison or consistency of the finite-volume measures under restriction.
Figure~\ref{fig:local-transfer} records the order of these limits.

\paragraph{CMR support at zero couplings.}
Blue/red overlap compatibility is a finite-alphabet local constraint and passes
to the limit directly. For the coupling-sign constraint, weak satisfaction
$J_e\sigma_x\sigma_y\ge0$ on blue edges (and its second-replica analogue)
passes as a closed support condition. Strict satisfaction follows by excluding
blue edges at zero coupling. At continuity radii $\delta>0$, the finite-volume
activation bound passes to the limit as
\[
 \nu(B_e=1,\ |J_e|\le\delta)\le1-e^{-4\beta\delta}\le4\beta\delta.
\]
Let $\delta\downarrow0$ to obtain $\nu(B_e=1,J_e=0)=0$, even when the
disorder has an atom at zero. The red activation rule gives the analogous
$2\beta\delta$ bound. Thus every finite cycle retains both parity constraints;
countability makes this simultaneous for all cycles.

\subsection{When existence has probability one}

Positive root probability alone does not imply almost-sure global existence for
an arbitrary dependent law. The following stronger exterior hypothesis supplies
that conclusion without ergodicity.

\begin{proposition}[Exterior conditioning and labelled existence]\label{prop:as-existence}
Let a countable, locally finite graph carry edge indicators $\omega$ and vertex labels
$\xi$ in a finite or countable set. Assume that open edges have equal endpoint
labels almost surely. Set
\[
 \mathcal E_v=\sigma((\xi_x)_{x\ne v},(\omega_e)_{e\notin\starv(v)}).
\]
Let $A\subseteq V$ be deterministic. Suppose, for every oriented edge $uv$ of $G[A]$,
$P(\omega_{uv}=1\mid\mathcal E_v)\ge p>0$ almost surely.
Suppose independent site percolation of parameter $p$ on $G[A]$ has positive
root-percolation probability at some $r\in A$, and that
$P(\xi_r=s\mid\mathcal E_r)\ge\delta_s>0$.
Then the full open graph has an infinite component labelled $s$ almost surely.
The statement applies separately to each label and also under conditional laws
with an environment first held fixed.
\end{proposition}

\begin{proof}
Let $T_s$ be absence of an infinite component in the open graph induced by
$\{v:\xi_v=s\}$. This event belongs to every $\mathcal E_v$: deleting one
vertex and finitely many incident edges cannot change global existence of an
infinite component. Defining the induced graph makes this measurability valid
on the whole configuration space, not just on configurations with compatible
edge labels.

If $P(T_s)>0$, the root-label bound gives $P(T_s,\xi_r=s)>0$. Condition on
this event and explore $G[A]$ from $r$. At each new target $v\ne r$, both
$T_s$ and the root label are exterior to $v$, as are all previously queried
edges. The conditional success bound is therefore preserved. Proposition~
\ref{prop:exploration} gives positive probability of an infinite root cluster.
All its open parent edges preserve label $s$, contradicting $T_s$. Hence
$P(T_s)=0$. Countably many such probability-one conclusions can be intersected.
\end{proof}

\subsection{A separate nonpercolation criterion}

\begin{proposition}[Path bound]\label{prop:path-bound}
On a countable graph of maximum degree $\Delta\ge2$, suppose every prescribed
self-avoiding path of $n$ edges is open with probability at most $u^n$.
Then, with $B_n$ the graph-distance ball around $o$,
\begin{equation}\label{eq:general-path-count}
 P(o\leftrightarrow\partial B_n)
 \le\frac{\Delta}{\Delta-1}[(\Delta-1)u]^n.
\end{equation}
If $(\Delta-1)u<1$, there is no infinite open component almost surely and
\begin{equation}\label{eq:general-susceptibility}
 E|C_\omega(o)|\le1+\frac{\Delta u}{1-(\Delta-1)u}.
\end{equation}
Uniform finite-volume path bounds pass to every local limit.
\end{proposition}

\begin{proof}
There are at most $\Delta(\Delta-1)^{n-1}$ rooted self-avoiding paths of length
$n$. Reaching distance $n$ requires one of them, giving
\eqref{eq:general-path-count}. Sum over lengths to bound the expected cluster
size. Decay at every root and countability exclude all infinite components.
A fixed path event is a cylinder, so the last assertion follows by local
convergence. One sufficient way to prove the path hypothesis is to bound each
next edge's conditional opening probability by $u$ given that the preceding path
edges are open. This history is off the next target's star.
\end{proof}

The lower comparison and the upper path estimate are separate inputs. Neither
one presumes stochastic monotonicity of the original edge law.

\section{A short proof from the Machta--Newman--Stein comparison}\label{sec:mns}

For laws with $\E e^{aW}<\infty$ for some $a>0$, the comparison lemmas
of Machta, Newman and Stein give an alternative verification and explicit
parameters for the percolating interval used in Section~\ref{sec:response}. This includes Gaussian and bounded disorder. The method still uses
the classical exploration from Section~\ref{sec:general}; its input is now
an independent mixed site/bond field.
Throughout this section $\Delta\ge2$ bounds the graph degree and
$F(t)=\E e^{tJ}=\E\cosh(tW)$ is finite at $t=4\beta$.

\subsection{The graph-local posterior comparison}

For the symmetric coupling law $\lambda$, let
$\lambda_h(dx)=e^{hx}\lambda(dx)/F(h)$. The following is the two-replica form
of \cite[Lemma 4.1]{MNS2008}. Its proof uses finite edge sums, not complete-graph
geometry, so it applies to the finite graphs considered here.

\begin{lemma}[Machta--Newman--Stein comparison]\label{lem:mns}
Fix the replica spins and arbitrary signs $\eta_e\in\{-1,1\}$. Write
$K_e=\eta_e J_e$ and
$\phi_e=\eta_e(\sigma_x\sigma_y+\tau_x\tau_y)$ for $e=xy$.
The posterior law of $K$, conditional on the replicas, satisfies
\begin{equation}\label{eq:mns}
 \bigotimes_e\lambda_{\beta(\phi_e-2)}
 \preceq\mathcal L(K\mid\sigma,\tau)
 \preceq\bigotimes_e\lambda_{\beta(\phi_e+2)}.
\end{equation}
Here $\preceq$ denotes stochastic domination in the coordinatewise order.
\end{lemma}

\begin{proof}
The posterior density is proportional to
\[
 Z(\eta K)^{-2}\exp\!\left(\beta\sum_e\phi_eK_e\right)
 \prod_e\lambda(dK_e).
\]
The functions $Z(\eta K)e^{\beta\sum_eK_e}$ and
$Z(\eta K)e^{-\beta\sum_eK_e}$ are respectively increasing and decreasing
in every coordinate: after expansion of $Z$, their exponential coefficients
are $\beta(\eta_e s_xs_y+1)\ge0$ and
$\beta(\eta_e s_xs_y-1)\le0$.
Relative to the upper product measure in \eqref{eq:mns}, the posterior density
is therefore decreasing; relative to the lower product measure it is increasing.
Positive association of product measures gives both comparisons, with truncation
for unbounded test functions. All required tilt normalizations are finite because
$|\phi_e\pm2|\le4$.
\end{proof}

Set
\begin{equation}\label{eq:mnsparameters}
 L_\beta=\tfrac12\E(1-e^{-4\beta W}),\qquad
 U_\beta=\frac{\E(e^{4\beta W}-1)}{2F(4\beta)},\qquad
 R_\beta=\frac{F(4\beta)}{F(2\beta)},\qquad
 a_\Delta=\frac1{1+R_\beta^\Delta}.
\end{equation}
At $\beta=0$, interpret these as $L_\beta=U_\beta=0$, $R_\beta=1$,
$a_\Delta=1/2$.

\begin{lemma}[Conditional blue and overlap bounds]\label{lem:overlap}
On any finite graph of degree at most $\Delta$:
\begin{enumerate}
\item Conditional on all replica spins, the blue field lies stochastically
between independent bond fields with respective probabilities $L_\beta$ and
$U_\beta$ on edges whose endpoint overlaps agree, and probability zero on
the remaining edges.
\item Every conditional overlap-site probability obeys
\begin{equation}\label{eq:overlapconditional}
 a_\Delta\le\nu_{G,\beta}(q_v=s\mid(q_x)_{x\ne v})
 \le1-a_\Delta,\qquad s=\pm1.
\end{equation}
\end{enumerate}
\end{lemma}

\begin{proof}
The first assertion is \cite[Lemma 4.2]{MNS2008}, with the same graph-local
proof. Indeed, choose $\eta_e=\sigma_x\sigma_y$. On an edge with equal endpoint
overlaps, $\phi_e=2$ and the blue probability is the increasing function
$g_\beta(K_e)=\mathbf1_{\{K_e>0\}}(1-e^{-4\beta K_e})$.
Lemma~\ref{lem:mns} and independent activation uniforms give the product bounds.
Integrating $g_\beta$ under $\lambda_0$ and $\lambda_{4\beta}$ gives
$L_\beta$ and $U_\beta$ respectively. On unequal-overlap edges blue is impossible.

For the second assertion, write $w(\sigma,\tau)$ for the annealed two-replica
spin mass and $\tau^v$ for $\tau$ with its spin at $v$ reversed. Exactly,
\begin{equation}\label{eq:flip}
 \frac{w(\sigma,\tau^v)}{w(\sigma,\tau)}
 =\E\left[\left.
   \exp\!\left(-2\beta\sum_{e=xy\ni v}\tau_x\tau_yJ_e\right)
   \right|\sigma,\tau\right].
\end{equation}
Apply Lemma~\ref{lem:mns} with $\eta_e=\tau_x\tau_y$ to this decreasing
observable. For an incident edge with equal endpoint overlaps, the lower and
upper factors are $F(2\beta)/F(4\beta)$ and $F(2\beta)$; for unequal overlaps
they are $1/F(2\beta)$ and $F(4\beta)/F(2\beta)$.
Log-convexity of $F$ gives $F(4\beta)\ge F(2\beta)^2$, so every factor belongs
to $[R_\beta^{-1},R_\beta]$. The ratio in \eqref{eq:flip} is consequently
between $R_\beta^{-\deg(v)}$ and $R_\beta^{\deg(v)}$.
Given $\sigma$ and all overlaps except $q_v$, the two choices of $q_v$ correspond
to this spin flip. Their probabilities lie in $[a_\Delta,1-a_\Delta]$.
Averaging over $\sigma$ gives \eqref{eq:overlapconditional}.
\end{proof}

The second part is the additional inference from the MNS comparison that supplies
the lattice geometry. At $\beta$ of order $1/\Delta$, its full conditional bounds
are close to $1/2$, even after arbitrary other overlap signs have been specified.

\subsection{From overlap sites to blue clusters}

Recall from Remark~\ref{rem:short-route} that $\thsite(p)$ is the probability
that the origin belongs to an infinite occupied cluster in iid site percolation
of parameter $p$. This includes the origin's occupancy; forcing the origin
occupied gives probability $\thsite(p)/p$ when $p>0$.

Recall from Remark~\ref{rem:exploration} that a mixed site/bond root cluster
with parameters $a,b$ contains the root-forced iid site cluster of parameter
$ab$. On $\Z^d$, $ab>\pc$ therefore implies percolation of the mixed model,
and its global existence event has probability one by translation ergodicity.

\begin{proposition}[Comparison criteria]\label{prop:criteria}
Set $\Delta=2d$ and suppose $F(4\beta)<\infty$. Define
\begin{equation}\label{eq:effective}
 p_\Delta=a_\Delta L_\beta,\qquad
 u_\Delta=(1-a_\Delta)U_\beta.
\end{equation}
Every $\nu\in\mathcal L_{d,\beta}$ has the following properties.
\begin{enumerate}
\item For either $s=\pm1$, conditional on $q_o=s$, the blue root cluster
dominates the root cluster of iid site percolation of parameter $p_\Delta$,
with its root forced occupied. Thus, if $p_\Delta>\pc(\Z^d)$,
\begin{equation}\label{eq:rootlower}
 \nu(o\blue\infty,q_o=s)
 \ge\frac{\thsite(p_\Delta)}{2p_\Delta}>0.
\end{equation}
Each overlap sign contains an infinite blue component almost surely.
\item Every prescribed self-avoiding $n$-edge path is blue with probability
at most $u_\Delta^n$. Hence
\begin{equation}\label{eq:armupper}
 \nu(o\blue\partial B_n)
 \le\frac{\Delta}{\Delta-1}\bigl[(\Delta-1)u_\Delta\bigr]^n.
\end{equation}
When $(\Delta-1)u_\Delta<1$, there is no infinite blue component and
\begin{equation}\label{eq:chiupper}
 \chi_{\mathrm b}(\nu):=\E_\nu|C_{\mathrm b}(o)|
 \le1+\frac{\Delta u_\Delta}{1-(\Delta-1)u_\Delta}.
\end{equation}
\end{enumerate}
\end{proposition}

\begin{proof}
First work on a finite torus. By \eqref{eq:overlapconditional}, the indicator
field of $q_x=s$ dominates iid sites $X$ of parameter $a_\Delta$, the
independent law at the worst conditional probability \cite[Eq.~(8)]{ABL1987}. Given the
replicas, couple an iid bond field $Y$ of parameter $L_\beta$ below blue on
equal-overlap edges using Lemma~\ref{lem:overlap}. These couplings can be combined
so that $X$ and $Y$ are independent: sample the replicas and their site coupling,
then sample the blue/bond coupling from a kernel depending only on the replicas.
The conditional marginal of $Y$ is the same product law for every replica
configuration, which proves the asserted independence. Consequently
\[
 X_xX_yY_{xy}\le B_{xy},\qquad X_x=1\Longrightarrow q_x=s.
\]
The original disorder is included by sampling it from its physical conditional
law given the replicas and blue indicators. This preserves the full finite-volume
joint law, including its quenched normalization.
The same construction conditional on $q_o=s$ forces $X_o=1$ and keeps all other
site bounds, because partial conditioning preserves \eqref{eq:overlapconditional}.
The mixed-model exploration gives the rooted comparison.

Fix a subsequence converging to the selected physical law $\nu$. Compactness of
the auxiliary binary spaces gives a further joint subsequence with the same
physical marginal. The independent product law and the displayed local
containment constraints pass to its limit. For the rooted construction,
$\{q_o=s\}$ is a clopen cylinder of probability $1/2$, so the conditioned
physical laws converge to $\nu(\cdot\mid q_o=s)$.
The unconditioned mixed model
percolates almost surely when $p_\Delta>\pc$, proving global existence in sector
$s$. Apply this separately to each sign; the two probability-one conclusions
hold simultaneously without requiring the two auxiliary models to be independent.
The root-forced comparison gives \eqref{eq:rootlower}, since the global flip of
one replica makes $q_o$ fair and preserves blue bonds.

For the upper bound, condition on the replicas. An $n$-edge path can be blue only
when all its overlap signs agree, and its conditional blue probability is then
at most $U_\beta^n$. Reveal the overlaps successively along a self-avoiding path.
Each subsequent sign matches the first with conditional probability at most
$1-a_\Delta$. Thus the unconditional path probability is at most $u_\Delta^n$.
Proposition~\ref{prop:path-bound} gives \eqref{eq:armupper} and
\eqref{eq:chiupper}, including passage to every local limit.
\end{proof}

\subsection{The high-dimensional scale}

For $s_2=\E W^2$ and $\beta\downarrow0$, exponential integrability gives
\begin{equation}\label{eq:expansions}
 L_\beta=2m\beta+O(\beta^2),\quad
 U_\beta=2m\beta+O(\beta^2),\quad
 \log R_\beta=6s_2\beta^2+O(\beta^4).
\end{equation}
At $\beta=c/(\Delta m)$, it follows that
\begin{equation}\label{eq:parameterlimits}
 \Delta p_\Delta\to c,\qquad (\Delta-1)u_\Delta\to c.
\end{equation}
Kesten's threshold estimate \eqref{eq:kesten} gives $\Delta\pc\to1$.
Proposition~\ref{prop:criteria} proves the two parts of
Theorem~\ref{thm:headline} under exponential integrability.
For the uniform subcritical interval, $u_\Delta$ is nondecreasing in $\beta$.
Indeed, $R_\beta$ increases by convexity of $\log F$, while $U_\beta$ is the
expectation of the increasing function $g_\beta$ under the increasing tilt
$\lambda_{4\beta}$. Choosing $\rho=(1+c)/2$ gives the stated bound for large $d$.

\subsection{A percolating temperature interval}\label{sec:consequences}

\begin{corollary}[A percolating interval]\label{cor:window}
Suppose $\E e^{aW}<\infty$ for some $a>0$, and write $s_2=\E W^2$.
For each fixed $\varepsilon>0$ and all sufficiently large $d$, both overlap signs
contain infinite blue components almost surely in every $\nu\in\mathcal L_{d,\beta}$
throughout the interval
\begin{equation}\label{eq:window}
 \frac{1+\varepsilon}{2dm}
 \le\beta\le\sqrt{\frac{\log(2d)}{24s_2d}}.
\end{equation}
\end{corollary}

\begin{proof}
Put $\Delta=2d$ and $b_\Delta=\sqrt{\log\Delta/(12s_2\Delta)}$.
Uniformly for $0<\beta\le b_\Delta$, \eqref{eq:expansions} and
$\Delta\beta^4=o(1)$ give
\[
 \Delta p_\Delta
 =\frac{\Delta L_\beta}{1+R_\beta^\Delta}
 \ge(1-o(1))m\Delta\beta e^{-6s_2\Delta\beta^2}.
\]
The function $\beta e^{-6s_2\Delta\beta^2}$ increases and then decreases, so its
minimum on the stated interval occurs at an endpoint. At the lower endpoint
the last bound tends to $1+\varepsilon$; at the upper it is asymptotic to
$m\sqrt{\log\Delta/(12s_2)}$, which diverges. Thus $p_\Delta>\pc(\Z^d)$
throughout the interval for large $d$. Proposition~\ref{prop:criteria} applies
at each temperature, without a monotonicity assertion about the physical blue law.
\end{proof}

\section{Response identities and limitations of alternative routes}\label{app:limitations}

This appendix collects additional physical interpretations and quantitative
limitations of the comparisons. They are not needed for the main percolation
or conditional-ordering results. The MNS parameters used below are defined in
\eqref{eq:mnsparameters} and \eqref{eq:effective}.

\subsection{Signed grey mass and magnetic response}

There is an exact response formula in a finite zero-field volume with free or
periodic boundary conditions. Given a full CMR graph, choose
red-parity signs $\xi_G(x)$ on each grey component $G$: they are constant across
blue edges and opposite across red edges. Put $S_G=\sum_{x\in G}\xi_G(x)$.
Conditional on the graph and $J$, the overall overlap orientations of its grey
components are independent fair signs: every assignment has exactly
$2^{N_{\mathrm b}}$ compatible spin pairs, where $N_{\mathrm b}$ counts blue
components, including isolated vertices. Thus, for a dimensionless replica field
$h\sum_x\sigma_x\tau_x$ on a finite zero-field volume of $N$ vertices,
\begin{equation}\label{eq:grey-response}
 \frac{Z_{2,J}(h)}{Z_J^2}
 =\E_{{\rm CMR},J}\prod_G\cosh(hS_G),\qquad
 \chi_{{\rm SG},N}=\frac1N\E\sum_G S_G^2.
\end{equation}
This is the signed connectivity identity of
\cite[Section II, Eq.~(16)]{MNS2007} summed over the volume. A grey charge is
an alternating sum of its blue-component masses. Squaring after that sum keeps
the cross terms that an unsigned blue susceptibility discards.
Higher response does not automatically remove the cancellation: the fourth
derivative of $N^{-1}\E\log[Z_{2,J}(h)/Z_J^2]$ at zero equals
\[
 \frac1N\E_J\left[
 3\operatorname{Var}_{{\rm CMR},J}\!\left(\sum_G S_G^2\right)
 -2\E_{{\rm CMR},J}\sum_G S_G^4\right].
\]
For the ordinary uniform magnetic field, let
$\psi_N(h)=N^{-1}\E\log Z_J(h)$. Fair-sign gauge invariance gives
$\psi_N''(0)=1$ and $\psi_N^{(4)}(0)=4-6\chi_{{\rm SG},N}$:
only spin tuples with every vertex repeated evenly survive disorder averaging.
Consequently the susceptibility bounded in Corollary~\ref{cor:separation} is
also the quantity controlling the averaged cubic magnetic response.

\subsection{Suppressing a sector and changing representation}

Suppressing the second infinite blue cluster is a separate, stronger geometric
question in a specified state. A unique cluster of maximal density, as in the
SK percolation picture, can coexist with a second positive-density cluster
\cite{MNS2008}. To exclude a chosen overlap sector on $\Z^d$, one sufficient
input would be a uniform bound
\[
 \Pp(q_x=-1\text{ for every }x\in\gamma)
 \le C_\beta\,b_\beta^{|\gamma|},\qquad (2d-1)b_\beta<1,
\]
for all rooted self-avoiding paths $\gamma$ in an appropriate state selected at
the same temperature. Path counting would then exclude negative-sector blue
percolation. No such low-temperature estimate is provided here.

Other graphical representations do not immediately bypass these missing inputs.
At $\beta=0$, the overlap signs are iid and fair, so their geometric clusters
already percolate in high dimension by \eqref{eq:kesten}, although the spatial
overlap is zero. Thus geometric overlap percolation alone cannot establish order.
The two-replica FK intersection (TRFK) studied in \cite{MNS2008} has a different
scale. Its MNS lower within-sector parameter is
$L_\beta^{\mathrm{TRFK}}=\tfrac12\E(1-e^{-2\beta W})^2$.
For exponential-moment disorder, fix $x>0$ and take
$\beta_\Delta=\sqrt{x/(\Delta s_2)}$. Our overlap comparison gives
\[
 \Delta a_\Delta L_\beta^{\mathrm{TRFK}}
 \longrightarrow\frac{2x}{1+e^{6x}}
 \le\frac1{3e}<1.
\]
The same mixed-site/bond criterion therefore cannot prove TRFK percolation at
any fixed value of this scaling parameter. This is a limitation of the present
bound, not a nonpercolation theorem. A useful representation change would need
an additional comparison or identity that controls overlap correlations at this
scale.

Finally, reducing the required dimension is a quantitative problem.
Theorem~\ref{thm:explicit} gives concrete dimensions for unit couplings. For a
general law the threshold in Theorem~\ref{thm:headline} remains unquantified.
Under exponential integrability, one sufficient condition is
\[
 \frac{\E(1-e^{-4\beta W})}
 {2\{1+[F(4\beta)/F(2\beta)]^{2d}\}}
 >\pc(\Z^d).
\]
For unit couplings, the proof of Theorem~\ref{thm:explicit} instead reveals the
overlap field, uses a bond floor that failed queries cannot lower, and compares
the overlap field with a weak Ising field (Section~\ref{sec:overlap-route}).
Restricting the cavity laws (Section~\ref{sec:ov-noisy}) and replacing oriented
paths by macrostep paths (Section~\ref{sec:macrostep}) take this route down to
dimension $7$. The fresh-star route of Appendix~\ref{app:fresh-star} reaches
dimension $16$, and Remark~\ref{rem:fs-limits} discusses dimension $15$.
Remark~\ref{rem:ms-limits} explains why the far-region bound of the macrostep
engine, with three lateral coordinates, is unavailable in dimension $6$.
Further gains may come from
other path families, from local inputs that are not uniform over cavity laws,
or from block explorations whose conditional crossing probabilities can be
bounded after every permitted history. Observed successes, unlike failures, can lower the bond floor
(Remark~\ref{rem:ov-successes}); more informative target histories would
therefore require stronger conditional estimates than these failure bounds.

There is a concrete limitation for unit couplings. Put $t=\tanh\beta$ and
$\Delta=2d$. Since $R_\beta\ge1+6t^2$ and $L_\beta=2t/(1+t)^2$,
\[
 p_\Delta\le\frac{t}{1+3\Delta t^2}
 \le\frac1{2\sqrt{3\Delta}}.
\]
At $d=6$ this is at most $1/12<1/11\le\pc(\Z^6)$. Moreover the independent
mixed model itself has path probability $a_\Delta^{n+1}L_\beta^n$ for a
prescribed $n$-edge path, so Proposition~\ref{prop:path-bound} makes it
subcritical there. Thus retaining repeated attempts inside that auxiliary model
cannot by itself reach dimension six. A stronger model-specific input is needed.

\section*{Computational reproducibility}
All certificate programs used in this paper are in \cite{CMRCode}, at revision
\href{https://github.com/PeaBrane/cmr-blue-percolation/tree/8da06c6d4f6be0d464eae4b212e0fbd2a27cd684}{\texttt{8da06c6d4f6b}}.
The default and supplementary runs need only the Python standard library.
Running \texttt{python run\_all.py} executes the following certificates.
\begin{itemize}
\item \href{https://github.com/PeaBrane/cmr-blue-percolation/blob/8da06c6d4f6be0d464eae4b212e0fbd2a27cd684/overlap_revealed/certify.py}{\texttt{overlap\_revealed/certify.py}}
decides, for each row of Table~\ref{tab:overlap-constants}, the local and
global conditions of Section~\ref{sec:overlap-route}: the bond floor, the
Holley line in all $2d+1$ environments, the mean-field test for the
plus-density, and the second-moment criterion with its constant, including the
refinement of Lemma~\ref{lem:ov-shared} at $d=10$. Its inputs are frozen exact
rationals in \path{overlap_revealed/params.json}.
\item \path{overlap_revealed/noisy_cavity/certify_noisy_cavity.py} decides, for
each row of Table~\ref{tab:ov-noisy}, the sharpened local inputs (N1)--(N2) of
Section~\ref{sec:ov-constants}, the comparison of the certified floor with the
aligned frozen values of Remark~\ref{rem:ov-nesting}, and (C3)--(C6) with the
score $\mathcal S_c$ or, at $t=7/50$, the score $\mathcal S'_c$ of
Lemma~\ref{lem:ov-covariance}. Its inputs are in
\path{overlap_revealed/noisy_cavity/params.json} and
\path{overlap_revealed/noisy_cavity/data/}.
\item \path{macrostep/certify_macrostep.py} decides, for each row of
Table~\ref{tab:ms-local}, conditions (V0)--(V2) of
Section~\ref{sec:ms-certificates}, and, for both certificates of every near set
of Table~\ref{tab:ms-certificates}, condition (V3), that is, conditions (i) and
(ii) of Theorem~\ref{thm:ms-criterion}, with $\theta_*$ computed exactly. It
also prints the aligned frozen values quoted in
Section~\ref{sec:ms-certificates}. It reads the frozen inputs in
\path{macrostep/params.json} and the stored engine outputs
$\bar{\mathsf M}$, $\bar\eta$, near Green rows and $\eta_F$ in
\path{macrostep/certificates/}. It also re-decides further stored certificates
that no theorem here uses.
\item \path{single_floor/verify_single_floor.py} evaluates the rows $h=0$ at the
certified points of Appendix~\ref{app:fresh-star}, the bounds of
Lemma~\ref{lem:collision-green}, the block covers of
Table~\ref{tab:fs-blocks}, the signed refinements of
Remark~\ref{rem:fs-signed} and the closed-form Corollary~\ref{cor:fs-closed-form}.
\item \path{verify_explicit_balance.py} decides the comparisons of
Corollary~\ref{cor:explicit-balance} in exact rational arithmetic and every
entry of Table~\ref{tab:explicit-chi}; \path{verify_explicit_balance.mjs},
run separately with Node.js, decides the comparisons again in integer
arithmetic.
\item \path{verify_physical_separation.py} checks the dimension-$180$ example of
Appendix~\ref{app:separation-certificate}, and \path{verify_signed_star.py}
recomputes the dimension-$22$ certificate of \cite{PeiV1}, whose row $h=0$ is
the floor $0.0586$ of Remark~\ref{rem:fs-signed}.
\end{itemize}
On a laptop the programs for Tables~\ref{tab:ov-noisy} and~\ref{tab:ms-local}
take about $14$ and $11$ seconds, and the latter uses about $105$\,MB of
memory.

The default run takes the stored macrostep engine outputs as given; that they
are upper bounds for the quantities of Theorem~\ref{thm:ms-criterion} rests on
the engine programs and their error analysis
(Section~\ref{sec:ms-certificates}). The entry point
\path{macrostep/reproduce_engine.py} recomputes them with the a priori error
program from \path{macrostep/params.json}, compares them with the stored
certificates, and re-decides the fresh certificates with two exact checkers;
with the option \texttt{-{}-l2} it also reruns the directed-rounding program and
compares the two programs entry by entry. This full reproduction needs NumPy
and Numba. With six worker processes on a $24$-core workstation it took about
$21$ minutes. It reproduced the stored engine outputs $\bar\eta$,
$\bar{\mathsf M}$, near Green rows and $\eta_F$ of every near set of
Table~\ref{tab:ms-certificates} bit for bit. At these near sets the values of
$\bar\eta$, the near Green rows and the entries of $\bar{\mathsf M}$ above
$10^{-14}$ of the two programs agreed within a relative $6\times10^{-8}$.

The option \texttt{-{}-supplementary} adds finite-model checks and an
independently written implementation of the local inputs of the rows of
Table~\ref{tab:overlap-constants}, which reproduces all $639$ local rationals of
\path{certify.py} exactly. It also adds exact checks of
Lemmas~\ref{lem:fs-balanced-tilt}, \ref{lem:fs-floor-identity}
and~\ref{lem:fs-monotone}, of the side results in
Remark~\ref{rem:fs-limits}, and of the nesting in Remark~\ref{rem:ov-nesting}.
The option \texttt{-{}-optional-deps} adds an independent evaluation of the
global criterion in interval arithmetic, which requires \texttt{mpmath}, and,
with \texttt{NumPy}, implementations of Lemmas~\ref{lem:ov-pair}
and~\ref{lem:ov-sym-holley} that enumerate all sign vectors, and exact
enumeration checks of Lemma~\ref{lem:ov-noisy} and of the pairing identity in
the proof of Lemma~\ref{lem:ov-pair}.

Every claimed inequality is decided in exact integer or rational arithmetic.
Transcendental constants are replaced by rational enclosures with explicit
remainders. The macrostep engine outputs are rigorous upper bounds computed in
binary64 arithmetic, and the checker reads each stored number as the dyadic
rational it represents. The frozen inputs, such as tangent points, line
weights, the rationals $\psi_U(x)$, the certificate vectors $w$ and the
rational entries of Tables~\ref{tab:overlap-constants}, \ref{tab:ov-noisy}
and~\ref{tab:ms-local}, were chosen by floating-point searches that use
\texttt{NumPy} and \texttt{SciPy}. Some of these searches are included in
\path{overlap_revealed/search/} and \path{macrostep/search/}; none is needed
for verification, and any admissible choice of inputs gives a valid
certificate.

\end{document}